\documentclass[11pt]{article}
\usepackage[margin=.7in,left=.7in]{geometry}
\usepackage{mathrsfs}
\usepackage{amsmath}
\usepackage{amsfonts}
\usepackage{amssymb}
\usepackage{amsthm}
\usepackage{xfrac}
\usepackage{multirow}
\usepackage{rotating}
\usepackage{graphicx}
\usepackage{color}
\usepackage{xcolor}
\usepackage{authblk}
\usepackage{hyperref}
\usepackage[all]{xy}

\usepackage{tikz}

\usepackage{floatflt}  
\usepackage{wrapfig} 
\usepackage{array}     
\newcolumntype{L}[1]{>{\raggedright\let\newline\\\arraybackslash\hspace{0pt}}m{#1}}
\newcolumntype{C}[1]{>{\centering\let\newline\\\arraybackslash\hspace{0pt}}m{#1}}
\newcolumntype{R}[1]{>{\raggedleft\let\newline\\\arraybackslash\hspace{0pt}}m{#1}}

\usepackage{adjustbox}

\usepackage{diagbox}  

\usepackage{stmaryrd}    

\usepackage{enumerate} 

\usepackage{helvet}   

\usepackage{tabularx}   

\makeatletter
\newif\if@sup
\newtoks\@sups
\def\append@sup#1{\edef\act{\noexpand\@sups={\the\@sups #1}}\act}%
\def\reset@sup{\@supfalse\@sups={}}%
\def\mk@scripts#1#2{\if #2/ \if@sup ^{\the\@sups}\fi \else%
  \ifx #1_ \if@sup ^{\the\@sups}\reset@sup \fi {}_{#2}%
  \else \append@sup#2 \@suptrue \fi%
  \expandafter\mk@scripts\fi}
\def\tensor#1#2{\reset@sup#1\mk@scripts#2_/}
\def\multiscripts#1#2#3{\reset@sup{}\mk@scripts#1_/#2%
  \reset@sup\mk@scripts#3_/}
\makeatother

\makeatletter
\newbox\slashbox \setbox\slashbox=\hbox{$/$}
\def\itex@pslash#1{\setbox\@tempboxa=\hbox{$#1$}
  \@tempdima=0.5\wd\slashbox \advance\@tempdima 0.5\wd\@tempboxa
  \copy\slashbox \kern-\@tempdima \box\@tempboxa}
\def\slash{\protect\itex@pslash}
\makeatother

\def\clap#1{\hbox to 0pt{\hss#1\hss}}

\def\mathrlap{\mathpalette\mathrlapinternal}

\def\mathrlapinternal#1#2{\rlap{$\mathsurround=0pt#1{#2}$}}

\let\oldroot\root
\def\root#1#2{\oldroot #1 \of{#2}}
\renewcommand{\sqrt}[2][]{\oldroot #1 \of{#2}}

\DeclareSymbolFont{symbolsC}{U}{txsyc}{m}{n}
\SetSymbolFont{symbolsC}{bold}{U}{txsyc}{bx}{n}
\DeclareFontSubstitution{U}{txsyc}{m}{n}

\DeclareSymbolFont{stmry}{U}{stmry}{m}{n}
\SetSymbolFont{stmry}{bold}{U}{stmry}{b}{n}

\DeclareFontFamily{OMX}{MnSymbolE}{}
\DeclareSymbolFont{mnomx}{OMX}{MnSymbolE}{m}{n}
\SetSymbolFont{mnomx}{bold}{OMX}{MnSymbolE}{b}{n}
\DeclareFontShape{OMX}{MnSymbolE}{m}{n}{
    <-6>  MnSymbolE5
   <6-7>  MnSymbolE6
   <7-8>  MnSymbolE7
   <8-9>  MnSymbolE8
   <9-10> MnSymbolE9
  <10-12> MnSymbolE10
  <12->   MnSymbolE12}{}

\makeatletter
\def\Decl@Mn@Delim#1#2#3#4{%
  \if\relax\noexpand#1%
    \let#1\undefined
  \fi
  \DeclareMathDelimiter{#1}{#2}{#3}{#4}{#3}{#4}}
\def\Decl@Mn@Open#1#2#3{\Decl@Mn@Delim{#1}{\mathopen}{#2}{#3}}
\def\Decl@Mn@Close#1#2#3{\Decl@Mn@Delim{#1}{\mathclose}{#2}{#3}}
\Decl@Mn@Open{\llangle}{mnomx}{'164}
\Decl@Mn@Close{\rrangle}{mnomx}{'171}
\Decl@Mn@Open{\lmoustache}{mnomx}{'245}
\Decl@Mn@Close{\rmoustache}{mnomx}{'244}
\makeatother

\makeatletter
\DeclareRobustCommand\widecheck[1]{{\mathpalette\@widecheck{#1}}}
\def\@widecheck#1#2{%
    \setbox\z@\hbox{\m@th$#1#2$}%
    \setbox\tw@\hbox{\m@th$#1%
       \widehat{%
          \vrule\@width\z@\@height\ht\z@
          \vrule\@height\z@\@width\wd\z@}$}%
    \dp\tw@-\ht\z@
    \@tempdima\ht\z@ \advance\@tempdima2\ht\tw@ \divide\@tempdima\thr@@
    \setbox\tw@\hbox{%
       \raise\@tempdima\hbox{\scalebox{1}[-1]{\lower\@tempdima\box
\tw@}}}%
    {\ooalign{\box\tw@ \cr \box\z@}}}
\makeatother

\makeatletter
\def\udots{\mathinner{\mkern2mu\raise\p@\hbox{.}
\mkern2mu\raise4\p@\hbox{.}\mkern1mu
\raise7\p@\vbox{\kern7\p@\hbox{.}}\mkern1mu}}
\makeatother

\newcommand{\maps}{\colon} 
\newcommand{\C}{\mathbb{C}} 
\renewcommand{\H}{{\mathbb H}}  
\renewcommand{\O}{{\mathbb O}}  
\newcommand{\K}{{\mathbb K}}  

\renewcommand{\Im}{\mathrm{Im}} 
\newcommand{\tr}{\mathrm{tr}} 

\newcommand{\U}{{\rm U}} 
\newcommand{\SU}{{\rm SU}} 
\newcommand{\SO}{{\rm SO}} 
\newcommand{\Pin}{{\rm Pin}} 
\newcommand{\Spin}{{\rm Spin}} 

\newcommand{\h}{\mathfrak{h}} 
\newcommand{\so}{\mathfrak{so}} 

\newcommand{\Spaces}{{\rm Spaces}} 
\newcommand{\Orb}{{\rm Orb}} 

\newcommand{\Cl}{{\rm C}\ell} 
\newcommand{\Ho}{{\rm Ho}} 

\newcommand{\id}{{\rm id}} 
\newcommand{\A}{\mathcal{A}} 
\newcommand{\B}{\mathcal{B}} 

\newcommand{\gt}{>}
\newcommand{\lt}{<}

\newcommand{\R}{\ensuremath{\mathbb R}}
\newcommand{\Z}{\ensuremath{\mathbb Z}}

\renewcommand{\(}{\begin{equation*}}
\renewcommand{\)}{\end{equation*}}
\newcommand{\bea}{\begin{eqnarray*}}
\newcommand{\eea}{\end{eqnarray*}}

\newcommand{\dslash}{\hspace{-1mm}\sslash \hspace{-1mm}}

\theoremstyle{italics}
\newtheorem{theorem}{Theorem}[section]
\newtheorem{lemma}[theorem]{Lemma}
\newtheorem{prop}[theorem]{Proposition}

\theoremstyle{definition}
\newtheorem{defn}[theorem]{Definition}

\newtheorem{example}[theorem]{Example}

\newtheorem{remark}[theorem]{Remark}
\newtheorem{note[theorem]}{Note}

\usepackage{amsfonts}

\begin{document}

\title{Real ADE-equivariant (co)homotopy and Super M-branes}

\author{John Huerta\thanks{CAMGSD, Instituto Superior T\'ecnico, Av.\ Ravisco Pais, 1049-001 Lisboa, Portugal },
  Hisham Sati\thanks{Division of Science and Mathematics, New York University, Abu Dhabi, UAE},
  Urs Schreiber\thanks{Division of Science and Mathematics, New York University, Abu Dhabi, UAE, on leave from Czech Acad. of Science, Prague}}

\maketitle

\begin{abstract}
A key open problem in M-theory is the identification of the degrees of freedom that are expected to be hidden
at ADE-singularities in spacetime. Comparison with the classification of D-branes by K-theory suggests that
the answer must come from the right choice of generalized cohomology theory for M-branes.
Here we show that \emph{real equivariant Cohomotopy on superspaces} is a consistent such choice,
at least rationally. After explaining this new approach, we demonstrate how to use  Elmendorf's Theorem
in equivariant homotopy theory to reveal ADE-singularities as part of the data of equivariant $S^4$-valued
super-cocycles on 11d super-spacetime. We classify these super-cocycles and find a detailed
\emph{black brane scan} that enhances the entries of the old brane scan to cascades of fundamental brane
super-cocycles on strata of intersecting black M-brane species.
We find that
on each singular stratum the black brane's 
instanton contribution, 
namely its super Nambu-Goto/Green-Schwarz action,
appears
as the homotopy datum associated to the morphisms
in the \emph{orbit category}.
\end{abstract}

\newpage

\tableofcontents

\newpage

\section{Introduction}
\label{Survey}

A {\bf homotopy theory} (e.g. \cite{Lurie09}, see Sec. \ref{HomotopyTheory}) is a mathematical theory in which the concept of
strict equality is generalized to that of \emph{homotopy}. A special case is higher gauge theory, where equality of
higher gauge field configurations  is generalized to that of \emph{higher gauge equivalence}. Homotopy theory is extremely rich, involving a zoo of higher-dimensional structures and exhibiting a web of interesting and often unexpected \emph{equivalences}, which say that very different-looking homotopy theories are, in fact, equivalent.
One such equivalence is \emph{Elmendorf's Theorem} (Proposition \ref{ElmendorfTheorem} below). This says, roughly, that
homotopy theory for \emph{equivariant} homotopies is equivalent to another homotopy theory where no
equivariance on homotopies exists anymore, but where instead extra structure appears \emph{on singularities},
namely on the fixed point strata of the original group action.

\medskip
\noindent {\bf String/M-theory} (e.g. \cite{Duff99B, BeckerBeckerSchwarz06}, see Sec. \ref{Physics}) is also extremely rich, involving a zoo of higher-dimensional objects
(branes) and expected to exhibit a web of interesting and often unexpected \emph{dualities}, which say that very
different-looking string theories are, in fact, equivalent. The most striking such duality is the one between all
superstring theories on the one hand, and something with the working title \emph{M-theory} on the other.
Under this duality, ``fundamental'' or ``perturbative'' strings and branes, whose sigma-model description is
equivariant with respect to certain finite group actions, are supposed to be related to ``black'' or ``non-perturbative''
branes located at the singular fixed points of this group action. This duality is crucial for M-theory to be viable at all,
since realistic gauge force fields can appear \emph{only} at these singularities (reviewed in Sec. \ref{MBraneInterpretation}). But with the mathematics of M-theory still elusive, the identification of the extra
degrees of freedom of M-theory, that ought to be ``hidden'' at these singularities, has remained a key open problem.

\medskip
We highlight that {\bf string theory and homotopy theory are closely related}: every homotopy theory
induces a flavor of \emph{generalized cohomology theories} (see \ref{Coh}), and a fundamental insight of string theory is that the true
nature of the F1/D$p$-branes (and the higher gauge fields that they couple to) is as cocycles in
the generalized cohomology theory \emph{twisted K-theory} \cite{Witten98, FreedWitten99, MooreWitten00, EvslinSati06, GS19} (however, see 
\cite{ADerivationofK, Evslin06, KS2, S4} for limitations) , or rather \emph{real} twisted K-theory
on \emph{real} orbifolds (``orientifolds'') \cite[Sec. 5.2]{Witten98}, \cite{Gukov99, Hori99, DFM09, DMR13}. See Example \ref{ExamplesOfCohomologyTheories} below, or for the general setting see the gentle survey  \cite{FSS19}.

\medskip
This suggests that the solution to the {\bf open problem of the elusive nature of M-branes} requires,
similarly, identifying the right generalized cohomology theory, hence the right homotopy theory, in which the M-branes
(their charges) are cocycles -- see \cite{S1, S2, S3, S4}. For the \emph{fundamental} M2/M5-brane, we already know
this generalized cohomology theory in the rational approximation:
it is \emph{Cohomotopy on superspaces} in degree 4 (\cite[Sec. 2.5]{S-top}\cite{cohomotopy}\cite{FSS16a}, recalled in \cite{FSS19}
and below in Section \ref{TheFundamentalBraneScan}).
Hence the open question is: which enhancement of rational Cohomotopy
also captures the \emph{black} M-branes located at real ADE-singularities?

\medskip
Here we present a {\bf candidate solution}:
we set up \emph{equivariant rational Cohomotopy on superspaces} (Sec. \ref{ERSHTForSuperBranes}) and
show that (Sec.  \ref{ADESingularitiesInSuperSpacetime}, \ref{ADEEquivariantRationalCohomotopy},  and \ref{ADEEquivariantMBraneSuperCucycles})
the equivalence of homotopy theories that is given by Elmendorf's Theorem
translates into a duality in string/M-theory that makes the black branes at real ADE-singularities appear from the equivariance
of the super-cocycle of the fundamental M2/M5-brane:

\vspace{.3cm}
\hspace{-.8cm}
\begin{tabular}{|C{6.5cm}  C{3cm}  C{7cm} |}
\hline
  $G$-equivariance &
  ${\xymatrix{\ar@{<->}[rrrr]^{
    \mbox{
      \footnotesize
      \begin{tabular}{c}
        Elmendorf's Theorem
      \end{tabular}
    }}&&&&}}$  & $G$-fixed points
  \\
  \hline
  \hline
  \begin{tabular}{c}
    Fundamental M2/M5-branes
    \\
    on 11d superspacetime with
    \\
     real ADE-equivariant sigma-model
  \end{tabular}
    & ${\xymatrix{\ar@{<->}[rrr]^{\mbox{\footnotesize Theorem \ref{RealADEEquivariantEndhancementOfM2M5Cocycle}}}&&&}}$ &
  \begin{tabular}{c}
    Fundamental F1/M2/M5-branes
    \\
    on intersecting black M-branes
    \\
    at real ADE-singularities
    \\
  \end{tabular}
  \\
  \hline
\end{tabular}

\vspace{.3cm}
Our {\bf main theorem}, Theorem \ref{RealADEEquivariantEndhancementOfM2M5Cocycle},
 shows that enhancement
of the fundamental M2/M5-brane cocycle from rational Cohomotopy of superspaces to
\emph{equivariant  rational Cohomotopy} exists.
Furthermore, the possible choices correspond to fundamental branes propagating
on intersecting \emph{black M-branes} at real ADE-singularities,
with the brane instanton contribution
\cite{BeckerBeckerStrominger95, HarveyMoore99}
of the superembedding of the black brane \cite{Sorokin99, Sorokin01}
providing the equivariant coherence.
Part of this statement is a classification of finite
group actions on super Minkowski super spacetime $\mathbb{R}^{10,1\vert \mathbf{32}}$
by super-isometries. Our {\bf first theorem}, Theorem \ref{SuperADESingularitiesIn11dSuperSpacetime},
shows that this classification accurately reproduces the local models for $\geq \sfrac{1}{4}$-BPS black brane
solutions in 11-dimensional supergravity.

\medskip

\noindent The {\bf outline} of this article is as follows.
We start by providing novel consequences for the understanding of M-branes
in Sec.  \ref{Physics}. In particular,  in Sec. \ref{MBraneInterpretation}, we explain the physical meaning
of Theorem \ref{RealADEEquivariantEndhancementOfM2M5Cocycle},
by comparison to the story told in the informal string theory literature, the main points of which we
streamline there and in Sec.   \ref{TheFundamentalBraneScan}.
Sec.   \ref{EquivariantHomotopyTheory} provides the proper mathematical setting for our formulation.
After collecting the concepts and techniques of equivariant homotopy theory in Sec. \ref{HomotopyTheoryOfGSpaces},
we provide an extension to the super setting in Sec. \ref{ERSHTForSuperBranes}, which we hope would also be
of independent interest.  Our main results on Real ADE-Equivariant Cohomotopy classification of super M-branes
are given in Sections \ref{ADESingularitiesInSuperSpacetime}, \ref{ADEEquivariantRationalCohomotopy},  and \ref{ADEEquivariantMBraneSuperCucycles}, where
we discuss equivariant enhancements (according to Example \ref{EquivariantEnhancementOfSuperCocycles})
of the M2/M5-brane cocycle (Prop. \ref{M2M5SuperCocycle}).
Group actions on the 4-sphere model space, as well as the resulting incarnation of the 4-sphere
as an object in equivariant rational super homotopy theory, are given in Sec.  \ref{ADEEquivariantRationalCohomotopy}.
 In Sec.  \ref{ADESingularitiesInSuperSpacetime}
 we describe real ADE-actions on 11-dimensional superspacetime $\mathbb{R}^{10,1\vert \mathbf{32}}$
and their fixed point super subspaces.
The corresponding super-orbifolds constitute a supergeometric refinement of the
du Val singularities in Euclidean space. Having discussed real ADE-actions both
on $\mathbb{R}^{10,1\vert \mathbf{32}}$ (Sec. \ref{ADESingularitiesInSuperSpacetime}) and on $S^4$ (Sec. \ref{ADEEquivariantRationalCohomotopy}), the possible equivariant enhancements of the M2/M5-brane cocycle
that are compatible with these actions are studied in Sec. \ref{ADEEquivariantMBraneSuperCucycles}. In particular,
in Sec. \ref{FundamentalBraneCocyclesOnSystemsOfIntersectingBlackBranes}
we show that the homotopy-datum in these enhancements is
the instanton contribution of the corresponding black brane.

\medskip
\noindent The final statement is Theorem \ref{RealADEEquivariantEndhancementOfM2M5Cocycle}.
This involves three ingredients:
\begin{center}
\fbox{
$
  \xymatrix@R=.5em{
    &
    \mathbb{R}^{10,1\vert \mathbf{32}}
    \ar@(ul,ur)[]^{G_{\mathrm{ADE}} \times G_{\mathrm{HW}}}
    \ar[rr]^-{\mu_{{}_{M2/M5}}}_<<<<<<<<{\ }="s"
    &&
    S^4
    \ar@(ul,ur)[]^{G_{\mathrm{ADE}} \times G_{\mathrm{HW}}}
    \\
    \\
    \\
    \\
    &
    \mathbb{R}^{p,1\vert \mathbf{N}}
    \ar@{^{(}->}[uuuu]|-{
      \mbox{
        \tiny
        \color{blue}
        \begin{tabular}{c}
          super-
          \\
          embedding
        \end{tabular}
      }
    }
    \ar[rr]^>>>>>>>>{\ }="t"
    &&
    S^d
    \ar@{^{(}->}[uuuu]
    \\
    &
    \mbox{
      \tiny
      \color{blue}
      \begin{tabular}{c}
        black brane
        \\
        at
        \\
        Real ADE-singularity
      \end{tabular}
    }
    &
    \mbox{
      \tiny
      \color{blue}
      \begin{tabular}{c}
        Real ADE-equivariant
        \\
        fundamental brane
        \\
        cocycle
      \end{tabular}
    }
    &
    \mbox{
      \tiny
      \color{blue}
      \begin{tabular}{c}
        coefficient for
        \\
        Real ADE-equivariant
        \\
        rational Cohomotopy
      \end{tabular}
    }
    \\
    &
    \mbox{\bf Sec. \ref{ADESingularitiesInSuperSpacetime}}
    &
    \mbox{\bf Sec. \ref{ADEEquivariantMBraneSuperCucycles}}
    &
    \mbox{\bf Sec. \ref{ADEEquivariantRationalCohomotopy}}
    \ar@{==>}|{
      \mbox{
        \tiny
        \color{blue}
        \begin{tabular}{c}
          brane instanton
          \\
          contribution from
          \\
          Green-Schwarz action
        \end{tabular}
      }
    }
    "s"; "t"
  }
$}
\end{center}
\begin{itemize}
    \vspace{-2mm}
  \item In Section \ref{ADESingularitiesInSuperSpacetime} we classify group actions on the \emph{domain space}, namely on $D = 11$, $\mathcal{N} =1$
    super-Minkowski spacetime, which have the same fixed point locus as an involution and are at least $\sfrac{1}{4}$-BPS.
    If orientation-preserving, these actions are given by finite subgroups of $\mathrm{SU}(2)$, hence finite groups in the ADE-series.
    By the comparison in Section \ref{Physics},
    these singular super-loci recover the \emph{superembedding}
    of super $p$-branes advocated in \cite{Sorokin99, Sorokin01}.

  \vspace{-2mm}
  \item In Section \ref{ADEEquivariantRationalCohomotopy} we discuss real ADE-space structure on the \emph{coefficient space}, namely the rational 4-sphere.
    This establishes the cohomology theory \emph{ADE-equivariant rational cohomotopy in degree 4} to which the M2/M5-cocycle enhances.

  \vspace{-3mm}
  \item In Section \ref{ADEEquivariantMBraneSuperCucycles} we discuss the possible equivariant enhancements of the M2/M5-cocycle itself, mapping
   between these real ADE-spaces.
   \newline
   In particular, in Section \ref{FundamentalBraneCocyclesOnSystemsOfIntersectingBlackBranes}
   we show
   that the homotopies appearing in the equivariant enhancement
   exhibit the super Nambu-Goto/Green-Schwarz Lagrangian 
   of the solitonic brane loci,
   hence the local density of their brane instanton contributions
   \cite{BeckerBeckerStrominger95, HarveyMoore99}.
\end{itemize}

In the appendices
we provide our spinorial conventions (Section \ref{SpacetimesAndSpin}) as well as some basic notions from homotopy theory and cohomology theories (Section \ref{HomotopyTheory}).

\newpage

\vspace{0cm}

\hypertarget{TableA}{$\,$}

{\footnotesize
\begin{center}
\begin{tabular}{|c|l|c|}
  \hline
  \begin{tabular}{c}
    {\bf Real ADE-Singularities}
    \\
    {\bf according to Thm \ref{SuperADESingularitiesIn11dSuperSpacetime}, Prop. \ref{IntersectingBlackBraneSpecies}}
  \end{tabular}
  &
  \begin{tabular}{c}
    {\bf Selected physics literature}
    \\
    {\bf on black brane species}
  \end{tabular}
  &
  \begin{tabular}{c}
    \bf Example
  \end{tabular}
  \\
  \hline
  \hline
  General &
  \begin{tabular}{l}
    \cite{AcharyaWitten01}
    \\
   \cite[Sec. 3]{Acharya02}
   \\
   \cite{AtiyahWitten03}
   \\
   \cite{AcharyaGukov04}
  \end{tabular}
  & (Sec. \ref{MBraneInterpretation})
  \\
  \hline
  \hline
  $\mathrm{MO}9$ &
  \begin{tabular}{l}
    \cite{HoravaWitten96a}
    \\
    \cite{HoravaWitten96b}
    \\
    \cite[Sec. 3]{GKST01}
  \end{tabular}
  & \ref{TheMO9}
  \\
  \hline
  $\mathrm{MO5}$ &
  \begin{tabular}{l}
    \cite{Witten95II}
    \\
    \cite{Hori97}
    \\
    \cite[Sec. 3.1]{HananyKol00}
  \end{tabular}
  &
  \ref{TheBlackM5}
  \\
  \hline
   $\mathrm{MO1}$
   &
   \begin{tabular}{l}
     \cite{Hull84}
     \\
     \cite[p. 94]{Philip05}
     \\
     \cite[Sec. 3.3]{HananyKol00}
   \end{tabular}
   & \ref{TheMWave}
   \\
   \hline
   \hline
  $\mathrm{M}2$ &
  \begin{tabular}{l}
    \cite{MFFGME10}
    \\
    \cite{BLMP13}
  \end{tabular}
  & \ref{TheBlackM2}
  \\
  \hline
  $\mathrm{MK}6$ &
  \begin{tabular}{l}
    \cite[p. 6-7]{Townsend95},
    \\
    \cite[Sec. 2]{Sen97},
    \\
    \cite[p. 17-18]{AtiyahWitten03}
  \end{tabular}
  & \ref{TheMK6}
  \\
  \hline
  \hline
  $\mathrm{M5}_{\mathrm{ADE}}$
  &
  \begin{tabular}{l}
    \cite[Sec.  8.3]{MF10}
    \\
    \cite{HMV13}
    \\
    \cite[Sec. 3]{ZHTV15}
    \\
  \end{tabular}
  & \ref{ADEM5}
  \\
  \hline
  \begin{tabular}{c}
    $\tfrac{1}{2}\mathrm{NS}5 =\mathrm{M5} \; \Vert \; \mathrm{MK6} \;\dashv\; \mathrm{MO9}_I$
  \end{tabular} &
  \begin{tabular}{l}
    \cite[Sec. 2.4]{BrodieHanany97}
    \\
    \cite{Berkooz98}
    \\
    \cite[Sec. 2]{EGKRS00}
    \\
    \cite[around Fig. 6.1, 6.2]{GKST01}
    \\
    \cite[Sec. 6, 7]{ZHTV15}
    \\
    \cite[p. 38 and around Fig. 3.9, 3.10]{Fazzi17}
  \end{tabular}
  & \ref{TheBlackNS5}
  \\
  \hline
  \begin{tabular}{c}
    $\mathrm{M1}= \mathrm{M2} \; \dashv \; \mathrm{M5} $
  \end{tabular}
  &
  \begin{tabular}{l}
    \cite[Sec. 2.2.1]{BPST10}
    \\
    \cite[Sec. 2.3]{HI13}
    \\
    \cite[HIKLV15]{HIKLV15}
  \end{tabular}
  & \ref{TheSelfDualString}
  \\
  \hline
  \begin{tabular}{c}
    $\mathrm{NS}1_{H}=\mathrm{M}2 \; \dashv \; \mathrm{MO}9_{H} $
  \end{tabular}
  &
  \begin{tabular}{l}
    \cite{LLO97},
    \\
    \cite{Kashima00}
    \\
    \cite[Sec. 2.2.2]{BPST10}
  \end{tabular}
  & \ref{TheBlackNS1H}
  \\
  \hline
  \begin{tabular}{c}
    $\mathrm{E}1=\mathrm{M2} \; \dashv \; \mathrm{MO}9_I$
  \end{tabular}
  &
  \begin{tabular}{l}
    \cite{KKLPV14}
  \end{tabular}
  & \ref{TheBlackNS1H}
  \\
  \hline
\end{tabular}
\end{center}
}

\noindent
{
{\bf {Table L.}} The list of symbols for the real ADE-singularities in 11d super spacetime,
as they appear in the classification of  Theorem \ref{SuperADESingularitiesIn11dSuperSpacetime} and Theorem \ref{RealADEEquivariantEndhancementOfM2M5Cocycle},
matched with pointers to selected references (out of many) in the physics literature, from which
the established name of the corresponding brane species may be identified, as discussed in detail in Sec. \ref{MBraneInterpretation}.
}

\begin{center}
\begin{tabular}{l}
{\bf List of results, tables and figures.}
\\
{
\footnotesize
\begin{tabular}{|l|ll|}
  \hline
  Singularities in $D = 11$, $\mathcal{N} = 1$ super spacetime & Table & \hyperlink{SingularitiesTable}{1}
  \\
  \hspace{.3cm} Simple singularities & Thm. & \ref{SuperADESingularitiesIn11dSuperSpacetime}
  \\
  \hspace{.3cm} Non-simple singularities & Prop. & \ref{NonSimpleRealSingularities}
  \\
  &Figs.&  \hyperlink{Figure1}{1}, \hyperlink{Figure2}{2}
  \\
  \hline
  Cocycles in equivariant super cohomotopy & Table & \hyperlink{EquivariantEnhancementsList}{2}
  \\
  & Thm.  & \ref{RealADEEquivariantEndhancementOfM2M5Cocycle}
  \\
  \hline
  Branes & &
  \\
  \hspace{.3cm} The old brane scan & Table & \hyperlink{TableB}{B}
  \\
  \hspace{.3cm} The fundamental brane bouquet & Figure & \hyperlink{Figure3}{3}
  \\
  \hspace{.3cm} Selected literature on black branes & Table & \hyperlink{TableA}{L}
  \\
  \hline
\end{tabular}
}
\end{tabular}
\end{center}

\newpage

\noindent {\bf Interpretation.} We suggest that the results of this article
may be understood as providing the {\bf black brane scan}:
we recall in Sec. \ref{TheFundamentalBraneScan} below, that for the \emph{fundamental branes}
(or ``probe branes'': the consistent Green--Schwarz-type sigma-models) such a cohomological classification of species is famously known as the \emph{old brane scan}
\cite{AETW87}\cite[p. 15]{Duff88}, recalled as Prop. \ref{TheOldBraneScan} below. Careful consideration of higher symmetries leads to a completion
 to the \emph{fundamental brane bouquet} \cite{FSS13}, indicated in \hyperlink{Figure2}{Figure 2} below.

\medskip
It has been an open problem (see \cite[p. 6-7]{Duff99}, \cite{Duff08}) to improve this ``fundamental brane scan/bouquet'' to a \emph{cohomological} classification\footnote{
  One may try to organize some of the branes missing from the old brane scan by other than cohomological means,
  such as by BPS solutions to supergravity \cite{DuffLu92, DuffKhuriLu95}. We discuss this in Section \ref{MBraneInterpretation}.
  } that includes the ``black'' brane species. Theorem \ref{RealADEEquivariantEndhancementOfM2M5Cocycle} suggests that the missing \emph{black brane scan}
is obtained by enhancing to equivariant cohomology; see Sec. \ref{MBraneInterpretation} for elaborations.

\medskip
Here it may be noteworthy that the emerging picture of M-theory thus obtained exhibits the foundational
paradigms of \emph{Klein geometry} and of \emph{Cartan geometry} (see e.g. \cite[Chapter 1]{CapSlovak09}):
\begin{enumerate}[\bf (i)]
\vspace{-2mm}
\item {\bf Klein geometry.} In the \emph{Erlangen program} of \cite{Klein1872} the \emph{basic shapes} of interest in geometry
are taken to be fixed loci of group actions on an ambient model space.\footnote{ From \cite[Sec. 1]{Klein1872}:
``As a generalization of geometry arises then the following comprehensive problem: {\it given a manifoldness
and a group of transformations of the same; to investigate the configurations belonging to the
manifoldness with regard to such properties as are not altered by the transformations of the group}.''.}
Theorem \ref{SuperADESingularitiesIn11dSuperSpacetime} shows that, when the ambient space is taken
 to be $D = 11$, ${\cal N} = 1$ super-spacetime, then the \emph{basic shapes} in the Kleinian sense are
precisely the black M-brane species and their bound states.

\begin{center}
{\small
\begin{tabular}{|c|c|c|c|}
  \hline
  & \bf General & \bf Black M-brane species
  \\
  \hline
  \hline
  \begin{tabular}{c}
    \bf Orbifold
    \\
    \bf Klein geometry
  \end{tabular}
    &
  $\Gamma \backslash G / H$ &
  $
  \underset{
    \mbox{
      cone with real ADE-singularity
    }
  }{
  \underbrace{
  \left( G_{\mathrm{ADE}} \times \mathbb{Z}_2 \right)
  \backslash
  \underset{
    \underset{
      \mbox{
        \begin{tabular}{c}
          super
          Minkowski spacetime
        \end{tabular}
      }
    }{ \mbox{\large $\mathbb{R}^{10,1\vert \mathbf{32}}  $ } }
  }{
  \underbrace{
  \overset{ \mbox{
    \begin{tabular}{c} super \\ Poincar{\'e} group\end{tabular}
  } }{\overbrace{\mathrm{Iso}(\mathbb{R}^{10,1\vert \mathbf{32}})}} / \mathrm{Spin}(10,1)
  }}
  }}
  $
  \\
  \hline
\end{tabular}
}
\end{center}

For example, the following picture (from \hyperlink{Figure1}{Figure 1} below) illustrates the super orbifold Klein geometry that is the local model
for M2-branes at an ADE-singularity
intersecting an MO9-plane at a Ho{\v r}ava-Witten $\mathbb{Z}_2$-singularity (see \hyperlink{TableA}{Table L} for a list of literature, and
see  Sec. \ref{Physics} for discussion of the physics background):

\hspace{3cm}
\scalebox{.9}{
\includegraphics[width=.5\textwidth]{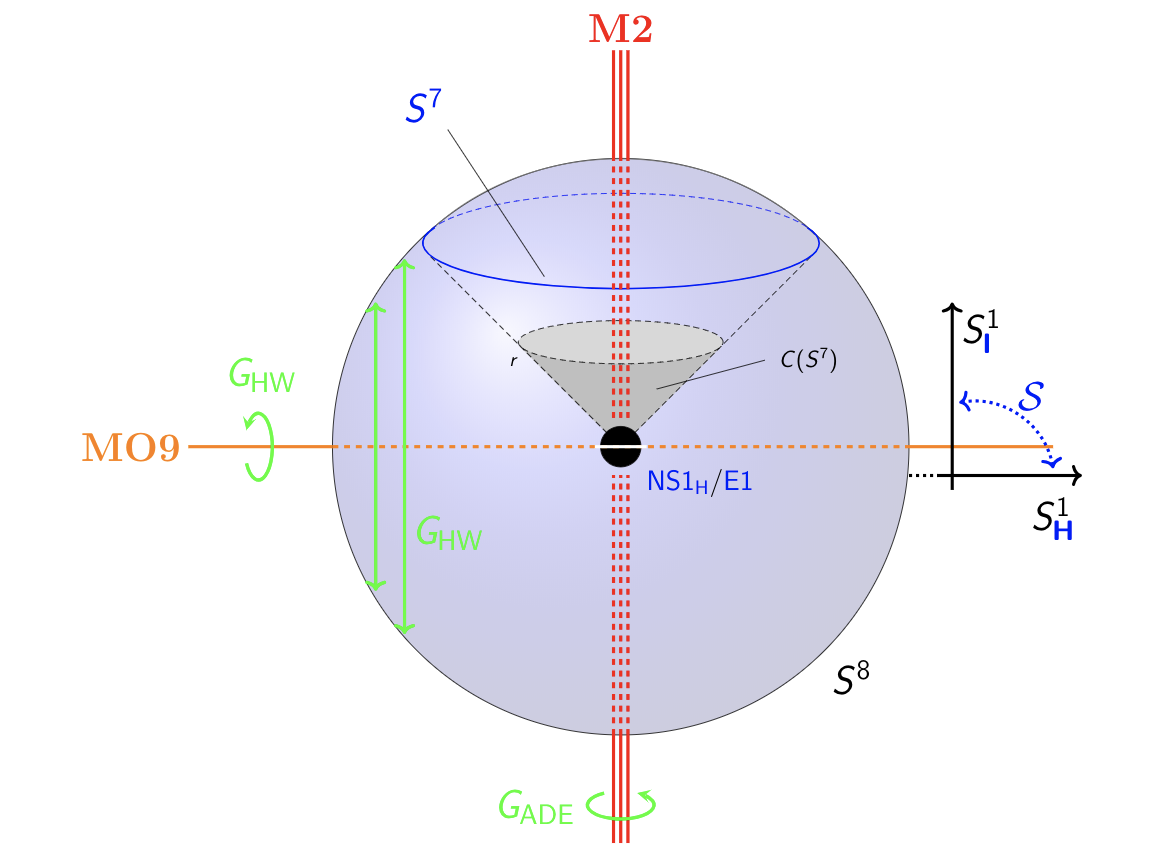}
}

Our main Theorem \ref{RealADEEquivariantEndhancementOfM2M5Cocycle} characterizes cocycles on these Klein geometries in equivariant super homotopy theory, with coefficients in
the 4-sphere, equipped with analogous group actions.

By Example \ref{EquivariantEnhancementOfSuperCocycles} below, these are, in particular, systems of maps from the fixed point strata of the singular super spacetime, to those of the (rational) 4-sphere, as indicated by the following picture:

\vspace{-1.4cm}
\hspace{3cm}
\raisebox{-190pt}{
\includegraphics[width=.6\textwidth]{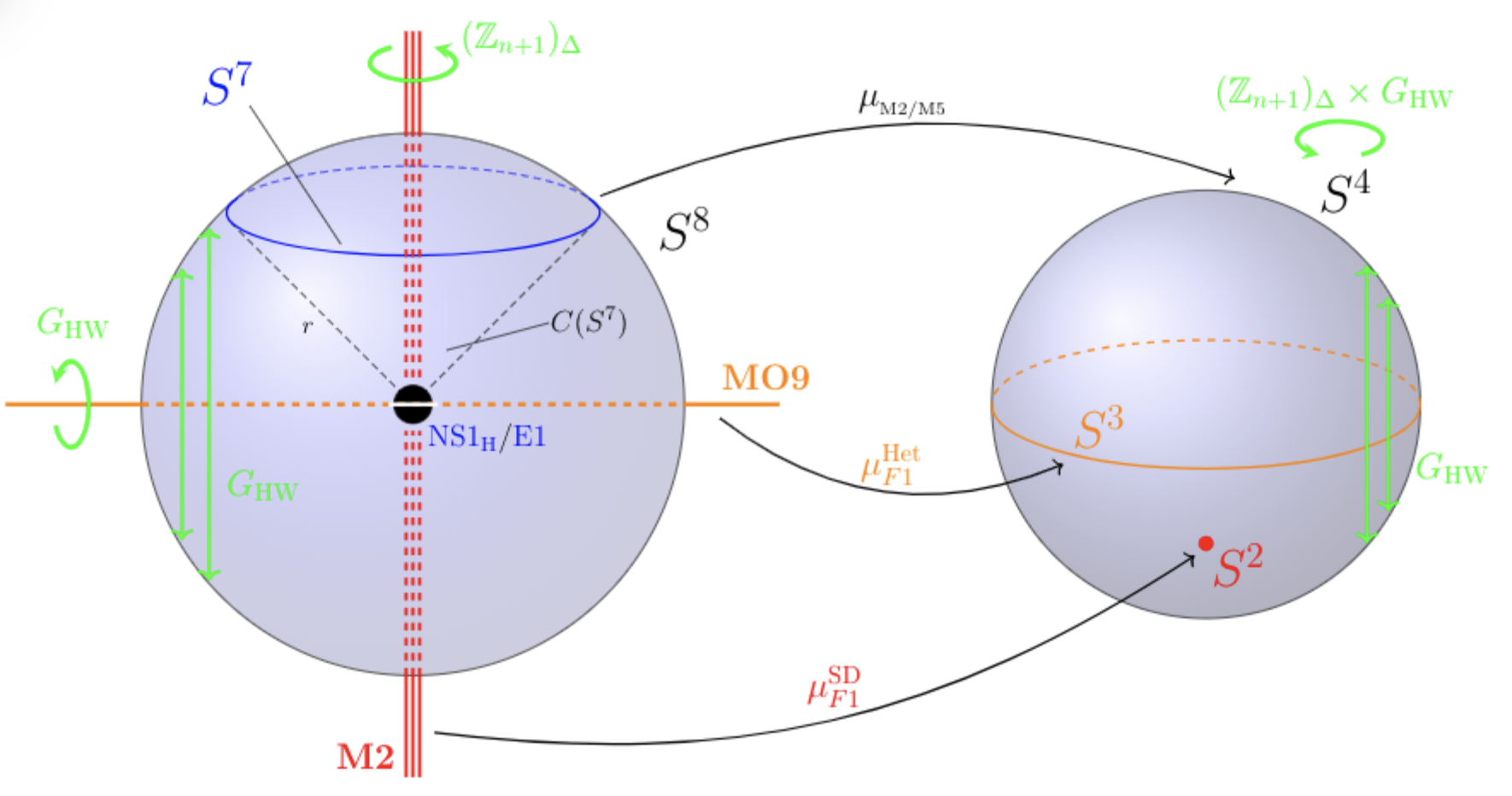}
}

Here the group actions on our \emph{coefficient 4-sphere} on the right (defined below in Sec. \ref{RealADEActionsOnThe4Sphere}) are those on the \emph{spacetime} 4-sphere
around a black M5-brane, we explain this in Sec. \ref{MBraneInterpretation}.
In addition to the ADE-singularities captured by ADE-equivariance, the \emph{real} structure (see Example \ref{ExamplesOfCohomologyTheories}) reflects the presence of M-theoretic O-planes, such as the MO9 (Example \ref{TheMO9}).
Each of the stratum-wise maps indicated in the above picture turns out to encode a super-cocycle that characterizes a fundamental brane species
propagating on a given black brane singularity. Moreover, the compatibility relations in the datum of an equivariant cocycle
makes the Green--Schwarz action functional for the corresponding sigma-model appear; this is explained in Sec. \ref{FundamentalBraneCocyclesOnSystemsOfIntersectingBlackBranes} below.
\label{On4Sphere}

\vspace{-1mm}
\item  {\bf Cartan geometry.} The geometry of \cite{Cartan1923} is the \emph{local-to-global principle} applied to Klein geometry: the local model space of Klein is
promoted to a \emph{moving frame} that characterizes each tangent space of a curved Cartan geometry as compatibly identified with the
local model space. Indeed, the geometry of supergravity is intrinsically Cartan geometric \cite{Lott90, EE}, and specifically
11d supergravity is \emph{equivalent} to torsion-free super Cartan geometry modeled on $\mathbb{R}^{10,1\vert \mathbf{32}}$ \cite{CL, Howe97, FOS}.
This generalizes to orbifold Cartan geometry (\emph{higher Cartan geometry} \cite{Schreiber15, Wellen}) locally modeled on orbifold Klein spaces.

The left half of the following picture illustrates a curved higher Cartan geometry locally modeled on the orbifold Klein geometry shown before.

\vspace{-.2cm}

\begin{center}
 \includegraphics[width=.8\textwidth]{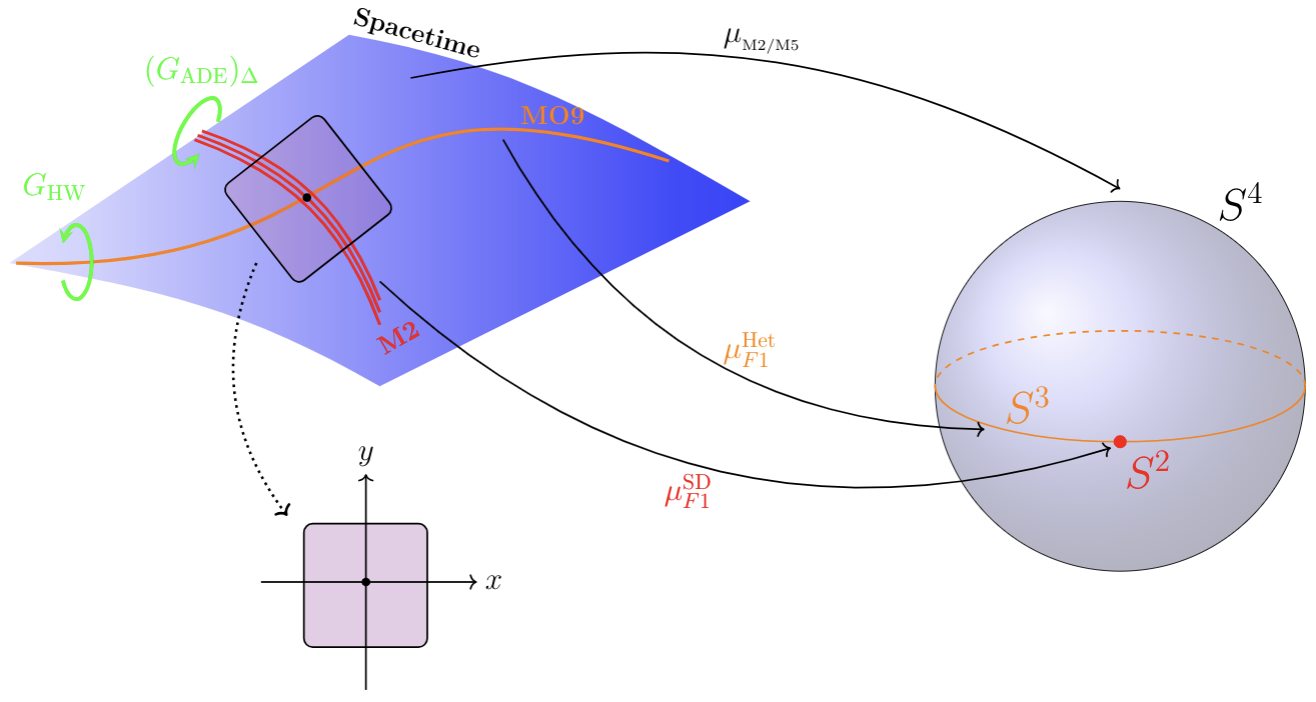}
\end{center}

\end{enumerate}

\vspace{-.5cm}

\noindent This perspective of orbifold Klein geometry controlling its curved generalizations, via higher Cartan geometry,
serves to conceptually explain how the equivariant cohomology of super-Minkowski spacetime itself, which we study here,
is able to see so much of the structure of M-theory.

\medskip
Therefore, since Theorem \ref{RealADEEquivariantEndhancementOfM2M5Cocycle} shows that equivariant cohomotopy \emph{locally},
i.e super tangent space-wise, captures M-brane physics,
this suggests that, at least rationally, the generalized cohomology theory \emph{real ADE-equivariant 4-cohomotopy of superspaces} (Sec. \ref{ADEEquivariantRationalCohomotopy}) serves as the missing {\bf definition of the concept of \emph{M-brane} species} also globally; in direct analogy to how the generalized (co-)homology theory \emph{K-theory} is understood
to provide the precise definition of the concept of D-branes:

{\small
\vspace{-2mm}
\hspace{1cm}
\begin{center}
\begin{tabular}{|c|c|l}
  \cline{1-2}
  {\bf Objects} & \begin{tabular}{c} \bf Cohomology theory \\ (Sec. \ref{Coh}) \end{tabular}
  \\
  \cline{1-2}
  \cline{1-2}
  M-branes & $\phantom{{A \atop A} \atop {A \atop A}}$ \begin{tabular}{c} {\color{blue} Real ADE-equivariant } \\ {\color{blue} Cohomotopy } \end{tabular} $\phantom{{A \atop A} \atop {A \atop A}}$
  & \multirow{2}{*}{ \raisebox{5pt}{\xymatrix{ \ar@{<->}@/^1.6pc/[d]^{\mbox{\begin{tabular}{c}
  \footnotesize stabilized \\ \footnotesize Ext/Cyc-adjunction \end{tabular}}} \\ \,  }} }
  \\
  \cline{1-2}
  D-branes & $\phantom{{A \atop A} \atop {A \atop A}}$ \begin{tabular}{c} Real \\ K-theory \end{tabular} $\phantom{{A \atop A} \atop {A \atop A}}$
  \\
  \cline{1-2}
\end{tabular}
\end{center}
}
\noindent To support this, there must be:
\begin{enumerate}[{\bf (1)}]
\vspace{-3mm}
\item a homotopy-theoretic formulation of the famous but informal idea
of ``compactifying M-theory on a circle'' (\cite{DuffHoweInamiStelle87, Witten95I}, see \cite[Sec. 6]{Duff99B} and via cohomology in \cite{MathaiSati04}), such that
\vspace{-3mm}
\item under this operation the cohomology theory \emph{degree-4 cohomotopy} transmutes into the cohomology theory \emph{K-theory}, matching how the M-branes are supposed to reduce to F1/D$p$-branes under {\bf double dimensional reduction}.
\footnote{
While such \emph{a derivation of K-theory from M-theory} is suggested by the title of \cite{ADerivationofK},
that article only checks that the behavior of the partition function of the 11d supergravity $C$-field is compatible
with the a priori K-theory classification of D-branes. }
\end{enumerate}
\vspace{-2mm}
Indeed, in \cite{FSS16a} \cite{FSS16b} it is shown that, rationally, {\bf (1)} is exhibited by the \emph{Ext/Cyc-adjunction} and then {\bf (2)} follows, since 6-truncated twisted K-theory appears, rationally, in the cyclic loop space of the 4-sphere.
In the companion article \cite{GaugeEnhancement} it is shown that the {\bf gauge enhancement} of this result to the
full, untruncated, twisted K-theory spectrum arises from the fiberwise stabilization of the unit of the
Ext/Cyc-adjunction applied to the A-type orbispace of the 4-sphere (Def.\ \ref{SuspendedHopfAction} below).

\medskip
In conclusion, equivariant rational cohomotopy of superspaces goes a long way towards capturing the folklore on
the zoo of brane species; and of course the vast majority of discussions in the
string theory literature is sensitive only to rational (non-torsion) effects, anyway. Nevertheless, as shown by the classification of
D-branes in string theory by K-theory, a {\it non-rational lift} of this cohomology theory will be necessary
to fully capture M-brane physics. For the moment we leave this as an open problem.
However, we suggest that detailed study of the
rational theory provides {crucial clues} for passing beyond the rational approximation.
This is because we are not just faced with one rational cohomology theory in isolation, but with a web of rational cohomology theories that are subtly related to each other and
which collectively paint a large coherent picture (see also \cite{FSS19}):
\begin{enumerate}
\vspace{-2mm}
\item {\bf Fundamental M2/M5-branes} \cite[Sec. 2.5]{S-top}\cite{cohomotopy}:
 The M2/M5-brane cocycle is, rationally, in equivariant cohomotopy.
 \vspace{-2mm}
\item {\bf Black M-branes} (Thm.\ \ref{RealADEEquivariantEndhancementOfM2M5Cocycle}):
The corresponding ADE-equivariant enhancement exhibits the black M-branes at ADE-singularities.
\vspace{-2mm}
\item  {\bf M/IIA duality} (\cite{FSS16a}
  The corresponding double dimensional reduction is cohomological cyclification and yields the F1/D$(p \leq 4)$-cocycle in type IIA string theory.
\vspace{-2mm}
\item  {\bf Gauge enhancement} \cite{GaugeEnhancement}:
The corresponding lift through the fiberwise stabilization of the Ext/Cyc-adjunction
yields the gauge enhancement to the full type IIA F1/D$p$ cocycle in twisted K-theory (rationally).
\vspace{-2mm}
\item {\bf IIA/IIB T-duality} \cite{FSS16b}:  The further double dimensional
reduction of that, via further cyclification, exhibits T-duality between the cocycles of the type IIA and type IIB F1/D$p$-branes.
\vspace{-2mm}
\item {\bf M/HET duality} \cite{Higher-T}\cite{PLB}:
 The higher analogue of this Fourier-Mukai transform applied to the M2/M5-brane cocycle itself yields a higher T-duality of
a 7-twisted cohomology theory \cite{Sa09} that connects to the Green--Schwarz mechanism of heterotic string theory.
\end{enumerate}
\vspace{-2mm}
\noindent This web of dualities between various rational cohomology theories accurately captures a fair bit of the web of dualities
expected in string/M-theory.
Since every non-rational lift of the cohomology theory for M-branes will have to lift that entire web of dualities, this
puts strong conditions on such a lift, considerably constraining the freedom in lifting an isolated rational cohomology theory.

\medskip

We read all this as indication that our analysis narrows in on the correct generalized cohomology theory classifying super M-branes,
and thus, at least in part, on the elusive definition of M-theory itself.


\section{Understanding M-Branes}
 \label{Physics}

We now provide an
 informal discussion, that is meant to put the formal results of
 Sections \ref{ADESingularitiesInSuperSpacetime}, \ref{ADEEquivariantRationalCohomotopy},  and \ref{ADEEquivariantMBraneSuperCucycles}
  into the perspective of string/M-theory.
For completeness and to highlight the concepts involved,
we  first quickly review a perspective on branes within string/M-theory. Then in
Sec. \ref{TheFundamentalBraneScan} we briefly recall the mathematical classification of \emph{fundamental branes}, which is the conceptual background for
the starting point of our mathematical discussion in Sec. \ref{ERSHTForSuperBranes}. Finally, Sec. \ref{MBraneInterpretation} concerns the physics interpretation
of our classification result from Sections \ref{ADESingularitiesInSuperSpacetime}, \ref{ADEEquivariantRationalCohomotopy},  and \ref{ADEEquivariantMBraneSuperCucycles}: we walk there through selected examples from the informal
string/M-theory literature (as listed in \hyperlink{TableA}{Table L}) and point out how to match, item by item, the entries of
our classification Tables \hyperlink{SingularitiesTable}{1} and \hyperlink{EquivariantEnhancementsList}{2} to structures in the folklore on M-branes.

\medskip

The concept of \emph{fundamental brane} is the evident higher-dimensional generalization of the concept of a
\emph{fundamental particle}:
a precise concept of \emph{fundamental particles}, in turn, is obtained by combining
perturbative quantum field theory\footnote{Contrary to wide-spread perception, perturbative quantum field theory,
such as pertaining to the standard model of particle physics,
has a perfectly rigorous mathematical formulation, going back to \cite{EpsteinGlaser73}, see e.g. \cite{Schreiber18}. }  with an insight called the \emph{worldline formalism} (reviewed in \cite{SchmidtSchubert95, Schubert96}).
Here, the trajectories of fundamental particles in some spacetime $X$ are represented by maps from the abstract worldline of the particle,
modeled by a 1-manifold $\Sigma_1$, to $X$
$$
  \xymatrix@R=-1pt{
    \Sigma_1 \ar@/^1pc/[rrrr]^-{\phi}_{\tiny \color{blue}
      \begin{tabular}{c}
      abstract
        \\
        worldline
      \end{tabular}
    }
    &&&& X
    \\
   { \tiny
    \begin{tabular}{c} particle \\ trajectory \end{tabular}
    }
    &&&&
  {
      \tiny \begin{tabular}{c} spacetime \end{tabular}
    }
  }
$$
The physically realizable particle trajectories are characterized as being the local extrema of a certain
non-linear functional
on the space of all these maps, the \emph{action functional}.
This has two contributions:
\begin{enumerate}[{\bf (i)}]
\vspace{-2mm}
\item The first contribution is that of the \emph{proper volume} of $\phi$, as measured by the
pseudo-Riemannian metric on $X$. This encodes the forces that a background field of gravity exerts on the particle,
it is known as the \emph{Nambu-Goto action functional}.

\vspace{-2mm}
\item The second contribution encodes the remaining forces felt by the particle, exerted by further background fields.
Notably, if the particle is charged under an electromagnetic field captured by a differential 2-form $\mu_2$ on $X$ (the \emph{Faraday tensor}),
then the corresponding contribution to the action functional is the holonomy functional of a principal connection
whose curvature 2-form is $\mu_2$.
If $X$ is Minkowski spacetime, then this connection is given by a differential 1-form $\Theta_1$ on $X$ (the \emph{vector potential})
and the holonomy functional is just the integration of $\Theta_1$ along $\phi$. This is the simplest example of what is called a \emph{WZW term}
in an action functional.
\end{enumerate}
\vspace{-2mm}
Hence, a fundamental charged particle is characterized by an action functional which is the integration
over the particle's worldline of a differential 1-form $\mathbf{L}_1$ (the \emph{Lagrangian density}) which, just slightly schematically, reads:
$$
  \mathbf{L}_1 \;=\; \underset{ \mathrm{NG} }{\underbrace{\mathrm{vol}_{1}}} + \underset{ \mathrm{WZW} }{\underbrace{ \mathbf{\Theta}_1 }}
  \;\;
 \xymatrix{\ar@{|->}[r]^d &}
  \mu_2
  \,.
$$

The Mellin transform of these action functionals yields distributions in two variables, called \emph{Feynman propagators}, which
may be interpreted as the probability amplitude for a \emph{quantum} fundamental particle to come into an accelerator experiment
on a fixed asymptotic trajectory, and emerge on the other end on some fixed asymptotic trajectory, \emph{without interacting}, in between,
with anything.
More generally, given a finite graph, a product of distributions may be assigned to it, with one Feynman propagator factor for each edge.
These products turn out to be well defined and unique away from coincident vertices, and may be extended
to the locus of coinciding vertices. The choice involved in these extensions of distributions is called \emph{(re-)normalization}.
The resulting distribution is called the \emph{Feynman amplitude} associated with the graph.
If the graph has external edges, this may be interpreted as the probability amplitude for
some number of quantum fundamental particles to come into an accelerator experiment on given asymptotic trajectories, interacting with each other,
as determined by the shape of the graph, and emerge on the other side on some given asymptotic trajectories.
The \emph{Feynman perturbation series} is the sum over all graphs of these Feynman amplitudes, as a formal power series in
powers of the number of loops of the graphs. This, finally, may be interpreted as the probability amplitude that may be compared to
experiment, describing an arbitrary scattering process of several quantum fundamental particles.

\medskip
What is striking about this worldline formulation of perturbative quantum field theory is that it immediately suggests
a tower of possible deformations: it is compelling, at least mathematically, to investigate the variants of this prescription where 1-dimensional graphs are replaced by $(p+1)$-dimensional manifolds $\Sigma_{p+1}$, where hence the particle trajectories are replaced by maps
out of this \emph{abstract worldvolume} of dimension $p + 1$

\vspace{-5mm}
\begin{equation}
  \label{pBraneTrajectory}
  \xymatrix@R=1pt{
    \Sigma_{p+1} \ar@/^1pc/[rrrr]^-{\phi}_{\tiny \color{blue}
      \begin{tabular}{c}
        abstract
        \\
        worldvolume
      \end{tabular}
    }
&&&&
X
\\
{
    \tiny \begin{tabular}{c} $p$-brane \\ trajectory \end{tabular}
    }
    &&&&
    \mbox{\tiny
      spacetime
    }
  }
\end{equation}
and where, finally the action functional is replaced by the integral of a suitable $(p+1)$-form
\begin{equation}
  \label{Lagrangian}
  \mathbf{L}_{p+1}
    \;:=\;
  \underset{ \mathrm{NG} }{\underbrace{\mathrm{vol}_{p+1}}} + \underset{ \mathrm{WZW} }{\underbrace{ \Theta_{p+1} }}
  \;\;
 \xymatrix{\ar@{|->}[r]^d &}
  \mu_{p+2}
  \,.
\end{equation}
This is naturally thought of as possibly producing probability amplitudes for higher dimensional fundamental objects
to scatter off of each other. For $p = 1$ these objects look like strings (whence the name); for $p = 2$ they look like
mem\emph{branes}. Hence for general $p$ one speaks of \emph{fundamental $p$-branes} \cite{DIPSS88}.
Moreover, at least for $p = 1$ there is a good candidate of what may replace the sum over all graphs: since 2-dimensional
surfaces have a nice classification by genus and punctures, there is a good mathematical definition
of a \emph{string perturbation series}, deforming the above concept of the Feynman perturbation series.
The study of this string perturbation series, thought of as a deformation of the Feynman perturbation series,
is the subject of \emph{perturbative String theory} (e.g. \cite{Witten15}, \cite{AMS} and references therein).


\begin{center}
 \includegraphics[width=.4\textwidth]{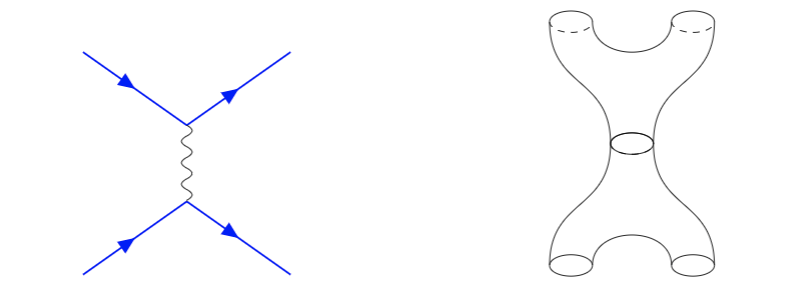}
\end{center}

\noindent Despite their immense successes, both the Feynman perturbation series as well as the string perturbation series
have a severe conceptual problem in their very perturbative nature. More precisely, while one would like to interpret the result
of these perturbation series as \emph{numbers}, characterizing concrete probability amplitudes that may be compared to
measurement results, mathematically
they are not numbers but just formal power series. Worse, simple arguments show that in all cases of interest, the radius of
convergence of these formal power series \emph{vanishes} (e.g. \cite[Sec. 1]{Suslov05}). This means that Feynman/string perturbation theory provides
no reason why it would make sense to sum up the first few terms of the perturbation theory and regard that as a decent approximation to
experiment. Sometimes, notably for computations in QED, it does happen to produce excellent agreement with experiment. But sometimes it does not
(as in much of QCD) and besides trial and error and experimental experience, there is no mathematical reason to tell.
Hence the perturbativeness of fundamental particles and of fundamental strings
is as much their intrinsic nature as it is their fatal shortcoming. The understanding of non-perturbative effects
in quantum field theory such as quark confinement, hence existence of ordinary baryonic matter,
is a wide open \emph{millenium problem} \cite{ClayInstitute}.

\medskip
In view of this, it is noteworthy that perturbative String theory exhibits concrete hints as
to the nature of its non-perturbative completion. Notably, one finds \cite{DuffHoweInamiStelle87} that when the fundamental
membrane mimics a fattened fundamental string by wrapping around a small circle fiber of spacetime
(\emph{double dimensional reduction}), then the volume of that
circle fiber is proportional to the strength of the fundamental string's interactions. Read the other way around, this says that the
fundamental string at strong coupling, hence beyond perturbation theory, should be nothing but the fundamental membrane.
This led to the speculation that a
non-perturbative version of perturbative string theory does exist and is embodied by M(embrane)-theory
 \cite{Townsend95, Witten95I}, even if its actual nature remains elusive.\footnote{
It may be worthwhile to recall that, in mathematics, it is not unusual to postulate the existence of certain theories before one actually knows about their nature.
Famous historical examples include the theory of \emph{motives} or the theory of the \emph{field with one element}. The former has meanwhile
been found. }
That is, there is no straightforward way to generalize the summation over all 1-dimensional graphs beyond a summation over surfaces
to a summation over higher-dimensional manifolds, because there is no classification parameter for higher dimensional manifolds,
that could organize such a sum as a formal power series.
Therefore, a potential \emph{Membrane theory}, along the above lines, if it makes sense at all,
must be more subtle than being a \emph{direct} variant of perturbative string theory, which itself \emph{is} a direct
variant of the Feynman perturbation theory of fundamental particles. This subtlety is the reason for the name ``M-theory'':
this is a non-committal shorthand\footnote{ \cite[p. 1]{HoravaWitten96a}: ``{\it As it has been proposed that [this] theory is a supermembrane theory but there are some reasons to doubt that interpretation, we will non-committedly call it the $M$-theory, leaving to the future the relation of $M$ to membranes.}'' }\raisebox{4pt}{\footnotesize ,\,}\footnote{
\cite[p. 2]{Witten95II}: ``{\it $M$ stands for magic, mystery, or membrane, according to taste}.''} for ``Membrane theory'',
to cautiously highlight that if a concept of fundamental membranes makes sense
then its relation to fundamental strings must be more subtle than that of the relation of fundamental strings to fundamental particles.

\medskip
Evidence for M-theory had been actively accumulated up to just around the turn of the millenium, see \cite{Duff99B}; then it waned, the community
getting distracted by other topics:

\begin{quote}
\it \footnotesize We still have no fundamental formulation of ``M-theory'' -
Work on formulating the fundamental principles underlying M-theory has noticeably waned. \it [\dots]. \it If history is a good guide, then we should expect that anything as profound and far-reaching as a fully satisfactory formulation of M-theory is surely going to lead to new and novel mathematics. Regrettably, it is a problem the community seems to have put aside - temporarily. But, ultimately, Physical Mathematics must return to this grand issue.
\hfill {\rm \cite[Sec. 12]{Moore14}}
\end{quote}

 This may remind one of an old prophecy\footnote{
\cite{Witten03}: ``{\it Back in the early '70s, the Italian physicist, Daniele Amati reportedly said that string theory was part of 21st-century physics that fell by chance into the 20th century. I think it was a very wise remark. How wise it was is so clear from the fact that 30 years later we're still trying to understand what string theory really is.}''}
which suggests that unraveling the true nature of string theory, hence of what came to be called M-theory, will require developments that become available only in the
21st century. One development that the new millenium has brought is the blossoming of homotopy theory into an immensely rich
(see for instance \cite{Ravenel03, HHR09}), powerful (\cite{Lurie09})
and foundational (\cite{Shulman17}) subject.
Above we have shown that in homotopy theory one finds curious
classifications, whose entries we have labeled with symbols used in the String/M-theory literature. Now we will discuss
how these entries match a zoo of phenomena that the informal string theory literature expects to play a role in M-theory.

\subsection{The fundamental brane scan}
  \label{TheFundamentalBraneScan}

Here we informally recall the some background on the cocycle $\mu_{ M2/M5}$ for the fundamental M2/M5-brane,
 which is the starting point of the analysis to come in Prop. \ref{M2M5SuperCocycle}, and indicate its place in a general
 cohomological classification of fundamental branes: the \emph{fundamental brane bouquet} in
\hyperlink{Figure3}{Figure 3} below.

\medskip
An early indication that homotopy theory plays a deep role in string theory was, in hindsight, the joint success
and failure of the old brane scan (recalled as Prop. \ref{TheOldBraneScan}).
Namely, via a sequence of somewhat involved arguments\footnote{
In \cite{Sorokin99, Sorokin01} it is shown that the traditional derivation of the Green--Schwarz-type sigma-models \eqref{GreenLagrangian},
is clarified drastically if one takes the worldvolume $\Sigma_{p+1}$ \eqref{pBraneTrajectory} to be a supermanifold locally modeled on the
relevant BPS super subspace $\mathbb{R}^{p,1\mathbf{N}} \hookrightarrow \mathbb{R}^{10,1\vert \mathbf{32}}$. This is exactly what we see
appear in Theorem \ref{RealADEEquivariantEndhancementOfM2M5Cocycle}, via Prop. \ref{TrivializationsOfRestrictionsOfM2M5Cocycle}.},
it was eventually found that for fundamental branes on a supergravity
background to be compatible with local supersymmetry,
the form $\mu_{p+2}$ in \eqref{TheOldBraneScan}, had to be a non-trivial cocycle in the supersymmetry super Lie algebra cohomology
of super-spacetime \cite{AzTo89}. The corresponding Lagrangians \eqref{Lagrangian}  are the \emph{Green--Schwarz-type} Lagrangians,
which in Prop. \ref{GreenSchwarzFromSuperVolumeForm} we have seen to be just the \emph{super} volume forms (hence the ``super Nambu-Goto Lagrangians'')\footnote{
  In the sprit of this physics section, we are deliberately suppressing notation for pullback of differential forms in these expressions,
  in order to bring out conceptual meaning of these formulas;
  see instead Sec. \ref{FundamentalBraneCocyclesOnSystemsOfIntersectingBlackBranes} for precise details.
}
\begin{equation}
  \label{GreenLagrangian}
  \mathbf{L}^{\mathrm{GS}}_{p+1}
    \;:=\;
  \underset{
    \underset{ \mbox{\footnotesize \begin{tabular}{c} supersymmetric volume \\ Prop. \ref{GreenSchwarzFromSuperVolumeForm} \end{tabular} }   }{
      \mathrm{svol}_{p+1}
    }
  }{
  \underbrace{
  \underset{ \mathrm{NG} }{\underbrace{\mathrm{vol}_{p+1}}} + \underset{ \mathrm{WZW} }{\underbrace{ \Theta_{p+1} }}
  }}
 \; \xymatrix{\ar@{|->}[r]^d&}
  \underset{ \mbox{ \footnotesize  \begin{tabular}{c} super- \\ cocycle\end{tabular} } }{ \underbrace{ \mu_{p+2} } }
  \!\!\!\!\!.
\end{equation}
These cocycles $\mu_{p+2}$ are what the (old) \emph{brane scan} \cite{Duff88} classifies:

\hypertarget{TableB}{$\,$}

\begin{center}
\scalebox{.8}{
\begin{tabular}{|c||c|c|c|c|c|c|c|c|c|c|c|}
  \hline
\backslashbox{~~~~~$d+1$}{\raisebox{-9pt}{$\!\!\!\!\!\!\!\!p$}}
     &   & $1$ & $2$ & $3$ & $4$ & $5$ & $6$ & $7$ & $8$ & $9$ & $10$
	 \\[5pt]
	 \hline \hline
	 $\!\!\!\!\!\!\!\!\!\!\!\!\!\!\!\!\!\!\!\!\!\!\!10+1$ &   & &
	  $\mu_{{}_{M2}}$  &
	  & &
	 &&& &  &
	 \\[5pt]
	 \hline
	 $\!\!\!\!\!\!\!\!\!\!\!\!\!\!\!\!\!\!\!\!\!\!9+1$ &  		
	    &
        $\mu_{{}_{F1}}^{H/I}$
		&
		&
        $ \phantom{ {A \atop A} \atop {A \atop A}} $
		&
		&
        $\mu^{H/I}_{{}_{\mathrm{NS5}}}$
	  &
	  &
	  &
      &
	  &
	 \\[5pt]
	 \hline
	 $\!\!\!\!\!\!\!\!\!\!\!\!\!\!\!\!\!\!\!\!\!\!8 + 1$ & & & & & $\ast$ & & & & & &
	 \\[5pt]
	 \hline
	 $\!\!\!\!\!\!\!\!\!\!\!\!\!\!\!\!\!\!\!\!\!\!7 + 1$  &  & & & $\ast$ & & & & & & &
	 \\[5pt]
	 \hline
	 $\!\!\!\!\!\!\!\!\!\!\!\!\!\!\!\!\!\!\!\!\!\!6 + 1$  &  & & $\ast$ & & & & & & & &
	 \\[5pt]
	 \hline
	 $\!\!\!\!\!\!\!\!\!\!\!\!\!\!\!\!\!\!\!\!\!\!5 + 1$  &  & $\ast$
		 & & $\ast$ & & & & & & &
	 \\[5pt]
	 \hline
	 $\!\!\!\!\!\!\!\!\!\!\!\!\!\!\!\!\!\!\!\!\!\!4 + 1$ &  &  & $\ast$ & & & & & & & &
	 \\[5pt]
	 \hline
	 $\!\!\!\!\!\!\!\!\!\!\!\!\!\!\!\!\!\!\!\!\!\!3 + 1$  & & $\ast$ & $\ast$ &&&&& & & &
	 \\[5pt]
	 \hline
	 $\!\!\!\!\!\!\!\!\!\!\!\!\!\!\!\!\!\!\!\!\!\!2 + 1$  & $ \phantom{ {A \atop A} \atop {A \atop A}} $ & $\mu_{{}_{F1}}^{D=3}$ &&&&& & & & &
     \\
     \hline
  \end{tabular}
  }
\end{center}
\begin{quote}
\footnotesize
\noindent {\bf Table B. The Old Brane Scan} classifies the non-trivial $\mathrm{Spin}$-invariant super $(p+2)$-cocycles on super Minkowski spacetimes,
for $p \geq 1$, see Prop. \ref{TheOldBraneScan} for details.
Via the associated Green--Schwarz-type Lagrangian densities \eqref{GreenLagrangian} these cocycles correspond to those fundamental super $p$-branes
propagating in $D$-dimensional super spacetimes,
that do not carry (higher) gauge fields on their worldvolume. The completion of the old brane scan to the remaining
branes and various further details is the \emph{fundamental brane bouquet},
parts of which is shown in \hyperlink{Figure3}{Figure 3}.
\end{quote}

\medskip
\noindent This scan does discover the brane species that have been argued for by various other means.
For example, the entry $\mu_{{}_{M2}}$ at $D = 11$, $p = 2$ in the old brane scan reflects the existence of the fundamental super membrane which was mentioned above, now known as the M2-brane, whose Green--Schwarz Lagrangian \eqref{GreenLagrangian} we see below in the proof of Prop. \ref{GreenSchwarzFromSuperVolumeForm}.
Similarly, the old brane scan shows that and why there
is a fundamental superstring, Example \ref{FundamentalF1Cocycle}, as well as a 5-brane propagating in 10d super spacetime, just as earlier found with rather different methods.
This suggests that super Lie algebra cohomology might be part of the missing mathematical formulation of the elusive
foundations of string/M-theory.
But this is a partial success only: a glorious insight associated with the ``second superstring revolution'' says that there
should be \emph{more} brane species than the old brane scan shows, and that jointly the enlarged system of brane species
in various super spacetimes exhibits subtle equivalence relations known as \emph{dualities}.

\medskip
In \cite{FSS13} it was pointed out that one may retain the foundational promise of the old brane scan, while improving it to
include also all these further brane species, if one passes from super Lie algebras to their homotopy theoretic incarnation,
called super \emph{strong homotopy Lie algebras} or super \emph{$L_\infty$-algebras}, for short.
In fact, this generalization emerges naturally from a closer look at the nature of cohomology on ordinary (super) Lie algebras:
it is a familiar fact that every \emph{2-cocycle} $\mu_2$ on a (super) Lie algebra classifies a central extension.
For example the type IIA superspacetime carries a Spin-invariant 2-cocycle $\mu_{{}_{D0}} = \overline{\psi}\Gamma^{10} \psi$, whose central extension
is $D = 1$, $\mathcal{N }= 1$ super Minkowski spacetime:
\vspace{-2mm}
$$
  \xymatrix@R=10pt{
    \mathbb{R}^{10,1\vert \mathbf{32}}
    \ar[dd]_-{ \mbox{ \tiny \begin{tabular}{c} central extension \\ by $\mu_{D0} = \overline{\psi}\Gamma^{10} \psi$ \end{tabular}  } }
    \\
    \\
    \mathbb{R}^{9,1\vert \mathbf{16} + \overline{\mathbf{16}}}\;.
  }
$$
Since $\mu_{D0}$ is the cocycle that defines the fundamental super D0-brane, and since informal string theory folklore has it that the
 above extension may be
understood the \emph{condensation} of D0-branes, it is natural to ask whether the other cocycles in the old brane scan correspondingly define extensions of sorts.

\medskip
In order to see how this could work, one observes that in the rational super homotopy theory,
a 2-cocycle as above is equivalently
a map, namely a map of the form  $\mathbb{R}^{9,1\vert \mathbf{16} + \overline{\mathbf{16}}} \overset{\mu_2}{\longrightarrow} B \mathbb{R}$,
and for every map in homotopy there is the corresponding homotopy fiber. Inspection shows that for a 2-cocycle, this is just the corresponding
central extension
\vspace{-4mm}
$$
  \xymatrix@R=10pt{
    \mathbb{R}^{10,1\vert \mathbf{32}}
    \ar[dd]_-{ \mbox{ \tiny \begin{tabular}{c} homotopy fiber \\ of $\mu_2 = \overline{\psi}\Gamma^{10} \psi$  \end{tabular} } }
    \\
    \\
    \mathbb{R}^{9,1\vert \mathbf{16} + \overline{\mathbf{16}}}
    \ar[rr]^-{ \mu_2 = \overline{\psi}\Gamma^{10}\psi }
    &&
    B^2 \mathbb{R}\;.
  }
$$
With this perspective, now it is clear what the generalization is: for instance, given the string cocycle
$\mu_{F1} = \tfrac{i}{2}\overline{\psi}\Gamma_{a_1 a_2} \psi \wedge e^{a_1}\wedge e^{a_2}$, its
homotopy fiber is not a super Lie algebra anymore, but a higher super Lie algebra, namely a \emph{super Lie 2-algebra} (see \cite{BaezHuerta11} for exposition).
Following established terminology for the bosonic analogue of this construction (see \cite[appendix]{FSS12} for details and further pointers),
this is called the \emph{superstring Lie 2-algebra}, and denoted $\mathfrak{string}_{\mathrm{IIA}}$.
\vspace{-2mm}
$$
  \xymatrix@R=10pt{
    \mathfrak{string}_{\mathrm{IIA}}
    \ar[dd]_-{ \mbox{\tiny \begin{tabular}{c} homotopy fiber \\ of $\mu_{F1}$
      \end{tabular} } }
    \\
    \\
    \mathbb{R}^{9,1\vert \mathbf{16} + \overline{\mathbf{16}}}
    \ar[rrrr]^-{\mu_{{}_{F1}} = \tfrac{i}{2}(\overline{\psi} \Gamma_{a} \psi) \wedge e^{a} }
    &&&&
    B^3 \mathbb{R}\;.
  }
$$
But, of course, the concept of super Lie algebra cohomology generalizes from plain super Lie algebras to
super $L_{\infty}$-algebras. Hence we may now ask whether the \emph{higher} extension of
super-spacetime by the super string Lie 2-algebra carries further Spin-invariant cohomology classes.
And indeed it does carry non-trivial Spin-invariant cocycles precisely for all the previously missing branes, namely
super D-branes of type IIA:
\vspace{-2mm}
$$
  \xymatrix@R=10pt{
    \mathfrak{d}2p\mathfrak{brane}
    \ar[dd]_-{ \mbox{ \tiny \begin{tabular}{c} homotopy fiber \\ of $\mu_{{}_{D(2p)}}$  \end{tabular}  } }
    \\
    \\
    \mathfrak{string}_{\mathrm{IIA}}
    \ar[rr]^{\mu_{D(2p)}}
    &&
    B^{2p + 2}\R\;.
  }
$$
As before, these cocycles are reflected in the homotopy fibers that they induce, which are now super Lie $2p+1$-algebras, which, following the emerging pattern,
we denote by $\mathfrak{d}2p\mathfrak{brane}$.

\medskip
By proceeding in this fashion, one finds that as soon as super spacetime is regarded in super homotopy theory,
there is, potentially, a whole \emph{bouquet} of iterative invariant higher central extensions emerging from it,
each corresponding to a fundamental brane species. In fact, as we saw with the $\mathrm{D0}$-brane
cocycle at the beginning, super-spacetime itself may emerge from the existence of super 0-branes
this way. This naturally leads one to look for a possible ``root'' of the fundamental brane bouquet.
The simplest non-trivial super spacetime is the $D = 0$, $\mathcal{N} = 1$ super spacetime, also
known as the \emph{superpoint}. Regarding the superpoint in super homotopy theory,
the bouquet of higher invariant central extensions that emerges out of it may be shown to be
the completion of the old brane scan to the full \emph{fundamental brane bouquet}.
Parts of this are shown in \hyperlink{Figure3}{Figure 3}.

\vspace{-5mm}

{\small
\hypertarget{Figure3}{$\,$}
$$
  \hspace{-.7cm}
  \xymatrix@R=1em@C=.9em{
    &
    &&&& \mathfrak{m}5\mathfrak{brane}
     \ar[dd]
    &&&&
    \mbox{ \cite{cohomotopy} }
    \\
    &
    \mbox{ \cite{Higher-T, PLB} }
    &&
    \mathbb{R}^{10,1\vert \mathbf{32}}_{\mathrm{exc},s}
    \ar[rrd]|-{\mathrm{comp}}
    \ar@/_1.1pc/[dddll]
    \ar@{<..>}@(ul,ur)[]^{ \mbox{ \color{blue} \begin{tabular}{c} Higher \\ T-duality \end{tabular}  } }
    && \;\;\;\;
    \ar@{<..>}@/^1.1pc/[drr]|{ \mbox{ \color{blue} \begin{tabular}{c} M/IIA \\ Duality  \end{tabular} } }
    &&
    &&
    \cite{FSS16a}
    \\
    &
    &&
     && \mathfrak{m}2\mathfrak{brane}
    \ar[dd]
    &&
    &
    \\
    &
    &
    \mathfrak{d}5\mathfrak{brane}
    \ar[ddr]
    &
    \mathfrak{d}3\mathfrak{brane}
    \ar[dd]
    &
    \mathfrak{d}1\mathfrak{brane}
    \ar[ddl]
    &
    & \mathfrak{d}0\mathfrak{brane}
    \ar@{}[ddd]|{\mbox{\tiny (pb)}}
    \ar[ddr]
    \ar@{..>}[dl]
    &
    \mathfrak{d}2\mathfrak{brane}
    \ar[dd]
    &
    \mathfrak{d}4\mathfrak{brane}
    \ar[ddl]
    \ar@<+22pt>@{..>}@/^2.1pc/[dd]|{ \mbox{ \color{blue} \begin{tabular}{c}  Gauge \\  enhancement \\  \end{tabular} } }
    \\
    &
    \mathbb{R}^{10,1\vert \mathbf{32}}_{\mathrm{exc}}
    \ar[dddddd]
    &
    \mathfrak{d}7\mathfrak{brane}
    \ar[dr]
    &
    &
    \mbox{
      \color{blue}
      \begin{tabular}{c}
        brane \\ bouquet
      \end{tabular}
    }
    & \mathbb{R}^{10,1\vert \mathbf{32}}
      \ar[ddr]
    &&&
    \mathfrak{d}6\mathfrak{brane}
    \ar[dl]
    &
    \\
    &
    &
    \mathfrak{d}9\mathfrak{brane}
    \ar[r]
    &
    \mathfrak{string}_{\mathrm{IIB}}
    \ar[dr]
    &
    \mbox{\cite{FSS13}}
    & \mathfrak{string}_{H}
      \ar[d]
    &&
    \mathfrak{string}_{\mathrm{IIA}}
    \ar[dl]
    &
    \mathfrak{d}8\mathfrak{brane}
    \ar[l]
    &
    \mbox{ \cite{GaugeEnhancement} }
    \\
    &
    &
    &
    &
    \mathbb{R}^{9,1 \vert \mathbf{16} + {\mathbf{16}}}
    \ar@{<-}@<-3pt>[r]
    \ar@{<-}@<+3pt>[r]
    & \mathbb{R}^{9,1\vert \mathbf{16}} \ar[dr]
    \ar@<-3pt>[r]
    \ar@<+3pt>[r]
    &
    \mathbb{R}^{9,1\vert \mathbf{16} + \overline{\mathbf{16}}}
    \\
    &
    &
    &
    &
    &
    \mathbb{R}^{5,1\vert \mathbf{8}}
    \ar[dl]
    &
    \mathbb{R}^{5,1 \vert \mathbf{8} + \overline{\mathbf{8}}}
    \ar@{<-}@<-3pt>[l]
    \ar@{<-}@<+3pt>[l]
    \\
    &
    &
    &
    &
    \mathbb{R}^{3,1\vert \mathbf{4}+ \mathbf{4}}
    \ar@{<-}@<-3pt>[r]
    \ar@{<-}@<+3pt>[r]
    &
    \mathbb{R}^{3,1\vert \mathbf{4}}
    \ar[dl]
    &
    \!\!\!\!\!\!\!
    \mbox{
      \color{blue}
      \begin{tabular}{c}
        emergent
        \\
        spacetime
      \end{tabular}
    }
    \!\!\!\!\!\!\!
    &
    &&
    \mbox{ \cite{HuertaSchreiber17} }
    \\
    &
    &
    &
    &
    \mathbb{R}^{2,1 \vert \mathbf{2} + \mathbf{2} }
    \ar@{<-}@<-3pt>[r]
    \ar@{<-}@<+3pt>[r]
    &
    \mathbb{R}^{2,1 \vert \mathbf{2}}
    \ar[dl]
    \\
    &
    \mathbb{R}^{0\vert 32}
    \ar@{<-}@<-4pt>[rrr]
    \ar@{<-}@<+4pt>[rrr]^-{\vdots}
    &
    &
    &
    \mathbb{R}^{0 \vert \mathbf{1}+ \mathbf{1}}
    \ar@{<-}@<-3pt>[r]
    \ar@{<-}@<+3pt>[r]
    &
    \mathbb{R}^{0\vert \mathbf{1}}
    \\
    \\
    & \fbox{Exceptional} && \fbox{Type IIB} \ar@{<..>}@/_2pc/[rrrr]|{\mbox{\color{blue}T-Duality}} && \fbox{Type I} && \fbox{Type IIA} && \mbox{ \cite{FSS16b} }
  }
$$
}
{\footnotesize
\noindent {\bf {Figure 3}. The fundamental brane bouquet \cite{FSS19}.} The rational homotopy theory of superspaces (Def.\ \ref{RationalSuperHomotopyTheory}) contains
a god-given object: the superpoint $\mathbb{R}^{0\vert 1}$. The diagram shows part of the web of higher rational superspaces that appears
when, starting with the superpoint, one iteratively applies the operations
of 1) doubling fermions and 2) passing to higher extensions invariant with respect to automorphisms modulo R-symmetry.
What appears are, first, the super-Minkowski spacetimes (Def.\ \ref{ExamplesOfSuperMinkowskiSpacetimes}) of the dimensions shown (see Example \ref{ExamplesOfSuperMinkowskiSpacetimes}).
These carry invariant 3-cocycles that correspond to the various species of \emph{fundamental string} (i.e. the fundamental 1-branes)
in the way reviewed above in Sec. \ref{TheFundamentalBraneScan}. The higher extensions classified by these string-cocycles, shown by the name of the
corresponding string species in the diagram, carry, in turn, further higher cocycles, these now corresponding to the D-branes and the M2-brane in their incarnation as
\emph{fundamental} or \emph{probe} branes. This process climbs up to a cocycle for the fundamental M5-brane on the higher extension classified by the cocycle
for the fundamental M2-brane on 11d super-Minkowski spacetime. By homotopical descent \cite{cohomotopy}, this is equivalently the datum of a single 4-sphere valued cocycle
on 11d super-Minkowski spacetime: the unified M2/M5-brane cocycle in rational super cohomotopy, from Prop. \ref{M2M5SuperCocycle}:
$$
  \xymatrix{
    \mathbb{R}^{10,1\vert \mathbf{32}}
    \ar[rr]^-{\mu_{{}_{M2/M5}}}
    &&
    S^4
  }.
$$
}

The brane bouquet climbs up to the fundamental membrane on $11d$ superspacetime, and then exhibits the
emergence of a further 5-brane on top of that.
By homotopical descent, as explained in detail in \cite{cohomotopy}, these two iterative higher central extensions unify to a single cocycle
on 11d super-spacetime, albeit no longer in ordinary cohomology, but in cohomotopy (Example \ref{ExamplesOfCohomologyTheories}), as controled by (at least the rational image of) the complex Hopf fibration (Def.\ \ref{HopfFibration}):
\begin{equation}
  \label{TheM2M5CocycleFromDescent}
  \raisebox{40pt}{
  \xymatrix@C=9pt@R=1.2em{
    \mathfrak{m}2\mathfrak{brane}
    \ar[dd]
    \ar[rr]^{\mu_{{}_{M5}}}
    &&
    S^7
    \ar[dd]^{ \mbox{ \footnotesize \begin{tabular}{c}  quaternionic \\ Hopf fibration \end{tabular}  } }
    \\
    \\
    \mathbb{R}^{10,1\vert \mathbf{32}}
    \ar[dr]_{\mu_{{}_{M2}}}
    \ar@{-->}[rr]^{ \mu_{{}_{M2/M5}} }
    &&
    S^4
    \ar[dl]
    \\
    & B \mathrm{SU}(2)_{\mathbb{R}}
  }
  }
  \phantom{A}
  \in \;
  \mathrm{Ho}\left( \mathrm{SuperSpaces}_{\mathbb{R}} \right).
\end{equation}
This is the fundamental M2/M5-cocycle $\mu_{{}_{M2/M5}}$, that in Prop. \ref{M2M5SuperCocycle} is the starting point for
our discussion of equivariant enhancement of fundamental brane cocycles.


\subsection{The black brane scan}
\label{MBraneInterpretation}

{\small
\begin{floatingtable}[r]{
\begin{tabular}{|c|c|}
  \hline
  \begin{tabular}{c}
    \bf Black brane
    \\
    \bf species
  \end{tabular}
  &
  \bf Example
  \\
  \hline
  \hline
  $\mathrm{MO9}$ & \ref{TheMO9}
  \\
  \hline
  $\mathrm{MO5}$ & \ref{TheBlackM5}
  \\
  \hline
  $\mathrm{MO1}$ & \ref{TheMWave}
  \\
  \hline
  \hline
  $\mathrm{M2}$ & \ref{TheBlackM2}
  \\
  \hline
  $\mathrm{MK6}$ & \ref{TheMK6}
  \\
  \hline
  \hline
  $\mathrm{M5}_{\mathrm{ADE}}$ & \ref{ADEM5}
  \\
  \hline
  $\tfrac{1}{2}\mathrm{M5}$ & \ref{TheBlackNS5}
  \\
  \hline
  $\mathrm{M1}$ & \ref{TheSelfDualString}
  \\
  \hline
  $\mathrm{NS1}_H$ & \ref{TheBlackNS1H}
  \\
  \hline
\end{tabular}
}
\end{floatingtable}
}
Here we informally recall aspects of the story of \emph{black branes}, and
then provide a list of commented pointers to statements in the
string theory literature, that serve to support  the interpretation of
Theorem \ref{SuperADESingularitiesIn11dSuperSpacetime} and Theorem
\ref{RealADEEquivariantEndhancementOfM2M5Cocycle} as
providing a precise definition and \emph{black brane scan}-classification of
bound states  of black and fundamental branes in M-theory.


\medskip

There is a curious analogy between fundamental particles and gravitional singuralities: the \emph{black hole uniqueness theorems} (``no hair theorems'')
of general relativity (see \cite{Mazur01, HollandsIshibashi12})
state
that isolated black hole spacetimes in equilibrium are completely characterized by just a handful of parameters; namely their mass, charge and spin (angular momentum).
These are of course also the \emph{quantum numbers} that characterize fundamental particles. Since, moreover, a black hole is a homogeneous spacetime,
except for a \emph{pointlike} singularity (or rather the 1+1-dimensional worldline of a point removed from spacetime, where it would become singular, if the point
were to be included),
it is natural to wonder if there is a secret relation between fundamental particles and black holes (see \cite[Sec. 5]{Duff99B}).

\medskip
In superstring theory this analogy becomes stronger: on the one hand, there is a zoo of \emph{fundamental $p$-branes} for various values of $p$ and in various spacetime dimensions,
as discussed above. On the other hand, the equations of motion of supergravity in these dimensions admit homogeneous solutions
which are much like black holes in 4d gravity, but whose singular locus is $p + 1$-dimensional, for specific values of $p$; these are called
\emph{black $p$-brane} solutions of supergravity \cite{DuffLu92, Gueven92, DuffLu93, DuffKhuriLu95}. Strikingly, one finds that, essentially, for each fundamental $p$-brane sigma-model there is a corresponding
black $p$-brane solution of supergravity which shares the same few defining parameters.

\medskip
In particular, one may therefore consider the
quantum fluctuations of fundamental $p$-branes that are aligned close to the singularity of their own black $p$-brane analogue: the result are
conformal field theories of fluctuations on asymptotically anti-de Sitter spacetimes (\cite{BlencoweDuff88, DuffSutton88}, see \cite[Sec. 5]{Duff99L}, \cite{PastiSorokinTonin99}).
This most intimate relation between fundamental $p$-branes and black $p$-branes has (later) come to be famous as the \emph{AdS/CFT correspondence}
(see \cite[Sec. 6]{Duff99B}).

\medskip
For these reasons, much of the informal literature in string theory terminologically blurs the distinction between fundamental $p$-branes and black $p$-branes,
tacitly anticipating a working \emph{M-theory} where it should make sense to, somehow, closely relate
macroscopic solutions of classical supergravity with fundamental quantum objects.
While all evidence indeed points to there being a unified perspective on these two phenomena, precise details on the conceptual relation have been emerging only gradually.
The precise unification of fundamental and black $p$-brane aspects in Theorem \ref{RealADEEquivariantEndhancementOfM2M5Cocycle}, mediated via
the discussion in Sec. \ref{FundamentalBraneCocyclesOnSystemsOfIntersectingBlackBranes}, could serve to clarify the situation.

\medskip
Concretely, asymptotically close to their horizon, the $\gt \sfrac{1}{4}$-BPS black $p$-brane spacetimes are all Cartesian products of an
anti-de Sitter spacetime with a free discrete quotient of the sphere around the singularity, such that
the result is a warped metric cone over the $p$-brane singularity, as shown here (\cite{FF98, MFFGME10}):

\begin{equation}
\label{NearHorizonGeometry}
\scalebox{.9}{
\mbox{
\begin{tabular}{|c|cccccc|}
   \hline
   \multirow{2}{*}{
   \begin{tabular}{c}
     Near horizon
     \\
     spacetime
   \end{tabular}
   }
   &
   \multicolumn{3}{c}{
     \raisebox{6pt}{\tiny
     \begin{tabular}{c}
        anti-de Sitter
\\        spacetime
     \end{tabular}
   }
   }
   &&
   &
   \\
    &
    \multicolumn{3}{c}{
     $\mathrm{AdS}_{p+2}$
    }
     &
     $\!\!\!\!\times\!\!\!\!$
     &
     $S^{D-p-2}$ &  \!\!\!\!\!\!\!\!\!\!\!\! $/G$
    \\
    &
    \multicolumn{3}{c}{
      $\overbrace{ \phantom{---------------} }$
    }
    &&&
    \\
    \begin{tabular}{c}
      Metric in
      \\
      horospheric coord.
    \end{tabular}
    &
    $\frac{R^2}{z^2} ds^2_{\mathbb{R}^{p,1}}$
    &
    $\!\!\!\!+\!\!\!\!$
    &
    $\frac{R^2}{z^2} dz^2$
    &
    $\!\!\!\!+\!\!\!\!$
    &
    $ d s^2_{S^{D-p-2}} $
    &
    \\
    \multirow{4}{*}{
    \begin{tabular}{c}
      Causal
      \\
      chart
    \end{tabular}
    }
    &{\tiny
    \begin{tabular}{c}
      singularity
    \end{tabular}
    }
    &&
    \raisebox{6pt}{\tiny
    \begin{tabular}{c}
      radial
      \\
      direction
    \end{tabular}
    }
    &&
   \raisebox{14pt}{\tiny
   \begin{tabular}{c}
     sphere
     around
     \\
     singularity
   \end{tabular}
   }
   &
    \\
    &
    $\mathbb{R}^{p,1}$ & $\!\!\!\!\times\!\!\!\!$ & $\mathbb{R}_+$  & $\times$&  $S^{D-p-2}$ &  \!\!\!\!\!\!$/G$
   \\
    && &
    \multicolumn{3}{c}{
      $ \underbrace{\phantom{-------------}} $
    }
    &
   \\
    && &
    \multicolumn{3}{c}{\tiny
      transversal space $\mathbb{R}^{D-p-2} \setminus \{0\}$
    }
    &
   \\
   \begin{tabular}{c}
     Metric in
     \\
     natural coord.
   \end{tabular}
   &
   $\frac{r^n}{\ell^n} \, d s^2_{\mathbb{R}^{p,1}}$
   &
   $\!\!\!\!+\!\!\!\!$
   &
   $\frac{\ell^2}{r^2} \, d r^2$
   &$+$ &
   $\ell^2 \, d s^2_{S^{D-p-2}}$
   &
   \\
   &&&
   \multicolumn{3}{c}{
     $\underbrace{ \phantom{ --------------- } }$
   }
   &
   \\
   & {\tiny $\frac{r^n}{\ell^n} \cdot$ Minkowski}
   &&
   \multicolumn{3}{c}{\tiny
     $\frac{\ell^2}{r^2} \cdot $
     metric cone $\underset{\phantom{a}}{C( S^{D-p-2} ) \setminus \{0\}}$
   }
   &
   \\
   \hline
\end{tabular}
}
}
\end{equation}
\begin{quote}
\footnotesize
 {\hspace{4cm} \bf Table C. $D$-dimensional black $p$-brane spacetimes}
\end{quote}

\medskip
\noindent The standard causal chart shown in the middle of
\eqref{NearHorizonGeometry}
exhibits that the tangent spaces \emph{at the singularity} (if it were to be included in the
underlying spacetime manifold) are naturally identified with
\vspace{-2mm}
$$
    T_{\mathrm{sing}} \;\simeq_{\mathbb{R}}\;
    \mathbb{R}^{p,1} \oplus
    \xymatrix{
      \mathbb{R}^{D-p-2}
      \hspace{-9mm}\ar@(ul,ur)[]^G
    }
  \;\;\;\;\;\;\;\;,
$$
where the $G$-action fixes the origin in the transversal space $\mathbb{R}^{D-p-2}$, hence in total fixes the black $p$-brane singularity.
This is just the (super) tangent space-wise situation (see \eqref{SpacetimeActions}) with which Theorem \ref{SuperADESingularitiesIn11dSuperSpacetime}
is concerned.
This way, for $G \neq \{e\}$ one may think of the black $p$-brane as sitting at a \emph{conical singularity}  \cite{AFFHS98}\cite{MP}.

\medskip
Indeed, the standard computations of BPS black $p$-brane solutions in supergravity all proceed, eventually,
by reducing a computation of Killing spinor fields on a curved supergravity spacetime, to a computation of
spinors on \emph{one tangent space} that are \emph{fixed} by an involution,
or, for intersecting branes, by several involutions (e.g. \cite[(4), (8), (11)]{Gauntlett97}), just as in Prop. \ref{ClassificationOfZ2Actions} below,
or, more generally, by larger finite ADE-groups \cite{MFFGME10}, \cite[Sec. 8.3]{MF10}, as in our Prop. \ref{M2AndMK6ADE} below. We may hence view
Theorem \ref{SuperADESingularitiesIn11dSuperSpacetime} below as a converse to these observations, saying that
indeed the spectrum of black $p$-brane species is entirely determined super tangent space-wise.
We suggest that it is useful think of this as exhibiting a higher form of the paradigm of
super Cartan geometry, as indicated in Sec. \ref{Survey}, in line with the result of \cite{CL, Howe97, FOS},
that the equations of motion of supergravity themselves are implied by the super torsion constraint, hence, via \cite{Guillemin65}, by the requirement that
the infinitesimal neighborhood of every point in super spacetime looks super-metrically like the
model super-Minkowski tangent spacetime.

\medskip
Notice how the singularity itself is not actually part of spacetime in
\eqref{NearHorizonGeometry}: the singularity would be at $r = 0$, which
is excluded from the spacetime manifold, since classical (super)gravity is not defined on singular spaces, it only sees everything
right outside the singularities. But the brane that is supposed to sit there at the singularity is meant to be part of the elusive
M-theory, of which 11d supergravity is meant to be just some approximation. There is a multitude of indirect informal arguments
that in full M-theory some extra physical degrees of freedom do appear
at singularities (\cite{AcharyaWitten01}, \cite[Sec. 3]{Acharya02}, \cite{AtiyahWitten03}, see \cite{AcharyaGukov04}
for review of the folklore).
One such indirect argument we recall as Example \ref{TheMK6} below.
Notice also that some singularities that are supposed to appear in M-theory do not have a supergravity description at all,
see Example \ref{TheMO9} below.
By the nature of these arguments, it seems plausible that identifying the missing M-theoretic degrees of freedom
at these singularities goes a long way towards identifying the elusive M-theory itself.

\medskip

This basic background on black M-branes already serves to illuminate the role of the 4-sphere as coefficient
object for measuring M-brane charge: by Prop. \ref{M2M5SuperCocycle} and the discussion in Sec. \ref{TheFundamentalBraneScan}, we know that the 4-sphere has the correct rational homotopy type
for measuring fundamental M-brane charge. What governs this is really the fact \eqref{TheM2M5CocycleFromDescent}, that the 4-sphere participates in the rational image of the quaternionic Hopf fibration
$S^7 \overset{H_{\mathbb{H}}}{\longrightarrow} S^4$ (Def.\ \ref{HopfFibration}).
But notice that, by \eqref{NearHorizonGeometry}, the two spheres involved here are exactly the
unit spheres around the singularities
in the near horizon geometries of the single M2-brane and of the M5-brane:

\medskip
\begin{center}
{\small
\begin{tabular}{|c||c|c|}
  \hline
  \begin{tabular}{c}
    \bf Black
   brane
  \end{tabular}
  &
  \begin{tabular}{c}
    \bf Near
    \bf horizon
    \bf geometry
  \end{tabular}
  &
  \begin{tabular}{c}
    \
    \bf Causal
    \bf chart
  \end{tabular}
  \\
  \hline
  \hline
  $ \phantom{{A \atop A}\atop {A \atop A} } \mathrm{M2}\phantom{{A \atop A}\atop {A \atop A} } $ & $\mathrm{AdS}_4 \times S^7$ & $\mathbb{R}^{2,1} \times \mathbb{R}_+ \times S^7$
  \\
  \hline
  $\phantom{{A \atop A}\atop {A \atop A} }\mathrm{M5}\phantom{{A \atop A}\atop {A \atop A} }$ & $\mathrm{AdS}_7 \times S^4$ & $\mathbb{R}^{5,1} \times \mathbb{R}_+ \times S^4$
  \\
  \hline
\end{tabular}
}
\end{center}
Hence, in the spirit of Dirac charge quantization (see \cite[Sec. 2]{Freed00}), the coefficient space $A$ for the generalized cohomology theory (Def.\ \ref{CohomoloyFromHomotopy}), which measures the presence of units of M-brane charge, should rationally be a 4-sphere, and should in addition have homotopy groups $\pi_4(A) \simeq \mathbb{Z}$ (for measuring the integer charge carried by
M5-branes) and $\pi_7(A) \simeq \mathbb{Z}$ (for measuring the integer charge carried by M2-branes). But the evident choice for this is just the
actual 4-sphere, $A = S^4$: this measures the presence of a single M5-brane by the identity map on the 4-sphere encircling it
\begin{equation}
  \label{MeasuringM2BraneCharge}
  \xymatrix{
    \mathrm{AdS}_7 \times S^4
    \ar[r]^-{\mathrm{pr}_2}
    &
    S^4
    \ar[r]^-{\mathrm{id}}
    &
    S^4
  },
  \phantom{AAA}
  [\mathrm{id}] = 1 \in \mathbb{Z} \simeq \pi_4(S^4)
\end{equation}
and measures the presence of a single M2-brane by way of the quaternionic Hopf fibration $H_{\mathbb{H}}$ (Def.\ \ref{HopfFibration}) from the 7-sphere encircling it:
\begin{equation}
  \label{MeasuringM5BraneCharge}
  \xymatrix{
    \mathrm{AdS}_4 \times S^7
    \ar[r]^-{\mathrm{pr}_2}
    &
    S^7
    \ar[r]^-{H_{\mathbb{H}}}
    &
    S^4
  },
  \phantom{AAA}
  [H_{\mathbb{H}}] = 1 \in \mathbb{Z} \simeq \pi_7(S^4)
  \,.
\end{equation}
That the M-brane charge should take values in degree-4 cohomotopy this way was first proposed in \cite[Sec. 2.5]{S-top}.

\begin{remark}[Origin of group actions on the 4-sphere as the cohomology theory for M-branes]
 \label{OriginOfGroupActionsOn4Sphere}
We may generalize the above reasoning to the presence of conical singularities, and thereby motivate the group actions
on the 4-sphere that we consider in Sec. \ref{ADEEquivariantRationalCohomotopy} below.
\item {\bf (i)} By \cite{MFFGME10} (see Prop. \ref{M2AndMK6ADE} below), the near horizon geometry of the
$\geq \sfrac{1}{4}$-BPS black M2-brane is $\mathrm{AdS}_4 \times S^7/G_{\mathrm{ADE}}$,
where $G_{\mathrm{ADE}}$ acts along the twisted diagonal, via the identification
$$
  \xymatrix{
    S^7
    \ar@(ul,ur)[]^{ \mathrm{SU}(2)_L  \times \mathrm{SU}(2)_R }
  }
    \hspace{-2mm}\simeq\;
  S(\hspace{-3mm}
    \xymatrix{
      \mathbb{H}
      \ar@(ul,ur)[]^{ \mathrm{SU}(2)_L }
    }
    \hspace{-2mm}
    \oplus
\hspace{-2mm}
    \xymatrix{
      \mathbb{H}
      \ar@(ul,ur)[]^{ \mathrm{SU}(2)_R }
    }
  \hspace{-3mm})\,,
$$
where both copies of $\mathrm{SU}(2)$ act by their defining representation on $\mathbb{H} \simeq_{\mathbb{R}} \mathbb{C}^2$; see \eqref{SU2ActionsOnQuaternions}.
\item {\bf (ii)} Hence, in order for the quaternionic Hopf fibration \eqref{MeasuringM5BraneCharge} to also measure the charge of M2-branes at such singularities, we need
an $\mathrm{SU}(2)_L \times \mathrm{SU}(2)_R$-action on $S^4$ which makes the quaternionic Hopf fibration be an equivariant map \eqref{GEquivarianceMap}
(see Remark \ref{EquivarianceCommutingDiagram} below on notation):
$$
  \xymatrix@=4em{
    S^7
    \ar@(ul,ur)[]^{ \mathrm{SU}(2)_L  \times \mathrm{SU}(2)_R }
    \ar[rr]^{H_{\mathbb{H}}}
    &&
    S^4
    \ar@(ul,ur)[]^{ \mathrm{SU}(2)_L  \times \mathrm{SU}(2)_R }
  }
  \hspace{-9mm}.
$$
By the explicit formula \eqref{HopfMap} for $H_{\mathbb{H}}$, this is the case precisely for the action
$$
  S^4 \;\simeq\; S( \mathbb{R} \oplus \!\!\!\!\!\!\!\!\!\!\!\!\!\!\!
  \xymatrix{ \mathbb{H} \ar@(ul,ur)[]^{ \mathrm{SU}(2)_L  \times \mathrm{SU}(2)_R } }\!\!\!\!\!\!\!\!\!\!\!\!\!\!\! )
  \,,
$$
where $\mathrm{SU}(2)_L$ acts on $\mathbb{H}$ by quaternion multiplication from the left, while $\mathrm{SU}(2)_R$ acts by quaternion multiplication
with the inverse from the right \eqref{SU2ActionsOnQuaternions}. These are the $\mathrm{SU}(2)$-actions on $S^4$ which we consider
in Def.\ \ref{SuspendedHopfAction} below.

\item {\bf (iii)} This is consistent also with the charge carried by M5-branes at singularities:
when the element acting from the right is trivial, then the remaining quotient by finite subgroups of $\mathrm{SU}(2)$ acting from the left
provides precisely the near horizon geometry $\mathrm{AdS}_7 \times S^4/(G_{\mathrm{ADE}})_{L}$ of black M5-branes at singularities;
see
Example \ref{TheBlackNS5}.

\item {\bf (iv)} Similarly, when the black M5-brane is situated at a Ho{\v r}ava-Witten $\mathbb{Z}_2$-singularity (Example \ref{TheBlackNS5}), then its near horizon geometry
is $\mathrm{AdS}_7 \times S^4\sslash(\mathbb{Z}_2)_{\mathrm{HW}}$, where the involution (Example \ref{Z2ActionsAreInvolutions}) acts by reflection of one of the coordinates
$$
  S^4 \;\simeq\; S(\mathbb{R}^4 \oplus \!\!\!\!\!\!\!\xymatrix{ \mathbb{R}  \ar@(ul,ur)[]^{ (\mathbb{Z}_2)_{\mathrm{HW}} } } \!\!\!\!\!\!)
  \,.
$$
Hence, by the same reasoning as before, the correct coefficient object to measure M-brane charge at Ho{\v r}ava-Witten singularities is again the 4-sphere,
now equipped with this $\mathbb{Z}_2$-action. This is what we consider in Def. \ref{RealADEActionsOnThe4Sphere}.

\item {\bf (v)} In conclusion, this says that the correct \emph{coefficient} space for the cohomology theory measuring M-brane charge is essentially identified
with the 4-sphere in \emph{spacetime} around a black M5-brane, as indicated in the figures on p. \pageref{On4Sphere}.
This way, analysis of the near horizon geometry of black M-branes supports the suggestion that the correct generalized cohomology theory
measuring M-brane charge is at least closely related to equivariant cohomotopy in degree 4.
\end{remark}

In order to substantiate that equivariant enhancement of fundamental brane cocycles is a plausible candidate for the M-theoretic degrees of freedom
that are ``hidden'' at spacetime singularities,
and that  Theorem \ref{RealADEEquivariantEndhancementOfM2M5Cocycle} below may reasonably be regarded as providing
a cohomological \emph{black brane scan} for branes at singularities (in fact a unified fundamental-and-black brane scan), we now walk through selected discussions in the literature, of black $p$-branes at singularities, and expand on how to match them to the equivariant cocycle data
found in Sections \ref{ADESingularitiesInSuperSpacetime}, \ref{ADEEquivariantRationalCohomotopy},  and \ref{ADEEquivariantMBraneSuperCucycles}, via this kind of translation.


\medskip

\noindent {\bf The $\geq \sfrac{1}{4}$-BPS branes.} In the following
we compare the items in the classification of simple super singularities from
Theorem \ref{SuperADESingularitiesIn11dSuperSpacetime} to the literature.

\subsubsection{The MO9}
\label{TheMO9}

The item denoted ``MO9'' in Theorem \ref{SuperADESingularitiesIn11dSuperSpacetime} is of course readiliy
identified with the $\mathbb{Z}_2$-fixed locus of Ho{\v r}ava-Witten theory \cite{HoravaWitten96a} \cite{HoravaWitten96b},
whence our notation ``$G_{\mathrm{HW}}$'' for the corresponding group action. The characteristic relation \cite[eq. (2.2)]{HoravaWitten96a}
is of course the content of Lemma \ref{1braneFixed}.
It is, however, noteworthy that the nature of the Ho{\v r}ava-Witten fixed locus among the other M-branes had been unclear,
not the least because, due to its singular nature, it is not a BPS soluion of 11-dimensional supergravity; see
the beginning of \cite{BeSc98}, where the term ``M9-brane'' for this object was first suggested. A clear identification of
the role of the MO9 among the other branes,
in its appearance as the O8-plane in type I string theory, is in \cite[Sec. 3]{GKST01}, see also Example \ref{TheBlackNS5}.

\medskip
Now, the point of \cite{HoravaWitten96a} \cite{HoravaWitten96b} is to argue that the worldvolume of the MO9
is to be identified with the spacetime that the heterotic string propagates in, and that, somehow, that
heterotic string is also to be identified with the boundary of the M2-brane ending on the MO9; see also Example \ref{TheBlackNS1H}.
That story evidently matches the data in the equivariant cocycle enhancement that appears labeled
$\mathrm{M2} \underset{\mathrm{NS1}_H}{\dashv} \mathrm{MO9}$ in \hyperlink{EquivariantEnhancementsList}{Table 3} of Theorem \ref{RealADEEquivariantEndhancementOfM2M5Cocycle}.

\subsubsection{The $\mathrm{MO5}$ and $\mathrm{M5}$}
\label{TheBlackM5}

Similarly to Example \ref{TheMO9}, the computation in \cite[Sec. 2.1]{Witten95II}, characterizing an  $\mathrm{MO5}$ is
precisely that in Lemma \ref{5braneFixed},
identifying the fixed locus of a $5$-brane involution, in the sense of Def.\ \ref{pBraneInvolution}.
The conclusion in \cite[Sec. 3.]{Witten95II} is that, for anomaly cancellation, some of these black M5-branes need to sit at the
singularity of an orbifold locally of the form
$$
  \mathbb{R}^{5,1} \times ( \mathbb{R}^5 \dslash \mathbb{Z}_2 )
  \,,
$$
where $\mathbb{Z}_2$ acts by reversing all the coordinates of $\mathbb{R}^5$. This is precisely the action of the 5-brane
involution of Lemma \ref{5braneFixed}.

\medskip
Notice that the M5 at such an $\mathbb{Z}_2$-singularity is not a solution of
supergravity anymore (just as for the $\mathrm{MO9}$ in Example \ref{TheMO9}), but must be something that M-theory
needs to make sense of. While hence an ordinary black 5-brane does not/need not sit at the $\mathbb{Z}_2$-singularity,
the analysis of flux quantization conditons in \cite{Hori97} shows that \emph{if} it meets a $\mathbb{Z}_2$-singularity,
then it cannot do so just partially.

\medskip
Also notice that the M5-branes at orientation-\emph{preserving} $\mathbb{Z}_2$-singularities arise from a further intersection
with an ADE-singularity; this is the $\mathrm{M5}_{\mathrm{ADE}}$ in Example \ref{ADEM5}.

\subsubsection{The $\mathrm{MO1}$ and the M-wave}
  \label{TheMWave}

The singularity identified as MO1 in Prop. \ref{ClassificationOfZ2Actions}
is discussed as such in \cite[Sec. 3.3]{HananyKol00}.
On the other hand, the \emph{M-wave} (MW)  is well-known as
a supergravity $\sfrac{1}{2}$-BPS solution (due to \cite{Hull84}, see \cite{Philip05} for decent review),
but remains somewhat neglected in the literature on M-branes. Where it turns out to intersect with the
M2-brane in \cite[Sec. 2.2.3]{BPST10}, the authors find that ``natural, if slightly unusual'' (bottom of p. 13).

\medskip
It seems that there is no previous reference saying that
the M-wave may sit at an MO1-singularity in the same
way that the M5 may sit at an MO5 singularity (Example \ref{TheBlackM5}), and as suggested by the classification in Prop. \ref{ClassificationOfZ2Actions}.
But observe that the image of the spinor-to-vector pairing on the M-wave
has the special property that it is just one of the two light rays (this is made fully explicit in \cite[p. 94]{Philip05}).
This identifies its worldvolume structure with that of the MO1 found in \eqref{SpinorPairingOn1Brane} in Lemma \ref{1braneFixed}.

\subsubsection{The $\mathrm{M2}$}
 \label{TheBlackM2}

The article \cite{MFFGME10} gives a complete classification of $\geq \sfrac{1}{4}$-BPS black M2-brane solutions \eqref{NearHorizonGeometry}. The result of this is that these are all of the form
$$
  \mathrm{AdS}_4 \times (S^7/G_{\mathrm{ADE}})
  \,,
$$
as in \eqref{NearHorizonGeometry}, where the quotient of the 7-sphere on the right is that induced by any one of the ADE-actions on $\mathbb{R}^8 \simeq \mathbb{R}^4 \oplus \mathbb{R}^4$ that are labeled ``M2'' in
Theorem \ref{SuperADESingularitiesIn11dSuperSpacetime}, under the identification $S^7 \simeq S(\mathbb{R}^8)$.
Comparison with \eqref{NearHorizonGeometry} shows that
the actual singularity itself, if it were included in the spacetime, would be sitting at $r =0$, hence at the origin of $\mathbb{R}^8$.
That origin, of course, is precisely the fixed point set of the $G_{\mathrm{ADE}}$-action \eqref{SpacetimeActions} on $\mathbb{R}^8$.
This situation
$$
  \underset{\mbox{M2}}{
  \underbrace{
    \mathbb{R}^{2,1}
  }}
  \times
  \underset{ C(S^7) }{
  \underbrace{  (\mathbb{R}_+ \cup \{0\}) \times S^7 }}\;
  / G_{\mathrm{ADE}}
  \;\simeq\;
  \mathbb{R}^{2,1} \oplus  \mathbb{R}^8 / G_{\mathrm{ADE}}
  \,,
$$
is what is illustrated in the two items in Figures \hyperlink{Figure1}{1} and
\hyperlink{Figure2}{2} that involve M2:

\begin{center}
\raisebox{-110pt}{
\includegraphics[width=.28\textwidth]{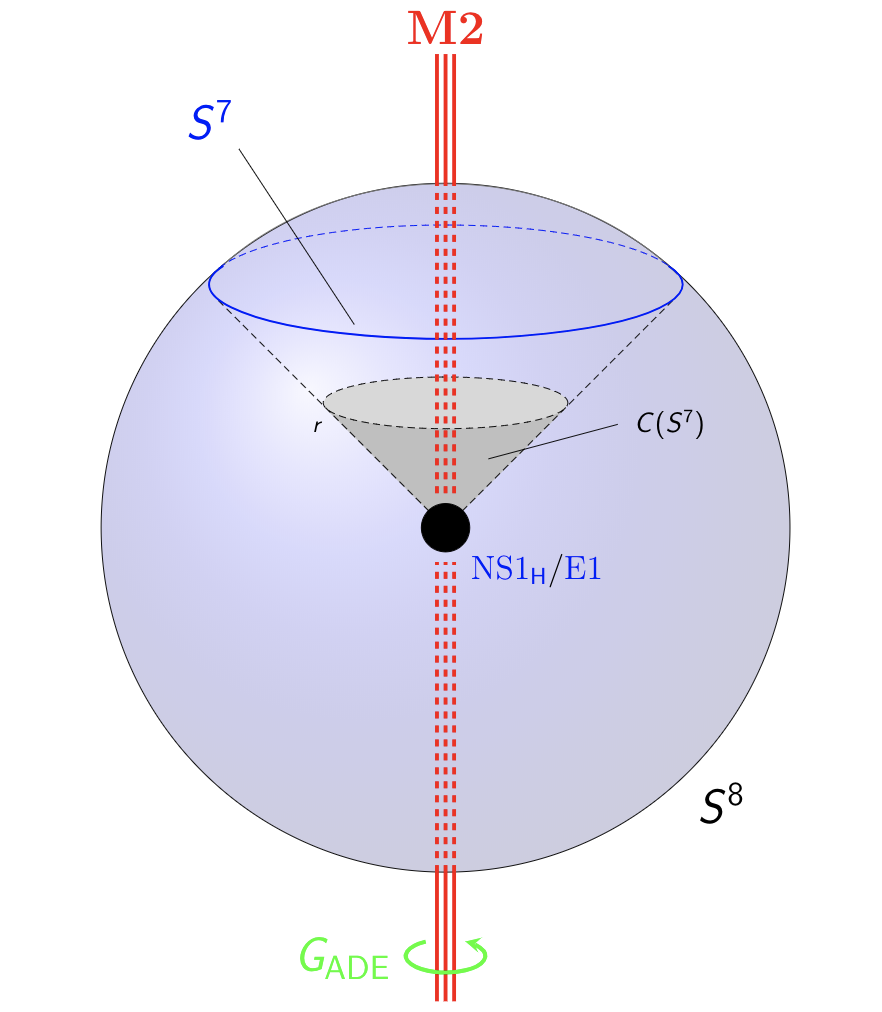}
}
\hspace{2cm}
\raisebox{-107pt}{
\includegraphics[width=.295\textwidth]{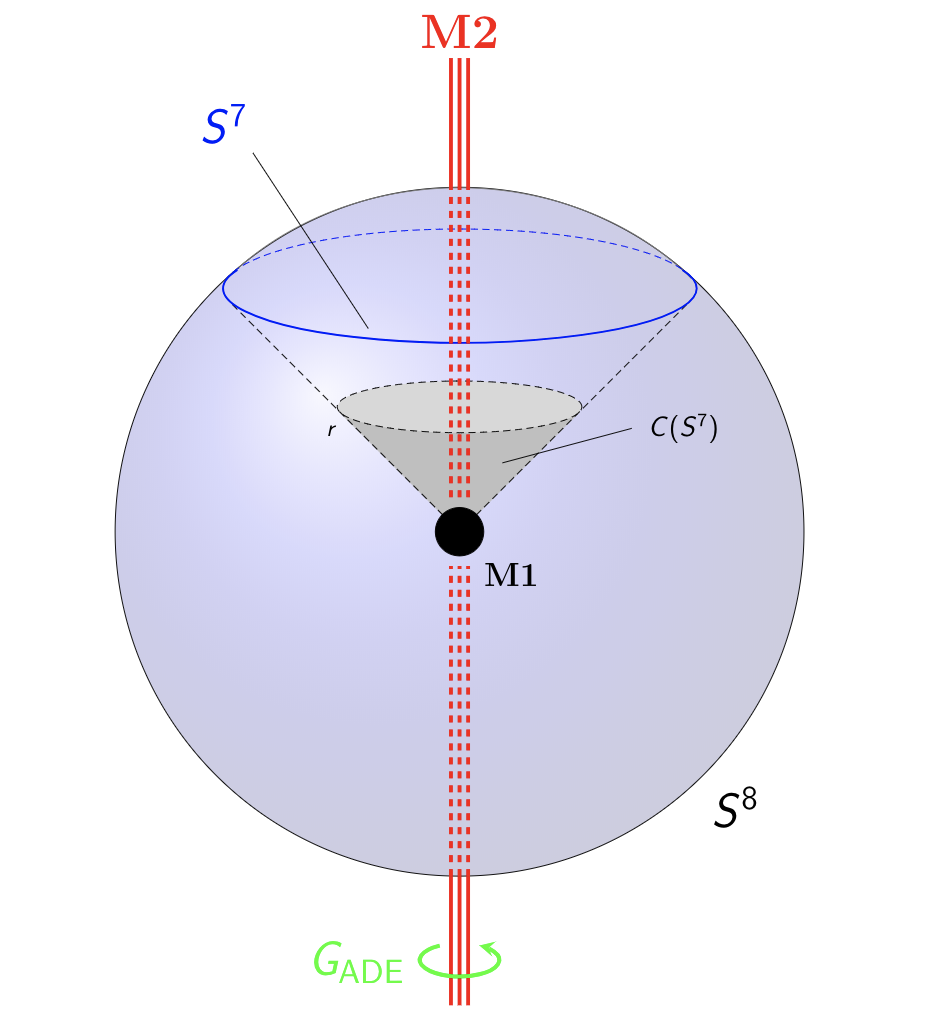}
}
\end{center}

\noindent In fact, this identification is precisely how the classification in \cite{MFFGME10} works:
using results of \cite{Wang}, the problem is first reduced to the case of spherical space forms $X_7 \simeq S^7/G$.
Now that $S^7/G$ is $\geq \sfrac{1}{4}$-BPS means equivalently that its space of
Killing spinors is at least $\sfrac{1}{4}$ of its maximally possible size.
But by a theorem of \cite{Baer}, Killing spinors on $S^7/G$ are equivalently $G$-constant spinors on the
\emph{metric cone} $C(S^7) \simeq \mathbb{R}^8$. These constant spinors, in turn, are precisely the spinorial
fixed points that appear in Theorem \ref{SuperADESingularitiesIn11dSuperSpacetime}.


\medskip
Last but not least, 3d superconformal field theories have famously been identified, which
have the properties expected of the worldvolume field theories of M2-branes, see \cite{BLMP13}.
These field theories have an ADE-classification, and inspection shows that their scalar fields
are as expected if the corresponding M2-branes sit at an ADE-singularity in the way just discussed.
(Of course this is the very motivation for the classification in \cite{MFFGME10}.)

\subsubsection{The $\mathrm{MK6}$}
 \label{TheMK6}

The \emph{Kaluza-Klein monopole} solution of plain 11-dimensional supergravity \cite{HK} is different from the other black brane
solutions, in that it does not feature a spacetime singularity. But it may be regarded as a circle fibration over the
10-dimensional type IIA super-spacetime base, and as such it \emph{is} singular, in that the circle fiber degenerates on
a $6+1$-dimensional locus. This is the 11-dimensional KK-monopole as a singular locus. Conversely, if that locus is
removed from the spacetime manifold, then on the complement
we have the total space of a non-singular fibration, in fact a principal $S^1$-bundle.

\medskip
Generally, given any $S^1$-fibration $X_{11} \to X_{10}$, one may consider the canonical inclusion of a cyclic group into
$S^1 \simeq U(1)$, as a subgroup of
roots of unity $\mathbb{Z}_n \hookrightarrow U(1)$ and hence the induced quotient bundle
$X_{11}/\mathbb{Z}_n \to X_{10}$. Since $S^1/\mathbb{Z}_n \simeq S^1$
is still a circle, just an ``$n$-times smaller'' circle, this is still a circle bundle over type IIA spacetime. Since the radius of the circle
fiber of the 11d spacetime over $X_{10}$ is supposed to be the M-theoretic incarnation of the coupling parameter of the type IIA string on $X_{10}$,
it is natural to regard such quotients, as $n$-varies. The limit where M-theory is supposed to asymptote to the perturbative type IIA superstring
would then be the limit $n \to \infty$.

\medskip
But now, if the circle fibration is actually degenerate, as it is for the Kaluza-Klein monopole, then the quotient spacetime $X_{11}/\mathbb{Z}_n$
is singular after all, with a $\mathbb{Z}_n$-singularity at the locus of the KK-monopole. One argues that from the point of view
of the type IIA string theory this configuration is the black D6-brane \cite[p. 6-7]{Townsend95}, \cite[Sec. 2]{Sen97}, \cite[p. 17-18]{AtiyahWitten03}.
Accordingly, the KK-monopole in this singular incarnation ought to be an M-brane, the MK6.

\medskip
We review the following more sophisticated (albeit still informal) argument for why the MK6-brane may occupy more general ADE-singularities,
and that there must be ``hidden M-theory degrees of freedom'' at these singularities, not seen in the supergravity approximation; namely degrees of
freedom that in the approximation of the type IIA string theory incarnate as nonabelian gauge fields on the D6-brane.

\medskip
This argument goes back to \cite[Sec. 2]{Sen97}, brief recollection may be found in \cite[Sec. 4.I]{Witten02}, \cite[Sec. 6.3.3]{IbanezUranga12}. We should amplify that, fascinating as the following
picture is, it remains a conjectural story that is waiting to be substantiated by actual mathematics of M-theory:

\medskip

Consider 11-dimensional spacetime that is locally the Cartesian product
$$
  \mathbb{R}^{6,1} \times \mathbb{C}^2 \dslash G_{\mathrm{ADE}}
$$
of 6+1-dimensional Minkowski spacetime with the orbifold quotient of the canonical action of a finite subgroup of $\mathrm{SU}(2)$
(Remark \ref{ADEGroups}). In terms of algebraic geometry, the underlying ordinary quotient $\mathbb{C}^2 / G_{\mathrm{ADE}}$
may naturally be regarded as a complex variety that is non-smooth -- hence \emph{singular} -- at the origin.
The specific singularities arising this way are known as \emph{du Val singularities} \cite{DuVal}.

\medskip
Algebraic geometry knows a canonical process of smoothing out singular points in varieties, called \emph{blowup of singularities}.
Now the blowups specifically of du Val singularities have the following striking property
(due to \cite[I, p. 1-3 (453-455)]{DuVal} see \cite{Reid} for a quick overview and \cite[Sec. 6]{Slodowy80} for a comprehensive account):

\medskip

{\bf Magic blowup property of du Val singularities.}
{\it The blowup of the du Val singularity in $\mathbb{C}^2/G_{\mathrm{ADE}}$ is a union of spheres that touch (``kiss'') each other such that
connecting the touching points by straight lines yields the Dynkin diagram given by the same ADE-label
that also classifies the group $G_{\mathrm{ADE}}$ according to the table in Remark \ref{ADEGroups}.}

\medskip

The following picture illustrates the situation of such spheres touching according to an A-type Dynkin diagram,
hence resolving the singular point of the quotient of $\mathbb{C}^2$ by the action of a cyclic group, via $\mathrm{SU}(2)$

\begin{center}
\includegraphics[width=.56\textwidth]{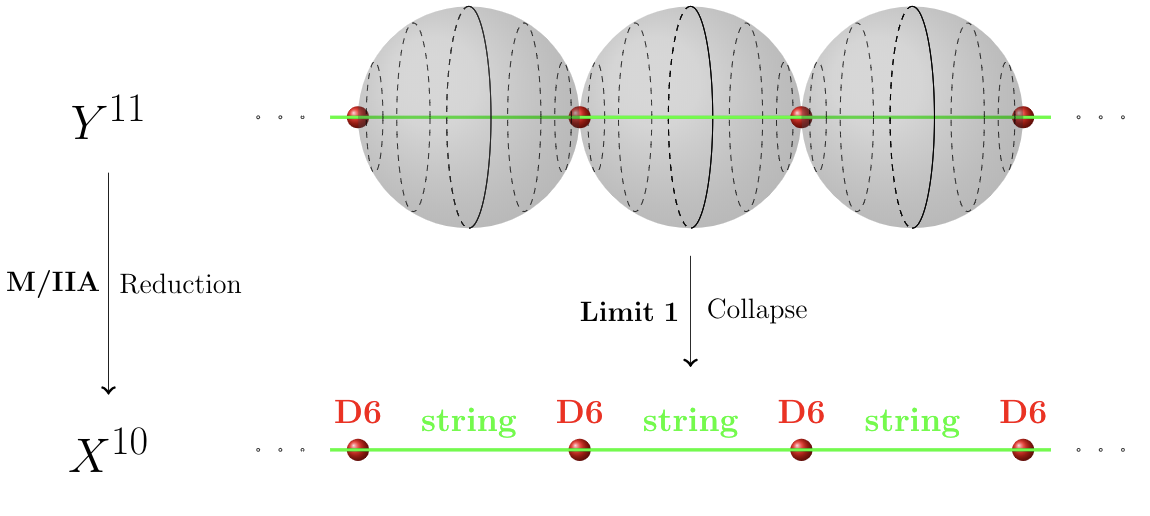}
\hspace{-1mm}
\includegraphics[width=.43\textwidth]{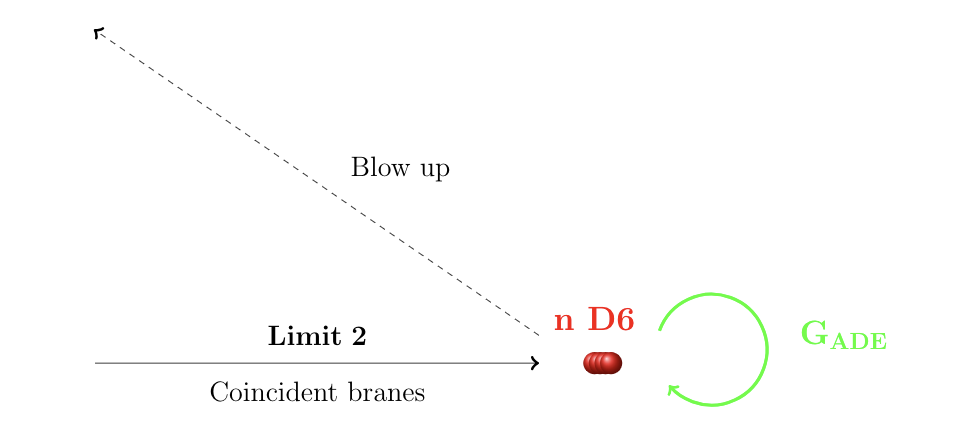}
\end{center}

So far this is the mathematics of algebraic geometry. Now the suggestion is that one interprets this situation in terms
of M-theoretic geometry:
as shown, one may think of these spheres as being $S^1$-fibrations over this Dynkin diagram, with degenerate fibers over the vertices of the diagram.
If we identify this with the M-theory circle fibration of 11-dimensional supergravity
fibered over a 10d type IIA spacetime, then this is a \emph{multi KK-monopole}-solution of 11d sugra,
with KK-monopoles centered at vertices of the Dynkin diagram.
From this perspective, one imagines that the original singularity may be understood as the result of taking
two consecutive limits, namely:
\begin{list}{}{}

\vspace{-1mm}
\item {\bf Limit 1:} first taking the type IIA-limit where the radius of the circle fibers is taken to zero;

\vspace{-2mm}
\item {\bf Limit 2:} and then taking the further limit where the vertices of the Dynkin diagram tend to coincide.
\end{list}

The string theoretic interpretation of the first limit is that the SuGra KK-monopoles becomes D6-branes in type IIA string theory,
and hence the second limit yields a configuration of $n$-coincident D6-branes. By turning this around, the original
du Val singularity must have been the M-theoretic incarnation of what in the approximation of string theory looks
like coincident D-branes: this is the black MK6-brane.

\medskip
 If one, moreover, considers M2-brane instantons wrapping these spheres, then the first of these two limits gives the
\emph{double dimensional reduction} of the M2-brane to non-perturbative strings stretching between D6-branes, and in
the second limit, where the D6-branes coincide, these D6-branes lose their tension energy and thus become \emph{massless}
perturbative strings. From perturbative string theory it is known that these are quanta of nonabelian gauge fields on the
worldvolume of the  D6-branes. Finally, if M-theory is supposed to be a refinement of string theory, a pre-image of these
nonabelian gauge field degrees of freedom must have existed already on the black MK6-brane at the singularity. This is
one incarnation of the mysterious M-theory degrees of freedom hidden at the ADE-singularity.

\medskip

\medskip

\medskip

\noindent {\bf Examples of intersecting M-branes.} Next we compare some of the examples
of intersecting singularities from Prop. \ref{NonSimpleRealSingularities} to the literature.

\subsubsection{The $\mathrm{M5}_{\mathrm{ADE}}$}
  \label{ADEM5}
  According to \cite[Sec. 3]{ZHTV15}, the configuration of an  $\mathrm{M5}$ placed inside an $\mathrm{MK6}$ is supposed to have worldvolume theory
  the $D = 6$, $\mathcal{N} = (1,0)$ superconformal QFT with the corresponding ADE-classification of its gauge field
  content.
  This corresponds to the item labeled $\mathrm{M5}_{\mathrm{ADE}}$ in Prop. \ref{NonSimpleRealSingularities}, as shown in \hyperlink{Figure1}{Figure 1}.
  Notice that this means that the M5-brane at an ADE-singularity necessarily sits inside a larger 6-brane, which it thereby
  ``divides in half''; see also
  the discussion in Example \ref{TheBlackNS5}.

\medskip
  With this in mind, we interpret the discussion in \cite[Sec.  8.3]{MF10}, which aims at a classification of the near horizon geometries for black M5-brane solutions of 11d supergravity.
  We may translate the considerations there to those here by exactly the same logic as in Example \ref{TheBlackM2}.
In that language, the result of \cite[Sec.  8.3]{MF10} is that the would-be M5 is necessarily inside a larger fixed locus, corresponding to a 6-brane:

\vspace{-.2cm}

\begin{center}
\raisebox{-30pt}{
\includegraphics[width=.4\textwidth]{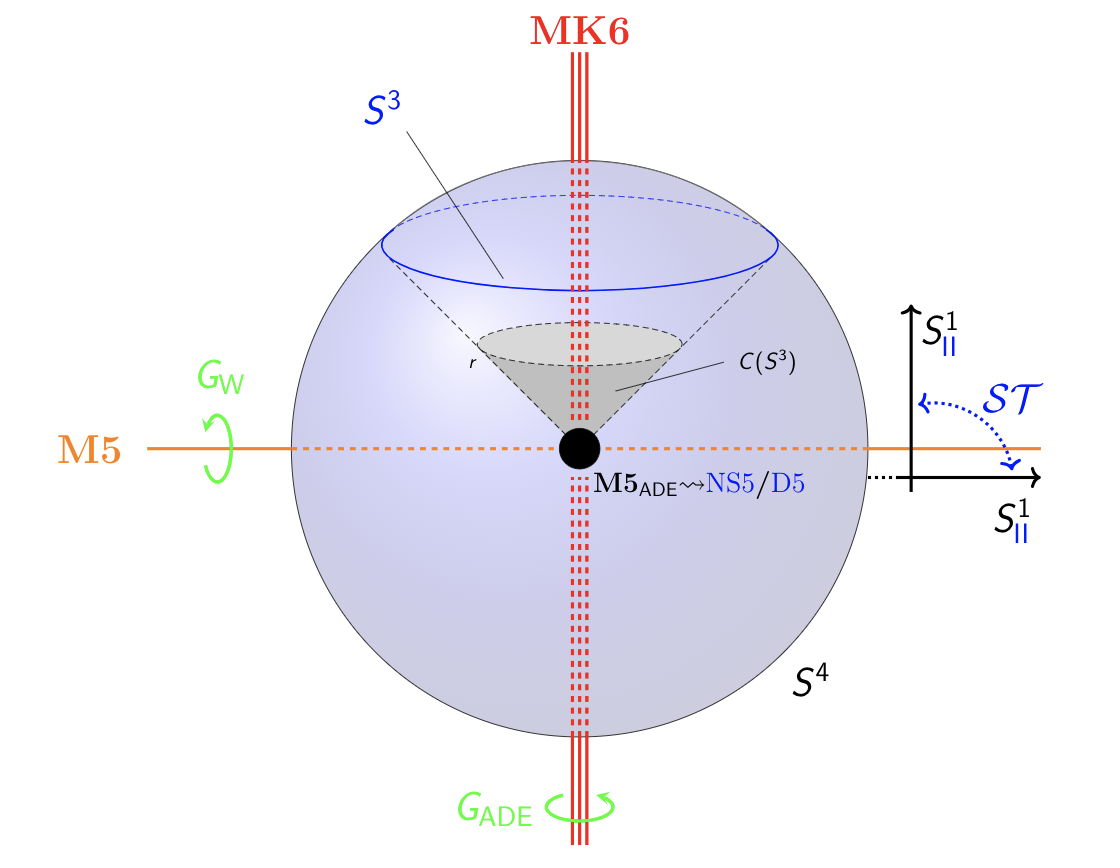}
}
\end{center}
We quote the \emph{mathematical} result of \cite[Sec.  8.3]{MF10} in the proof of our Prop. \ref{M2AndMK6ADE},
where the corresponding actions and fixed points appear labeled ``MK6''.

\medskip
  Notice that there is not supposed to be a realization of the $D =6$, $\mathcal{N} = (2,0)$ superconformal field theories in the ADE-series
  by M-branes \emph{at singularities} (the A-series however is supposed to come from coincident $\mathrm{M5}$-branes not placed at a singularity \cite{Strominger96}).
  On the other hand, such realizations are supposed to exist in F-theory at ADE-singularities \cite{HMV13}.  But the starting point
  of our discussion here does also exist for F-theory (\cite[Sec. 8]{FSS16b})
  and hence it should be possible to go through an analogous analysis of equivariant cohomology for F-branes.

\subsubsection{The $\tfrac{1}{2}\mathrm{NS5}$}
\label{TheBlackNS5}

The intersection of a black $\mathrm{M5}$ (Example \ref{TheBlackM5}) with an $\mathrm{MO9}$ (Example \ref{TheMO9})
is not a solution to supergravity, but may be argued to
be visible as an object of M-theory via the 6d superconformal worldvolume theory that it is supposed to carry \cite{Berkooz98}.
The resulting near horizon geometry involves the quotient $S^4\sslash(\mathbb{Z}_2)_{\mathrm{HW}}$ induced by the
$\mathbb{Z}_2$-action that we consider in Def.\ \ref{SuspendedHopfAction}.

\medskip
More generally, one may consider the $\mathrm{M5}_{\mathrm{ADE}}$ (Example \ref{ADEM5}) to intersect the $\mathrm{MO9}$
\cite[Sec. 2.4]{BrodieHanany97}, \cite[EGKRS00]{EGKRS00}, \cite[around Fig. 6.1, 6.2]{GKST01}, \cite[Sec.s 6]{ZHTV15}, \cite[p. 38 and around Fig. 3.9, 3.10]{Fazzi17}.
This then accordingly involves the full $G_{\mathrm{ADE}} \times \mathbb{Z}_2$-action that we consider in Def.\ \ref{SuspendedHopfAction}, see
specifically \cite[Sec. 7]{ZHTV15}.
In \cite{GKST01} this situation (regarded from the type I-perspective) is referred to as the ``$\tfrac{1}{2}\mathrm{NS5}$''-brane:

\begin{center}
\scalebox{.8}{
\raisebox{-96pt}{
\includegraphics[width=.5\textwidth]{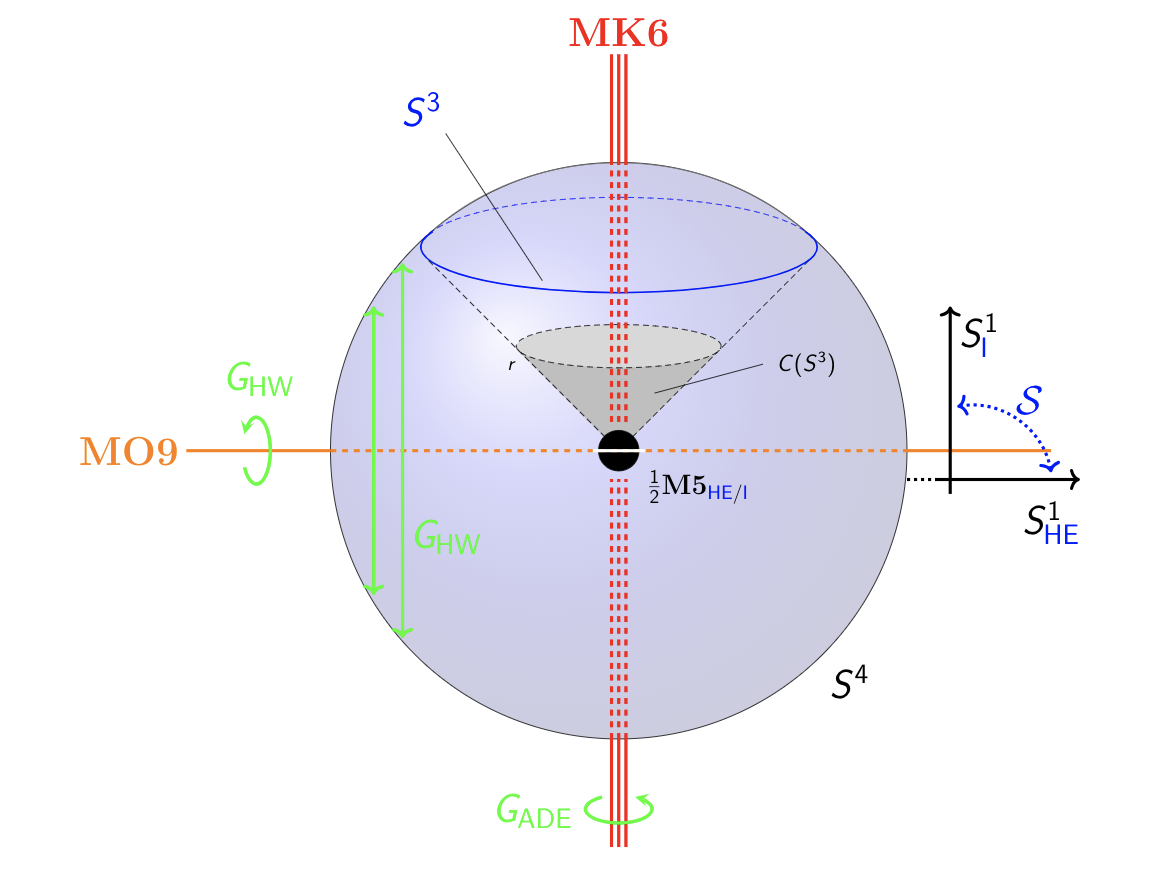}
}}
\scalebox{.8}{
\raisebox{-100pt}{
\includegraphics[width=.5\textwidth]{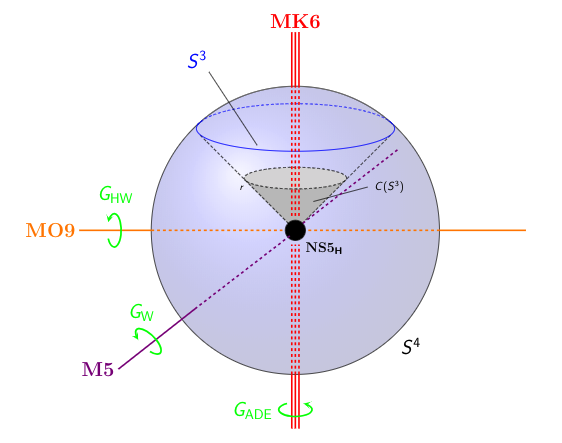}
}}

\end{center}

\subsubsection{The $\mathrm{M1}$}
\label{TheSelfDualString}

The intersection of the $\mathrm{M2}$ (Example \ref{TheBlackM2}) with the $\mathrm{M5}$ (Example \ref{TheBlackM5})
is also called the \emph{M-string} \cite[Sec. 2.3]{HI13} \cite{HIKLV15}.
In terms of boundary conditions for the expected worldvolume conformal field theory on the M2, this is argued to indeed exist
in \cite[Sec. 2.2.1]{BPST10}.
This clearly corresponds to the item $\mathrm{M1} =  \mathrm{M2} \dashv \mathrm{5}$ in Prop. \ref{NonSimpleRealSingularities}.

\subsubsection{The $\mathrm{NS1}_H$/$\mathrm{E1}$}
\label{TheBlackNS1H}

That the intersection of the $\mathrm{M2}$ (Example \ref{TheBlackM2}) with the $\mathrm{MO9}$ (Example \ref{TheMO9})
should be the heterotic string, if the $\mathrm{MO9}$ is orthogonal to the M-theory circle fiber, is the key claim of
 Ho{\v r}ava-Witten theory \cite{HoravaWitten96a} \cite{HoravaWitten96b}(if it is longitudinal to it, then this intersection
 is called the $\mathrm{E1}$ \cite{KKLPV14}).
The actual black brane configurations exhibiting this have been considered in \cite{LLO97}\cite{Kashima00}.
That the expected superconformal worldvolume theory on the M2-branes may have boundary conditions corresponding to an
MO9 is argued in \cite[Sec. 2.2.2]{BPST10}.
This clearly corresponds to the item $\mathrm{NS1}_H =  \mathrm{M2} \dashv \mathrm{MO9}$ in Prop. \ref{NonSimpleRealSingularities}.

\medskip
Notice that in the folklore the distinction, if any, between the fundamental heterotic string and its black brane incarnation remains
ambiguous, not the least because, without actual M-theory in hand, the corresponding M2-brane is really known only in its supergravity approximation.
But a look at the item denoted $\mathrm{NS1}_H = \mathrm{M2} \dashv \mathrm{MO9}$ in \hyperlink{EquivariantEnhancementsList}{Table 3}
\begin{equation}
\label{AgainNS1H}
\raisebox{30pt}{
\scalebox{.8}{
$$
  \xymatrix@C=10pt@R=4pt{
    &&&&
    &&& &&&&&&&&&
    S^4
    \ar@(ul,ur)[]^{ (\mathbb{Z}_{n+1})_{\Delta} \times G_{\mathrm{HW}} }
    \\
    \\
    &&  &&
    &&&
    \mathbb{R}^{10,1\vert \mathbf{32}}
    \ar@(ul,ur)[]^{ ( \mathbb{Z}_{n+1} )_{\Delta} \times G_{\mathrm{HW}} }
    \ar[uurrrrrrrrr]^{ \color{blue} 0  }_<<<<<<<<<<<<{\ }="s1"_<<<<<<<<<<{\ }="s2"
    \\
    &
    \mathrm{ST}
    &&&
    &&&&&&&&& S^2
    \ar@{..>}[uuurrr]|0
    &&& &&&
    S^3
    \ar[uuulll]|0
    \\
    \\
    \mathrm{M2} &&\mathrm{MO9}&&
    \mathbb{R}^{2,1\vert 8 \cdot \mathbf{2}}
    \ar@{^{(}->}[uuurrr]
    \ar@{..>}[uurrrrrrrrr]|<<<<<<<<{ \color{blue} \mu_{{}_{F1}}}^>>>>>>>>>>>>>>>>>>>>>>>>>>>>>>>>>>>>>>{\ }="t2"_<<<<<<<<<<<<<<<<<<<{\ }="s4"
    &&& &&&
    \mathbb{R}^{9,1 \vert \mathbf{16}}
    \ar@{_{(}->}[uuulll]
    \ar[uurrrrrrrrr]|>>>>>>>>>>>>>{\color{blue} \mu^{H}_{{}_{F1}}}^>>>>>>>>>>>>>>>>>>>>>>>>>>>>>>>>>>>>{\ }="t1"_<<<<<<<<<{\ }="s3"
    \\
    & \mathrm{NS1}_H &&&
    &&&   &&&&&&&&& S^1
    \ar[uuurrr]|0
    \ar@{..>}[uuulll]|0
    \\
    \\
    &&&&
    &&&
    {\mathbb{R}^{1,1\vert 8 \cdot \mathbf{1}}}
    \ar@{^{(}->}[uuurrr]
    \ar@{_{(}->}[uuulll]
    \ar[uurrrrrrrrr]|{0}^>>>>>>>>>>>>>>>>>>>>>>>>>>>>>>>>>>>{\ }="t3"^<<<<<<<<<<<<<<{\ }="t4"
    \ar@{=>}|<<<<<<<<<<<{ \color{cyan} \phantom{{A \atop A}} 0  \phantom{{A \atop A}} } "s1"; "t1"
    \ar@{::>}|<<<<<{ \color{cyan} \phantom{{A \atop A}} 0  \phantom{{A \atop A}} } "s2"; "t2"
    \ar@{=>}|{ \color{cyan} \phantom{{A \atop A}} \mathrm{svol}_{1+1} } "s3"; "t3"
    \ar@{::>}|<<<<<{ \color{cyan} \phantom{{A \atop A}} \mathrm{svol}_{1+1} } "s4"; "t4"
  }
$$
}}
\end{equation}

\noindent indicates that Theorem \ref{RealADEEquivariantEndhancementOfM2M5Cocycle}
serves to resolve the subtle distinctions and identifications involved in Ho{\v r}ava-Witten theory:
the equivariant cocycle data (via Example \ref{EquivariantEnhancementOfSuperCocycles}) shown in \eqref{AgainNS1H}
 exhibits, on the right,
the fundamental heterotic string cocycle $\mu_{F1}^{H}$ (Example \ref{FundamentalF1Cocycle}) on the worldvolume
of the MO9. At the same time, on the bottom
right, it relates the fundamental heterotic string to its own black brane incarnation, via the Green--Schwarz functional
$\mathrm{svol}_{1+1}$ appearing on the superembedding
of the string worldsheet into the MO9 worldvolume, according to Prop. \ref{TrivializationsOfRestrictionsOfM2M5Cocycle}.
Finally, the left part of the diagram witnesses that this black string is indeed the
boundary of the black M2-brane ending on the MO9, according to Prop. \ref{NonSimpleRealSingularities}.

\medskip

This concludes our comparison of selected items from Theorem \ref{SuperADESingularitiesIn11dSuperSpacetime} and Theorem \ref{RealADEEquivariantEndhancementOfM2M5Cocycle}
to existing classification of supergravity solutions and informal arguments from the string/M-theory literature. One could discuss more examples,
but this should suffice to support the suggestion that (rationally) M-branes, both fundamental branes, black branes, as well as their
various ``bound states'', are classified by real equivariant cohomotopy of super-spacetimes.

\medskip

Using this precise formulation, one may now proceed and compile comprehensive classifications of
real equivariant cohomotopy classes on super-spacetimes, explore dualities and, eventually, search for the all important lift
beyond the rational approximation.
But since the present article is clearly long enough already, we relegate such investigations to elsewhere.

\newpage

\section{Equivariant super homotopy theory}
  \label{EquivariantHomotopyTheory}

Here we establish the context of homotopy theory within which the results in Sections \ref{ADESingularitiesInSuperSpacetime}, \ref{ADEEquivariantRationalCohomotopy},  and \ref{ADEEquivariantMBraneSuperCucycles} are cast. First we briefly review ordinary equivariant
homotopy theory in Sec. \ref{HomotopyTheoryOfGSpaces}. Then, in Sec. \ref{ERSHTForSuperBranes},
we set up the \emph{equivariant rational super homotopy theory} in which
our main result, Theorem \ref{RealADEEquivariantEndhancementOfM2M5Cocycle}, will take place.


\medskip

Throughout, we take $G$ to be a finite group (equipped with the discrete topology), such as for instance a cyclic group
$$
  G = \mathbb{Z}_n \;:=\; \mathbb{Z}/(n \mathbb{Z})
$$
or more generally a finite subgroup of $\mathrm{SU}(2)$ (see Remark \ref{ADEGroups} below).
We denote by $\{ e \}$, or simply by $1$, the trivial group, i.e., the group whose only element is the neutral element.
All of the following generalizes to the case that $G$ is allowed to be a compact Lie group, such as the circle group $U(1)$,
if one considers fixed point loci for the \emph{closed} subgroups only.
But for brevity we will not explicitly discuss this generalization here.

\subsection{Ordinary equivariant homotopy theory}
\label{HomotopyTheoryOfGSpaces}

We recall just enough of the background on \emph{equivariant homotopy theory}, i.e.,
of  the homotopy theory of topological spaces equipped with $G$-actions, in order to state and explain
the relevance of  \emph{Elmendorf's Theorem}
(Theorem \ref{ElmendorfTheorem} below). This is the basis for the  generalization to equivariant
super homotopy theory in Sec.\ref{ERSHTForSuperBranes}.
For a comprehensive introduction to equivariant homotopy theory see \cite{Blu17}; for further reading
see \cite{Ma96}, \cite[appendix]{HHR09}. Some basic concepts of general homotopy theory
are recalled in Sec. \ref{HomotopyTheory}.

\subsubsection*{Homotopy theory of $G$-Spaces}

To fix notation, we begin by recalling some standard facts.
\begin{defn}[Group actions]
 \label{GroupActions}
 Let $G$ be a topological group and $X$ a topological space.
 Then a \emph{continuous action} of $G$ on $X$ is a continuous function
\begin{equation}
  \label{MapAction}
  \rho \maps G\times X \longrightarrow X
\end{equation}
such that
\begin{equation}
  \label{ActionProperty}
  \rho(1) x = x
  \qquad
  \text{and}
  \qquad
  \rho(g_1) \rho(g_1) x = \rho(g_1 g_2) x \;,
\end{equation}
where we write $\rho(g,x)$ as $\rho(g)x$, as is conventional for group actions. When
$\rho$ is understood, we will write $\rho(g)x$ simply as $gx$.
\end{defn}
\begin{remark}[Shorthand notation]
  \label{ActionShorthandNotation}
As is typical in physics, we will write:
$$
  G_\rho
    \;:=\;
  (G,\rho)
$$
for the pair of data consisting of a group with a chosen action. For instance, in Sections \ref{ADESingularitiesInSuperSpacetime}, \ref{ADEEquivariantRationalCohomotopy},  and \ref{ADEEquivariantMBraneSuperCucycles} three different actions of the group
$\mathbb{Z}_2$ play a role and we will denote them $G_{\mathrm{ADE}}$, $G_{\mathrm{HW}}$ and
$G_{\mathrm{ADE}, \mathrm{HW}}$, respectively.
\end{remark}

\begin{example}[$\mathbb{Z}_2$-actions are involutions]
  \label{Z2ActionsAreInvolutions}
  An action (Def.\ \ref{GroupActions}) of the cyclic group of order two, $\mathbb{Z}_2 = \{e,\sigma\}$,
  is equivalently an \emph{involution} on a topological space, namely a continuous function
  $$
    X \xrightarrow{\;\rho(\sigma)\;} X
  $$
which squares to the identity, $\rho(\sigma)^2 = 1_X$.
\end{example}

\begin{defn}[Spaces associated with a $G$-action]
  \label{SpacesAssociatedWithAGAction}
Let $G$ be a group equipped with an action (Def.\ \ref{GroupActions}) on some topological space $X$.
This naturally induces the following structures (we now use the shorthand notation of Remark \ref{ActionShorthandNotation}):
\begin{enumerate}[{\bf (i)}]
\vspace{-2mm}
\item  The \emph{orbit} of $x$ is the subspace of $X$ given by
$G(x):=\{ gx\;\vert\; g\in G\}$.

\vspace{-2mm}
\item   The \emph{isotropy group} of a point $x \in X$ is the subgroup of $G$ defined as
$G_x:=\{ g \in G \;:\; gx=x\}$.

Having fixed $x \in X$, the natural map $G \to X$ given by $g \mapsto gx$ induces a homeomorphism
$G/G_x \xrightarrow{\;\simeq\;} G(x)$.
Note also that the isotropy groups $G_{gx}$ of any other element $gx$
is related to $G_x$ by conjugation with $g$: $G_{gx}=gG_x g^{-1}$.

\vspace{-2mm}
\item   The \emph{fixed point space} of $G$ acting on $X$ is the subspace
\begin{equation}
  \label{FixedPointSpace}
  X^G := \{ x \in X\;|\; gx=x \text{\;for\;all\;} g \in G\}\;.
\end{equation}

\vspace{-5mm}
\item  The \emph{orbit space} $X/G$ is the quotient topological space of $X$ by the
equivalence relation generated by setting $x \sim gx$ for some $g \in G$.  Note that if $G$
acts freely on $X$ then the quotient map $X \to X/G$ is a regular covering with $G$ as a group
of deck transformation.
\end{enumerate}
\end{defn}

\begin{example}[Group actions on $\R^n$]
 Note that not every finite group action (Def.\ \ref{GroupActions}) on $\mathbb{R}^n$
 needs to have fixed points (Def.\ \ref{SpacesAssociatedWithAGAction}).
 Indeed, in \cite{CF} first examples of $\Z_n$-actions on $\R^n$
 without fixed points are given.
Later, smooth fixed point free actions on $\R^n$
 of $G=\Z_{pq}$,  for two relatively prime integers $p, q\geq 2$ are given in
 \cite[pp. 58-61]{Br}. As a consequence,  one has to pick the appropriate
 action in order to achieve gauge enhancement of M-branes.
\end{example}

\begin{defn}[Types of group actions]
  A group action (Def.\ \ref{GroupActions})
  is called
    \item {\bf (i)} \emph{free} if for any two points $x,y \in X$ there is \emph{at most} one element $g \in G$ with $g (x) = y$;
    \item {\bf (ii)} \emph{semi-free} if it is free away from the fixed points (Def.\ \ref{SpacesAssociatedWithAGAction}).
\end{defn}

\begin{defn}[Topological $G$-spaces]
  \label{GSpace}

     \begin{enumerate}[{\bf (i)}]

\item A \emph{topological $G$-space} is a topological space $X$ equipped with a continuous $G$-action (Def.\ \ref{GroupActions}). For $(X_1, \rho_1)$ and $(X_2, \rho_2)$ two topological $G$-spaces, a
\emph{$G$-equivariant map} between them is a continuous function $X_1 \overset{f}{\longrightarrow} X_2$
between the corresponding topological spaces which respects the $G$-action, in that
\begin{equation}
  \label{GEquivarianceMap}
  f(\rho_1(g) x) = \rho_2(g) f(x)
  \phantom{AAAAA}
  \mbox{for all $x \in X_1$ and $g \in G$.}
\end{equation}
\vspace{-3mm}
\item We write $G\mathrm{Spaces}$ for the corresponding category of topological $G$-spaces. Moreover, for $X_1, X_2 \in G \mathrm{Spaces}$, we write
\begin{equation}
  \label{SpaceOfEquivariantMaps}
  \mathrm{Maps}(X_1, X_2)^G \subset \mathrm{Maps}(X_1, X_2) \in \mathrm{Spaces}
\end{equation}
for the mapping space of $G$-equivariant continuous functions between them, equipped with the compact-open topology.
\end{enumerate}
\end{defn}

\begin{remark}
   \label{EquivarianceCommutingDiagram}
   We will sometimes write a $G$-equivariant map as follows:
\begin{equation}
  \label{AbbreviatedDiagramForEquivariance}
  \xymatrix{
  X_1
  \ar@(ul,ur)[]^{G}
  \ar[rr]^-{f}
  &&
  X_2
  \ar@(ul,ur)[]^{G}
  }
\end{equation}
which is to be understood as saying that $f$ is a continuous function from $X_1$ to
$X_2$ and $G$-equivariant according to (\ref{GEquivarianceMap}).
\end{remark}

\begin{example}[Ordinary topological spaces as topological $G$-spaces]
  \label{OrdinaryTopologicalSpacesAsTopologicalGSpaces}
  For $G = \{e\}=1$ the trivial group, a topological $G$-space is just a topological space.
  Similarly, the $G$-equivariant homotopy theory described in the following reduces to classical
  homotopy theory in this case.
\end{example}

\begin{example}[$G$-invariance as $G$-equivariance]
  \label{InvarianceAsEquivariance}
  If $X$ is a topological $G$-space (Def.\ \ref{GSpace}), but $A$ is just a topological space (Def.\ \ref{TopologicalSpaces}),
  regarded as a topological $G$-space with \emph{trivial} $G$-action, via Example \ref{OrdinaryTopologicalSpacesAsTopologicalGSpaces},
  then $G$-equivariant functions \eqref{GEquivarianceMap} from $X$ to $A$, which, following Remark \ref{EquivarianceCommutingDiagram},
  we may denote by
  \begin{equation}
    \xymatrix{
      X
      \ar@(ul,ur)[]^{G}
      \ar[rr]^-{f}
      &&
      A
    }
  \end{equation}
  are equivalently \emph{$G$-invariant} functions, satisfying
  $
    f(g x) = f(x)
    $.
\end{example}

\begin{example}[Real spaces]
  \label{RealSpace}
  For $G = \mathbb{Z}_2 =\{e,\sigma\}$ the cyclic group of order two,
  a topological $G$-space $X$ (Def.\ \ref{GSpace}), hence a topological \emph{$\mathbb{Z}_2$-space}
  is also called a \emph{real space} (\cite[Sec. 1]{Atiyah66}). By Example \ref{Z2ActionsAreInvolutions},
  this is a topological space equipped with a topological involution.
  For instance, if $X$ is the underlying topological space of a complex algebraic variety, it becomes
  a $\mathbb{Z}_2$-space or real space via the involution induced by complex conjugation. In this sense,
  real structure on a topological space is a generalization of real structure on a complex vector space,
  making it a real vector space.
\end{example}

\begin{example}[Basic kinds of $G$-spaces]
  \label{GSpacesExamples}
  Basic families of topological $G$-spaces (Def.\ \ref{GSpace}) include the following:
\begin{itemize}

\item
  Any topological space $X$ becomes a $G$-space by equipping it with the trivial action
  $\rho(g,x) = x$. If we do not specify a $G$-action otherwise, then this trivial action will be understood.

\item   For $H \subset G$ a subgroup of $G$, the coset space $G/H$ inherits a $G$-action from the left
multiplication of $G$ on itself. We will always understand these coset spaces to be $G$-spaces via this
choice of $G$-action.

\begin{list}{$\circ$}{}

\vspace{-2mm}
\item  Observe that given any point $x \in X$ in a topological $G$-space $X$, then the \emph{orbit} of $x$ under the $G$-action (Def.\ \ref{SpacesAssociatedWithAGAction}) looks
like $G/H$, for $H \subset G$ the \emph{stabilizer subgroup} which fixes $x$. Hence we may think of the cosets $G/H$ as the
possible \emph{orbit spaces} of $G$.

\item Observe that for the degenerate case when $H = G$, the coset $G/G = \ast$ is the point. We will find below that
equivariant homotopy theory is like ordinary homotopy theory, but with the single point $\ast = G/G$ promoted
 to a systems of generalized points given by the orbit spaces $G/H$. This is formalized by the statement of Elmendorf's Theorem (Theorem \ref{ElmendorfTheorem} below).
\end{list}

\vspace{-3mm}
\item   For $X_1$ and $X_2$ two $G$-spaces, their Cartesian product space $X_1 \times X_2$ becomes a $G$-space via the diagonal action
$g(x_1, x_2) = (gx_1, gx_2)$.
\end{itemize}
\end{example}

Using the basic cases from Example \ref{GSpacesExamples} as building
blocks yields the following concept of \emph{$G$-cell complexes}.\footnote{The cells
in a cell complex are the spatial analogues of algebra generators in an algebra.}  The
\emph{equivariant Whitehead theorem} (Theorem \ref{EquivariantWhiteheadTheorem}
below) states that the homotopy category of these complexes yields the full
equivariant homotopy theory.

\begin{defn}[$G$-cell complexes (see {\cite[Def.\ 1.2.1]{Blu17}})]
\label{GCW}
\vspace{-.5mm}
\item {\bf (i)}  For $n \in \mathbb{N}$, and $H \subset G$ a subgroup, we say that the
basic $n$-dimensional \emph{$G$-space cell} at stage $H$ is the Cartesian product
$$
  D^n \times G/H
$$
of the topological unit $n$-ball $D^n$ equipped with the trivial $G$-action and a coset space equipped with its canonical $G$-action, as in Example \ref{GSpacesExamples}.
\vspace{-1mm}
\item {\bf (ii)} A \emph{$G$-CW-complex} $X$ is the $G$-space defined inductively,
  starting with $X_0$ a disjoint union of 0-dimensional $G$-space cells and then,
  given $X_{n-1}$, gluing $n$-dimensional $G$-space cells via $G$-equivariant maps to
  obtain $X_n$. The colimit of this sequence is $X$.
\item {\bf (iii)} We write
\begin{equation}
  \label{GCWCategory}
  \xymatrix{
    G \mathrm{CWComplexes}
    \; \ar@{^{(}->}[rr]^-{I}
    &&
    G \mathrm{Spaces}
  }
\end{equation}
for the full subcategory of topological $G$-spaces (Def.\ \ref{GSpace}) on the $G$-CW-complexes.
\end{defn}

Next we consider the actual homotopy theory of topological $G$-spaces, and pass to
the corresponding homotopy categories of these two models for $G$-spaces. These
homotopy categories will turn out to be equivalent to each other, thus providing us
with two different but equivalent perspectives, each with its own advantages, on
$G$-equivariant homotopy theory.

\begin{defn}[Equivariant homotopy]
  \label{EquivariantHomotopy}
Given two topological $G$-spaces $X_1$, $X_2$ (Def.\ \ref{GSpace})
and given two $G$-equivariant maps $f_0,f_1 \maps X_1 \to X_2$ between them (see \eqref{GEquivarianceMap}), we say that
a \emph{$G$-equivariant homotopy}
 from $f_0$ to $f_1$ is a $G$-equivariant map of the form
\begin{equation}
  \label{LeftHomotopyMap}
  \eta \maps  X_1 \times [0,1] \longrightarrow X_2\;,
\end{equation}
where the interval $[0,1]$ is equipped with the trivial $G$ action, and $\eta$ satisfies:
\begin{equation}
  \label{BoundaryConditionsForLeftHomotopy}
  \eta(x,0) = f_0(x)\;,
  \quad
  \eta(x,1) = f_1(x)\;.
\end{equation}
In other words, this is a 1-parameter family of $G$-equivariant maps that
continuously interpolates between $f_0$ and $f_1$.  We denote this homotopy by
$$
  \xymatrix{
    X_1
    \ar@/^1pc/[rr]^-{f_0}_-{\ }="s"
    \ar@/_1pc/[rr]_-{f_0}^-{\ }="t"
    &&
    X_2
    \ar@{=>}^\eta "s"; "t"
  }.
$$
\end{defn}

\begin{defn}[Equivariant homotopy equivalences]
\label{EquivariantHomotopyEquivalence}
A \emph{$G$-equivariant homotopy equivalence} is a $G$-equivariant map $f \maps X_1
\to X_2$ which has an inverse up to $G$-equivariant homotopy (Def.\
\ref{EquivariantHomotopy}). This means that there exist a $G$-equivariant function
$\widetilde f \maps X_2 \to X_1$ and $G$-equivariant homotopies

\vspace{-3mm}
$$
  \xymatrix{
    X_1
    \ar@/^1pc/[rr]^-{\mathrm{id}}_-{\ }="s"
    \ar@/_1pc/[rr]_-{\widetilde f \circ f}^-{\ }="t"
    &&
    X_1
    \ar@{=>} "s"; "t"
  }
  \phantom{AAAAAA}
  \xymatrix{
    X_2
    \ar@/^1pc/[rr]^-{\mathrm{id}}_-{\ }="s"
    \ar@/_1pc/[rr]_-{f \circ \widetilde f}^-{\ }="t"
    &&
    X_2
    \ar@{=>} "s"; "t"
  }.
$$
\end{defn}
In ordinary homotopy theory, a homotopy equivalence between spaces turns out to be too strong. Instead, we pass to
weak homotopy equivalence, given by any map that induces isomorphisms between homotopy groups. This weaker
notion also has an analogue in the equivariant setting.
\begin{defn}[Weak equivariant homotopy equivalence]
\label{EquivariantWeakHomotopyEquivalence}
Let $X_1$ and $X_2$ be two topological $G$-spaces (Def.\ \ref{GSpace}).
Then a $G$-equivariant function between them (\ref{GEquivarianceMap})
$$
  f \maps X_1 \longrightarrow X_2
$$
is called a \emph{weak $G$-equivariant homotopy equivalence} if for every subgroup $H \subset G$
the induced map on $H$-fixed point spaces (Def.\ \ref{FixedPointSpaceRepresented})
$$
  f^H \maps X_1^H \longrightarrow X_2^H
$$
is an ordinary weak homotopy equivalence (Def.\ \ref{ClassicalHomotopyCategories}).
\end{defn}
\begin{remark}
As in the non-equivariant setting,  every equivariant homotopy equivalence (Def.\ \ref{EquivariantHomotopyEquivalence}) is also a weak equivariant homotopy equivalence (Def.\ \ref{EquivariantWeakHomotopyEquivalence}), but not conversely.
\end{remark}
\begin{defn}[Homotopy theory of topological $G$-spaces]
  \label{HomotopyCategoriesOfGSpaces}
  We equip the categories of topological $G$-spaces from Def.\ \ref{GSpace} and Def.\ \ref{GCW} with weak equivalences (Def.\ \ref{CategoryWithWeakEquivalences}) as follows:

\begin{enumerate}[{\bf (i)}]

    \item   On the category $G \mathrm{CWComplexes}$ (Def.\ \ref{GCW}) we take the weak equivalences to be the equivariant homotopy equivalences
     from Def.\ \ref{EquivariantHomotopyEquivalence}.

    \item  On the category $G \mathrm{Spaces}$ (Def.\ \ref{GSpace}) we take the weak equivalences to be the weak equivariant homotopy equivalences from Def.\ \ref{EquivariantWeakHomotopyEquivalence}.
\end{enumerate}

  \noindent Given any category with weak equivalences like the examples above, we can
  form its \emph{homotopy category} by inverting the weak equivalences. The resulting
  homotopy categories (Def.\ \ref{CategoryWithWeakEquivalences}) are as follows:
  \begin{align*}
    \mathrm{Ho}\left(
      G \mathrm{CWComplexes}
    \right)
    &:=\;
    \mathrm{Ho}\big(
      G \mathrm{CWCplx}
        \big[
          \left\{
            \mbox{equivariant homotopy equivalences}
          \right\}^{-1}
        \big]
    \big)
  \\
    \mathrm{Ho}\left(
      G\mathrm{Spaces}
    \right)
    & :=\;
    \mathrm{Ho}\big(
      G \mathrm{Spaces}
      \big[
        \left\{
          \mbox{weak equivariant homotopy equivalences}
        \right\}^{-1}
      \big]
    \big)
    .
  \end{align*}
\end{defn}
The following fact is the first indication that equivariant homotopy theory elevates the collection of fixed point loci to
a special role in the theory.
\begin{prop}[Equivariant Whitehead theorem ({\cite[Thm 3.4]{Wan80}}, see {\cite[Cor. 1.2.14]{Blu17}})]
  \label{EquivariantWhiteheadTheorem}
  Under passage to the homotopy categories of Def.\ \ref{HomotopyCategoriesOfGSpaces}, the inclusion $I$ in \eqref{GCWCategory} from Def.\ \ref{GCW}
   induces an equivalence of categories:
  \begin{equation}
     \xymatrix{
      \mathrm{Ho}\left(
        G \mathrm{CWComplexes}
      \right)
      \ar[rr]^-{I}_-{\simeq}
      &&
      \mathrm{Ho}\left(
        G \mathrm{Spaces}
      \right)
    \!.}
  \end{equation}
\end{prop}

\noindent Therefore, we now turn our full attention to these systems of fixed point loci.

\subsubsection*{Systems of fixed point loci}
\label{SystemsOfFixedPointLoci}

For any $G$-space $X$, an orbit in $X$ is a $G$-space of the form $G/H$, where $H$ is
the stabilizer of a point in the given orbit. This means we can form the collection
of all possible orbits for all possible $G$-spaces: they are given by the coset
spaces $G/H$. We thus call these coset spaces \emph{orbit spaces}, as in Def.\
\ref{GSpacesExamples}.

\begin{defn}[The orbit category, see {\cite[Def.\ 1.3.1]{Blu17}}]
  \label{OrbitCategory}
  We write
  \begin{equation}
    \mathrm{Orb}_G  \xymatrix{ \ar@{^{(}->}[r]&} G \mathrm{Spaces}
  \end{equation}
  for the full subcategory of topological $G$-spaces (Def.\ \ref{GSpace}) which are orbit spaces (Def.\ \ref{GSpacesExamples}), called the \emph{orbit category} of $G$.
  That is, the objects in this category are the coset spaces $G/H$, one for each
  subgroup $H \subset G$, and the morphisms are the continuous $G$-equivariant functions (\ref{GEquivarianceMap}) between the coset spaces $G/H_1 \longrightarrow G/H_2$.

\end{defn}

\begin{example}[Orbit category of $\mathbb{Z}_2$]
  \label{OrbitCategoryForZMod2}
  Consider the orbit category (Def.\ \ref{OrbitCategory}) of the cyclic group of order two: $\mathbb{Z}_2 = \{e, \sigma \}$
  with a single non-trivial element $\sigma$, squaring to the neutral element $\sigma \cdot \sigma = e$.
  This has precisely two subgroups, namely itself and the trivial group $1 = \{e\}$. Hence its orbit spaces
  are $\mathbb{Z}_2/\mathbb{Z}_2 = 1$ and $\mathbb{Z}_2/1 =\mathbb{Z}_2$. The non-trivial morphisms
   in the orbit category
  are depicted succinctly as follows:
  $$
    \mathrm{Orb}_{\mathbb{Z}_2}
    \;=\;
     {\footnotesize \left\{ \hspace{-2mm}
    \raisebox{10pt}{
    \xymatrix@R=1.8em{
      \mathbb{Z}_2/1 \ar@(ul,ur)[]^\sigma
      \ar[d]
      \\
      \mathbb{Z}_2/\mathbb{Z}_2
    }} \hspace{-1.5mm}
    \right\}
    }.
  $$
\end{example}
\begin{example}[Systems of fixed point spaces]
  \label{FixedPointSpaceRepresented}
  If $X$ is a topological $G$-space (Def.\ \ref{GSpace}), and $H \subset G$ a
  subgroup, then a $G$-equivariant map
  $$
    G/H \overset{f}{\longrightarrow} X
  $$
  from the orbit space $G/H$ (Def.\ \ref{GSpacesExamples}) must, by $G$-equivariance (\ref{GEquivarianceMap}), send the
  equivalence class of the neutral element
  $e \in G$ to an $H$-\emph{fixed point} of $X$, since the action of $H \subset G$ on $G/H$ is trivial. Moreover, still by
  equivariance, the choice of the
  image of the neutral element uniquely fixes the value of $f$ on all other points of $G/H$. This means that the
  equivariant mapping space (\ref{SpaceOfEquivariantMaps}) out of $G/H$ into $X$ is equivalently the subspace of $H$-fixed
  points (Def.\ \ref{SpacesAssociatedWithAGAction})
  \begin{equation}
    \label{FixedPointSpace}
    X^H
    \; \simeq \;
    \mathrm{Maps}(G/H, X)^{G}
    .
  \end{equation}
  Accordingly, for
  \begin{equation}
    \label{AFunctionBetweenOrbitSpaces}
    G/H_1 \overset{f}{\longrightarrow}  G/H_2
  \end{equation}
  a $G$-equivariant map between two orbit spaces, precomposition with $f$ yields a
  continuous function between mapping spaces, going in the opposite direction:
  \[
  \xymatrix{ \mathrm{Maps}(G/H_2, X)^G \ar[rr]^{(-) \circ f } && \mathrm{Maps}(G/H_1,
  X)^G}\!.
  \]
  Under the equivalence with fixed-point spaces, this becomes a map:
  \begin{equation}
    \label{InducedMapOnFixedPoints}
    \xymatrix{
    X^{H_2} \ar[rr]^{X^f} && X^{H_1}. }
  \end{equation}
  We can use equivariance to describe this map very explicitly. As noted above, $f$
  is determined by where it sends the class $[e] \in G/H_1$. Let us call this value
  $[g_f] \in G/H_2$, that is $f[1] = [g_f]$, for some choice of representative $g_f
  \in G$. Then $X^f(x) = g_f x$ for any $H_2$-fixed point $x \in X^{H_2}$. The reader
  can check that this is well-defined and lands in the space of $H_1$-fixed points,
  $X^{H_1}$.

  Moreover, this construction respects composition and identities:
  \[ X^{f \circ g} = X^g X^f, \quad X^{\id} = \id . \]
  We summarize this by saying that the system of $H$-fixed point spaces $X^H$ of $X$
  as $H \subset G$ varies is a \emph{presheaf of topological spaces on the orbit
  category} $\mathrm{Orb}_G$ (Def.\ \ref{OrbitCategory}). This is denoted:
  $$
    \xymatrix@R=.6em{
      \mathrm{Orb}_G^{\mathrm{op}}
      \ar[rrr]^{X^{(-)}}
      &&&
      \mathrm{Spaces}
      \\
      G/H_1
        \ar[dd]^{f_1}
        \ar@/_2.5pc/[dddd]_{f_2 \circ f_1}
        && &
        X^{H_1}
      \\
      \\
      G/H_2 \ar[dd]^{f_2}  &&& X^{H_2} \ar[uu]^{X^{f_1}}
      \\
      \\
      G/H_3
        &&&
      X^{H_3}
      \ar@/_2.3pc/[uuuu]_{ X^{ f_2 \circ f_1 } }
      \ar[uu]^{X^{f_2}}
    }
  $$
\end{example}

It will be useful to isolate the structure of systems of fixed point spaces, like
in Example \ref{FixedPointSpaceRepresented},
as a concept in itself:
\begin{defn}[Systems of topological spaces indexed over the orbit category]
  \label{PresheavesOnOrbitCatgegory}
    \vspace{-1mm}
\item {\bf (i)}
  A \emph{system of topological space indexed by the orbit category} $\mathrm{Orb}_G$ (Def.\ \ref{OrbitCategory}),
  also called a \emph{presheaf of topological spaces on the orbit category},
  is an assignment of a topological space $X^H \in \mathrm{Spaces}$ to each subgroup $H \subset G$
  and of a continuous function $X^f \maps X^{H_2} \to X^{H_1}$ to each $G$-equivariant map
  $f \maps G/H_1 \to G/H_2$ such that this assignment respects composition identities:
  \[ X^{f \circ g} = X^g X^f, \quad X^\id = \id . \]
  This is denoted:
 $$
    \xymatrix@R=.6em{
      \mathrm{Orb}_G^{\mathrm{op}}
      \ar[rrr]^{X}
      &&&
      \mathrm{Spaces}
      \\
      G/H_1
        \ar[dd]^{f_1}
        \ar@/_2.5pc/[dddd]_{f_2 \circ f_1}
        &&&
        X^{H_1}
      \\
      \\
      G/H_2 \ar[dd]^{f_2}  &&& X^{H_2} \ar[uu]^{X^{f_1}}
      \\
      \\
      G/H_3
        &&&
      X^{H_3}
      \ar@/_2.3pc/[uuuu]_{ X^{f_2 \circ f_1} }
      \ar[uu]^{X^{f_2}}
    }
  $$

   \vspace{-2mm}
\item {\bf (ii)}  Given two such system $X_1$ and $X_2$, then a \emph{homomorphism} between them, denoted
  \begin{equation}
    \label{PresheafMap}
    F \maps X_1 \longrightarrow X_2
  \end{equation}
  is an assignment of continuous functions
  \begin{equation}
    \label{PresheafMapComponent}
    F^H \maps X_1^H \longrightarrow X_2^H
  \end{equation}
  for each subgroup $H \subset G$, such that this respects all the equivariant
  functions $G/H_1 \overset{f}{\to} G/H_2$ between orbit spaces, meaning that $X_2^f
  \circ F^{H_2} = F^{H_1} \circ X_1^f$, which we can summarize by saying that the
  following square commutes for all $f$:
  \[ \xymatrix@R=1em{
  X_1^{H_1} \ar[rr]^-{ F^{H_1} } && X_2^{H_1}
  \\
  \\
  X_1^{H_2} \ar[rr]_-{ F^{H_2} } \ar[uu]^{X_1^f} && X_2^{H_2} \ar[uu]_{X_2^f} . }
  \]
\item {\bf (iii)} We write $\mathrm{PSh}(\mathrm{Orb}_G, \mathrm{Spaces})$ for the
  category of systems of topological spaces indexed by the orbit category, with
  homomorphisms between them.
\end{defn}
\begin{example}
  \label{ExternalYonedaFromGSpacesToPresheavesOnOrbitCategory}
  The construction that associates a topological $G$-space $X$ (Def.\ \ref{GSpace}) to its system of
  fixed-point spaces $X^{(-)}$, according to Example \ref{FixedPointSpaceRepresented},
   gives a functor
  \begin{equation}
    \label{ExternalYoneda}
    \xymatrix@R=1pt{
      G\mathrm{Spaces}
      \ar[rr]^-{Y}
      &&
      \mathrm{PSh}(\mathrm{Orb}_G, \mathrm{Spaces})
      \\
      X \ar@{|->}[rr] && X^{(-)}
    }
  \end{equation}
  to the category of systems of topological spaces indexed over the orbit category (Def.\ \ref{PresheavesOnOrbitCatgegory}).
\end{example}

This gives us a machine for turning $G$-spaces into systems of spaces indexed by the
orbit category. It turns out that systems of spaces actually have the \emph{same
homotopy theory} as $G$-spaces, giving us a third model of equivariant homotopy
theory. The key idea here is that the following weak equivalences are, on each orbit
space, the same as the ordinary weak equivalences of classical homotopy theory:
\begin{defn}[Homotopy theory of systems of spaces over $\Orb_G$]
  \label{HomotopyCategoryOfGFixedPointSystems}

  \item {\bf (i)} We call a morphism $F^{(-)} \maps X^{(-)} \to Y^{(-)}$ (see \eqref{PresheafMap}) in
  the category $\mathrm{PSh}(\mathrm{Orb}_G, \mathrm{Spaces})$ from Def.\ \ref{PresheavesOnOrbitCatgegory}
  a \emph{weak equivalence} if for each subgroup $H \subset G$ its component $F^{H}$ (see \eqref{PresheafMapComponent})
  is a weak homotopy equivalence of spaces (Def.\ \ref{ClassicalHomotopyCategories}).

\item {\bf (ii)} We denote the resulting homotopy category (Def.\ \ref{CategoryHomotopy})  by
  $$
    \mathrm{Ho}\left(
      G\mathrm{FixedPointSystems}
    \right)
    :=
    \mathrm{Ho}\Big(
      \mathrm{PSh}( \mathrm{Orb}_G, \mathrm{Spaces} )
      \big[
        \left\{
          \mbox{subgroup-wise weak homotopy equivalences}
        \right\}^{-1}
      \big]
    \Big).
  $$
\end{defn}

The following proposition, known as Elmendorf's Theorem, says that the homotopy
theory of $G$-spaces and of systems of spaces over $\Orb_G$ are the same. In the next
section, we will use Elemendorf's Theorem to generalize equivariant homotopy theory
to situations that do not admit ``point-set models'', such as the 11d
super-spacetimes on which the M2/M5-brane cocycle is defined:
\begin{prop}[Elmendorf's Theorem ({\cite{Elmendorf83}}, see {\cite[Thm. 1.3.6 and 1.3.8]{Blu17}})]
  \label{ElmendorfTheorem}
  Under passage to the homotopy categories of Def.\ \ref{CategoryHomotopy}, the
  functor $Y$ (in \eqref{ExternalYoneda}) from Def.\
  \ref{ExternalYonedaFromGSpacesToPresheavesOnOrbitCategory} constitutes an
  equivalence of categories:
  \begin{equation}
    \xymatrix{
      \mathrm{Ho}\left(
        G \mathrm{Spaces}
      \right)
      \ar[rr]^-\simeq_-{\tiny \mbox{\rm Elmendorf} }
      &&
      \mathrm{Ho}\left(
        G\mathrm{FixedPointSystems}
      \right)
    }
  \end{equation}
  between the homotopy theory of $G$-spaces (Def.\ \ref{HomotopyCategoriesOfGSpaces}) and that of systems of spaces over $\Orb_G$ (Def.\ \ref{HomotopyCategoryOfGFixedPointSystems}).
\end{prop}

In summary, we have the following system of homotopy categories
  $$
    \xymatrix@R=1em{
      \mathrm{Ho}\left(
        G \mathrm{CWComplexes}
      \right)
      \ar[dd]
      \ar[rr]^-{\simeq}_-{\tiny\mbox{\rm Equivariant} \atop \mbox{\rm Whitehead} }
      &&
      \mathrm{Ho}\left(
        G \mathrm{Spaces}
      \right)
      \ar[dd]
      \ar[rr]^-\simeq_-{\tiny \mbox{\rm Elmendorf} }
      &&
      \mathrm{Ho}\left(
        G\mathrm{FixedPointSystems}
      \right)
      \ar[dd]
      \\
      \\
      \mathrm{Ho}\left(
        \mathrm{CWComplexes}
      \right)
      \ar[rr]^-{\simeq}_-{ \tiny \mbox{\rm Whitehead} }
      &&
      \mathrm{Ho}\left(
        \mathrm{Spaces}
      \right)
      \ar[rr]^-=
      &&
      \mathrm{Ho}\left(
        \{e\}\mathrm{FixedPointSystems}
     \right)
    }
  $$
  Since all the homotopy categories in the top row and those in the bottom row are
  equivalent to each other, we use them interchangeably. So, we will often write
  $\mathrm{Ho}(G \mathrm{Spaces})$ and $\mathrm{Ho}(\mathrm{Spaces})$ for the top row
  and bottom row, respectively.

\begin{example}[Homomorphisms of systems of fixed points up to homotopy]
  \label{HomomorphismsOfSystemsOfFixedPointsUpToHomotopy}
  Consider the case that $G = \mathbb{Z}_2$ is the cyclic group of order 2, so that the orbit category is as in Example \ref{OrbitCategoryForZMod2}.
  Consider a topological $G$-space $A$ (Def.\ \ref{GSpace}) with
  precisely one fixed point under the non-trivial element in $\mathbb{Z}_2$, so that its system of fixed point spaces
  according to Example \ref{FixedPointSpaceRepresented} is

  \vspace{-5mm}
  $$
    \xymatrix@R=4pt@C=1.5em{
      &
      \mathbb{Z}_2/1
      \ar@(ul,ur)[]^\sigma
      \ar[dd]
      &&
      A
      \ar@(ul,ur)[]^{\rho(\sigma)}
      \\
      A^{(-)}
      \maps
      &&
      \longmapsto &&
      \\
      &
      \mathbb{Z}_2/\mathbb{Z}_2
      &&
      \ast
      \ar@{^{(}->}[uu]_a
    }
  $$
  Now let $X$ be another $G$-space. Then a homomorphism $X^{(-)} \to A^{(-)}$ of
  systems of fixed point spaces, according to Def.\ \ref{PresheavesOnOrbitCatgegory},
  is a pair of continuous functions $X^1 \to A^1$ and $X^{\Z_2} \to A^{\Z_2}$ such
  that following square commutes:
  \begin{equation}
    \label{MorphismInPShOrbSpaces}
    \xymatrix@R=4pt{
      &&&& X^{(-)} \ar[rr] && A^{(-)}
      \\
      \\
      &&
      \mathbb{Z}_2/1
      \ar@(ul,ur)[]^\sigma
      \ar[dd]
      &&
      X^{1}
      \ar@(ul,ur)[]
      \ar[rr]
      &&
      A
      \ar@(ul,ur)[]
      \\
      &&&
      \\
      &&
      \mathbb{Z}_2/\mathbb{Z}_2
      &&
      X^{\mathbb{Z}_2}
      \ar[rr]
      \ar[uu]
      &&
      \ast
      \ar@{^{(}->}[uu]_a
    }
  \end{equation}
However, as we pass to the homotopy category $\mathrm{Ho}\left( \mathrm{PSh}(\mathrm{Orb}_{\mathbb{Z}_2}, \mathrm{Spaces}) \right)$
from Def.\ \ref{HomotopyCategoriesOfGSpaces},   the system $A^{(-)}$ becomes equivalent to ``more flexible'' systems. In particular,
according to Example \ref{BasedPathSpaces} there is a weak equivalence
  to the system which assigns to $\mathbb{Z}_2/\mathbb{Z}_2$ not the point, but the based path space of $A$:\footnote{
  In the language of model category theory (see e.g. \cite[Sec. 2]{Schreiber17b}), the system involving the based path space is
  a \emph{fibrant resolution} of the original
  system $A^{(-)}$ in the projective model category structure on functors (see e.g. \cite[Thm. 3.26]{Schreiber17b}).}
  \begin{equation}
    \label{ExampleForFibrantResolution}
    \xymatrix@R=4pt{
      &&
      \mathbb{Z}_2/1
      \ar@(ul,ur)[]^\sigma
      \ar[dd]
      &&
      A
      \ar@(ul,ur)[]
      \ar@{=}[rr]
      &&
      A
      \ar@(ul,ur)[]
      \\
      &&&
      \\
      &&
      \mathbb{Z}_2/\mathbb{Z}_2
      &&
      \ast
      \ar@{^{(}->}[uu]
      \ar[rr]_{\simeq}
      &&
      P_a A
      \ar[uu]_{\mathrm{ev}_1}
    }
  \end{equation}
  Still by Example \ref{BasedPathSpaces}, this means that in the homotopy category the
  commutative squares involved in the definition of the map (\ref{MorphismInPShOrbSpaces})
  may be filled by a homotopy
  \begin{equation}
    \label{HomotopyFromPathSpace}
    \raisebox{65pt}{
    \xymatrix@R=.8em{
      &
      A
        \ar@{=}[dr]
      \\
      X
        \ar[rr]
        &
        &
      A
      \\
      &
      \ast
      \ar[dr]^-\simeq
      \ar[uu]|-{\phantom{AA} \atop \phantom{AA}}
      \\
      X^{\mathbb{Z}_2}
      \ar@{^{(}->}[uu]
      \ar[rr]
      &
      &
      P_a A
      \ar[uu]_{\mathrm{ev}_1}
    }}
    \phantom{AAAA}
    \raisebox{20pt}{
    \xymatrix{\ar@{<->}[r] &}
    }
    \phantom{AAAA}
    \raisebox{40pt}{
    \xymatrix@R=.8em{
      X
      \ar[rr]_<{\ }="s"
      &&
      A
      \\
      \\
      X^{\mathbb{Z}_2}
      \ar[uu]
      \ar[rr]^>{\ }="t"
      &&
      \ast
      \ar[uu]_{a}
      \ar@{=>} "s"; "t"
    }}
  \end{equation}
\end{example}

With equivariant homotopy theory in hand, our concern below in Sections \ref{ADESingularitiesInSuperSpacetime}, \ref{ADEEquivariantRationalCohomotopy},  and \ref{ADEEquivariantMBraneSuperCucycles} will be to
find \emph{equivariant enhancements} of given cocycles in non-equivariant cohomology (see Example \ref{EquivariantEnhancementOfSuperCocycles} below):

\begin{defn}[Enhancement of cohomology to equivariant cohomology]
\label{EnhancementToEquivariantCohomology}
For $X, A$ topological spaces (Def.\ \ref{TopologicalSpaces}), let
$$
  [c]=[
    X \overset{c}{\longrightarrow} A
  ]
  \;\in\; \mathrm{Ho}(\mathrm{Spaces})
$$
be the homotopy class of a map, hence the cohomology class of a cocycle on $X$ with coefficients in $A$.
We will say that an \emph{enhancement} of this to the cohomology class of a $G$-equivariant
cocycle is a lift of this map through
the forgetful functor
$$
  \xymatrix@R=1.5em{
    && \mathrm{Ho}(G\mathrm{Spaces}) \ar[d]
    \\
    \ast
    \ar[rr]_-{[c]}
    \ar[urr]^-{[c_G]}
    &&
    \mathrm{Ho}( \mathrm{Spaces} )\;.
  }
$$
\end{defn}

The following phenomenon will be of key importance in Sections
\ref{ADESingularitiesInSuperSpacetime}, \ref{ADEEquivariantRationalCohomotopy},
\ref{ADEEquivariantMBraneSuperCucycles}. It explains how equivariant enhancement
yields what physicists would call ``extra degrees of freedom'':

\begin{remark}[Equivariant enhancement is extra structure]
  \label{EquivariantEnhancementIsExtraStrujcture}
  Note that an equivariant enhancement as in Def.\
  \ref{EnhancementToEquivariantCohomology} may not exist, and if it does, it
  involves a choice. This is because it may happen that there is a plain homotopy
  between $G$-equivariant maps but not a $G$-equivariant homotopy, so that the two
  maps represent the same homotopy class in $\Ho(\Spaces)$, but two different classes
  in $\Ho(G\Spaces)$:
$$
  \xymatrix@R=4pt@C=.3em{
    X
    \ar@/^2.1pc/[dd]^{c_1}_{\ }="s1top"
    \ar@/^0pc/[dd]_{\ }="s2top"^{\ }="t1top"
    \ar@/_2.1pc/[dd]_{c_2}^{\ }="t2top"
    \\
    &&&&&& [c_1] \neq [c_2]  & \in  & \mathrm{Ho}(G\mathrm{Spaces}) \ar[dddd]
    \\
    A
    \ar@{}|{\times} "s1top"; "t1top"
    \ar@{=>} "s2top"; "t2top"
    \\
    \\
    X
    \ar@/^2.1pc/[dd]^{c_1}_{\ }="s1"
    \ar@/^0pc/[dd]_{\ }="s2"^{\ }="t1"
    \ar@/_2.1pc/[dd]_{c_2}^{\ }="t2"
    \\
    &&&&&& [c_1] = [c_2]  & \in &  \mathrm{Ho}(\mathrm{Spaces})
    \\
    A
    \ar@{=>} "s1"; "t1"
    \ar@{=>} "s2"; "t2"
  }
$$
\end{remark}

In Example \ref{EquivariantEnhancementOfSuperCocycles} we spell out what this extra structure of equivariant enhancement
means for general super-cocycles. Our main Theorem \ref{RealADEEquivariantEndhancementOfM2M5Cocycle} determines
the ADE-equivariant enhancements of the cocycle of the fundamental M2/M5-brane.

\subsection{Equivariant rational super homotopy theory}
  \label{ERSHTForSuperBranes}

In this section we set the scene for the discussion in Sections \ref{ADESingularitiesInSuperSpacetime},
\ref{ADEEquivariantRationalCohomotopy},  and \ref{ADEEquivariantMBraneSuperCucycles} by establishing the
homotopy theory in which the equivariant gauge enhancement of the M2/M5-cocycle (Theorem \ref{RealADEEquivariantEndhancementOfM2M5Cocycle} below) takes
place, namely \emph{equivariant rational super homotopy theory} (Def. \ref{GEquivariantRationalSuperHomotopy} below).

\medskip
Here \emph{rational super homotopy theory} (Def. \ref{RationalSuperHomotopyTheory} below), is the evident generalization
of plain rational homotopy theory (recalled as Def. \ref{RationalHomotopyTheory} below) regarded via
Sullivan's equivalence (recalled as Prop. \ref{SullivanEquivalence} below). This equivalence identifies the rational homotopy theory
of sufficiently well-behaved spaces with that of sufficiently well-behaved differential-graded commutative algebras (see Def. \ref{dgAlgebrasAnddgModules}).
Rational super homotopy theory (Def. \ref{RationalSuperHomotopyTheory} below) results from generalizing the latter to differential-graded \emph{super-}algebras. In the supergravity
literature, the cofibrant objects among these are known as ``FDAs'', following \cite{vanNieuwenhuizen82, CDF}.

\medskip
We had studied the rational super homotopy theory of super $p$-branes in {\cite{FSS16b,FSS16a, cohomotopy, FSS13}};
 for expository review see \cite{Schreiber16}. From this discussion we here need that the M2/M5-brane WZW-term is a cocycle in
 non-equivariant degree four rational super cohomotopy, which we recall as Prop. \ref{M2M5SuperCocycle} below.
Then we use the perspective provided by Elmendorf's Theorem (Prop. \ref{ElmendorfTheorem}) to introduce the equivariant refinement of
rational super homotopy theory (Def. \ref{EquivariantAbstractHomotopyTheory} and Def. \ref{GEquivariantRationalSuperHomotopy} below).
We point out how super Lie algebras with $G$-action provide examples (Example \ref{SuperLieAlgebrasWithGActionAsGSuperspaces} below)
and we highlight which data is involved in a $G$-equivariant enhancement of a given super-cocycle (Example \ref{EquivariantEnhancementOfSuperCocycles} below).
These are key ingredients in the proof of our main result, Theorem \ref{RealADEEquivariantEndhancementOfM2M5Cocycle}, below.

%
%

\begin{defn}[Rational super homotopy theory]
  \label{RationalSuperHomotopyTheory}

\item {\bf (i)} We write
  $\mathrm{dgcSuperAlg}$ for the category   whose objects are differential graded-commutative super-$\mathbb{R}$-algebras,
  whose morphisms are homomorphisms $\phi : A_1 \to A_2$.
  This means, equivalently, that an object $A \in \mathrm{dgcSuperAlg}$ is a $\mathbb{Z} \times \mathbb{Z}_2$-graded
  differential algebra, with $\mathbb{Z}$ the ``cohomological grading'' and
  with $\mathbb{Z}_2 = \{\mbox{even}, \mbox{odd}\}$ the \emph{super-grading},
  hence where elements $a_1, a_2 \in A$ of homogeneous bi-degree $(n_i,\sigma_i)$ satisfy (as in \cite[II.2.106-109]{CDF}, \cite[Sec. 6]{DeligneFreed99})
  the sign rule
  $$
    a_1 a_2 \;=\; (-1)^{n_1 n_2}(-1)^{\sigma_1 \sigma_2} a_2 a_1
    \,,
  $$
  and such that the differential is of bidegree $(1,\mathrm{even})$.

 \item {\bf (ii)}
 We take the weak equivalences (Def. \ref{CategoryWithWeakEquivalences}) in $\mathrm{dgcSuperAlg}$
 to be the quasi-isomorphisms (as in Def. \ref{dgAlgebrasAnddgModules}).
 Mimicking the ``bosonic'' Sullivan equivalence (Prop. \ref{SullivanEquivalence}),  the resulting homotopy category
  (Def. \ref{CategoryHomotopy}) restricted to the connected and finite-type dgc-superalgebras
  (as in Def. \ref{dgAlgebrasAnddgModules}) we denote by
  $$
    \mathrm{Ho}\left( {\color{blue}\mathrm{Super}}\mathrm{Spaces}_{\mathbb{R}, \mathrm{cn}, \mathrm{nil}, \mathrm{fin}} \right)
     \;:=\;
    \mathrm{Ho}\big( \mathrm{dgc}{\color{blue}\mathrm{Super}}\mathrm{Alg}^{\mathrm{op}}_{\mathrm{cn}, \mathrm{fin}}\big[ \left\{\mbox{quasi-isomorphisms}\right\}^{-1}\big] \big)
    .
  $$
  As in plain rational homotopy theory (see \eqref{NonConnectedSullivanEquivalence}) we may drop the connectedness condition by considering
  the category of indexed tuples of dgc-superalgebras on the right:
  $$
    \mathrm{Ho}\left( {\color{blue}\mathrm{Super}}\mathrm{Spaces}_{\mathbb{R}, \mathrm{nil}, \mathrm{fin}} \right)
     \;:=\;
    \underset{S \in \mathrm{Set}}{\int} \mathrm{Ho}\big( \mathrm{dgc}{\color{blue}\mathrm{Super}}\mathrm{Alg}^{\mathrm{op}}_{\mathrm{cn}, \mathrm{fin}}\big[ \left\{\mbox{quasi-isomorphisms}\right\}^{-1}\big] \big)^S.
  $$
\end{defn}
\begin{example}[Rational homotopy types as rational superspaces]
  \label{RationalCohomotopySuperSpaces}
  Every ordinary dgc-algebra (Def. \ref{dgAlgebrasAnddgModules}) becomes a dgc-superalgebra
   (Def. \ref{RationalSuperHomotopyTheory}) by regarding
  each element in even super-degree. Hence we have a full inclusion of rational homotopy theory 
  (via Prop. \ref{SullivanEquivalence})
  into rational super homotopy theory
  $$
    \mathrm{Ho}\left(  \mathrm{Spaces}_{\mathbb{R}, \mathrm{nil},\mathrm{fin}} \right)
   \xymatrix{\ar@{^{(}->}[r]&}
    \mathrm{Ho}\left(  \mathrm{SuperSpaces}_{\mathbb{R}, \mathrm{nil},\mathrm{fin}} \right)\!.
  $$
\end{example}
\begin{remark}
  \label{NonNilSuperRational}
  One may generalize super-geometric homotopy theory beyond the rational, nilpotent and finite-type situation considered
  in Def. \ref{RationalSuperHomotopyTheory} (see \cite{dcct}). For brevity and focus here we will not further discuss this,
  except that, to ease notation in the following, we note that this yields a full embedding
  $$
    \mathrm{Ho}\left(  \mathrm{SuperSpaces}_{\mathbb{R}, \mathrm{nil},\mathrm{fin}} \right)
    \xymatrix{\ar@{^{(}->}[r]&}
    \mathrm{Ho}\left(  \mathrm{SuperSpaces}_{\mathbb{R}} \right)
  $$
  into a less restrained homotopy category. Hence we may safely suppress the subscripts when discussing morphisms 
  in the homotopy category
  (hence cocycles!, see Def. \ref{CohomoloyFromHomotopy}) between given nilpotent finite-type superspaces.
\end{remark}
\begin{example}[Rational cohomotopy of superspaces]
  \label{RationalCohomotopyOfSuperSpaces}
  The minimal dgc-algebra model for the rational 4-sphere (Example \ref{SpheredgcAlgebraModel}) may be regarded as a
  dgc-superalgebra (Def. \ref{RationalSuperHomotopyTheory}) via Example \ref{RationalCohomotopySuperSpaces},
  and as such makes the 4-sphere represent an object in rational super homotopy theory
  $$
    S^4 \;\in\; \mathrm{Ho}\left( \mathrm{SuperSpaces}_{\mathbb{R},\mathrm{nil}, \mathrm{fin}} \right).
  $$
  This means (via Example \ref{ExamplesOfCohomologyTheories}) that for $X$ any super space the \emph{rational degree
   four cohomotopy} of $X$ is
  the set of morphisms
  $$
    X \longrightarrow S^4 \phantom{AAA} \in \mathrm{Ho}\left( \mathrm{SuperSpaces}_{\mathbb{R}} \right).
  $$
  On the right we are now using notation as in Remark \ref{NonNilSuperRational}.
\end{example}

\begin{example}[Super Lie algebras as superspaces]
  \label{SuperLieAlgebraAsSuperSpaces}
  Let
  $$
    \mathfrak{g} \simeq_{\mathbb{R}} \mathfrak{g}_{\mathrm{even}} \oplus \mathfrak{g}_{\mathrm{odd}}
  $$
  be a finite-dimensional super Lie algebra. Then its Chevalley-Eilenberg algebra $\mathrm{CE}(\mathfrak{g})
  \in \mathrm{dgcSuperAlgebra}$ is a dgc-superalgebra (Def. \ref{RationalSuperHomotopyTheory}) and hence
  defines a rational superspace, which we will denote by the same symbol:
  \begin{equation}
    \label{SuperSpaceSuperLieAlgebra}
    \mathfrak{g} \in \mathrm{Ho}\left( \mathrm{SuperSpaces}_{\mathbb{R}} \right).
  \end{equation}
  If $\mathfrak{g} \simeq_{\mathbb{R}} \mathfrak{g}_{\mathrm{even}}$ happens to be an ordinary Lie algebra
  (i.e., concentrated in even degree),  then $\mathrm{CE}(\mathfrak{g})$ is an ordinary dgc-algebra and hence in this case \eqref{SuperSpaceSuperLieAlgebra} is in the inclusion
  of ordinary rational spaces from Example \ref{RationalCohomotopySuperSpaces}.
  This construction extends to a functor from the category of finite-dimensional super Lie algebras to the homotopy category
  of rational super spaces:
  \begin{equation}
    \label{SuperLieAlgebrasToSuperSpaces}
    \mathrm{SuperLieAlg}_{\mathbb{R}}
    \longrightarrow
    \mathrm{Ho}\left(
      \mathrm{SuperSpaces}_{\mathbb{R}}
    \right).
  \end{equation}
\end{example}
\begin{remark}
  \label{Ambiguity}
 Observe that, if, in Example \ref{SuperLieAlgebraAsSuperSpaces}, $\mathfrak{g} = (\mathbb{R}^n_{\mathrm{even}}, [-,-] = 0)$
  is the abelian (hence in particular nilpotent) Lie algebra on $n$-generators, then the
  rational space corresponding to its Chevalley-Eilenberg algebra under the Sullivan equivalence (Prop. \ref{SullivanEquivalence})
  is not the Cartesian space $\mathbb{R}^n$
  (which instead is equivalent to the point in $\mathrm{Ho}(\mathrm{Spaces}_{\mathbb{R}})$) but the $n$-torus
  $\mathbb{T}^n = \mathbb{R}^n/\mathbb{Z}^n$.
  This highlights that in \eqref{SuperSpaceSuperLieAlgebra} the Lie bracket structure is important. Nevertheless, we will often leave this
  notationally implicit, such as in Def. \ref{MinkowskiSuper}.
\end{remark}

\begin{defn}[Spin-invariant Super Minkowski spacetime (e.g. {{\cite[p. 10]{FSS13}}\cite{HuertaSchreiber17}})]
\label{MinkowskiSuper}

\item {\bf (i)} For $p \in \mathbb{N}$,  let $\mathbf{N}$ be a real representation of the group $\mathrm{Spin}(p,1)$ (Def. \ref{LorentzGroupsAndTheirSpinCovers})
which is of real dimension $N \in \mathbb{N}$.\footnote{If there are different real representations of the same real dimension we will distinguish them
by extra decoration of their dimension in boldface, for instance $\mathbf{N}$ and $\overline{\mathbf{N}}$.}
This gives rise to the following DGC-superalgebra (Def. \ref{RationalSuperHomotopyTheory})
\begin{equation}
  \label{CEAlgebraSuperMinkowski}
  \mathrm{CE}\big(
    \mathbb{R}^{p,1\vert \mathbf{N}}/\mathrm{Spin}(p,1)
  \big)
  \;:=\;
  \Bigg(
  \mathbb{R}\big[
    \underset{ \mathrm{deg} = (1,\mathrm{even}) }{\underbrace{(e^a)_{a =0}^p}} , \;
    \underset{  (1, \mathrm{odd}) }{\underbrace{ ( \psi^\alpha )_{\alpha = 1}^{N} }} \;
  \big]/
  {\small
  \left(
    \begin{aligned}
      d e^a & = \overline{\psi}\Gamma^a \psi
      \\[-1mm]
      d \psi^\alpha & = 0
    \end{aligned}
  \right)}
  \Bigg)^{\mathrm{Spin}(p,1)}
  \;\;
  \in
  \mathrm{dgcSuperAlg}
  \,.
\end{equation}
Here the expression $\overline{\psi} \Gamma^a \psi$ on the right is obtained from the spinor-to-vector pairing (see \eqref{ViaDiracConjugateSpinorToVectorPairing} in Appendix), and the superscript $(-)^{\mathrm{Spin}(p,1)}$
means that we consider the sub dgc-algebra of $\mathrm{Spin}(p,1)$-invariant elements inside the dgc-algebra
that is defined in the parenthesis.

\item {\bf (ii)} Accordingly this defines (still by Def. \ref{RationalSuperHomotopyTheory}) superspaces
\begin{equation}
  \label{SuperMinkowskiAsSuperSpace}
  \mathbb{R}^{p,1\vert \mathbf{N}}/\mathrm{Spin}(p,1)
  \;\in\;
  \mathrm{Ho}\left(
    \mathrm{SuperSpaces}_{\mathbb{R}}
  \right)
\end{equation}
which are the incarnation of \emph{super Minkowski spacetimes}, in rational super homotopy theory, such that all
maps out of them are forced to be $\mathrm{Spin}(p,1)$-invariant. Since this is the only case in which we are
interested, we will henceforth suppress the notation for Spin-invariance, and will just write
\begin{equation}
  \label{AbusingNotation}
  \mathbb{R}^{p,1\vert \mathbf{N}}
\end{equation}
for the superspaces in \eqref{SuperMinkowskiAsSuperSpace}. Note that this is deliberate abuse of notation, since the
symbol \eqref{AbusingNotation} more properly refers to the superspace that corresponds to the full dgc-algebra inside the
parenthesis in \eqref{CEAlgebraSuperMinkowski}.
\end{defn}

These super Minkowski spacetimes are a special case of Example \ref{SuperLieAlgebraAsSuperSpaces}, in that $\mathrm{CE}(\mathbb{R}^{p,1\vert \mathbf{N}})$
is the Chevalley-Eilenberg algebra of the super-translation part of the $D = p+1$, $\mathbf{N}$-supersymmetry super Lie algebra.
By Remark \ref{Ambiguity} this means that, when regarding super-Minkowski spacetime as an object in $\mathrm{Ho}\left( \mathrm{SuperSpace}_{\mathbb{R}}\right)$,
it is crucial that we do take the super Lie bracket into account, and that we could make this more explicit by instead of \eqref{SuperMinkowskiAsSuperSpace}
writing
$$
  \mathbb{T}^{p,1\vert \mathbf{N}} \in  \mathrm{Ho}\left( \mathrm{SuperSpace}_{\mathbb{R}}\right).
$$

\begin{example}[Examples of super Minkowski spacetimes]
  \label{ExamplesOfSuperMinkowskiSpacetimes}
  Consider the general construction of Def. \ref{MinkowskiSuper} for the real Spin representations listed in Example \ref{RelevantExamplesOfRealSpinRepresentations}.
  This yields, among others, the following super Minkowski spacetimes (with $D := d + 1$ the total spacetime dimension and $\mathcal{N}$ the ``number of supersymmetries'', according to Remark \ref{NumberOfSupersymmetries}):
  {\footnotesize
  \begin{center}
  \begin{tabular}{|c|c||l|}
  \hline
    Dimension & Supersymmetry   & Super-Minkowski \\
    \hline
    \hline
    $D = 11$ & \hspace{-3mm} $\mathcal{N} = 1$ & $\phantom{ {A \atop A} \atop {A \atop A} }\mathbb{R}^{\mathrlap{10,1\vert \mathbf{32}}}\phantom{ {A \atop A} \atop {A \atop A} }$
    \\
    \hline
    $D = 10$  & \hspace{-3mm} $\mathcal{N} = (1,0)$ & $\phantom{ {A \atop A} \atop {A \atop A} }\mathbb{R}^{\mathrlap{9,1\vert \mathbf{16}}}\phantom{ {A \atop A} \atop {A \atop A} }$
    \\
    \hline
    $D = 10$  & \hspace{-3mm} $\mathcal{N} = (1,1)$ & $\phantom{ {A \atop A} \atop {A \atop A} }\mathbb{R}^{\mathrlap{9,1\vert \mathbf{16} + \overline{\mathbf{16}}}}\phantom{ {A \atop A} \atop {A \atop A} }$
    \\
    \hline
    $D = 10$  & \hspace{-3mm} $\mathcal{N} = (2,0)$ & $\phantom{ {A \atop A} \atop {A \atop A} }\mathbb{R}^{\mathrlap{9,1\vert \mathbf{16} + \mathbf{16} }}\phantom{ {A \atop A} \atop {A \atop A} }$
    \\
    \hline
    $D = 7$  & \hspace{-3mm} $\mathcal{N} = 1$ & $\phantom{ {A \atop A} \atop {A \atop A} }\mathbb{R}^{\mathrlap{6,1\vert \mathbf{16}}}\phantom{ {A \atop A} \atop {A \atop A} }$
    \\
    \hline
    $D = 6$  & \hspace{-3mm} $\mathcal{N} = (2,0)$ & $\phantom{ {A \atop A} \atop {A \atop A} }\mathbb{R}^{\mathrlap{5,1\vert \mathbf{8} + \mathbf{8}}}\phantom{ {A \atop A} \atop {A \atop A} }$
    \\
    \hline
    $D = 6$  & \hspace{-3mm} $\mathcal{N} = (1,0)$ & $\phantom{ {A \atop A} \atop {A \atop A} }\mathbb{R}^{\mathrlap{5,1\vert \mathbf{8} }}\phantom{ {A \atop A} \atop {A \atop A} }$
    \\
    \hline
    $D = 3$  & \hspace{-3mm} $\mathcal{N} = 1$ & $\phantom{ {A \atop A} \atop {A \atop A} }\mathbb{R}^{\mathrlap{2,1\vert \mathbf{2} }}\phantom{ {A \atop A} \atop {A \atop A} }$
    \\
    \hline
    $D = 2$  & \hspace{-3mm} $\mathcal{N} = (1,0)$ & $\phantom{ {A \atop A} \atop {A \atop A} }\mathbb{R}^{\mathrlap{1,1\vert \mathbf{1} }}\phantom{ {A \atop A} \atop {A \atop A} }$
    \\
    \hline
  \end{tabular}
  \end{center}
  }
\end{example}

The following Def. \ref{BPSSuperLieSubalgebra} reflects standard physics terminology for dimensionality of fermionic
subspaces in super spacetimes. This will play a key role in the classification of real ADE-singularities in
Sec. \ref{ADESingularitiesInSuperSpacetime}.
\begin{defn}[BPS super subspaces]
  \label{BPSSuperLieSubalgebra}
  Let $\mathbb{R}^{p,1\vert \mathbf{N}}$ be a super-Minkowski spacetime (Def. \ref{MinkowskiSuper}).
  Then consider a sub-superspace, which is itself a super-Minkowski spacetime of the
  same bosonic dimension but with real Spin representation $\mathbf{N}/\mathbf{k}$ of dimension some fraction $\tfrac{1}{k} N$
  $$
    \mathbb{R}^{p,1\vert \mathbf{N}/\mathbf{k}}
      \xymatrix{\ar@{^{(}->}[r]&}
    \mathbb{R}^{p,1\vert \mathbf{N}}.
  $$
  We call this a \emph{$1/k$ BPS super subspace}.
\end{defn}

Even though super Minkowski spacetimes (Def. \ref{MinkowskiSuper}) are a fairly mild variant of plain Minkowski spacetime,
in contrast to the latter they have interesting ordinary cohomology, in fact \emph{exceptional} cohomology: A finite number of
invariant cocycles appears for special combinations of dimension $D$, number of supersymmetries $\mathcal{N}$ and cocycle degree $p + 2$.
Since these exceptional cocycles witness fundamental branes propagating on these super Minkowski spacetimes,
this classification was  known as the \emph{brane scan}, and has come to be known as the ``old brane scan''
(e.g. \cite[Sec. 2]{DuffLu92}, \cite[Sec. 3.1]{Duff08}), since, interestingly, it misses some branes
(as discussed in Sec. \ref{TheFundamentalBraneScan}).

\begin{prop}[Old Brane Scan ({\cite{AETW87,AzTo89}}, see {\cite{FSS13}})]
 \label{TheOldBraneScan}

Let
\begin{enumerate}
\vspace{-2mm}
 \item $\mathbb{R}^{d,1\vert \mathbf{N}_{\mathrm{irr}}} \in \mathrm{Ho}\left( \mathrm{SuperSpaces}_{\mathbb{R}}\right)$ be one of the
super Minkowski spacetimes from Example \ref{ExamplesOfSuperMinkowskiSpacetimes}, for $\mathbf{N}_{\mathrm{irr}}$
irreducible (i.e. for $\mathcal{N} = 1$, see Remark \ref{NumberOfSupersymmetries});

\vspace{-2mm}
\item $B^{p+2} \mathbb{R} \in \mathrm{Ho}\left( \mathrm{SuperSpaces}_{\mathbb{R}} \right)$
be the image of the Eilenberg-MacLane space of degree $p + 2$ in rational spaces (Example \ref{dgCocyclesAsMaps}), regarded as a
superspace via Example \ref{RationalCohomotopySuperSpaces}.
\end{enumerate}

\vspace{-2mm}
\noindent Then there are nontrivial maps of the form
$$
  \xymatrix{
    \mathbb{R}^{d,1\vert \mathbf{N}}
    \ar[rr]^-{ \mu_{p+2} }
    &&
    B^{p+2} \mathbb{R}
  }
  \;
  \in
  \mathrm{Ho}\left( \mathrm{SuperSpaces}_{\mathbb{R}} \right)
$$
in the homotopy category of rational super spaces (Def. \ref{RationalSuperHomotopyTheory}).
Hence, by Example \ref{dgCocyclesAsMaps}, there is a cohomology class of nontrivial Spin-invariant
cocycles in $\mathrm{CE}(\mathbb{R}^{10,1\vert \mathbf{N}})$ (see \eqref{CEAlgebraSuperMinkowski}),
$$
  [\mu_{p+2}]
  \;\in\;
  H^{p+2}
  \big(
   \mathbb{R}^{p,1\vert \mathbf{N}}
  \big)^{ \mathrm{Spin}(p,1) }
$$
precisely for the combinations $(d,p)$ that are checked in \hyperlink{TableB}{Table B}.
Moreover, for each entry, there is precisely one such map, up to rescaling by $\mathbb{R} \setminus \{0\}$ and,
via the translation in Example \ref{dgCocyclesAsMaps}, it is represented by the element of the form
\begin{equation}
  \label{OldBraneScanCocycle}
  \mu_{p+1}
  \;\propto\;
  \tfrac{1}{p!}
  \left( \overline{\psi} \Gamma_{a_1 \cdots a_p} \psi\right) \wedge e^{a_1} \wedge \cdots \wedge e^{a_p}
  \;\;
  \in
  \mathrm{CE}\big( \mathbb{R}^{d+1 \vert \mathbf{N}_{\mathrm{irr}}} \big)
  \,,
\end{equation}
where we are using spinor notation as in Prop. \ref{RealSpinorsViaMajoranaConditionsOnDiracRepresentations}
We call these elements the \emph{fundamental $p$-brane cocycles}.
\end{prop}

\begin{example}[Fundamental superstring cocycles]
  \label{FundamentalF1Cocycle}

 \item  {\bf (i)} On the $D = 10$, $\mathcal{N} = 1$ super Minkowski spacetime $\mathbb{R}^{9,1\vert \mathbf{16}}$ from Example \ref{ExamplesOfSuperMinkowskiSpacetimes},
  the old brane scan (Prop. \ref{OldBraneScanCocycle}) recognizes a cocycle for a fundamental 1-brane, corresponding to the
  entry $(D = 9+1, p =1)$ in \hyperlink{TableB}{Table B}:
  $$
    \xymatrix@R=-2pt{
      \mathbb{R}^{9,1\vert \mathbf{16}}
      \ar[rr]^-{\mu_{{}_{F1}}^{H/I}}
      &&
      B^3 \mathbb{R}
      \ar[r]^{\simeq_{\mathbb{Q}}}
      &
      S^3.
      \\
      && h_3 & h_3 \ar@{|->}[l]
    }
  $$
  This corresponds to the fundamental \emph{heterotic string} or \emph{type I string} (see Sec. \ref{Physics} for terminology).

\item   {\bf (ii)} On the $D = 3$, $\mathcal{N} = 1$ super Minkowski spacetime $\mathbb{R}^{2,1\vert \mathbf{2}}$ from Example \ref{ExamplesOfSuperMinkowskiSpacetimes},
  the old brane scan (Prop. \ref{TheOldBraneScan}) recognizes a cocycle for a fundamental 1-brane, corresponding to the
  to the entry $(D = 2+1, p = 1)$ in \hyperlink{TableB}{Table B}:
  $$
    \xymatrix@R=-2pt{
      \mathbb{R}^{2,1\vert \mathbf{2}}
      \ar[rr]^-{\mu_{{}_{F1}}^{D=3}}
      &&
      B^3 \mathbb{R}
      \ar[r]
      &
      S^2.
      \\
      &&
      h_3
      &
      \omega_3 \ar@{|->}[l]
      \\
      &&
      0
      &
      \omega_2 \ar@{|->}[l]
    }
  $$
  In both cases we have indicated on the right that, as a morphism in
  $\mathrm{Ho}\left( \mathrm{SuperSpaces}_{\mathbb{R}}\right)$, these cocycles
  may equivalently be regarded as taking values in suitable spheres, by Example \ref{SpheredgcAlgebraModel}
  and Example \ref{RationalCohomotopySuperSpaces}.
  These re-identifications of rational coefficients will be used in the statement and proof of Theorem \ref{RealADEEquivariantEndhancementOfM2M5Cocycle} below.
\end{example}

\begin{remark}[Recognizing fundamental brane cocycles via normed division algebra]
For checking that the 1-brane cochains \eqref{OldBraneScanCocycle} in dimensions 3, 4, 6, 10, and the 2-brane cochains in
dimensions 4, 5, 6, and 11 are indeed cocycles, as claimed by the Old Brane Scan (Prop. \ref{TheOldBraneScan})
it is useful to represent the corresponding real spinor representations in terms of
normed division algebra, as briefly explained in Sec. \ref{SpacetimesAndSpin}. This
streamlined computation is spelled out in \cite{BaezHuerta10, BaezHuerta11}.
\end{remark}

In fact, super Minkowski spacetimes carry more Spin-invariant  cocycles than what is
seen by the old brane scan (Prop. \ref{TheOldBraneScan}),
albeit not in ordinary cohomology, but in generalized cohomology (Def. \ref{CohomoloyFromHomotopy}).

\begin{defn}[M-brane super-cochains {\cite[Def. 4.2]{FSS16b}}]
  \label{M2M5SuperCochains}
  Consider the $D=11$, $\mathcal{N} = 1$ super Minkowski spacetime $\mathbb{R}^{10,1\vert \mathbf{32}}$ from Example \ref{ExamplesOfSuperMinkowskiSpacetimes}.
  We say that the \emph{M-brane super-cochains} are the following two elements in the corresponding DGC-superalgebra $\mathrm{CE}(\mathbb{R}^{10,1\vert \mathbf{32}})$
  \begin{equation}
    \label{MBraneCochains}
    \begin{aligned}
      \mu_{{}_{M2}}
        & :=
      \tfrac{i}{2} \overline{\psi} \Gamma_{a_1 a_2} \psi \wedge e^{a_1} \wedge e^{a_2}\;,
      \\
      \mu_{{}_{M5}}
        & :=
      \tfrac{1}{5!} \overline{\psi} \Gamma_{a_1 \cdots a_5} \psi \wedge e^{a_1} \wedge \cdots \wedge e^{a_5}\;,
    \end{aligned}
  \end{equation}
  where we are using spinor notation as in Prop. \ref{RealSpinorsViaMajoranaConditionsOnDiracRepresentations}.
\end{defn}

Observe that, by Prop. \ref{TheOldBraneScan}, $\mu_{{}_{M2}}$ is a cocycle in ordinary cohomology of degree 4, but not $\mu_{{}_{M5}}$
is not a cocycle by itself.

\medskip
Our investigations below revolve around the following exceptional structure in the rational super homotopy theory;
see Sec. \ref{TheFundamentalBraneScan} for discussion of its \emph{physical meaning}.

\begin{prop}[M2/M5-brane cocycle in rational super cohomotopy ({\cite[Prop. 4.3]{FSS16b}, \cite[Cor. 2.3]{FSS16a}})]
  \label{M2M5SuperCocycle}
  The M-brane super-cochains $\mu_{{}_{M2}}, \mu_{{}_{M5}}$ from Def. \ref{M2M5SuperCochains}
  constitute a cocycle on $D=11$ super Minkowski spacetimes (Example \ref{ExamplesOfSuperMinkowskiSpacetimes})
  with values in degree four rational cohomotopy (Example \ref{RationalCohomotopyOfSuperSpaces}), as follows:
  \begin{equation}
    \label{M2M5CocycleMap}
    \xymatrix@R=-2pt{
      \mathbb{R}^{10,1\vert \mathbf{32}}
      \ar[rr]^-{\mu_{{}_{M2/M5}}}
      &&
      S^4
      \\
      \mu_{{}_{M2}}
      &&
      \omega_4\ar@{|->}[ll]
      \\
      \mu_{{}_{M5}}
      &&
      \omega_7\ar@{|->}[ll]
    }
    \phantom{A}
    \in \mathrm{Ho}\left( \mathrm{SuperSpaces}_{\mathbb{R}} \right).
  \end{equation}
  Here on the right $\omega_4$, $\omega_7$ are the two generators of the minimal dgc-algebra model for the 4-sphere,
  from Example \ref{SpheredgcAlgebraModel}, and as a homomorphism of dgc-superalgebras,
  $\mu_{{}_{M2/M5}}$ takes them to the two
  M-brane cochains $\mu_{{}_{M2}}$, $\mu_{{}_{M5}}$ from Def. \ref{M2M5SuperCochains}, respectively.
\end{prop}

Now we turn to the equivariant refinement of rational super homotopy theory:

\begin{defn}[$G$-equivariant abstract homotopy theory]
  \label{EquivariantAbstractHomotopyTheory}
  Let $(\mathcal{C}, W)$ be a homotopy theory embodied by a category with weak equivalences (Def. \ref{CategoryWithWeakEquivalences}),
  and let $G$ be a finite group.
  Then the corresponding \emph{$G$-equivariant homotopy theory} is the category

  \vspace{-3mm}
  \begin{equation}
    \label{SystemsOfFixedPointsInAbstractHomotopyTheory}
    \mathrm{PSh}\left( \mathrm{Orb}_G, \mathcal{C} \right)
    \;=\;
    \left\{
      \raisebox{30pt}{
      \xymatrix@R=4pt{
        \mathrm{Orb}_G^{\mathrm{op}}
        \ar[rr]^-{X(-)}
        &&
        \mathcal{C}
        \\
        G/H_1 \ar[dd]_{[g]} && X(H_1)
        \\
        & \longmapsto
        \\
        G/H_2 && X(H_2) \ar[uu]_{X(g)}
      }
      }
    \right\}
  \end{equation}
  of systems $X(-)$ of objects $X(H)$ of  $\mathcal{C}$ parameterized by the category of orbit spaces $G/H$ of the group $G$ (Def. \ref{OrbitCategory})
  whose weak equivalences are those natural transformations \eqref{PresheavesOnOrbitCatgegory}
  \begin{equation}
    \xymatrix@R=6pt{
      & X_1(-) \ar[rr]^-{ F(-) } && X_2(-)
      \\
      G/H_1 \ar[dd]_g & X_1(H_1) \ar[rr]^{F(H_1)} && X_2(H_1)
      \\
      \\
      G/H_2  & X_1(H_2) \ar[rr]^{F(H_1)} \ar[uu]^{ X_1(g) } && X_2(H_2) \ar[uu]_{ X_2(g) }
    }
  \end{equation}
 such that for each subgroup $H \subset G$ the component
  $F(H)$ from \eqref{PresheafMapComponent} is a weak equivalence of $\mathcal{C}$ (an element of $W$).
  We denote  the resulting homotopy category (Def. \ref{CategoryHomotopy}) by
  $$
    \mathrm{Ho}\left(
      G \mathcal{C}
    \right)
    \;:=\;
    \mathrm{Ho}\Big(
      \mathrm{PSh}\big(
        \mathrm{Orb}_G, \mathcal{C}
      \big)
      \big[
        \left\{
          \mbox{subgroup-wise weak equivalences of $\mathcal{C}$}
        \right\}^{-1}
      \big]
    \Big).
  $$
  The operation that extracts from systems \eqref{SystemsOfFixedPointsInAbstractHomotopyTheory} of
  fixed point loci in $\mathcal{C}$ just the total spaces, hence ``forgetting'' the group action, constitutes
  a ``forgetful functor'' from
  $G$-equivariant abstract homotopy theory to the underlying plain homotopy theory:
  \begin{equation}
    \label{ForgetfulFromGEquivariantToUnderlying}
    \raisebox{20pt}{
    \xymatrix{
      X(-)
      \ar@{|->}[d]
      &\in & \mathrm{Ho}\left( G\mathcal{C}\right)
      \ar[d]^{\mbox{\tiny forget equivariance}}
      \\
      X(1) & \in & \mathrm{Ho}(\mathcal{C})
    }
    }
  \end{equation}
\end{defn}
\begin{example}[Equivariant rational homotopy theory]
  \label{EquivariantrationalHomotopyTheory}
  The construction in Def. \ref{EquivariantAbstractHomotopyTheory}
  applied to dg-algebraic homotopy theory (Def. \ref{dgAlgebrasAnddgModules}) yields the homotopy theory shown on the right in \eqref{EquivariantRationnalHomotopyTheorydg}.
  The Sullivan equivalence (Prop. \ref{SullivanEquivalence}) extends to exhibit this as equivalent to the corresponding rational version of
  $G$-equivariant homotopy theory.
\begin{equation}
  \label{EquivariantRationnalHomotopyTheorydg}
  \mathrm{Ho}\left(
    G \mathrm{Spaces}_{\mathbb{Q}, \mathrm{nil}, \mathrm{fin}}
  \right)
  \;\simeq\;
  \mathrm{Ho}\left(
    \mathrm{PSh}\big(\mathrm{Orb}_G, \mathrm{dgAlg}^{\mathrm{op}}_{\mathrm{cn},\mathrm{fin}}\big)
    \big[\!
      \left\{\mbox{subgroup-wise quasi-isomorphisms}\right\}^{-1}
    \! \big] \!
  \right)\!.
\end{equation}
This is the DG-algebraic model for \emph{equivariant rational homotopy theory},
discussed in \cite{Scull01}\cite[Ch. III]{Ma96}\cite[Sec. 3.3 \& 3.4]{AP}.
\end{example}

The super-geometric homotopy theory that we need here is now obtained from Example \ref{EquivariantrationalHomotopyTheory}
by generalizing dgc-algebras to dgc-superalgebras (Def. \ref{RationalSuperHomotopyTheory}):
\begin{defn}[$G$-equivariant rational super homotopy theory]
  \label{GEquivariantRationalSuperHomotopy}
 For $G$ a finite group (or, more generally, a compact Lie group) the
 \emph{$G$-equivariant rational super homotopy theory} is the result of applying
 the general construction of $G$-equivariant homotopy theory
 from Example \ref{EquivariantrationalHomotopyTheory} to
 the rational super homotopy theory of Def. \ref{RationalSuperHomotopyTheory}:
$$
  \mathrm{Ho}\left(
    G{\color{blue}\mathrm{Super}}\mathrm{Spaces}_{\mathbb{R},\mathrm{nil},\mathrm{fin}}
  \right)
 \! :=
  \mathrm{Ho} \hspace{-1mm}\left(
    \mathrm{PSh}\left(
      \mathrm{Orb}_G, \mathrm{dgc}{\color{blue}\mathrm{Super}}\mathrm{Alg}^{\mathrm{op}_{\mathrm{cn},\mathrm{fin}}}
    \right)
    \hspace{-1mm}\big[\!
      \left\{\mbox{subgroup-wise quasi-isomorphisms}\right\}^{-1}
    \! \big] \!
  \right)\!.
$$
\end{defn}
\begin{example}[$G$-Spaces as $G$-Superspaces]
  \label{GSpacesAsGSuperspaces}
  Let $G$ be a finite group and let $A$ be  $G$-space (Def. \ref{GSpace}) such that for all subgroups $H \subset X$
  the fixed point space $A^H$ (Def. \ref{SpacesAssociatedWithAGAction}) is nilpotent and of finite rational type (Def. \ref{RationalHomotopyTheory}).
  Then the Sullivan equivalence (Prop. \ref{SullivanEquivalence}) implies that the system of fixed point spaces (Example \ref{ExternalYonedaFromGSpacesToPresheavesOnOrbitCategory}) of the rationalization of $A$, is equivalently given by the corresponding system of
  tuples of connected finite-type dgc-algebras

 \vspace{-2mm}
  $$
    \mathcal{O}\big( A^{(-)}\big)
    \;:\;
    \raisebox{28pt}{
    \xymatrix@R=1.2em{
      G/H_1
      \ar[dd]_{[g]}
      && \mathcal{O}(A^{H_1})
      \\
      \\
      G/H_2 && \mathcal{O}(A^{H_2}) \ar@{<-}[uu]_{\mathcal{O}(A^{[g]})}
    }
    }
  $$
  By Example \ref{RationalCohomotopySuperSpaces} this defines a rational $G$-superspace, which we will denote by the same symbol:
  $$
    A \in \mathrm{Ho}\left( G \mathrm{SuperSpaces}_{\mathbb{R}}\right)
    .
  $$
\end{example}
\begin{example}[Super Lie algebras with $G$-action as $G$-Superspaces]
  \label{SuperLieAlgebrasWithGActionAsGSuperspaces}
  Let $\mathfrak{g} \simeq_{\mathbb{R}} \mathfrak{g}_{\mathrm{even}} \oplus \mathfrak{g}_{\mathrm{odd}}$ be a finite-dimensional super Lie algebra,
  regarded as a rational superspace as in Example \ref{SuperLieAlgebraAsSuperSpaces}. Consider an action of a finite group $G$
  on $\mathfrak{g}$ by super Lie algebra automorphisms
  $$
    \rho(g)
    \;:\;
    \mathfrak{g}  \overset{\simeq}{\longrightarrow} \mathfrak{g}
    \,,
    \phantom{AA}
    g \in G
    \,.
  $$
  These are, equivalently, linear actions on the underlying real vector spaces $\mathfrak{g}_{\mathrm{even}}$
   and $\mathfrak{g}_{\mathrm{odd}}$, respecting the super Lie bracket. Hence, for each subgroup $H \subset G$,
   the fixed point spaces \eqref{FixedPointSpace} constitute a super Lie subalgebra:
  $$
    \mathfrak{g}^H
      \simeq
    (\mathfrak{g}_{\mathrm{even}})^H
      \oplus
    (\mathfrak{g}_{\mathrm{odd}})^H \subset \mathfrak{g}
    \,.
  $$
  As in Example \ref{FixedPointSpaceRepresented} \eqref{InducedMapOnFixedPoints}, this yields a system of fixed super Lie algebras, indexed by the orbit category (Def. \ref{OrbitCategory}):
   \vspace{-3mm}
  \begin{equation}
    \label{SuperLieAlgebraSysyemOfFixedLoci}
    \xymatrix@R=6pt{
      \mathrm{Orb}_G^{\mathrm{op}}
      \ar[rr]^-{ \mathfrak{g}^{(-)} }
      &&
      \mathrm{SuperLieAlg}
      \\
      G/H_1
      \ar[dd]_{g}
      &&
      \mathfrak{g}^{H_1}
      \\
      \\
      G/H_2
      &&
      \mathfrak{g}^{H_2} \ar[uu]_{ \rho(g) }
    }
  \end{equation}
  Via \eqref{SuperLieAlgebrasToSuperSpaces} this represents a $G$-superspace (Def. \ref{GEquivariantRationalSuperHomotopy}), which we will denote by
  the same symbol (see Remark \ref{Ambiguity}):
  $$
    \mathfrak{g}
    \;\in\;
    \mathrm{Ho}\left(
      G \mathrm{SuperSpaces}_{\mathbb{R}}
    \right).
  $$
\end{example}

Finally we may combine these examples to say what it means to enhance super-cocycles to equivariant rational super homotopy theory
(see Def. \ref{EnhancementToEquivariantCohomology} and Remark \ref{EquivariantEnhancementIsExtraStrujcture} for
the corresponding discussion in plain homotopy theory):
\begin{example}[Equivariant enhancement of super-cocycles]
  \label{EquivariantEnhancementOfSuperCocycles}
  Consider
  \begin{itemize}
  \vspace{-2mm}
    \item $\mathfrak{g} \in \mathrm{SuperLieAlg} \longrightarrow \mathrm{Ho}\left(\mathrm{SuperSpaces}_{\mathbb{R}}\right)$
     a super Lie algebra, regarded as a superspace via Example \ref{SuperLieAlgebraAsSuperSpaces};
   \vspace{-2mm}
    \item $A \in \mathrm{Spaces}_{\mathrm{nil},\mathrm{fin}} \longrightarrow \mathrm{Ho}\left( \mathrm{SuperSpaces}_{\mathbb{R}} \right)$
     a nilpotent space of finite rational type, regarded as a superspace via Example \ref{RationalCohomotopySuperSpaces};
    \vspace{-2mm}
    \item
     $\mathfrak{g} \overset{\mu}{\longrightarrow}\; A  \in \mathrm{Ho}\left( \mathrm{SuperSpaces}_{\mathbb{R}} \right)$
     a morphism between these in the homotopy category, hence (Def. \ref{CohomoloyFromHomotopy}) a class in the
     super rational $A$-cohomology of $\mathfrak{g}$.
  \end{itemize}
  Then for $G$ a finite group, a \emph{$G$-equivariant enhancement} of $\mu$ is a lift through the forgetful functor \eqref{ForgetfulFromGEquivariantToUnderlying}
  \vspace{-2mm}
  $$
    \xymatrix@C=5pt@R=1em{
      \mbox{
        \color{gray}
        \tiny
        \begin{tabular}{c}
          equivariant
          \\
          enhancement
        \end{tabular}
      }
      &
      \color{gray}
      \mathfrak{g}
      \ar@[gray][rr]^-{\color{gray}\mu}
      &&
      \color{gray} A
       &\in & \mathrm{Ho}\left(G\mathrm{SuperSpaces}_{\mathbb{R}}\right)
      \ar[dd]^{\mbox{ \tiny forget $G$-equivariance }}
      \\
      \\
      \mbox{\tiny cocycle}
      &
      \mathfrak{g}
      \ar[rr]^-{\mu}
      &&
      A
      &\in & \mathrm{Ho}\left(\mathrm{SuperSpaces}_{\mathbb{R}}\right)
    }
  $$
  By unwinding Def. \ref{EquivariantAbstractHomotopyTheory},
  such an enhancement amounts to and encodes all of the following extra data:
  \begin{enumerate}
  \vspace{-2mm}
    \item a system $A(-)$ of super-spaces indexed by the orbit category, with $A(1) = A$ the given coefficient object, for example given by an actual $G$-space structure on $A$ (Def. \ref{GSpace}) on the topological space $A$, via Example \ref{GSpacesAsGSuperspaces};
    \vspace{-2mm}
    \item a system $\mathfrak{g}(-)$ of super-spaces indexed by the orbit category, with $\mathfrak{g}(1) = \mathfrak{g}$ the given domain object, for example given be an actual action of $G$ by super Lie algebra automorphisms,
    as in Example \ref{SuperLieAlgebrasWithGActionAsGSuperspaces};
    \vspace{-2mm}
    \item compatible $G$-equivariant cocycle structure on $\mu$, which means:
     \begin{enumerate}
     \vspace{-2mm}
       \item for each non-trivial subgroup $H \subset G$ a super-cocycle of the form
       $$
         \mu(H) \;:\; \mathfrak{g}(H) \overset{}{\longrightarrow} A(H)
         \;\;\;\;
         \in
         \mathrm{Ho}\left(  \mathrm{SuperSpaces}_{\mathbb{R}} \right);
       $$
       \vspace{-6mm}
      \item for each $G$-equivariant map $G/H_1 \overset{f}{\longrightarrow} G/H_2$ between orbit spaces
       a choice of homotopy\footnote{
      This follows by representing fibrant resolutions in terms of homotopies, as in Example \ref{HomomorphismsOfSystemsOfFixedPointsUpToHomotopy}, in particular \eqref{ExampleForFibrantResolution},\eqref{HomotopyFromPathSpace}.
    }
      \begin{equation}
        \label{HomotopyDataForEnhancements}
        \raisebox{40pt}{
        \xymatrix@R=1.3em{
           G/H_1
           \ar[dd]_f
           &&
           \mathfrak{g}(H_1) \ar[rr]^{ \mu(H_1) }_<{\ }="s"    && A(H_1)
           \\
           \\
           G/H_2
           &&
           \mathfrak{g}(H_2) \ar[uu]^{\mathfrak{g}(f)} \ar[rr]_{\mu_{H_2}}^>{\ }="t" && A(H_2) \ar[uu]_{\mathfrak{g}(f)}
           \ar@{=>}^{\mu(f)} "s"; "t"
        }
        }
      \end{equation}
     \end{enumerate}
     such that the assignment $f \mapsto \mu(f)$ respects composition and identities.
  \end{enumerate}
\end{example}

In particular we may study the possible equivariant enhancements, as in Example \ref{EquivariantEnhancementOfSuperCocycles},
of the M2/M5-brane cocycle from Prop. \ref{M2M5SuperCocycle}.
This is what we turn to next.

\section{Real ADE-Singularities in Super-spacetimes}
  \label{ADESingularitiesInSuperSpacetime}

Here we classify finite group actions on 11-dimensional superspacetime $\mathbb{R}^{10,1\vert \mathbf{32}}$
(Def.\ \ref{MinkowskiSuper}) which have the same bosonic fixed locus as an involution (this is
Prop. \ref{ClassificationOfZ2Actions} below), and whose full super fixed point locus
\begin{equation}
  \label{Superembedding}
  \xymatrix{
    \mathbb{R}^{p,1\vert\mathbf{N}}
    ~\ar@{^{(}->}[rr]
    &&
    \mathbb{R}^{10,1\vert \mathbf{32}}
    \ar@(ul,ur)^{G}
  }
\end{equation}
is at least $\geq \sfrac{1}{4}$-BPS (Def.\ \ref{BPSSuperLieSubalgebra}).
These inclusions \eqref{Superembedding} of
BPS super fixed loci are \emph{superembeddings}
in the sense of \cite{Sorokin99, Sorokin01}.

\medskip
Using results of \cite{MFFGME10}, \cite[Sec. 8.3]{MF10},
we find that these are, if orientation-preserving, given by subgroups of $\mathrm{SU}(2)$, hence by finite groups in the ADE-series
(this is Prop. \ref{M2AndMK6ADE} below). We collect this classification as Theorem \ref{SuperADESingularitiesIn11dSuperSpacetime} below.
The corresponding quotient spaces constitute a supergeometric refinement of the du Val singularities in Euclidean space (Sec. \ref{TheMK6}).
By intersecting several such simple singularities one obtains more actions of interest. We discuss some of these
non-simple `super singularities' in Prop. \ref{NonSimpleRealSingularities}.
In the following Sec. \ref{ADEEquivariantMBraneSuperCucycles} we demonstrate how these fixed point superspaces
appear as part of the data of the equivariant cohomotopy of 11d superspacetime. Earlier, in Sec. \ref{TheFundamentalBraneScan},
we gave a  physical interpretation of these results: we interpret them as black brane species localized at singularities.

\medskip
In what follows, as we work with the super Minkowski spacetimes $\R^{p,1|{\bf N}}$,
we will often refer to the even part $\R^{p,1}$ as \emph{bosonic} and the odd part
${\bf N}$ as \emph{fermionic}.

\begin{defn}[Super singularities]
  \label{SimpleSingularities}
  Let $G$ be a finite group acting on super-Minkowski spacetime $\mathbb{R}^{p,1\vert \mathbf{N}}$ (Def.\ \ref{MinkowskiSuper})
  by isometries (Example \ref{SuperLieAlgebrasWithGActionAsGSuperspaces}).
  If every non-trivial subgroup $\{e\} \neq H \subset G$ has the same bosonic fixed space, we say that this action exhibits
  a \emph{simple singularity}. Otherwise we call it \emph{non-simple}.
\end{defn}
\begin{example}[Systems of fixed subspaces and brane intersections]
  \label{ConstantSystemsOfFixedSubspaces}
Consider a system of fixed subspaces $\left( \mathbb{R}^{10,1\vert \mathbf{32}}\right)^{(-)}$ (see \eqref{SuperLieAlgebraSysyemOfFixedLoci})
indexed by the orbit category (Def.\ \ref{OrbitCategory}), that corresponds to a given action on super Minkowski spacetime (Example \ref{SuperLieAlgebrasWithGActionAsGSuperspaces}).
In terms of this systems of fixed subspaces, the singularity is \emph{simple}, according to Def.\ \ref{SimpleSingularities},
if the system of underlying bosonic subspaces is a \emph{constant functor} on the orbit category, away from the trivial subgroup. In the following diagram, we display such a functor with constant bosonic value $\R^{p,1}$:

\vspace{-.2cm}

$$
  \xymatrix@C=12pt@R=7pt{
    & \mathrm{Orb}^{\mathrm{op}}_G
    \ar[rrrrrr]^{ \left( \mathbb{R}^{10,1\vert \mathbf{32}}\right)^{(-)} }
    &&& &&&
    \mathrm{SuperSpaces}_{\mathbb{R}}
    \\
    & G/\{e\}
    \ar@{->>}[dr]
    \ar@{->>}[dl]
    && &&&& \mathbb{R}^{10,1\vert \mathbf{32}} \; 
    \\
    G/H_1
    \ar@{->>}[dr]
    &\cdots &
    G/H_n
    \ar@{->>}[dl]
    &&\longmapsto&&
    \mathbb{R}^{p, 1 \vert \mathbf{N_1}}
    \;\ar@{^{(}->}[ur]
    &\cdots&
    \mathbb{R}^{p, 1 \vert \mathbf{N_n}}
    \; \ar@{_{(}->}[ul]
    \\
    & G/G
    &&& &&&
    \mathbb{R}^{p, 1 \vert \mathbf{N_{\rm int}}}
    \ar@{=}[ul]
    \ar@{=}[u]
    \ar@{=}[ur]
  }
$$
We may interpret a simple singularity as reflecting an `elementary brane'.

\medskip
\noindent {\bf System of fixed subspaces of a \emph{simple singularity}}, in the special case when not only the bosonic fixed space,
but also the preserved spinors are independent of which non-trivial subgroup acts (for instance the case of the MK6 in Theorem \ref{SuperADESingularitiesIn11dSuperSpacetime}).

\medskip
On the other hand, a non-simple singularity corresponds to a system of fixed subspaces that is non-constant even away from the
trivial subgroup. Here is such a functor:

$$
  \xymatrix@C=12pt@R=7pt{
    & \mathrm{Orb}^{\mathrm{op}}_G
    \ar[rrrrrr]^{ \left( \mathbb{R}^{10,1\vert \mathbf{32}}\right)^{(-)} }
    &&& &&&
    \mathrm{SuperSpaces}_{\mathbb{R}}
    \\
    & G/\{e\}
    \ar@{->>}[dr]
    \ar@{->>}[dl]
    && &&&& \mathbb{R}^{10,1\vert \mathbf{32}}
    \\
    G/H_1
    \ar@{->>}[dr]
    &\cdots &
    G/H_n
    \ar@{->>}[dl]
    &&\longmapsto&&
    \mathbb{R}^{p_1, 1 \vert \mathbf{N}_1}
    \ar@{^{(}->}[ur]
    &\cdots&
    \mathbb{R}^{p_n, 1 \vert \mathbf{N}_n}
    \ar@{_{(}->}[ul]
    \\
    & G/G
    &&& &&&
    \mathbb{R}^{p_{\mathrm{int}}, 1 \vert \mathbf{N}_{\mathrm{int}}}
    \ar@{_{(}->}[ul]
    \ar@{^{(}->}[u]
    \ar@{^{(}->}[ur]
  }
$$
We have labeled the subspace $\R^{p_{\rm int}|{\bf N_{\rm int}}}$, fixed by the entire group action, by \emph{int} for intersection. This is because we interpret non-simple singularities as `intersecting branes'.

\medskip
\noindent {\bf System of fixed super subspaces of a \emph{non-simple singularity}.}
The fixed locus $\mathbb{R}^{p_{\mathrm{int}},1\vert \mathbf{N}_{\mathrm{int}}}$ of the full group $G$ is exhibited as the non-trivial intersection of
fixed loci of some non-trivial subgroups $H_k$, as in Def.\ \ref{IntersectionOfSimpleSingularities} below; for instance
the intersection $\mathrm{M2} \dashv \mathrm{M5}$ in Prop. \ref{NonSimpleRealSingularities} below.
\end{example}

In the following table, the notation $(\mathbb{Z}_2)_{L,R}$ refers to the induced action of the center of $\mathrm{SU}(2)_{L,R}$,
while  $\overset{f}{\subset}$ denotes a subgroup inclusion which factors through a given group homomorphism $f$. We will also
encounter $\Delta$, which denotes the diagonal, and $\tau$, which denotes
a non-trivial outer automorphism.


\vspace{1cm}
\hspace{-1cm}
\hypertarget{SingularitiesTable}{
\scalebox{.8}{
\begin{tabular}{|lc|ccc|lc|ccccccccc|}
  \hline
  \multicolumn{2}{|c}{
  \multirow{2}{*}{
  \begin{tabular}{c}
    \bf Black brane
    \\
    \bf species
  \end{tabular}
  }
  }
  &
  \multicolumn{3}{|c|}{
  \multirow{2}{*}{
    {\bf BPS}
  }
  }
  &
  \multicolumn{2}{c|}{
    \begin{tabular}{c}
    \bf Fixed
    \\
    \bf locus
    \end{tabular}
  }
  &
  \begin{tabular}{c}
    {\bf Type of}
    \\
    {\bf singularity}
  \end{tabular}
  &
  \multicolumn{8}{l|}{\hspace{2.5cm} {\bf Intersection law} }
  \\
  &&&&&
  \multicolumn{2}{|c|}{
    {\bf in} $\mathbb{R}^{10,1\vert \mathbf{32}}$
  }
  &
   {\bf in} $\mathbb{R}^{10,1}$
  & $\simeq$
  & $\mathbb{R}^{1,1}$
  & $\!\!\!\!\!\!\!\!\oplus\!\!\!\!\!\!\!\!$
  & $\mathbb{R}^4$
  & $\!\!\!\!\!\!\!\!\oplus\!\!\!\!\!\!\!\!$
  & $\mathbb{R}^4$
  & $\!\!\!\!\!\!\!\!\oplus\!\!\!\!\!\!\!\!$
  & $\mathbb{R}^1$
  \\
  \hline
  \multicolumn{7}{l}{ \bf Black brane species  $ \phantom{ {{{A^A}^A}^A}^A} $  }
  &
  \multicolumn{9}{l}{ \bf Simple singularities }
  \\
  \hline
  $\mathrm{MO9}$ &
  &
  $\sfrac{1}{2}$ &&
  &
  $\mathbb{R}^{9,1\vert \mathbf{16}}$
  &
  &
  $\mathbb{Z}_{2}$
  &
  $\!\!\overset{\phantom{\Delta}}{=}\!\!$
  &
  \multicolumn{5}{c}{
    --------------------------------------------
  }
  &
  &
  $(\mathbb{Z}_2)_{\mathrm{HW}}$
  \\
  \hline
  $\mathrm{MO5}$ &
  &
  $\sfrac{1}{2}$ &&
  &
  $\mathbb{R}^{5,1\vert 2 \cdot \mathbf{8}}$ &
  &
  $\mathbb{Z}_{2}$
  &
  $\!\!\overset{\Delta}{\subset}\!\!$
  &
  \multicolumn{4}{c}{
    ------------------------------
  }
  &
  $(\mathbb{Z}_2)_R$
  &
  $\!\!\!\!\times\!\!\!\!\!$
  &
  $(\mathbb{Z}_2)_{\mathrm{HW}}$
  \\
  \hline
  $\mathrm{MO1}$ &
  &
  $\sfrac{1}{2}$ &&
  &
  $\mathbb{R}^{1,1\vert 16 \cdot \mathbf{1}}$ &
  &
  $\mathbb{Z}_{2}$
  &
  $\!\!\overset{\Delta}{\subset}\!\!$
  &
  ------
  &
  &
  $(\mathbb{Z}_2)_L$
  &
  $\!\!\!\!\times\!\!\!\!\!$
  &
  $(\mathbb{Z}_2)_R$
  &
  $\!\!\!\!\times\!\!\!\!\!$
  &
  $(\mathbb{Z}_2)_{\mathrm{HW}}$
  \\
  \hline
  \hline
  $\mathrm{MK6}$ &
  &
  $\sfrac{1}{2}$ &&
  &
  $\mathbb{R}^{6,1\vert \mathbf{16}}$ &
  &
  \begin{tabular}{c}
    $\mathbb{Z}_{n+1}, 2 \mathbb{D}_{n+2}$,
    \\
    $2T, 2O  , 2I$
  \end{tabular}
  &
  $\!\!\overset{}{\subset}\!\!$
  &
  \multicolumn{4}{c}{
   ------------------------------
  }
  &
  $\mathrm{SU}(2)_R$
  &
  &
  ------
  \\
  \hline
  $\mathrm{M2}$ &
  &
  $\sfrac{1}{2}$  &\hspace{-2.2cm}   $=\!\!\!\!\!\!$& \hspace{-1.6cm} $\sfrac{8}{16}$
  &
  $\mathbb{R}^{2,1\vert 8 \cdot \mathbf{2}}$ &
  &
  $\mathbb{Z}_{2}$
  &
  $\!\!\overset{\Delta}{\subset}\!\!$
  &
  ------
  &&
  $\mathrm{SU}(2)_L$
  &
  $\!\!\!\!\times\!\!\!\!\!$
  &
  $\mathrm{SU}(2)_R$
  &&
  ------
  \\
  \hline
  $\mathrm{M2}$ &
  &
  && \hspace{-1.6cm} $\sfrac{6}{16}$
  &
  $\mathbb{R}^{2,1\vert 6 \cdot \mathbf{2}}$ &
  &
  $\mathbb{Z}_{n+3}$
  &
  $\!\!\overset{\Delta}{\subset}\!\!$
  &
  ------
  &
  &
  $\mathrm{SU}(2)_L$
  &
  $\!\!\!\!\!\!\!\!\times\!\!\!\!\!\!\!\!\!$
  &
  $\mathrm{SU}(2)_R$
  &&
  ------
  \\
  \hline
  $\mathrm{M2}$ &
  &
  &&  \hspace{-1.6cm} $\sfrac{5}{16}$
  &
  $\mathbb{R}^{2,1\vert 5 \cdot \mathbf{2}}$ &
  &
  \begin{tabular}{c}
    $2 \mathbb{D}_{n+2}$,
    \\
    $2T, 2O  , 2I$
  \end{tabular}
  &
  $\!\!\overset{\Delta}{\subset}\!\!$
  &
  ------
  &
  &
  $\mathrm{SU}(2)_L$
  &
  $\!\!\!\!\!\times\!\!\!\!\!$
  &
  $\mathrm{SU}(2)_R$
  &&
  ------
  \\
  \hline
  $\mathrm{M2}$ &
  &
  $\sfrac{1}{4}$ && \hspace{-1.8cm}$=\sfrac{4}{16}$
  &
  $\mathbb{R}^{2,1\vert 4 \cdot \mathbf{2}}$ &
  &
  \begin{tabular}{c}
    $2 \mathbb{D}_{n+2}$,
    \\
    $2O  , 2I$
  \end{tabular}
  &
  $\!\!\overset{(\mathrm{id},\tau)}{\subset}\!\!$
  &
  ------
  &&
  $\mathrm{SU}(2)_L$
  &
  $\!\!\!\!\times\!\!\!\!\!$
  &
  $\mathrm{SU}(2)_R$
  &&
  ------
  \\
  \hline
  \multicolumn{7}{l}{ \bf Bound/intersecting brane species  $ \phantom{ {{{A^A}^A}^A}^A} $  }
  &
  \multicolumn{9}{l}{ \bf Non-simple singularities }
  \\
  \hline
  \multirow{3}{*}{ $\mathrm{M2}$ }
  & $\mathrm{M2}$
  &
  \multirow{3}{*}{
    $\sfrac{1}{4}$
  }
  &&&
  \multirow{3}{*}{
    $\mathbb{R}^{2,1\vert 4 \cdot \mathbf{2}}$
  }
  & $\mathbb{R}^{2,1\vert 5 \cdot \mathbf{2}}$
  &
  \begin{tabular}{c}
    $2\mathbb{D}_{n+2}$
    \\
    $ 2T, 2O, 2I$
  \end{tabular}
  &
   $\overset{
     \Delta
   }{\subset}$
  &
  &&
  $ \mathrm{SU}(2)_L $
  &
  $\!\!\!\!\!\times\!\!\!\!\!$
  &
  $\mathrm{SU}(2)_R$
  &&
  \\
  &
  \raisebox{5pt}{
  \hspace{4pt}\begin{rotate}{90} $\Vert$ \end{rotate}
  }
  &
  &&
  &
  &
  \raisebox{5pt}{
  \begin{rotate}{90} $\Vert$ \end{rotate}
  }
  &
  $\!\!\!\!\times\!\!\!\!\!$
  &&
  ------
  &&&&&&
  ------
  \\
  &
  $\mathrm{MK6}$
  &&&
  &&
    $\mathbb{R}^{6,1\vert \mathbf{16}}$
  &
  $\mathbb{Z}_2$
  &
  $=$
  &&
  &
  &&
  $ (\mathbb{Z}_2)_R $
  &&
  \\
  \hline
  \multirow{3}{*}{
    $\mathrm{NS1}_H$
  }
  & $\mathrm{M2}$
  &
  \multirow{3}{*}{
    $\sfrac{3}{16}$ or $\sfrac{1}{4}$
  }
  &&
  &
  & $\mathbb{R}^{2,1\vert \geq 6 \cdot \mathbf{2}}$
  &
  $\mathbb{Z}_{n+1}$
  &
    $\overset{\Delta}{\subset}$
  &
  &&
   $ \mathrm{SU}(2)_L $
  &
  $\!\!\!\!\times\!\!\!\!\!$
  &
    $ \mathrm{SU}(2)_R $
  &&
  \\
  & $\bot$
  &
  &&&
  $\mathbb{R}^{1,1\vert \geq 6 \cdot \mathbf{1}}$ &  $\bot$
  &
  $\!\!\!\!\times\!\!\!\!\!$
  &&
  ------
  &&&&&&
  \\
  & $\mathrm{MO9}_H$
  &&&&
  & $\mathbb{R}^{9,1\vert \mathbf{16}}$
  &
  $\mathbb{Z}_2$
  & $=$ &
  &&&&&&
  $(\mathbb{Z}_2)_{\mathrm{HW}}$
  \\
  \hline
  \multirow{3}{*}{ $\mathrm{M1}$ }
  & $\mathrm{M2}$
  &
  \multirow{3}{*}{
    $\sfrac{3}{16}$ or $\sfrac{1}{4}$
  }
  &&
  &
  \multirow{3}{*}{ $\mathbb{R}^{1,1\vert \geq 6 \cdot \mathbf{1}}$  }
  & $ \mathbb{R}^{2,1\vert \geq 6 \cdot \mathbf{2}}$
  &
  $\mathbb{Z}_{n+1}$
  &
    $\overset{\Delta}{\subset}$
  &
  &&
   $ \mathrm{SU}(2)_L $
  &
  $\!\!\!\!\times\!\!\!\!\!$
  &
    $ \mathrm{SU}(2)_R $
  &&
  \\
  & $\bot$
  &
  &&&
   &  $\bot$
  &
  $\!\!\!\!\times\!\!\!\!\!$
  &&
  ------
  &&&&&&
  \\
  &  $\mathrm{MO5}$
  &&&&
  & $\mathbb{R}^{5,1\vert 2 \cdot \mathbf{8}}$
  &
  $\mathbb{Z}_2$
  & $ \overset{\Delta}{\subset} $
  &
  &&
  &
  &
  $ (\mathbb{Z}_2)_R $
  &
  $\!\!\!\!\times\!\!\!\!\!$
  &
  $(\mathbb{Z}_2)_{\mathrm{HW}}$
  \\
  \hline
  \multirow{3}{*}{ $\mathrm{MW}_{\mathrm{ADE}}$ }
  & $\mathrm{M2}$
  &
  \multirow{3}{*}{
    $\sfrac{3}{16}$ or $\sfrac{1}{4}$
  }
  &&
  &
  \multirow{3}{*}{ $\mathbb{R}^{1,1\vert \geq 6 \cdot \mathbf{1}}$ }
  &
   $ \mathbb{R}^{2,1\vert \geq 6 \cdot \mathbf{2}}$
  &
  $\mathbb{Z}_{n+1}$
  &
    $\overset{\Delta}{\subset}$
  &
  &&
   $ \mathrm{SU}(2)_L $
  &
  $\!\!\!\!\times\!\!\!\!\!$
  &
    $ \mathrm{SU}(2)_R $
  &&
  \\
  &
  \hspace{4pt}\begin{rotate}{90} $\Vert$ \end{rotate}
  &
  &&&
  &
  \begin{rotate}{90} $\Vert$ \end{rotate}
  &
  $\!\!\!\!\times\!\!\!\!\!$
  &&
  ------
  &&&&&&
  \\
  & $\mathrm{MO1}$
  &&&&
  & $\mathbb{R}^{1,1\vert 16 \cdot \mathbf{1}}$
  &
  $\mathbb{Z}_2$
  & $ \overset{\Delta}{\subset} $ &
  &&
  $(\mathbb{Z}_2)_L$
  &
  $\!\!\!\!\times\!\!\!\!\!$
  &
  $(\mathbb{Z}_2)_R$
  &
  $\!\!\!\!\times\!\!\!\!\!$
  &
  $(\mathbb{Z}_2)_{\mathrm{HW}}$
  \\
  \hline
  \multirow{3}{*}{ $\mathrm{M5}_{\mathrm{ADE}}$ }
  & $\mathrm{MO5}$
  &
  \multirow{3}{*}{
    $\sfrac{1}{4}$
  }
  &&
  &
  \multirow{3}{*}{ $\mathbb{R}^{5,1\vert 1 \cdot \mathbf{8}}$ }
  &
   $ \mathbb{R}^{5,1\vert 2 \cdot \mathbf{8}}$
  &
  $\mathbb{Z}_{2}$
  &
    $\overset{\Delta}{\subset}$
  &
  &&
  &
  &
   $ (\mathbb{Z}_2)_R $
  &
  $\!\!\!\!\times\!\!\!\!\!$
  &
  $ (\mathbb{Z}_2)_{\mathrm{HW}} $
  \\
  &
  \hspace{4pt}\begin{rotate}{90} $\Vert$ \end{rotate}
  &
  &&&
  &
  \hspace{4pt}\begin{rotate}{90} $\Vert$ \end{rotate}
  &
  $\!\!\!\!\times\!\!\!\!\!$
  &&
  \multicolumn{4}{c}{
    ------------------------------
  }
  &&
  &
  \\
  & $ \mathrm{MK6} $
  &&&&
  & $\mathbb{R}^{6,1\vert \mathbf{16}}$
  &
  $
  \begin{tabular}{c}
    $\mathbb{Z}_{n+1}, 2 \mathbb{D}_{n+2}$,
    \\
    $2T, 2O  , 2I$
  \end{tabular}
  $
  & $ \subset $ &
  &&
  &&
  $ \mathrm{SU}(2)_R $
  &
  &
  \\
  \hline
  \multirow{3}{*}{ $ \mathrm{NS5}_H $ }
  & $\mathrm{MO9}_H$
  &
  \multirow{3}{*}{
    $\sfrac{1}{4}$
  }
  &&
  &
  \multirow{3}{*}{ $\mathbb{R}^{5,1\vert 1 \cdot \mathbf{8}}$  }
  &
   $ \mathbb{R}^{9,1\vert \mathbf{16}}$
  &
  $\mathbb{Z}_{2}$
  &
    $ \overset{\phantom{\Delta}}{=}$
  &
  &&
  &
  &
  &
  &
  $ (\mathbb{Z}_2)_{\mathrm{HW}} $
  \\
  &
  \hspace{4pt}\begin{rotate}{90} $\Vert$ \end{rotate}
  &
  &&&
  &
  \begin{rotate}{90} $\Vert$ \end{rotate}
  &
  $\!\!\!\!\times\!\!\!\!\!$
  &&
  \multicolumn{4}{c}{
    ------------------------------
  }
  &&
  &
  \\
  & $ \mathrm{M5} $
  &&&&
  & $\mathbb{R}^{5,1\vert \mathbf{16}}$
  &
  $
  \mathbb{Z}_2
  $
  & $ \overset{\Delta}{\subset} $ &
  &&
  &&
  $ (\mathbb{Z}_2)_R $
  &
  $\!\!\!\!\times\!\!\!\!\!$
  &
  $ (\mathbb{Z}_2)_{\mathrm{HW}} $
  \\
  \hline
  \multirow{3}{*}{ $\tfrac{1}{2}\mathrm{NS5}_I$ }
  & $\mathrm{MO9}_{I}$
  &
  \multirow{3}{*}{
    $\sfrac{1}{4}$
  }
  &&
  &
  \multirow{3}{*}{ $\mathbb{R}^{5,1\vert 1 \cdot \mathbf{8}}$ }
  &
   $ \mathbb{R}^{9,1\vert 2 \cdot \mathbf{16}}$
  &
  $\mathbb{Z}_{2}$
  &
    $\overset{\phantom{\Delta}}{=}$
  &
  &&
  &
  &
  &
  &
  $ (\mathbb{Z}_2)_{\mathrm{HW}} $
  \\
  &
  \raisebox{-3pt}{$\top$}
  &
  &&&
  &
  \raisebox{-1pt}{$\top$}
  &
  $\!\!\!\!\times\!\!\!\!\!$
  &&
  \multicolumn{4}{c}{
    ------------------------------
  }
  &&
  &
  \\
  & $ \mathrm{MK6} $
  &&&&
  & $\mathbb{R}^{6,1\vert \mathbf{16}}$
  &
  $
  \begin{tabular}{c}
    $\mathbb{Z}_{n+1}, 2 \mathbb{D}_{n+2}$,
    \\
    $2T, 2O  , 2I$
  \end{tabular}
  $
  & $ \subset $ &
  &&
  &&
  $ \mathrm{SU}(2)_R $
  &
  &
  \\
  \hline
\end{tabular}
}
}
\vspace{.3cm}

\noindent {\footnotesize {\bf Table 1. Singularities in $D = 11$, $\mathcal{N} = 1$ super Minkowski spacetime.}
The \emph{simple} singularities
in the top half (Def.\ \ref{SimpleSingularities}) are classified by Theorem \ref{SuperADESingularitiesIn11dSuperSpacetime}. The \emph{non-simple} singularities
in the bottom part are the intersections of the former, established in Prop. \ref{NonSimpleRealSingularities}.
The label ``black brane species'' attached with each type of singularity is explained in Sec. \ref{TheFundamentalBraneScan}.
The symbol ``$\Vert$'' indicates that two intersecting fixed loci are \emph{parallel}, in that one is contained in the other.
Otherwise we use ``$\perp$'' to indicate that they are \emph{perpendicular} to each other.
}

\newpage

\subsection{Simple super singularities and Single black brane species}
\label{SingleBlackBraneSpecies}

\begin{theorem}[Classification of simple real singularities in 11d super-Minkowski spacetime]
  \label{SuperADESingularitiesIn11dSuperSpacetime}
The following table classifies, up to conjugacy in $\mathrm{Pin}^+(10,1)$, the $\geq \sfrac{1}{4}$ BPS (Def.\ \ref{BPSSuperLieSubalgebra})
simple singularities (Def.\ \ref{SimpleSingularities})
in $D = 11$, $\mathcal{N} = 1$ super-Minkowski spacetime (Example \ref{ExamplesOfSuperMinkowskiSpacetimes})
which are fixed (Example \ref{SuperLieAlgebrasWithGActionAsGSuperspaces}) at least by a non-trivial $\mathbb{Z}_2$-action (as in Def.\ \ref{Z2ActionsBypBraneInvolutions} below).

\vspace{.3cm}

\hspace{-.8cm}
\scalebox{.9}{
\begin{tabular}{|lc|ccc|lc|ccccccccc|}
  \hline
  \multicolumn{2}{|c}{
  \multirow{2}{*}{
  \begin{tabular}{c}
    \bf Black brane
    \\
    \bf species
  \end{tabular}
  }
  }
  &
  \multicolumn{3}{|c|}{
  \multirow{2}{*}{
    {\bf BPS}
  }
  }
  &
  \multicolumn{2}{c|}{
    \begin{tabular}{c}
    \bf Fixed
    \\
    \bf locus
    \end{tabular}
  }
  &
  \begin{tabular}{c}
    {\bf Type of}
    \\
    {\bf singularity}
  \end{tabular}
  &
  \multicolumn{8}{l|}{ \hspace{2.5cm} {\bf Intersection law} }
  \\
  &&&&&
  \multicolumn{2}{|c|}{
    {\bf in} $\mathbb{R}^{10,1\vert \mathbf{32}}$
  }
  &
   {\bf in} $\mathbb{R}^{10,1}$
  & $\simeq$
  & $\mathbb{R}^{1,1}$
  & $\!\!\!\!\!\!\!\!\oplus\!\!\!\!\!\!\!\!$
  & $\mathbb{R}^4$
  & $\!\!\!\!\!\!\!\!\oplus\!\!\!\!\!\!\!\!$
  & $\mathbb{R}^4$
  & $\!\!\!\!\!\!\!\!\oplus\!\!\!\!\!\!\!\!$
  & $\mathbb{R}^1$
  \\
  \hline
  \multicolumn{7}{l}{ \bf Black brane species  $ \phantom{ {{{A^A}^A}^A}^A} $  }
  &
  \multicolumn{9}{l}{ \bf Simple singularities }
  \\
  \hline
  $\mathrm{MO9}$ &
  &
  $\sfrac{1}{2}$ &&
  &
  $\mathbb{R}^{9,1\vert \mathbf{16}}$
  &
  &
  $\mathbb{Z}_{2}$
  &
  $\!\!\overset{\phantom{\Delta}}{=}\!\!$
  &
  \multicolumn{5}{c}{
    --------------------------------------------
  }
  &
  &
  $(\mathbb{Z}_2)_{\mathrm{HW}}$
  \\
  \hline
  $\mathrm{MO5}$ &
  &
  $\sfrac{1}{2}$ &&
  &
  $\mathbb{R}^{5,1\vert 2 \cdot \mathbf{8}}$ &
  &
  $\mathbb{Z}_{2}$
  &
  $\!\!\overset{\Delta}{\subset}\!\!$
  &
  \multicolumn{4}{c}{
    ------------------------------
  }
  &
  $(\mathbb{Z}_2)_R$
  &
  $\!\!\!\!\times\!\!\!\!\!$
  &
  $(\mathbb{Z}_2)_{\mathrm{HW}}$
  \\
  \hline
  $\mathrm{MO1}$ &
  &
  $\sfrac{1}{2}$ &&
  &
  $\mathbb{R}^{1,1\vert 16 \cdot \mathbf{1}}$ &
  &
  $\mathbb{Z}_{2}$
  &
  $\!\!\overset{\Delta}{\subset}\!\!$
  &
  ------
  &
  &
  $(\mathbb{Z}_2)_L$
  &
  $\!\!\!\!\times\!\!\!\!\!$
  &
  $(\mathbb{Z}_2)_R$
  &
  $\!\!\!\!\times\!\!\!\!\!$
  &
  $(\mathbb{Z}_2)_{\mathrm{HW}}$
  \\
  \hline
  \hline
  $\mathrm{MK6}$ &
  &
  $\sfrac{1}{2}$ &&
  &
  $\mathbb{R}^{6,1\vert \mathbf{16}}$ &
  &
  \begin{tabular}{c}
    $\mathbb{Z}_{n+1}, 2 \mathbb{D}_{n+2}$,
    \\
    $2T, 2O  , 2I$
  \end{tabular}
  &
  $\!\!\overset{}{\subset}\!\!$
  &
  \multicolumn{4}{c}{
   ------------------------------
  }
  &
  $\mathrm{SU}(2)_R$
  &
  &
  ------
  \\
  \hline
  $\mathrm{M2}$ &
  &
  $\sfrac{1}{2}$ &$\!\!\!\!\!\!=\!\!\!\!\!\!$& $\sfrac{8}{16}$
  &
  $\mathbb{R}^{2,1\vert 8 \cdot \mathbf{2}}$ &
  &
  $\mathbb{Z}_{2}$
  &
  $\!\!\overset{\Delta}{\subset}\!\!$
  &
  ------
  &&
  $\mathrm{SU}(2)_L$
  &
  $\!\!\!\!\times\!\!\!\!\!$
  &
  $\mathrm{SU}(2)_R$
  &&
  ------
  \\
  \hline
  $\mathrm{M2}$ &
  &
  && $\sfrac{6}{16}$
  &
  $\mathbb{R}^{2,1\vert 6 \cdot \mathbf{2}}$ &
  &
  $\mathbb{Z}_{n+3}$
  &
  $\!\!\overset{\Delta}{\subset}\!\!$
  &
  ------
  &
  &
  $\mathrm{SU}(2)_L$
  &
  $\!\!\!\!\!\!\!\!\times\!\!\!\!\!\!\!\!\!$
  &
  $\mathrm{SU}(2)_R$
  &&
  ------
  \\
  \hline
  $\mathrm{M2}$ &
  &
  && $\sfrac{5}{16}$
  &
  $\mathbb{R}^{2,1\vert 5 \cdot \mathbf{2}}$ &
  &
  \begin{tabular}{c}
    $2 \mathbb{D}_{n+2}$,
    \\
    $2T, 2O  , 2I$
  \end{tabular}
  &
  $\!\!\overset{\Delta}{\subset}\!\!$
  &
  ------
  &
  &
  $\mathrm{SU}(2)_L$
  &
  $\!\!\!\!\!\times\!\!\!\!\!$
  &
  $\mathrm{SU}(2)_R$
  &&
  ------
  \\
  \hline
  $\mathrm{M2}$ &
  &
  $\sfrac{1}{4}$ &$\!\!\!\!\!\!=\!\!\!\!\!\!$& $\sfrac{4}{16}$
  &
  $\mathbb{R}^{2,1\vert 4 \cdot \mathbf{2}}$ &
  &
  \begin{tabular}{c}
    $2 \mathbb{D}_{n+2}$,
    \\
    $2O  , 2I$
  \end{tabular}
  &
  $\!\!\overset{(\mathrm{id},\tau)}{\subset}\!\!$
  &
  ------
  &&
  $\mathrm{SU}(2)_L$
  &
  $\!\!\!\!\times\!\!\!\!\!$
  &
  $\mathrm{SU}(2)_R$
  &&
  ------
  \\
  \hline
\end{tabular}
}

\vspace{.2cm}

\noindent Moreover, the actions on the underlying bosonic spacetime are as shown on the right of the table, induced by the vector space decomposition
\vspace{-5mm}
\begin{equation}
  \label{SpacetimeActions}
  \mathbb{R}^{10,1}
  \;\;\simeq \;\;
    \mathbb{R}^{1,1}
    \oplus
  \xymatrix{
    \mathbb{R}^4
    \ar@(ul,ur)[]^{ \mathrm{SU}(2)_L }
  }
    \oplus
  \xymatrix{
    \mathbb{R}^4
    \ar@(ul,ur)[]^{ \mathrm{SU}(2)_R }
  }
    \oplus
  \xymatrix{
    \mathbb{R}
    \ar@(ul,ur)[]^{ (\mathbb{Z}_2)_{\mathrm{HW}} }
  }
\end{equation}
where  both copies of $\mathrm{SU}(2)$ act via their defining action on $\mathbb{C}^2 \simeq_{\mathbb{R}} \mathbb{R}^4$,
while $(\mathbb{Z}_2)_{\mathrm{HW}}$ acts by multiplication by $-1$ on $\mathbb{R}^1$.
\end{theorem}

We now work towards the {\it proof} of Theorem \ref{SuperADESingularitiesIn11dSuperSpacetime}, which
will be given by combining Proposition \ref{ClassificationOfZ2Actions} and Proposition \ref{M2AndMK6ADE} below.

\begin{defn}[$p$-brane involution]
  \label{pBraneInvolution}
  Consider an involution $\sigma \in \Pin^+(10,1)$ acting canonically on
  $\mathbb{R}^{10,1}$ (as in Def.\ \ref{LorentzGroupsAndTheirSpinCovers}).  As a
  linear transformation of $\R^{10,1}$, it may be diagonalized with eigenvalues $\pm
  1$:
  \[
  \sigma = \begin{bmatrix} 1 & 0 \\ 0 & -1 \end{bmatrix}.
  \]
  That is, for some $p \in \mathbb{N}$, the space $\R^{10,1}$ decomposes into orthogonal
  subspaces
  \[
  \R^{10,1} \simeq \R^{p,1} \oplus \R^{10-p}\;,
  \]
  such that $\sigma$ acts as 1 on one summand and as $-1$ on the other. In fact,
  since $\sigma$ preserves time orientation, it acts as 1 on the $\R^{p,1}$ summand,
  and $-1$ on the spacelike $\R^{10-p}$ summand. To highlight the natural number $p$,
  we call $\sigma$ a \emph{$p$-brane involution}.
\end{defn}

\begin{example}[Trivial $p$-brane involutions]
  \label{TrivialPBraneInvolution}
  For $p = 10$, a $p$-brane involution (Def.\ \ref{pBraneInvolution}) acts trivially on $\R^{10,1}$, so
  we say $\sigma$ is a \emph{trivial} $p$-brane involution. For other values of $p$, we
  say $\sigma$ is \emph{nontrivial}.
  A trivial $p$-brane involution $\sigma$ is just
  an element of the kernel of the double cover $\Pin^+(10,1) \to {\rm O}^+(10,1)$, so
  $\sigma = \pm 1$.
\end{example}

\begin{defn}[$\mathbb{Z}_2$-actions by $p$-brane involutions on 11d super-Minkowski spacetime]
  \label{Z2ActionsBypBraneInvolutions}
Any $p$-brane involution (Def.\ \ref{pBraneInvolution}) defines an action of the group $\mathbb{Z}_2$ (Example \ref{Z2ActionsAreInvolutions})
on $\mathbb{R}^{10,1\vert \mathbf{}32} \in \mathrm{Ho}\left( \mathrm{SuperSpaces}_{\mathbb{R}}\right)$ (Example \ref{ExamplesOfSuperMinkowskiSpacetimes})
by super Lie algebra automorphisms (Example \ref{SuperLieAlgebrasWithGActionAsGSuperspaces}).
\end{defn}

The following classification of $\mathbb{Z}_2$-actions on super-Minkowski spacetime has also been briefly sketched in \cite[around (3.2)]{HananyKol00}.
We need the following detailed analysis for the proof of our main result
in Theorem \ref{SuperADESingularitiesIn11dSuperSpacetime}, Prop. \ref{NonSimpleRealSingularities} and Theorem \ref{RealADEEquivariantEndhancementOfM2M5Cocycle}.
\begin{prop}[Classification of $\mathbb{Z}_2$-actions on 11d super-Minkowski spacetime]
  \label{ClassificationOfZ2Actions}

The $\mathbb{Z}_2$-actions on $\mathbb{R}^{10,1\vert \mathbf{32}}$ according to Def.\ \ref{Z2ActionsBypBraneInvolutions}), are,
  up to conjugacy in $\mathrm{Pin}^+(10,1)$, in bijection with
   the entries in the following table:

   {\small
\begin{center}
\begin{tabular}{|lc|ccc|lc|ccccccc|}
  \hline
  \multicolumn{2}{|c}{
  \begin{tabular}{c}
    \bf Black
    \\
    \bf brane
    \\
    \bf species
  \end{tabular}
  }
  &
  \multicolumn{3}{|c|}{ {\bf BPS} }
  &
  \multicolumn{2}{c|}{
  \begin{tabular}{c}
    \bf Singular
    \\
    \bf locus
    \\
    $\subset \mathbb{R}^{10,1\vert \mathbf{32}}$
  \end{tabular}
  }
  &
  \begin{tabular}{c}
    {\bf Type of}
    \\
    {\bf singularity}
  \end{tabular}
  &
  \multicolumn{6}{l|}{\hspace{1cm} {\bf Intersection law} }
  \\
  \hline
  \hline
  $\mathrm{MO9}$ &
  &
  $\sfrac{1}{2}$ &&
  &
  $\mathbb{R}^{9,1\vert \mathbf{16}}$
  &
  &
  $\mathbb{Z}_{2}$
  &
  $\!\!\overset{\phantom{\Delta}}{=}\!\!$
  &
  &
  &
  &
  &
  $(\mathbb{Z}_2)_{\mathrm{HW}}$
  \\
  \hline
  $\mathrm{MO5}$ &
  &
  $\sfrac{1}{2}$ &&
  &
  $\mathbb{R}^{5,1\vert 2 \cdot \mathbf{8}}$ &
  &
  $\mathbb{Z}_{2}$
  &
  $\!\!\overset{\Delta}{\subset}\!\!$
  &
  $(\mathbb{Z}_2)_L$
  &
  &
  &
  $\!\!\!\!\times\!\!\!\!\!$
  &
  $(\mathbb{Z}_2)_{\mathrm{HW}}$
  \\
  \hline
  $\mathrm{MO1}$ &
  &
  $\sfrac{1}{2}$ &&
  &
  $\mathbb{R}^{1,1\vert 16 \cdot \mathbf{1}}$ &
  &
  $\mathbb{Z}_{2}$
  &
  $\!\!\overset{\Delta}{\subset}\!\!$
  &
  $(\mathbb{Z}_2)_L$
  &
  $\!\!\!\!\times\!\!\!\!\!$
  &
  $(\mathbb{Z}_2)_R$
  &
  $\!\!\!\!\times\!\!\!\!\!$
  &
  $(\mathbb{Z}_2)_{\mathrm{HW}}$
  \\
  \hline
  \hline
  $\mathrm{MK6}$ &
  &
  $\sfrac{1}{2}$ &&
  &
  $\mathbb{R}^{6,1\vert \mathbf{16}}$ &
  &
  $\mathbb{Z}_2$
  &
  $\!\!=\!\!$
  &
  $(\mathbb{Z}_2)_L$
  &&&&
  \\
  \hline
  $\mathrm{M2}$ &
  &
  $\sfrac{1}{2}$ &&
  &
  $\mathbb{R}^{2,1\vert 8 \cdot \mathbf{2}}$ &
  &
  $\mathbb{Z}_{2}$
  &
  $\!\! \overset{\Delta}{\subset} \!\!$
  &
  $(\mathbb{Z}_2)_L$
  &
  $\!\!\!\!\times\!\!\!\!\!$
  &
  $(\mathbb{Z}_2)_R$
  &&
  \\
  \hline
\end{tabular}
\end{center}
}
\end{prop}
\medskip

\noindent The {\it proof} of Prop. \ref{ClassificationOfZ2Actions} will be established by the following
series of lemmas.

\begin{lemma}[Existence of $p$-brane involutions]
  \label{ExistenceOfpBraneInvolutions}
  If $\sigma \in \Pin^+(10,1)$ is a $p$-brane involution (Def.\ \ref{pBraneInvolution}), then $p \equiv 1$ or $2$ mod
  $4$. Conversely, if $p \equiv 1$ or $2$ mod $4$, then for any orthogonal decomposition of
  $\R^{10,1}$ of the form
  \[ \R^{10,1} \simeq \R^{p,1} \oplus \R^{10-p} \]
  there is a $p$-brane involution acting as $1$ on the first summand and $-1$ on the
  second. Hence, there are $p$-brane involutions $\sigma$ precisely for $p = 1, 2, 5, 6, 9$ and
  $10$, though for $p = 10$, $\sigma$ is trivial (in the sense of Def.\ \ref{pBraneInvolution}).
\end{lemma}

\begin{proof}
  Let $\sigma$ be a $p$-brane involution, and let $\R^{10-p}$ be the subspace of
  $\R^{10,1}$ on which $\sigma$ acts as $-1$. Writing $q := 10 - p$ for brevity, let
  $\{b_1, \ldots, b_q\}$ be an orthonormal basis of $\R^{10-p}$. Then the Clifford
  algebra element $b_1 \cdots b_q$ lies in $\Pin^+(10,1)$ and acts on $\R^{p,1}$ as 1
  and on $\R^{10-p}$ as $-1$. Because $\Pin^+(10,1)$ is a double cover of ${\rm
  O}^+(10,1)$ with kernel $\{1,-1\}$, this implies $\sigma = \pm b_1 \cdots
  b_q$. Since, as elements of the Clifford algebra, the vectors $b_1, \ldots, b_q$
  square to 1 and anticommute with each other, we compute:
  \[ \sigma^2 = (-1)^{\frac{q(q-1)}{2}} b_1^2 \cdots b_q^2 = \pm 1\,. \]
  This gives 1 if and only if $q \equiv 0$ or 1 mod 4, which implies $p \equiv 1$ or
  2 mod 4.

  Conversely, if $p \equiv 1$ or 2 mod 4, choose an orthogonal decomposition of
  $\R^{10,1}$ of the form:
  \[ \R^{10,1} \simeq \R^{p,1} \oplus \R^{10-p} . \]
  As above, let $\{b_1, \ldots, b_q\}$ be an orthonormal basis of $\R^{10-p}$. The
  same calculation as above shows that $\sigma = b_1 \cdots b_q$ is the desired
  $p$-brane involution.
\end{proof}

\begin{lemma}[Conjugacy of $p$-brane involutions]
  \label{pBraneInvolutionConjugacy}
  All nontrivial $p$-brane involutions (Def.\ \ref{pBraneInvolution}) for given $p = 1$ or 2 mod 4
  (Lemma \ref{ExistenceOfpBraneInvolutions})
  are conjugate by an element of $\Pin^+(10,1)$.
\end{lemma}

\begin{proof}
  Let $\sigma, \sigma' \in \Pin^+(10,1)$ be two nontrivial $p$-brane
  involutions. Then $p \leq 9$, and $\R^{10,1}$ decomposes into an orthogonal direct
  sum
  \[
  \R^{10,1} \simeq \R^{p,1} \oplus \R^{10-p}
  \]
  on which $\sigma$ acts diagonally. Write $q := 10 - p$ for brevity, and let $\{b_1,
  \ldots, b_q\}$ be an orthonormal basis of the $\R^{10-p}$ summand. Then the
  Clifford algebra element $b_1 \cdots b_q$ lies in $\Pin^+(10,1)$ and acts on
  $\R^{10,1}$ in the same way $\sigma$ does. Thus $\sigma = \pm b_1 \cdots b_q$, and
  changing the sign of a basis vector if necessary, we can assume $\sigma = b_1
  \cdots b_q$. Similarly, we can find $q$ spacelike orthonormal vectors $\{b_1',
  \ldots, b_q'\}$ such that $\sigma' = b_1' \cdots b_q'$.

  We may now extend the set of $q$ orthonormal vectors $\{b_1, \cdots, b_q\}$ to a
  time-oriented orthonormal basis of $\R^{10,1}$, and similarly for the $q$
  orthonormal vectors $\{b_1', \ldots, b_q'\}$. There is a unique element of ${\rm
  O}^+(10,1)$ taking one basis to the other, and this lifts to an element $g \in
  \Pin^+(10,1)$ that maps $b_i \mapsto b_i'$ for all $1 \leq i \leq q$.

  Let us write the double covering map $R \maps \Pin^+(10,1) \to {\rm
  O}^+(10,1)$. Then, by construction,
  \[
   R(g) R(\sigma) R(g^{-1}) = R(\sigma') .
  \]
  Thus, in $\Pin^+(10,1)$, we must have either $g \sigma g^{-1} = \sigma'$, or $g
  \sigma g^{-1} = -\sigma'$. In the first case, we are done. In the second, if we can
  find an element $h \in \Pin^+(10,1)$ that anticommutes with $\sigma$, and redefine
  $g$ to be $gh$, we are done.

  Let us find $h \in \Pin^+(10,1)$ that anticommutes with $\sigma$. That is, that
  satisfies $hgh^{-1} = -\sigma$. If $q = 10 - p$ is odd, we can take any spacelike
  unit vector $v$ orthogonal to the basis $\{b_1, \ldots, b_q\}$. Then $h = v$
  anticommutes with $\sigma$. If $q$ is even, we can choose $h = b_1$.
  This completes the proof.
\end{proof}

Now we compute the fixed loci of $p$-brane involutions one by one, according to their essentially unique
existence established by Lemma \ref{ExistenceOfpBraneInvolutions} and Lemma \ref{pBraneInvolutionConjugacy}.

\begin{lemma}[Fixed locus of the 1-brane involution]
  \label{1braneFixed}
  Let $\sigma$ be a $1$-brane involution (Def.\ \ref{pBraneInvolution}). It fixes the $D = 1+1$, $\mathcal{N} = (16,0)$
  super-Minkowski subalgebra:
  \[
  \R^{1,1|16 \cdot {\bf 1}_+} \simeq \big(\R^{10,1|{\bf 32}} \big)^\sigma \,.
  \]
  Here, ${\bf 1}_+$ denotes the 1-dimensional real spinor representation of
  \[
   \Spin(1,1) = \{ \exp(\tfrac{t}{2} {\bf \Gamma}_{01}) : t \in \R \}
   \]
  on which the group element $\exp(\frac{t}{2} {\bf \Gamma}_{01})$ acts by
  multiplication by $e^{\frac{t}{2}}$ (Example \ref{RelevantExamplesOfRealSpinRepresentations}).
  Moreover, the bilinear spinor pairing \eqref{ViaDiracConjugateSpinorToVectorPairing} in $\R^{1,1|16 \cdot
  {\bf 1}_+}$ spans the lightlike subspace ${\rm span}\{ b_0 + b_1 \} \subseteq
  \R^{1,1}$:
  \begin{equation}
    \label{SpinorPairingOn1Brane}
    \mathrm{im}
      \big(
        \xymatrix{
          16 \cdot {\bf 1}_+ \,\otimes\, 16 \cdot {\bf 1}_+
           \ar[rr]^-{\overline{(-)}(-)}
           &&
          \mathbb{R}^{1,1}
        }
      \big)
      \;=\;
      {\rm span}\{ b_0 + b_1 \}
        \,\subseteq\,
      \R^{1,1}\;.
  \end{equation}
\end{lemma}

\begin{proof}
  Since all 1-brane involutions are conjugate, by Lemma \ref{pBraneInvolutionConjugacy}, we may choose:
  \[
  \sigma = -{\bf \Gamma}_{234543678910} .
  \]
  The minus sign ensures that we will get the spinor representation ${\bf 1}_+$ of
  the lemma above, rather than the other spinor representation ${\bf 1}_-$ on which
  $\exp(\frac{t}{2}\Gamma_{01})$ acts as $e^{-\frac{t}{2}}$. The two are related, of
  course, by an element of $\Pin^+(10,1)$.

  Clearly, $\sigma$ fixes the bosonic subspace $\R^{1,1} = {\rm span}\{b_0,
  b_1\}$.
  To see which fermionic subspace it fixes, we make use of the octonionic presentation of the $\mathbf{32}$
  of $\mathrm{Pin}^+(10,1)$ from Example \ref{OctonionPresentationOf32OfPin11}.
  Multiplying out the $\bf \Gamma$-matrices \eqref{11dGammaOct}, we see that $\sigma$ acts on
  $\bf 32 \simeq \O^4$ as the diagonal matrix:
  \[ \sigma = {\rm diag}(-1, 1, 1, -1)\;.
  \]
  Thus the space of spinors fixed by $\sigma$ is:
  \[ S = \{ (0, a, b, 0) \in \O^4 \} . \]
  Since ${\bf \Gamma}_{01} = {\rm diag}(-1, 1, 1, -1)$ then $\Spin(1,1)$ acts on
  $S$ as $16 \cdot {\bf 1}_+$. Thus, as a representation of $\Spin(1,1)$, the fixed
  point locus is indeed $\R^{1,1|16 \cdot {\bf 1}_+}$.

  Finally, note that $\R^{1,1}$ is not an irreducible representation of
  $\Spin(1,1)$. It decomposes into the sum of two lightlike subrepresentations:
  \[ \R^{1,1} \simeq {\rm span}\{b_0 + b_1 \} \oplus {\rm span}\{ b_0 - b_1 \} . \]
  On the first summand, $\exp(\frac{t}{2}{\bf \Gamma}_{01}) \in \Spin(10,1)$ acts
  as $e^t$, and on the second as $e^{-t}$. It is thus immediate from the
  $\Spin(1,1)$-equivariance of the bracket operation that the bracket of spinors in
  $\R^{1,1|16 \cdot {\bf 1}_+}$ lands in the first summand. The brackets of spinors
  span this 1-dimensional subspace as long as they are not all zero, which we leave
  to the reader to check.
\end{proof}

\begin{lemma}[Fixed locus of the 2-brane involution]
  \label{2braneFixed}
  Let $\sigma$ be a 2-brane involution (Def.\ \ref{pBraneInvolution}). Then it fixes the $D = 2 + 1$, $\mathcal{N} = 8$
  super-Minkowski subalgebra:
  \[ \R^{2,1|8 \cdot {\bf 2}} \simeq \big(\R^{10,1|{\bf 32}} \big)^\sigma . \]
\end{lemma}

\begin{proof}
Since all 2-brane involutions are conjugate, by Lemma \ref{pBraneInvolutionConjugacy}, we may choose:
$\sigma = {\bf \Gamma}_{23456789}$.
This clearly fixes the bosonic subspace $\R^{2,1} \simeq {\rm span} \{ b_0, b_1,
b_{10} \}$.

In order to compute the fermionic fixed space, we make use of the octonionic presentation of the
$\mathbf{32}$ of $\mathrm{Pin}^+(10,1)$ from Example \ref{OctonionPresentationOf32OfPin11}.
Multiplying out the $\bf \Gamma$ matrices, we see that $\sigma$ acts on
${\bf 32} \simeq \O^4$ as the diagonal matrix:
\[ \sigma = {\rm diag}(1, -1, 1, -1).
\]
It thus fixes the 16-dimensional space of spinors
\[ S \simeq \{ (a, 0, b, 0) \in \O^4 \} . \]
Because $\Spin(2,1)$ has a unique 2-dimensional real spinor representation (Prop. \ref{RealSpinRepresentationFromDivisionAlgebra}), we must
have $S \simeq 8 \cdot {\bf 2}$.
Thus the fixed point locus is the super-Minkowski subalgebra $\R^{2,1|8 \cdot {\bf
2}}$.
\end{proof}

A brief indication of following result was given in \cite[Sec. 2.1]{Witten95}; see  Example \ref{TheBlackM5}.

\begin{lemma}[Fixed locus of the 5-brane involution]
  \label{5braneFixed}
  Let $\sigma$ be a 5-brane involution (Def.\ \ref{pBraneInvolution}). Then it fixes
  the $D = 5 + 1$, $N = (2,0)$ super-Minkowski subalgebra:
  \[ \R^{5,1|2 \cdot {\bf 8}} \simeq \big(\R^{10,1|{\bf 32}} \big)^\sigma . \]
\end{lemma}

\begin{proof}
  Since all 5-brane involutions are conjugate by Lemma
  \ref{pBraneInvolutionConjugacy}, we may choose:
  $\sigma = {\bf \Gamma}_{678910}$.
  This clearly fixes the bosonic subspace $\R^{5,1} = {\rm span}\{b_0, b_1, \ldots,
  b_5 \}$.

  In order to compute the fermionic fixed space, we make use of the octonionic
  presentation of the $\mathbf{32}$ of $\mathrm{Pin}^+(10,1)$ from Example
  \ref{OctonionPresentationOf32OfPin11}.  Multiplying out the $\bf \Gamma$ matrices,
  we see that $\sigma$ acts on ${\bf 32} \simeq \O^4$ as the diagonal matrix:
  \[ \sigma = {\rm diag}(\theta, \theta, -\theta, -\theta) \]
  where each $\theta$ is the linear transformation of $\O$ given by $\theta = L_{e_4}
  L_{e_5} L_{e_6} L_{e_7}$; that is, by successive left multiplication by the
  imaginary octonions $e_4, e_5, e_6$ and $e_7$. Conjugating by an element in
  $\Pin^+(10,1)$ if necessary, we can order the $\Gamma$ matrices and thus the basis
  of $\Im \, \O$ so that $e_1 = i$, $e_2 = j$ and $e_3 = k$ are the imaginary units of a
  quaternionic subalgebra $\H \subseteq \O$, while $e_4 = \ell, e_5 = i\ell, e_6 =
  j\ell$ and $e_7 = k \ell$ span the orthogonal complement $\H\ell \subseteq \O$,
  where $\ell$ is a unit imaginary octonion orthogonal to $\H$. Using the
  Cayley--Dickson construction, we can check that:
  $
   L_{e_1} L_{e_2} \cdots L_{e_7} = -1
  $
  on $\O$. We thus have
  $
   \theta = L_{e_3} L_{e_2} L_{e_1}
  $
so that $\theta = kji = 1$ on $\H$.  We can compute that it is $-1$ on
  $\H\ell$. Consequently, the fixed fermionic subspace is:
  \begin{equation}
  \label{M5FixedFermions}
  S = \H^2 \oplus \H^2\ell = (\O^4)^\sigma .
  \end{equation}
  By Prop. \ref{RealSpinRepresentationFromDivisionAlgebra}, the action of $\Spin(5,1)$ on $\H^2$ is clearly as the representation ${\bf 8}$,
  since the ${\bf \Gamma}$ matrices ${\bf \Gamma}_0, \ldots, {\bf \Gamma}_5$ are
  quaternionic.

  It is less clear that its action on $\H^2 \ell$ is again ${\bf 8}$, but this is
  still true. We use an argument from \cite{HuertaSchreiber17}.
  Consider the action of generators of $\Spin(5,1)$ on $\H^2$ and
  $\H^2\ell$, respectively. A pair of unit vectors $A, B \in \R^{5,1}$, viewed as $2
  \times 2$ Hermitian matrices over $\H$, act on $\psi \in \H^2$ as
  $\tilde{A}_L(B_L \psi)$,
  whereas they act on $\psi \ell \in \H^2 \ell$ as
   $A_L( \tilde{B}_L (\psi \ell))$.
  In both expressions the subscript $L$ indicates that each matrix element of $A$ and
  $B$ acts by left multiplication. Yet we can use the Cayley--Dickson construction (Def.\ \ref{CayleyDicksonConstruction}) to
  see that
  \[
  A_L(\tilde{B}_L (\psi \ell)) = (A_R ( \tilde{B}_R \psi ))\ell\;,
  \]
  where the subscript $R$ indicates that matrix elements now act by right
  multiplication.
  Therefore, we are now reduced to showing that this right multiplication action
  $$
    \psi \mapsto  A_R (\tilde{B}_R (\psi))
  $$
  is isomorphic to the left multiplication action with the position of the trace reversal
  exchanged:
  $$
    \psi \mapsto \tilde A_L (B_L (\psi))
    \,.
  $$
  We claim that such an isomorphism is established by
  $
    F \maps \psi \mapsto J \overline{\psi}
    \,,
  $
  with $J=\binom{0\;-1}{1\;\;\;0}$.
  To see this, first note the relation
  $
    J \overline{A} = - \tilde A J
  $
holds for any hermitian $2 \times 2$ matrix $A$, which follows from a quick calculation. This is then used to deduce that $F$ is an equivariant map:
  $$
    \begin{aligned}
      F( \tilde A_R (B_R (\psi)) )
      & =
      J \, \overline{ \tilde A_R (B_R (\psi)) }
      \\
      & =
      J \, \overline{\tilde A}_L \, ( \overline{B}_L \, \overline{\psi} )
      \\
      & =
      A_L ( \tilde B_L J \overline{\psi} )
      \\
      & =
      A_L (\tilde B_L (F(\psi)))\;.
    \end{aligned}
  $$
  In the second step, we have used that conjugation turns right multiplication into left multiplication.
\end{proof}

\begin{lemma}[Fixed locus of the 6-brane involution]
  \label{6braneFixed}
 The  $p$-brane involution $\sigma$ (Def.\ \ref{pBraneInvolution}) for $p = 6$, which by Lemma \ref{pBraneInvolutionConjugacy}
 exists uniquely up to conjugacy,
 fixes the $D = 5+1$, $N = (2,0)$ super Minkowski subalgebra (Example \ref{ExamplesOfSuperMinkowskiSpacetimes})
 $$
   \mathbb{R}^{6,1\vert \mathbf{16}}
   \simeq
   \big(\mathbb{R}^{10,1\vert \mathbf{32}}\big)^\sigma
   \,.
 $$
\end{lemma}
\begin{proof}
  By Lemma \ref{pBraneInvolutionConjugacy} we may take $ \sigma =
  \mathbf{\Gamma}_{6789} $.  It is clear that the bosonic subspace of the
  corresponding fixed subalgebra is $\R^{6,1} = {\rm span}\{b_0, b_1, \ldots, b_5,
  b_{10}\} \subset \mathbb{R}^{10,1}$.  In order to compute the fermionic fixed space
  we use the octonionic presentation of $\mathbf{32}$ from Example
  \ref{OctonionPresentationOf32OfPin11}: From (\ref{11dGammaAsOctonionicAndTensor})
  we get
 \begin{equation}
   \label{Octonionic6BraneInvolution}
   \begin{aligned}
     \mathbf{\Gamma}_{6789}
     &
     = \varepsilon^4 \otimes J^4 L_{e_4} L_{e_5} L_{e_6} L_{e_7}
     \\
     & = L_{e_4} L_{e_5} L_{e_6} L_{e_7} \, 1_{4 \times 4}\;,
   \end{aligned}
 \end{equation}
 so that
 \begin{equation}
   \label{MK6FixedFermions}
   \left( \mathbb{O}^{4}\right)^{\mathbf{\Gamma}_{6789}}
   =
   \Big( \left( \mathbb{O}\right)^{L_{e_4} L_{e_5} L_{e_6} L_{e_7}}\Big)^{4}
   =
   \mathbb{H}^4
   \,.
 \end{equation}
 To see the second equality, observe that $\theta := L_{e_4} L_{e_5} L_{e_6} L_{e_7}$  is the same linear transformation that we studied in the proof of Lemma
 \ref{5braneFixed}. As in that proof, reordering the basis of $\R^{10,1|{\bf 32}}$ by
 an element of $\Pin^+(10,1)$ if necessary, we find that $\theta$ fixes a
 quaternionic subspace $\H \subseteq \O$. Thus ${\bf 32}^\sigma \simeq
 \H^4$. Restricting to the group $\Spin(6,1)$, we use
 Prop. \ref{RealSpinRepresentationFromDivisionAlgebra} to conclude $\mathbb{H}^{4}
 \simeq \mathbf{16}$, and this completes the proof.
\end{proof}

\begin{lemma}[Fixed locus of the 9-brane involution]
  \label{9braneFixed}
  Let $\sigma$ be a 9-brane involution (Def.\ \ref{pBraneInvolution}). Then it fixes
  the $D = 9 + 1$, $\mathcal{N} = 1$ super-Minkowski subalgebra:
  \[
  \R^{9,1|{\bf 16}} \simeq \big(\R^{10,1|{\bf 32}} \big)^\sigma.
  \]
\end{lemma}

\begin{proof}
  Since all 9-brane involutions are conjugate by Lemma
  \ref{pBraneInvolutionConjugacy}, we may choose:
  $\sigma = {\bf \Gamma}_{10}$.
  This clearly fixes the bosonic subspace $\R^{9,1} = {\rm span}\{b_0, b_1, \ldots,
  b_9 \}$. Moreover, ${\bf \Gamma}_{10}$ acts diagonally on ${\bf 32} \simeq \O^4$:
  \[
  {\bf \Gamma}_{10} = {\rm diag}(1,1,-1,-1) .
  \]
  It thus fixes $\O^2 \simeq \{(a,b,0,0) \in \O^4 \}$. Examining the octonionic
  representation of $\Pin^+(10,1)$ via Prop. \ref{RealSpinRepresentationFromDivisionAlgebra}, we see that $\Spin(9,1)$ acts on this $\O^2$ as
  ${\bf 16}$.
\end{proof}

\noindent In summary, these lemmas constitute the proof of Prop. \ref{ClassificationOfZ2Actions}.

\medskip
We now consider the generalization of the above $\mathbb{Z}_2$-actions on superspacetime to actions of larger groups,
while retaining the fixed point locus of every non-trivial subgroup.

\begin{prop}[Classification of $G_{\mathrm{ADE}}$-actions fixing M2 and MK6]
  \label{M2AndMK6ADE}
  The two (up to conjugacy) orientation-preserving $\mathbb{Z}_2$-actions on $\mathbb{R}^{10,1\vert \mathbf{32}}$ in Prop. \ref{ClassificationOfZ2Actions}
  (those labled $\mathrm{MK6}$ and $\mathrm{M2}$)
  extend to actions of all the finite subgroups of $\mathrm{SU}(2)$ (Remark \ref{ADEGroups})
  such that their fixed point locus retains the same bosonic part $\mathbb{R}^{p,1}$  ($p = 2,6$, respectively) and
  is $\geq \sfrac{1}{4}$-BPS (Def.\ \ref{BPSSuperLieSubalgebra}) as shown in the following table:

\vspace{.5mm}
{\small
\begin{center}
\begin{tabular}{|lc|ccc|lc|ccccccc|}
  \hline
  \multicolumn{2}{|c}{
  \begin{tabular}{c}
    \bf Black
    \\
    \bf brane
    \\
    \bf species
  \end{tabular}
  }
  &
  \multicolumn{3}{|c|}{ {\bf BPS} }
  &
  \multicolumn{2}{c|}{
  \begin{tabular}{c}
    \bf Singular
    \\
    \bf locus
    \\
    $\subset \mathbb{R}^{10,1\vert \mathbf{32}}$
  \end{tabular}
  }
  &
  \begin{tabular}{c}
    {\bf Type of}
    \\
    {\bf singularity}
  \end{tabular}
  &
  \multicolumn{6}{l|}{\hspace{8mm} {\bf Intersection law} }
  \\
  \hline
  \hline
  $\mathrm{MK6}$ &
  &
  $\sfrac{1}{2}$ &&
  &
  $\mathbb{R}^{6,1\vert \mathbf{16}}$ &
  &
  \begin{tabular}{c}
    $\mathbb{Z}_{n+1}, 2 \mathbb{D}_{n+2}$,
    \\
    $2T, 2O  , 2I$
  \end{tabular}
  &
  $\!\!\overset{}{\subset}\!\!$
  &
  $\mathrm{SU}(2)_L$
  &&&&
  \\
  \hline
  $\mathrm{M2}$ &
  &
  $\sfrac{1}{2}$ &$\!\!\!\!\!\!=\!\!\!\!\!\!$& $\sfrac{8}{16}$
  &
  $\mathbb{R}^{2,1\vert 8 \cdot \mathbf{2}}$ &
  &
  $\mathbb{Z}_{2}$
  &
  $\!\!\overset{\Delta}{\subset}\!\!$
  &
  $\mathrm{SU}(2)_L$
  &
  $\!\!\!\!\times\!\!\!\!\!$
  &
  $\mathrm{SU}(2)_R$
  &&
  \\
  \hline
  $\mathrm{M2}$ &
  &
  && $\sfrac{6}{16}$
  &
  $\mathbb{R}^{2,1\vert 6 \cdot \mathbf{2}}$ &
  &
  $\mathbb{Z}_{n+3}$
  &
  $\!\!\overset{\Delta}{\subset}\!\!$
  &
  $\mathrm{SU}(2)_L$
  &
  $\!\!\!\!\!\!\!\!\times\!\!\!\!\!\!\!\!\!$
  &
  $\mathrm{SU}(2)_R$
  &&
  \\
  \hline
  $\mathrm{M2}$ &
  &
  && $\sfrac{5}{16}$
  &
  $\mathbb{R}^{2,1\vert 5 \cdot \mathbf{2}}$ &
  &
  \begin{tabular}{c}
    $2 \mathbb{D}_{n+2}$,
    \\
    $2T, 2O  , 2I$
  \end{tabular}
  &
  $\!\!\overset{\Delta}{\subset}\!\!$
  &
  $\mathrm{SU}(2)_L$
  &
  $\!\!\!\!\!\times\!\!\!\!\!$
  &
  $\mathrm{SU}(2)_R$
  &&
  \\
  \hline
  $\mathrm{M2}$ &
  &
  $\sfrac{1}{4}$ &$\!\!\!\!\!\!\!\!\!\!=\!\!\!\!\!\!\!\!\!\!$& $\sfrac{4}{16}$
  &
  $\mathbb{R}^{2,1\vert 4 \cdot \mathbf{2}}$ &
  &
  \begin{tabular}{c}
    $2 \mathbb{D}_{n+2}$,
    \\
    $2O  , 2I$
  \end{tabular}
  &
  $\!\!\overset{(\mathrm{id},\tau)}{\subset}\!\!$
  &
  $\mathrm{SU}(2)_L$
  &
  $\!\!\!\!\times\!\!\!\!\!$
  &
  $\mathrm{SU}(2)_R$
  &&
  \\
  \hline
\end{tabular}
\end{center}
}
\noindent  Moreover, this exhausts those actions such that every non-trivial subgroup fixes precisely a 6-brane,
  and it exhausts those actions such that every non-trivial subgroup fixes precisely a 2-brane in spacetime and is $\geq \sfrac{1}{4}$-BPS.
\end{prop}
\begin{proof}
  Consider the given bosonic fixed point inclusion $\mathbb{R}^{p,1} \hookrightarrow \mathbb{R}^{10,1}$.
  This induces a branching of the real Spin representation $\mathbf{32}$ along the correspondingly broken
  $\mathrm{Spin}$ group
  $$
    \mathrm{Spin}(10-p) \xymatrix{\ar@{^{(}->}[r]&} \mathrm{Spin}(10,1)
  $$
  as a direct sum of irreps (see Example \ref{RelevantExamplesOfRealSpinRepresentations}) as follows:
  $$
    \mathbf{32}
    \simeq
    \left\{
    \begin{array}{cccccc}
      \mathbb{R}^2 &\!\!\!\!\otimes\!\!\!\!& ( \mathbf{8} \oplus \overline{\mathbf{8}} ) & \in \mathrm{Rep}(\mathrm{Spin}(8))  &\vert& p = 2,
      \\
      \mathbb{R}^{16} &\!\!\!\!\otimes\!\!\!\!& \mathbf{4} & \in \mathrm{Rep}(\mathrm{Spin}(4)) &\vert& p = 6.
    \end{array}
    \right.
  $$
  This way the question is reduced to classifying the finite subgroups of $\mathrm{Spin}(8)$ and $\mathrm{Spin}(4)$
  whose canonical action on $\mathbb{R}^8$ and $\mathbb{R}^4$, respectively, fixes only the origin
  (equivalently: whose action on the unit spheres $S^7 \simeq S(\mathbb{R}^8)$ and $S^3 \simeq S(\mathbb{R}^4)$ is free), and computing the
  fixed point subspaces of their action on the representation $\mathbf{8} \oplus \overline{\mathbf{8}}$ of $\mathrm{Spin}(8)$
  and $\mathbf{4}$ of $\mathrm{Spin}(4)$, respectively.
    But now these two classification problems have been solved in \cite{MFFGME10} and in \cite[Sec. 8.3]{MF10}, respectively.\footnote{\label{FigOFarMistake}
    The discussion in \cite[Sec. 8.3]{MF10} is motivated by classifying supersymmetric supergravity solutions corresponding
    to the near horizon limit of black M5 branes, but ends up classifying fixed loci $\mathbb{R}^{6,1} \hookrightarrow \mathbb{R}^{10,1}$
    corresponding to 6-branes, as in the above statement. In physics lingo, this reflects the phenomenon that multiple D6-branes end pairwise on NS5-branes.
    We discussed this subtle point in the physics interpretation in Sec. \ref{MBraneInterpretation}.
  }
\end{proof}

In fact the argument in \cite[Sec. 8.3]{MF10} implies that all of $\mathrm{SU}(2)$ fixes the 6-brane, and all the non-trivial subgroups
fix one and the same $\sfrac{1}{2}$-BPS $\mathbb{R}^{6,1\vert \mathbf{16}}$. For completeness,
we now make this full $\mathrm{SU}(2)$-action, fixing the 6-brane locus, fully explicit:

\begin{lemma}[6-brane locus fixed by $\mathrm{Spin}(2)$]
  \label{Spin2Fixes6BraneLocus}
  The super subspace $\R^{6,1|{\bf 16}} \hookrightarrow \mathbb{R}^{10,1\vert \mathbf{32}}$ is fixed by every element in every $\mathrm{Spin}(2)$ subgroup of $\Spin(4)$ of the form
  \begin{equation}
    \label{CircleActingFixing6Brane}
     \Spin(2)
       \simeq
     \left\{
     \exp\left(
       \tfrac{t}{2} \mathbf{\Gamma}_6 \mathbf{\Gamma}_7
     \right)
     \exp\left(
       -\tfrac{t}{2} \mathbf{\Gamma}_8 \mathbf{\Gamma}_9
     \right)
     \;:\; t \in \R \right\}
\end{equation}
  where $\{b_6, b_7, b_8, b_9\}$ is an orthonormal basis of the subspace orthogonal to $\R^{6,1}$.
\end{lemma}
\begin{proof}
It is clear that the bosonic subspace $\mathbb{R}^{6,1}$ is fixed. We need to check that
also the fermionic subspace is fixed.
Setting $t = \pi$ in (\ref{CircleActingFixing6Brane}), we see that the subgroup in question
contains a 6-brane involution $\sigma = -\mathbf{\Gamma}_{6789}$. We already know from Lemma \ref{6braneFixed} that this
fixes $\mathbb{R}^{6,1\vert \mathbf{16}}$.
Since $\mathrm{Spin}(2)$ is a connected 1-dimensional Lie group,
generated by the single Lie algebra element $X = \mathbf{\Gamma}_{67} - \mathbf{\Gamma}_{89}$, it suffices to show that $\ker(X) = \mathbf{16}$.
For this we use the octonionic presentation of $\mathbf{32}$ from Example \ref{OctonionPresentationOf32OfPin11}.
With this, and under the identification from Lemma \ref{6braneFixed}
where $\mathbf{16} \simeq \mathbb{H}^{4}$, we need to show that $\ker(X) = \mathbb{H}^{4}$.
Multiplying out the gamma matrices using (\ref{11dGammaAsOctonionicAndTensor}), we see:
\begin{align*}
X & = \varepsilon^2 \otimes J^2 L_{e_4} L_{e_5} - \varepsilon^2 \otimes J^2 L_{e_6} L_{e_7}
\\
&= -1 \otimes \left( L_{e_4} L_{e_5} - L_{e_6} L_{e_7} \right).
\end{align*}
So, we must check that $L_{e_4} L_{e_5} = L_{e_6} L_{e_7}$ as linear transformations
of $\H^4$.


Letting $i,j,k$ be the unit imaginary quaternions and $\ell$ any imaginary unit in
$\O$  orthogonal to $\H$, according to (\ref{CayleyDicksonConstruction}), we may choose $e_4 = \ell$, $e_5 = \ell i$, $e_6 = \ell j$
and $e_7 = \ell k$. Using the Cayley--Dickson relations (\ref{RelationsForCayleyDickson})
we compute for any $\psi \in \H^4$ as follows:
\begin{eqnarray*}
X \psi & = & \left( -L_{e_4} L_{e_5} + L_{e_6} L_{e_7} \right) \psi \\
& = & -\ell((\ell i) \psi) + (\ell j) ((\ell k) \psi) \\
& = & -\ell((\overline{i} \ell) \psi) + (\ell j)((k \ell^{-1}) \psi) \\
& = & -\ell((\overline{i} \, \overline{\psi})\ell) + (\ell j)((k \overline{\psi}) \ell^{-1}) \\
& = & \ell((\overline{i} \, \overline{\psi})\ell^{-1}) + (\ell j)((k \, \overline{\psi}) \ell^{-1}) \\
& = & \psi i + \overline{jk \overline{\psi}} \\
& = & \psi i - \psi i \\
& = & 0
\end{eqnarray*}
where in the third line we used (\ref{SomeRelationAmongImaginaryQuaternions}).
\end{proof}

\begin{prop}[6-brane locus fixed by $\mathrm{Spin}(3)$]
  The super subspace $\R^{6,1|{\bf 16}} \hookrightarrow \mathbb{R}^{10,1\vert \mathbf{32}}$ is fixed by every element
  in $\mathrm{SU}(2)_L$.
\end{prop}
\begin{proof}
  By Lemma \ref{Spin2Fixes6BraneLocus} every nontrivial element of $\mathrm{Spin}(2)$
  fixes the 6-brane. The copy of $\Spin(2)$ that appears in the statement of Lemma
  \ref{Spin2Fixes6BraneLocus} is the lift of $\U(1) \subseteq \SU(2) \subseteq
  \SO(4)$ to $\Spin(4)$, where $\U(1)$ is included as the diagonal matrices
  \[ e^{it} \in \U(1) \mapsto \begin{pmatrix} e^{it} & 0 \\ 0 & e^{-it} \end{pmatrix}
  \in \SU(2) . \]
  Yet the proof of Lemma \ref{Spin2Fixes6BraneLocus} applies equally well to any
  other choice of $\U(1)$ subgroup of $\SU(2)$ lifted to $\Spin(4)$. Hence, it
  applies to all of $\SU(2)$.
\end{proof}

This concludes our list of results expanding on the statement of Prop. \ref{M2AndMK6ADE}.

\subsection{Non-simple singularities and Brane intersection laws }
\label{IntersectingBlackBraneSpecies}

There is an evident concept of intersections of simple singularities in super spacetime (Def.\ \ref{IntersectionOfSimpleSingularities} below)
as well as of higher order intersections. This leads to a multitude of singularities of decreasing BPS degree. Here we discuss
most of the intersections of two simple singularities (Prop. \ref{NonSimpleRealSingularities} below). Higher order
intersections may be discussed analogously. In the interpretation
via the black brane scan (Sec. \ref{MBraneInterpretation}) these correspond to bound/intersecting black brane species.

\begin{defn}[Intersection of simple singularities by Cartesian product group actions]
  \label{IntersectionOfSimpleSingularities}
  Let $G_1$ and $G_2$ be two finite group actions (Example \ref{SuperLieAlgebrasWithGActionAsGSuperspaces}) on $\mathbb{R}^{10,1\vert \mathbf{32}}$ (Def.\ \ref{ExamplesOfSuperMinkowskiSpacetimes}), each corresponding to a simple singularity (Def.\ \ref{SimpleSingularities}), such that these actions
  commute with each other, hence such that the Cartesian product action $G_1 \times G_2$ exists.
  Then we say that the \emph{intersection} of the simple singularity of $G_1$ with the simple singularity of $G_2$ is that corresponding to
  the product group action $G_1 \times G_2$.
\end{defn}

\begin{prop}[Non-simple real ADE-Singularities in 11d super-Minkowski spacetime]
\label{NonSimpleRealSingularities}
The non-simple singularities (Def.\ \ref{SimpleSingularities}) in $D = 11$, $\mathcal{N} = 1$ super-Minkowski spacetime
arising as intersections (Def.\ \ref{IntersectionOfSimpleSingularities}) of two simple singularities from Theorem \ref{SuperADESingularitiesIn11dSuperSpacetime} include
those shown in the following table:

\vspace{.3cm}

\hspace{-.5cm}
\scalebox{.78}{
\begin{tabular}{|lc|ccc|lc|ccccccccc|}
  \hline
  \multicolumn{2}{|c}{
  \multirow{2}{*}{
  \begin{tabular}{c}
    \bf Black brane
    \\
    \bf species
  \end{tabular}
  }
  }
  &
  \multicolumn{3}{|c|}{
  \multirow{2}{*}{
    {\bf BPS}
  }
  }
  &
  \multicolumn{2}{c|}{
    \begin{tabular}{c}
    \bf Fixed
    \\
    \bf locus
    \end{tabular}
  }
  &
  \begin{tabular}{c}
    {\bf Type of}
    \\
    {\bf singularity}
  \end{tabular}
  &
  \multicolumn{8}{l|}{\hspace{2.5cm} {\bf Intersection law} }
  \\
  &&&&&
  \multicolumn{2}{|c|}{
    {\bf in} $\mathbb{R}^{10,1\vert \mathbf{32}}$
  }
  &
   {\bf in} $\mathbb{R}^{10,1}$
  & $\simeq$
  & $\mathbb{R}^{1,1}$
  & $\!\!\!\!\!\!\!\!\oplus\!\!\!\!\!\!\!\!$
  & $\mathbb{R}^4$
  & $\!\!\!\!\!\!\!\!\oplus\!\!\!\!\!\!\!\!$
  & $\mathbb{R}^4$
  & $\!\!\!\!\!\!\!\!\oplus\!\!\!\!\!\!\!\!$
  & $\mathbb{R}^1$
  \\
  \hline
  \multicolumn{7}{l}{ \bf Bound/intersecting brane species  $ \phantom{ {{{A^A}^A}^A}^A} $  }
  &
  \multicolumn{9}{l}{ \bf Non-simple singularities }
  \\
  \hline
  \multirow{3}{*}{ $\mathrm{M2}$ }
  & $\mathrm{M2}$
  &
  \multirow{3}{*}{
    $\sfrac{1}{4}$
  }
  &&&
  \multirow{3}{*}{
    $\mathbb{R}^{2,1\vert 4 \cdot \mathbf{2}}$
  }
  & $\mathbb{R}^{2,1\vert 5 \cdot \mathbf{2}}$
  &
  \begin{tabular}{c}
    $ 2T, 2O, 2I$
  \end{tabular}
  &
   $\overset{
     \Delta
   }{\subset}$
  &
  &&
  $ \mathrm{SU}(2)_L $
  &
  $\!\!\!\!\!\times\!\!\!\!\!$
  &
  $\mathrm{SU}(2)_R$
  &&
  \\
  &
  \raisebox{0pt}{
  \hspace{4pt}\begin{rotate}{90} $\Vert$ \end{rotate}
  }
  &
  &&
  &
  &
  \raisebox{5pt}{
  \begin{rotate}{90} $\Vert$ \end{rotate}
  }
  &
  $\!\!\!\!\times\!\!\!\!\!$
  &&
  ------
  &&&&&&
  ------
  \\
  &
  $\mathrm{MK6}$
  &&&
  &&
    $\mathbb{R}^{6,1\vert \mathbf{16}}$
  &
  $\mathbb{Z}_2$
  &
  $=$
  &&
  &
  &&
  $ (\mathbb{Z}_2)_R $
  &&
  \\
  \hline
  \multirow{3}{*}{
    $\mathrm{NS1}_H$
  }
  & $\mathrm{M2}$
  &
  \multirow{3}{*}{
    $\sfrac{3}{16}$ \emph{or} $\sfrac{1}{4}$
  }
  &&
  &
  & $\mathbb{R}^{2,1\vert \geq 6 \cdot \mathbf{2}}$
  &
  $\mathbb{Z}_{n+1}$
  &
    $\overset{\Delta}{\subset}$
  &
  &&
   $ \mathrm{SU}(2)_L $
  &
  $\!\!\!\!\times\!\!\!\!\!$
  &
    $ \mathrm{SU}(2)_R $
  &&
  \\
  & $\bot$
  &
  &&&
  $\mathbb{R}^{1,1\vert \geq 6 \cdot \mathbf{1}}$ &  $\bot$
  &
  $\!\!\!\!\times\!\!\!\!\!$
  &&
  ------
  &&&&&&
  \\
  & $\mathrm{MO9}_H$
  &&&&
  & $\mathbb{R}^{9,1\vert \mathbf{16}}$
  &
  $\mathbb{Z}_2$
  & $=$ &
  &&&&&&
  $(\mathbb{Z}_2)_{\mathrm{HW}}$
  \\
  \hline
  \multirow{3}{*}{ $\mathrm{M1}$ }
  & $\mathrm{M2}$
  &
  \multirow{3}{*}{
    $\sfrac{3}{16}$ \emph{or} $\sfrac{1}{4}$
  }
  &&
  &
  \multirow{3}{*}{ $\mathbb{R}^{1,1\vert \geq 6 \cdot \mathbf{1}}$  }
  & $ \mathbb{R}^{2,1\vert \geq 6 \cdot \mathbf{2}}$
  &
  $\mathbb{Z}_{n+1}$
  &
    $\overset{\Delta}{\subset}$
  &
  &&
   $ \mathrm{SU}(2)_L $
  &
  $\!\!\!\!\times\!\!\!\!\!$
  &
    $ \mathrm{SU}(2)_R $
  &&
  \\
  & $\bot$
  &
  &&&
   &  $\bot$
  &
  $\!\!\!\!\times\!\!\!\!\!$
  &&
  ------
  &&&&&&
  \\
  &  $\mathrm{MO5}$
  &&&&
  & $\mathbb{R}^{5,1\vert \mathbf{16}}$
  &
  $\mathbb{Z}_2$
  & $ \overset{\Delta}{\subset} $
  &
  &&
  &
  &
  $ (\mathbb{Z}_2)_R $
  &
  $\!\!\!\!\times\!\!\!\!\!$
  &
  $(\mathbb{Z}_2)_{\mathrm{HW}}$
  \\
  \hline
  & $\mathrm{M2}$
  &
  \multirow{3}{*}{
    $\sfrac{3}{16}$ \emph{or}  $\sfrac{1}{4}$
  }
  &&
  &
  \multirow{3}{*}{ $\mathbb{R}^{1,1\vert \geq 6 \cdot \mathbf{1}}$ }
  &
   $ \mathbb{R}^{2,1\vert \geq 6 \cdot \mathbf{2}}$
  &
  $\mathbb{Z}_{n+1}$
  &
    $\overset{\Delta}{\subset}$
  &
  &&
   $ \mathrm{SU}(2)_L $
  &
  $\!\!\!\!\times\!\!\!\!\!$
  &
    $ \mathrm{SU}(2)_R $
  &&
  \\
  &
  \hspace{4pt}\begin{rotate}{90} $\Vert$ \end{rotate}
  &
  &&&
  &
  \begin{rotate}{90} $\Vert$ \end{rotate}
  &
  $\!\!\!\!\times\!\!\!\!\!$
  &&
  ------
  &&&&&&
  \\
  & $\mathrm{MO1}$
  &&&&
  & $\mathbb{R}^{1,1\vert 16 \cdot \mathbf{1}}$
  &
  $\mathbb{Z}_2$
  & $ \overset{\Delta}{\subset} $ &
  &&
  $(\mathbb{Z}_2)_L$
  &
  $\!\!\!\!\times\!\!\!\!\!$
  &
  $(\mathbb{Z}_2)_R$
  &
  $\!\!\!\!\times\!\!\!\!\!$
  &
  $(\mathbb{Z}_2)_{\mathrm{HW}}$
  \\
  \hline
  \multirow{3}{*}{ $\mathrm{M5}_{\mathrm{ADE}}$ }
  & $\mathrm{M5}$
  &
  \multirow{3}{*}{
    $\sfrac{1}{4}$
  }
  &&
  &
  \multirow{3}{*}{ $\mathbb{R}^{5,1\vert 1 \cdot \mathbf{8}}$ }
  &
   $ \mathbb{R}^{5,1\vert 2 \cdot \mathbf{8}}$
  &
  $\mathbb{Z}_{2}$
  &
    $\overset{\Delta}{\subset}$
  &
  &&
  &
  &
   $ (\mathbb{Z}_2)_R $
  &
  $\!\!\!\!\times\!\!\!\!\!$
  &
  $ (\mathbb{Z}_2)_{\mathrm{HW}} $
  \\
  &
  \hspace{4pt}\begin{rotate}{90} $\Vert$ \end{rotate}
  &
  &&&
  &
  \hspace{4pt}\begin{rotate}{90} $\Vert$ \end{rotate}
  &
  $\!\!\!\!\times\!\!\!\!\!$
  &&
  \multicolumn{4}{c}{
    ------------------------------
  }
  &&
  &
  \\
  & $ \mathrm{MK6} $
  &&&&
  & $\mathbb{R}^{6,1\vert \mathbf{16}}$
  &
  $
  \begin{tabular}{c}
    $\mathbb{Z}_{n+1}, 2 \mathbb{D}_{n+2}$,
    \\
    $2T, 2O  , 2I$
  \end{tabular}
  $
  & $ \subset $ &
  &&
  &&
  $ \mathrm{SU}(2)_R $
  &
  &
  \\
  \hline
  \multirow{3}{*}{ $ \mathrm{NS5}_H $ }
  & $\mathrm{MO9}_H$
  &
  \multirow{3}{*}{
    $\sfrac{1}{4}$
  }
  &&
  &
  \multirow{3}{*}{ $\mathbb{R}^{5,1\vert 1 \cdot \mathbf{8}}$  }
  &
   $ \mathbb{R}^{9,1\vert \mathbf{16}}$
  &
  $\mathbb{Z}_{2}$
  &
    $ \overset{\phantom{\Delta}}{=}$
  &
  &&
  &
  &
  &
  &
  $ (\mathbb{Z}_2)_{\mathrm{HW}} $
  \\
  &
  \hspace{4pt}\begin{rotate}{90} $\Vert$ \end{rotate}
  &
  &&&
  &
  \begin{rotate}{90} $\Vert$ \end{rotate}
  &
  $\!\!\!\!\times\!\!\!\!\!$
  &&
  \multicolumn{4}{c}{
    ------------------------------
  }
  &&
  &
  \\
  & $ \mathrm{M5} $
  &&&&
  & $\mathbb{R}^{5,1\vert 2 \cdot \mathbf{8}}$
  &
  $
  \mathbb{Z}_2
  $
  & $ \overset{\Delta}{\subset} $ &
  &&
  &&
  $ (\mathbb{Z}_2)_R $
  &
  $\!\!\!\!\times\!\!\!\!\!$
  &
  $ (\mathbb{Z}_2)_{\mathrm{HW}} $
  \\
  \hline
  \multirow{3}{*}{ $\tfrac{1}{2}\mathrm{NS5}_I$ }
  & $\mathrm{MO9}_{I}$
  &
  \multirow{3}{*}{
    $\sfrac{1}{4}$
  }
  &&
  &
  \multirow{3}{*}{ $\mathbb{R}^{5,1\vert 1 \cdot \mathbf{8}}$ }
  &
   $ \mathbb{R}^{9,1\vert \mathbf{16}}$
  &
  $\mathbb{Z}_{2}$
  &
    $\overset{\phantom{\Delta}}{=}$
  &
  &&
  &
  &
  &
  &
  $ (\mathbb{Z}_2)_{\mathrm{HW}} $
  \\
  &
  \raisebox{-3pt}{$\top$}
  &
  &&&
  &
  \raisebox{-1pt}{$\top$}
  &
  $\!\!\!\!\times\!\!\!\!\!$
  &&
  \multicolumn{4}{c}{
    ------------------------------
  }
  &&
  &
  \\
  & $ \mathrm{MK6} $
  &&&&
  & $\mathbb{R}^{6,1\vert \mathbf{16}}$
  &
  $
  \begin{tabular}{c}
    $\mathbb{Z}_{n+1}, 2 \mathbb{D}_{n+2}$,
    \\
    $2T, 2O  , 2I$
  \end{tabular}
  $
  & $ \subset $ &
  &&
  &&
  $ \mathrm{SU}(2)_R $
  &
  &
  \\
  \hline
\end{tabular}
}

\vspace{.2cm}
  The cases $\mathrm{M2} \dashv \mathrm{MO9}, \mathrm{M5}, \mathrm{MW}$ are $\sfrac{1}{4}$ BPS when
  the $\mathrm{M2}$-singularity is of type $\mathbb{Z}_2 \overset{\Delta}{\subset} \mathrm{SU}(2)_L \times \mathrm{SU}(2)_R$,
  and $\sfrac{3}{16}$ BPS when the $\mathrm{M2}$-singularity is of type $\mathbb{Z}_{\geq 3} \overset{\Delta}{\subset} \mathrm{SU}(2)_L \times \mathrm{SU}(2)_R$.

\end{prop}
\begin{proof}
  In each case the existence is  given by checking that the two factors in the Cartesian product group $G_1 \times G_2$,
  that labels the singularity type, indeed have commuting actions on super spacetime.
  Apart from the immediate commutativity following from \eqref{SpacetimeActions}, this makes use of the fact that
  $(\mathbb{Z}_2)_{L,R} \subset \mathrm{SU}(2)_{L,R}$ is the inclusion of the group center.

  What remains to verify is that the intersections have spinor reps/BPS-degrees as shown, but this is directly checked
  in each case by means of the explicit formulas for the fixed loci established in Sec. \ref{SingleBlackBraneSpecies}.

  For example, in the case labeled $\mathrm{M5}_{\mathrm{ADE}}$, the proofs of Lemmas \ref{5braneFixed} and \ref{6braneFixed} show that
  for $G_{\mathrm{ADE}} = \mathbb{Z}_2$ the
   fermionic fixed locus is the intersection of $\mathbb{H}^2 \oplus \mathbb{H}^2 \ell$ from \eqref{M5FixedFermions} with
   $\mathbb{H}^2 \oplus \mathbb{H}^2$ from \eqref{MK6FixedFermions}, which is $\mathbb{H}^2 \simeq 1 \cdot \mathbf{8}$;
   and by Lemma \ref{M2AndMK6ADE} this conclusion still holds for general $G_{\mathrm{ADE}}$.
  \end{proof}

 The case which we denote $\mathrm{M2} \,\Vert\, \mathrm{MK6}$ also appears as \cite[7.2 and around (68)]{MF10}.

\newpage
\hypertarget{Figure1}{$\,$}

\vspace{-0cm}

\begin{center}
\begin{tabular}{l|l}
\vspace{-0cm}
\\
\raisebox{-12pt}{\footnotesize
\begin{tabular}{ccll}
  $\mathrm{MK6}$ &$=$& $\big( \mathbb{R}^{10,1\vert\mathbf{32}} \big)^{(G_{\mathrm{ADE}})_{R}}$
  & (\ref{TheMK6})
  \\
  $\phantom{A}$
  \\
  $\mathrm{MO9}$ & $=$ & $\big( \mathbb{R}^{10,1\vert\mathbf{32}} \big)^{G_{\mathrm{HW}}}$
  &  (\ref{TheMO9})
  \\
  $\phantom{A}$
 \\
 $\tfrac{1}{2}\mathrm{NS5}_I$ & $=$ & $\mathrm{MK6} \dashv \mathrm{MO9}_I  $
 & (\ref{TheBlackNS5})
 \\
  $\phantom{A}$
 \\
 $\tfrac{1}{2}\mathrm{NS5}_H$ & $=$ & $\mathrm{MK6} \dashv \mathrm{MO9}_H $
 & (\ref{TheBlackNS5})
\end{tabular}
}
&

\raisebox{-110pt}{
\includegraphics[width=.5\textwidth]{half-M5}
}

\vspace{+0cm}

\\
\hline
\vspace{-0cm}
\\
\raisebox{-30pt}{\footnotesize
\begin{tabular}{ccll}
  $\mathrm{M2}$ &$=$& $\big( \mathbb{R}^{10,1\vert\mathbf{32}} \big)^{ (G_{\mathrm{ADE}})_{\Delta}  }$
  & (\ref{TheBlackM2})
  \\
  $\phantom{A}$
  \\
  $\mathrm{MO9}$ & $=$ & $\big( \mathbb{R}^{10,1\vert\mathbf{32}} \big)^{G_{\mathrm{HW}}}$
  &  (\ref{TheMO9})
  \\
  $\phantom{A}$
 \\
 $\mathrm{NS1}_H$ & $=$ & $\mathrm{M2} \dashv \mathrm{MO9}_H$
 & (\ref{TheBlackNS1H})
 \\
  $\phantom{A}$
 \\
 $\mathrm{E1}$ & $=$ & $\mathrm{M2} \dashv \mathrm{MO9}_I $
 & (\ref{TheBlackNS1H})
\end{tabular}
}
&
\raisebox{-130pt}{
\includegraphics[width=.5\textwidth]{half-M2}
}

\vspace{-0cm}

\\
\hline
\vspace{-1.6cm}
\\
\raisebox{-128pt}{\footnotesize
\begin{tabular}{ccll}
  $\mathrm{MK6}$ &$=$& $\big( \mathbb{R}^{10,1\vert\mathbf{32}} \big)^{ (G_{\mathrm{ADE}})_{R} }$
  & (\ref{TheMK6})
  \\
  $\phantom{A}$
  \\
  $\mathrm{M5}$ & $=$ & $\big( \mathbb{R}^{10,1\vert\mathbf{32}} \big)^{ G_{\mathrm{W}} }$
  &  (\ref{TheBlackM5})
  \\
  $\phantom{A}$
 \\
 $\mathrm{M5}_{\mathrm{ADE}}$ & $=$ & $\mathrm{M5} \; \Vert \; \mathrm{MK6}$
 & (\ref{TheBlackNS1H})
\end{tabular}
}
&
\raisebox{-230pt}{
\includegraphics[width=.5\textwidth]{half-NS5}
}
\end{tabular}
\end{center}
{\footnotesize
{\bf Figure 1. Intersections of simple singularities in 11d super spacetime}.  The main axes indicate
fixed loci of the simple singularities from Theorem \ref{SuperADESingularitiesIn11dSuperSpacetime};
 their intersections are
shown according to Prop. \ref{NonSimpleRealSingularities}.
Also indicated is the sphere around the vertical singularity,
together with its cone, according to \eqref{NearHorizonGeometry}.
The inset also displays behavior under S/T-duality.
}

\newpage

$\,$

\vspace{-1.5cm}
\begin{center}
\begin{tabular}{l|l}
\\
\begin{tabular}{ccll} \footnotesize
  $\mathrm{M2}$ &$=$& \footnotesize $\big( \mathbb{R}^{10,1\vert\mathbf{32}} \big)^{ (G_{\mathrm{ADE}})_{\Delta}  } $
  & \footnotesize \eqref{TheBlackM2}
  \\
  $\phantom{A}$
  \\ \footnotesize
  $\mathrm{MO9}$ & $=$ & \footnotesize $\big( \mathbb{R}^{10,1\vert\mathbf{32}} \big)^{G_{\mathrm{HW}}}$
  &  \footnotesize \eqref{TheMO9}
  \\
  $\phantom{A}$
 \\ \footnotesize
 $\mathrm{M1}$ & $=$ &  \footnotesize $ \mathrm{M2} \dashv \mathrm{M5}$
 & \footnotesize \eqref{TheSelfDualString}
\end{tabular}
&

\raisebox{-110pt}{
\includegraphics[width=.5\textwidth]{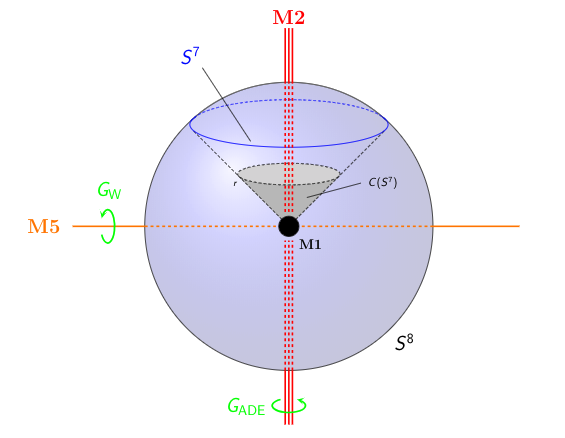}
}

\vspace{-0cm}

\\
\hline
\\
\begin{tabular}{ccll} \footnotesize
  $\mathrm{MK6}$ &$=$&   \footnotesize $\big( \mathbb{R}^{10,1\vert\mathbf{32}} \big)^{ (G_{\mathrm{ADE}})_{R} }$
  & \footnotesize  \eqref{TheMK6}
  \\
  $\phantom{A}$
  \\ \footnotesize
  $\mathrm{MO9}$ & $=$ & \footnotesize $\big( \mathbb{R}^{10,1\vert\mathbf{32}} \big)^{G_{\mathrm{HW}}}$
  & \footnotesize  \eqref{TheMO9}
  \\ \footnotesize
  $\phantom{A}$
 \\ \footnotesize
  $\mathrm{M5}$ &$=$& \footnotesize $\big( \mathbb{R}^{10,1\vert\mathbf{32}} \big)^{ G_W }$
  & \footnotesize \eqref{TheBlackM5}
  \\ \footnotesize
  $\phantom{A}$
  \\ \footnotesize
  $\mathrm{NS5}_H$ &  $=$ &  \footnotesize $\mathrm{MO9} \;\Vert\;  \mathrm{M5} \;\Vert\; \mathrm{MK6} $
\end{tabular}
&

\raisebox{-100pt}{
\includegraphics[width=.5\textwidth]{triple-pic}
}

\end{tabular}
\end{center}
\label{IntersectionGraphicsB}
{\footnotesize
{\bf \hypertarget{Figure2}{Figure 2}. Further intersections of simple singularities in 11d super spacetime}.  The main axes
indicate the simple singularities from Theorem \ref{SuperADESingularitiesIn11dSuperSpacetime}; their intersections are
displayed according to Prop. \ref{NonSimpleRealSingularities}.
}


\vspace{5mm}
\begin{example}[M2 inside MK6]
The example $\mathrm{M2} \Vert \mathrm{MK6}$ in Prop. \ref{NonSimpleRealSingularities} is implicitly considered in
 \cite[Sec. 7.2 and around (68)]{MF10}, there
viewed as a black M2 in an orbifold background. As in footnote \ref{FigOFarMistake}, we may identify the orbifold away
from the M2-orbifold singularity,
which is $\mathbb{R}^{10,1\vert\mathbf{32}} \sslash (\mathbb{Z}_2)_L$,
as that of a black MK6 according to Theorem \ref{SuperADESingularitiesIn11dSuperSpacetime}, hence the complete
orbifold as the intersection of the two.
\end{example}

Notice how the group theory implies \emph{intersection laws}:

\begin{example}[M2 intersecting M5]
In order to have an intersection of singularities (Def.\ \ref{IntersectionOfSimpleSingularities}), the corresponding group
actions need to commute with each other, so that their product group action makes sense. For possible intersections of the
black M2 with the black M5, as in Prop. \ref{NonSimpleRealSingularities},
this means that the action of the diagonal subgroup in $(\mathbb{Z}_2)_L \times (\mathbb{Z}_2)_{\mathrm{HW}}$
has to commute with the action of the diagonal subgroup action in $\mathrm{SU}(2)_L \times \mathrm{SU}(2)_R$. But
the first $(\mathbb{Z}_2)_L$ is in the center of the latter. Hence the condition is that $(\mathbb{Z}_2)_{\mathrm{HW}}$
commutes with the diagonal in $\mathrm{SU}(2)_L \times \mathrm{SU}(2)_R$. This diagonal
 is generated from elements of the form $\mathbf{\Gamma}_{a b} + \mathbf{\Gamma}_{cd}$
with indices $a,b,c,d$ ranging, say in $\{2,\cdots, 9\}$. But $(\mathbb{Z}_2)_{\mathrm{HW}}$ is generated from a single
$\mathbf{\Gamma}_{e}$. For that to commute with the others, we must have either $e =1$ or $e = 10$. But this implies
that the fixed spaces of the two actions do not contain each other, hence that the M2 intersects the M5 without being
contained in it.
\end{example}

\newpage

\section{Real ADE-equivariant rational cohomotopy}
 \label{ADEEquivariantRationalCohomotopy}

Here we discuss certain group actions on the 4-sphere (Def.\ \ref{SuspendedHopfAction} below), as well as the resulting incarnation of the 4-sphere
as an object in equivariant rational super homotopy theory (Def.\ \ref{GEquivariantRationalSuperHomotopy}).
The motivation for the particular actions considered was explained in Remark \ref{OriginOfGroupActionsOn4Sphere} of Sec.
 \ref{MBraneInterpretation}.
The main result of this section is the explicit system of minimal dgc-algebra models for this equivariant 4-sphere, this is Prop. \ref{ADEWEquivariant4SpheredgAlgebraicCoefficients} below.

\medskip
By Example \ref{RationalCohomotopyOfSuperSpaces} and Example \ref{GSpacesAsGSuperspaces},
taking this equivariant 4-sphere as the coefficient for a generalized cohomology theory (Def.\ \ref{CohomoloyFromHomotopy})
defines \emph{real ADE-equivariant rational super cohomotopy} in degree 4. We study this cohomology theory on super-spacetimes
 in Sec. \ref{ADEEquivariantMBraneSuperCucycles} below.

\subsection{Real ADE-actions on the 4-sphere}
\label{RealADEActionsOnThe4Sphere}

Following Remark \ref{OriginOfGroupActionsOn4Sphere}, we consider the following group actions on the 4-sphere:

\begin{defn}[Actions on the 4-sphere]
 \label{SuspendedHopfAction}
We may regard the 4-sphere as the unit sphere in the direct sum of the real numbers $\mathbb{R}$ with the quaternions $\mathbb{H}$
(Example \ref{TheFourRealNormedDivisionAlgebras}):
\begin{equation}
  \label{The4SphereInRPlusH}
  S^4
    \simeq
  S(
    \mathbb{R} \oplus
    \underset{\mathbb{H}}{\underbrace{{ \overset{ \mathrm{Im}(\mathbb{H}) }{\overbrace{\mathbb{R} \oplus \mathbb{R}^2 } } \oplus \mathbb{R}  }}}
  )
  \,.
\end{equation}
This decomposition induces the following group actions (Def.\ \ref{GroupActions}) on the 4-sphere (we use the shorthand notation for actions
from Remark \ref{ActionShorthandNotation}):
\begin{enumerate}[{\bf (i)}]
\vspace{-1mm}
\item Multiplication on the left and right by unit quaternions in $\mathbb{H}$
  \eqref{SU2ActionsOnQuaternions} preserves the 4-sphere
  \eqref{UnitQuaternionsActOrthogonallyOnQuaternions}, giving two actions of
  $\mathrm{SU}(2)$, which we denote by $\mathrm{SU}(2)_L$ and $\mathrm{SU}(2)_R$,
  respectively.  \vspace{-2mm}
  \item These two actions manifestly commute with each other, and hence we have the corresponding action of the
  Cartesian product of $\mathrm{SU}(2)$ with itself, which we denote by $\mathrm{SU}(2)_L \times \mathrm{SU}(2)_R$.
  \vspace{-2mm}
  \item
    We denote the action induced from this via the diagonal homomorphism $\mathrm{SU}(2) \overset{\Delta}{\longrightarrow} \mathrm{SU}(2) \times \mathrm{SU}(2)$
    by $\mathrm{SU}(2)_\Delta$. This factors (via \eqref{ImaginaryQuaternionsRotationAction}) through the action induced by the canonical
    action of $\mathrm{SO}(3)$ on $\mathrm{Im}(\mathbb{H}) \simeq_{\mathbb{R}} \mathbb{R}^3$.
  \vspace{-2mm}
  \item There is then an inclusion $S^1 \hookrightarrow \mathrm{SU}(2)$ such that the corresponding restriction of the
   diagonal action fixes the second coordinate in \eqref{The4SphereInRPlusH}.
   We accordingly denote this induced action by $S^1_{\Delta}$.
  \vspace{-2mm}
  \item We denote by $(\mathbb{Z}_2)_{\mathrm{HW}}$ the $\mathbb{Z}_2$-action induced by the involution (Example \ref{Z2ActionsAreInvolutions}) given by reflection of the
  last coordinate in \eqref{The4SphereInRPlusH} (i.e. multiplication by $-1$ on the real part of the quaternionic
  coordinate in \eqref{The4SphereInRPlusH}).
  \vspace{-2mm}
  \item    This commutes with the $\mathrm{SU}(2)_{\Delta}$-action, so that there is the corresponding action of the Cartesian product group,
    which we accordingly denote by $\mathrm{SU}(2)_{\Delta} \times (\mathbb{Z}_2)_{\mathrm{HW}}$.
\end{enumerate}

\end{defn}
\begin{remark}[Summary of actions]
With the evident shorthand notation, the action on the 4-sphere, from Def.\ \ref{SuspendedHopfAction} are given, in terms of the decomposition \eqref{The4SphereInRPlusH},
as follows:
\vspace{-3mm}
\begin{equation}
  \label{ActionsOnThe4SphereInRPlusH}
  S^4
    \simeq
  S(
    \mathbb{R} \;\; \oplus \;\;
    \underset{
      \footnotesize
      \xymatrix{ \mathbb{H} \ar@(dl,ul)[]^{\mathrm{SU}(2)_L} \ar@(ur,dr)[]^{\mathrm{SU}(2)_R}  }
    }{
      \underbrace{
        {
          \overset{
            \footnotesize
            \xymatrix{ \mathrm{Im}(\mathbb{H}) \ar@(ul,ur)[]^{ \mathrm{SU}(2)_{\Delta}  }
 }
          }{
            \overbrace{
              \mathbb{R}
                \oplus
              \xymatrix{ \mathbb{R}^2 \ar@(dl,dr)[]|{ S^1_\Delta } }
            }
          }
          \oplus
          { \xymatrix{ \mathbb{R} \ar@(ul,ur)[]^{ (\mathbb{Z}_2)_{\mathrm{HW}} } } }
        }
      }
    }
  )
\end{equation}
which may by collected into two actions of Cartesian products as
$$
  \xymatrix{
    S^4 \ar@(ul,ur)[]^{ \mathrm{SU}(2)_L \times \mathrm{SU}(2)_R }
  }
\qquad \text{and} \qquad
  \xymatrix{
    S^4 \ar@(ul,ur)[]^{ \mathrm{SU}(2)_\Delta \times (\mathbb{Z}_2)_{\mathrm{HW}} }
  }
  \hspace{-1cm}.
$$
\end{remark}

\begin{example}[Suspended Hopf action]
Consider the canonical inclusion $S^1 \simeq U(1) \hookrightarrow \mathrm{SU}(2)_L$
of the circle group into the special unitary group (as the subgroup of diagonal matrices).
The induced action $S^1_L$ on the 4-sphere, by Def.\ \ref{SuspendedHopfAction}, is the image under topological suspension of the $S^1$-action that exhibits
the complex Hopf fibration $S^3 \to S^2$ (Def.\ \ref{HopfFibration}) as an $S^1$-principal bundle.
\begin{center}
 \includegraphics[width=.4\textwidth]{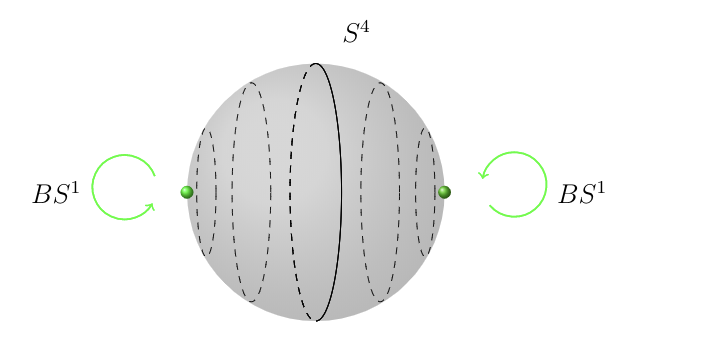}
\end{center}
\end{example}


\subsection{Systems of fixed loci of the 4-Sphere}

\begin{lemma}[The systems of real ADE-fixed points of the 4-sphere]
  \label{GADEGWFixedInFourSphere}
  Let $G_{\mathrm{ADE}} \subset \mathrm{SU}(2)$ be a finite subgroup of $\mathrm{SU}(2)$ (Remark \ref{ADEGroups})
  and write $(G_{\mathrm{ADE}})_{L,R,\Delta}$ for its various actions
  on the 4-sphere, via Def.\ \ref{SuspendedHopfAction}.
  The following displays the systems of fixed point spaces of $S^4$
  (according to Example \ref{FixedPointSpaceRepresented}) for the two actions of product groups
  from Def.\ \ref{SuspendedHopfAction}.
  Here we let $G \subset G_{\mathrm{ADE}}$ denote any \emph{non-cyclic} subgroup. If there is none such, the corresponding entries of the diagrams
  do not appear.

  \item {\bf (i)} For the action $\mathrm{SU}(2)_\Delta \times (\mathbb{Z}_2)_{\mathrm{HW}}$ in Def.\ \ref{SuspendedHopfAction}
  the system of fixed points
  $$
    \left( S^4\right)^{(-)}
    \maps
    \mathrm{Orb}^{\mathrm{op}}_{ (G_{\mathrm{ADE}})_{\Delta} \times (\mathbb{Z}_2)_{\mathrm{HW}}}
    \longrightarrow
    \mathrm{Spaces}
  $$
  is, in part, as follows:
  $$
    \raisebox{45pt}{
    \xymatrix@C=-4pt@R-14pt{
      &
      \frac{
        (G_{\mathrm{ADE}})_{\Delta} \times G_{\mathrm{HW}}
      }{
        \{e_{\mathrm{ADE}}\} \times \{e_{\mathrm{HW}}\}
      }
      \ar@(ul,ur)[]^{(G_{\mathrm{ADE}})_{\Delta} \times G_{\mathrm{HW}} }
      \ar[rr]
      \ar[ddl]
      \ar[dddd]|{ \phantom{ {A \atop A} \atop {A \atop A} } }
      &&
      \frac{
        (G_{\mathrm{ADE}})_{\Delta} \times G_{\mathrm{HW}}
      }{
        \{e_{\mathrm{ADE}}\} \times G_{\mathrm{HW}}
      }
      \ar@(ul,ur)[]^{(G_{\mathrm{ADE}})_{\Delta}  }
      \ar[ddl]
      \ar[dddd]
      \\
      \\
      \frac{
        (G_{\mathrm{ADE}})_{\Delta} \times G_{\mathrm{HW}}
      }{
        (\mathbb{Z}_{n+1})_\Delta \times \{e_{\mathrm{HW}}\}
      }
      \ar@(ul,ur)[]^{   }
      \ar[rr]
      \ar[ddr]
      &&
      \frac{
        (G_{\mathrm{ADE}})_{\Delta} \times G_{\mathrm{HW}}
      }{
        (\mathbb{Z}_{n+1})_{\Delta} \times G_{\mathrm{HW}}
      }
      \ar@(ul,ur)[]^{    }
      \ar[ddr]
      \\
      \\
      &
      \frac{
        (G_{\mathrm{ADE}})_{\Delta} \times G_{\mathrm{HW}}
      }{
        G_{\Delta} \times \{e_{\mathrm{HW}}\}
      }
      \ar[rr]
      &&
      \frac{
        (G_{\mathrm{ADE}})_{\Delta} \times G_{\mathrm{HW}}
      }{
        G_{\Delta} \times G_{\mathrm{HW}}
      }
    }
    }
   \xymatrix{\hspace{.8cm}\ar@{|->}[r]&\hspace{.8cm}}
    \raisebox{45pt}{
    \xymatrix@C=17pt@R=15pt{
      &
      S^4
      \ar@(ul,ur)[]^{(G_{\mathrm{ADE}})_{\Delta} \times G_{\mathrm{HW}} }
      \ar@{<-^{)}}[rr]
      \ar@{<-^{)}}[ddl]
      \ar@{<-^{)}}[dddd]|{ \phantom{ {A \atop A} \atop {A \atop A} } }
      &&
      S^3
      \ar@(ul,ur)[]^{(G_{\mathrm{ADE}})_{\Delta}  }
      \ar@{<-^{)}}[ddl]
      \ar@{<-^{)}}[dddd]
      \\
      \\
      S^2
      \ar@{<-^{)}}[rr]
      \ar@{<-^{)}}[ddr]
      &&
      S^1
      \ar@{<-^{)}}[ddr]
      \\
      \\
      &
      S^1
      \ar@{<-^{)}}[rr]
      &&
      S^0
    }
    }
  $$

\item {\bf (ii)} For the action $\mathrm{SU}(2)_L \times \mathrm{SU}(2)_R $ in Def.\ \ref{SuspendedHopfAction}
  the system of fixed points
  $$
    \left( S^4\right)^{(-)}
    \maps
    \mathrm{Orb}^{\mathrm{op}}_{
      (G_{\mathrm{ADE}})_{L}
      \times
      (G_{\mathrm{ADE}})_{R}
    }
    \longrightarrow
    \mathrm{Spaces}
  $$
  is, in part, as follows:
  $$
    \raisebox{45pt}{
    \xymatrix@C=12pt@R-15pt{
      &
      \frac{
        (G_{\mathrm{ADE}})_{L} \times (G_{\mathrm{ADE}})_{R}
      }{
        \{e_{\mathrm{ADE}}\} \times \{e_{\mathrm{ADE}}\}
      }
      \ar@(ul,ur)[]^{(G_{\mathrm{ADE}})_{L} \times (G_{\mathrm{ADE}})_{R} }
      \ar[dr]
      \ar[rr]
      \ar[dddd]
      &&
      \frac{
        (G_{\mathrm{ADE}})_{L} \times (G_{\mathrm{ADE}})_{R}
      }{
        \{e_{\mathrm{ADE}}\} \times (G_{\mathrm{ADE}})_{R}
      }
      \ar@(ul,ur)[]^{(G_{\mathrm{ADE}})_{L}  }
      \ar[dddd]
      \\
      &&
      \frac{
        (G_{\mathrm{ADE}})_{L} \times (G_{\mathrm{ADE}})_{R}
      }{
        (\mathbb{Z}_{n+1})_\Delta
      }
      \ar[dd]
      \\
      &&
      \\
      &&
      \frac{
        (G_{\mathrm{ADE}})_{L} \times (G_{\mathrm{ADE}})_{R}
      }{
        G_\Delta
      }
      \ar[dr]
      \\
      &
      \frac{
        (G_{\mathrm{ADE}})_{\Delta} \times G_{\mathrm{HW}}
      }{
        (G_{\mathrm{ADE}})_L \times \{ e_{\mathrm{ADE}} \}
      }
      \ar[rr]
      &&
      \frac{
        (G_{\mathrm{ADE}})_{L} \times (G_{\mathrm{ADE}})_{R}
      }{
        (G_{\mathrm{ADE}})_{L} \times (G_{\mathrm{ADE}})_{R}
      }
    }
    }
    \xymatrix{\hspace{.5cm}\ar@{|->}[r]&\hspace{.2cm}}
    \raisebox{45pt}{
    \xymatrix@C=16pt@R17pt{
      &
      S^4
      \ar@(ul,ur)[]^{(G_{\mathrm{ADE}})_{L} \times (G_{\mathrm{ADE}})_{R} }
      \ar@{<-^{)}}[dr]
      \ar@{<-^{)}}[rr]
      \ar@{<-^{)}}[dddd]
      &&
      S^0
      \ar@(ul,ur)[]^{ (G_{\mathrm{ADE}})_{L}  }
      \ar@{=}[dddd]
      \\
      &&
      S^2
      \ar@{<-^{)}}[dd]
      \\
      &&
      \\
      &&
      S^1
      \ar@{<-^{)}}[dr]
      \\
      &
      S^0
      \ar@{=}[rr]
      &&
      S^0
    }
    }
  $$
\end{lemma}
\begin{proof}
  This follows directly by careful inspection, manifestly so using the arrangement of the actions in \eqref{ActionsOnThe4SphereInRPlusH}.
\end{proof}

To first approximation, and for the purposes of the present article, we will be interested in the system of fixed point spaces from Lemma \ref{GADEGWFixedInFourSphere}
only rationally (Def.\ \ref{RationalHomotopyTheory}). Hence in the remainder of this section we work out an explicit model for the corresponding image of the 4-sphere
in equivariant rational homotopy theory (Example \ref{EquivariantrationalHomotopyTheory}) and hence, via Example \ref{GSpacesAsGSuperspaces},
in equivariant rational super homotopy theory (Def.\ \ref{GEquivariantRationalSuperHomotopy}). The result is Prop. \ref{ADEWEquivariant4SpheredgAlgebraicCoefficients} below.


\begin{lemma}[Rational image of real ADE-actions on the 4-sphere]
  \label{ADEActionOn4SphereRationallyTrivial}
  \vspace{-1mm}
\item {\bf (i)}   For each element $g \in \mathrm{SU}(2)_{L,R,\Delta}$ the corresponding action on the 4-sphere, via Def.\ \ref{SuspendedHopfAction}
  is an orientation-preserving isometry (of the round 4-sphere)
  while the nontrivial element $\sigma \in (\mathbb{Z}_2)_{\mathrm{HW}}$ acts as an orientation-\emph{reversing} isometry.

 \item {\bf (ii)}  It follows that the DGC-algebra homomorphism corresponding to $\rho(g)$ after passing to the minimal
  Sullivan model of the $n$-spheres (Example \ref{SpheredgcAlgebraModel}) is the \emph{identity} for all
  $g \in G_{\mathrm{ADE}}$,  while the DGC-algebra homomorphism corresponding to $\rho(\sigma)$ acts as minus the identity on $\omega_4$ and as the identity on $\omega_7$:
  \begin{equation}
    \label{TransformationOfGeneratorsFor4Sphere}
  \mbox{
  \begin{tabular}{|c||c|c|}
    \hline
    & $g \in G_{\mathrm{ADE}}$ & $e \neq \sigma \in G_{\mathrm{HW}}$
    \\
    \hline
    $\omega$ & $\rho(g)^\ast(\omega)$ & $\rho(\sigma)^\ast(\omega)$
    \\
    \hline
    \hline
    $\omega_4$ & $\phantom{-}\omega_4$ & $-\omega_4$
    \\
    \hline
    $\omega_7$ & $\phantom{-}\omega_7$ & $\phantom{-}\omega_7$
    \\
    \hline
  \end{tabular}
  }
  \end{equation}
\end{lemma}
\begin{proof}
  The orientation-preserving isometries   of $S^4$ form the group $\mathrm{SO}(5)$, acting on
  $S^4 \simeq S(\mathbb{R}^5)$ via its canonical action on $\mathbb{R}^5$. Now by
  Def.\ \ref{SuspendedHopfAction} and using (\ref{UnitQuaternionsActOrthogonallyOnQuaternions})
  it follows that the $\rho(g)$ acts through the subgroup $\mathrm{SO}(4) \hookrightarrow \mathrm{SO}(5)$,
  and hence itself as an orientation-preserving isometry.

  This means that under pullback of differential forms along $\rho(g)$, the canonical volume form of $S^4$
  is sent to itself.
  But, under passing to minimal Sullivan models in Example \ref{SpheredgcAlgebraModel}, this volume form is identified with the
  generator $\omega_4 \in \mathrm{CE}(\mathfrak{l}(S^4))$. Hence
  $$
    \rho(g)^\ast(\omega_4) = \omega_4
    \,.
  $$
  With this, respect for the CE-differential implies that also $\rho(g)^\ast(\omega_7) = \omega_7$, since this is the
  only primitive for $-\tfrac{1}{2} \omega_4 \wedge \omega_4$ in $\mathrm{CE}(\mathfrak{l}(S^4))$, according to (\ref{CEAlgebraFor4Sphere}):
  \begin{equation}
    \label{omega7TransformationFromOmega4}
    \xymatrix{
      \rho(g)^\ast(\omega_7)
      \ar@{|->}[d]_{d_{\mathfrak{l}(S^4)}}
      && \omega_7 \ar@{|->}[ll]_{\rho(g)^\ast}
      \ar@{|->}[d]^{d_{\mathfrak{l}(S^4)}}
      \\
      -\tfrac{1}{2} \omega_4 \wedge \omega_4
      &&
      -\tfrac{1}{2} \omega_4 \wedge \omega_4\;.
      \ar@{|->}[ll]_{\rho(g)^\ast}
    }
  \end{equation}

  \vspace{-9mm}
\end{proof}

\medskip
\begin{remark}[Equivariant rational homotopy improving on plain rational homotopy]
Lemma \ref{ADEActionOn4SphereRationallyTrivial} implies that the ADE-actions on the 4-sphere are invisible to
plain rational homotopy theory (Def.\ \ref{RationalHomotopyTheory}). But \emph{equivariant} rational homotopy theory (Example \ref{EquivariantrationalHomotopyTheory}),
and hence \emph{equivariant} rational super homotopy theory (Def.\ \ref{GEquivariantRationalSuperHomotopy}),
is more fine-grained and does
recognize that the ADE-action on the 4-sphere is not in fact trivial, by remembering its fixed point loci.
This is made explicit by Prop. \ref{ADEWEquivariant4SpheredgAlgebraicCoefficients}.
\end{remark}

\newpage

\begin{prop}[System of dgc-algebra models for fixed points real ADE-action on the 4-Sphere]
  \label{ADEWEquivariant4SpheredgAlgebraicCoefficients}
 Under passage to minimal dgc-algebra models $\mathrm{CE}\left( \mathfrak{l}(S^n)\right)$ for the $n$-spheres (Example \ref{SpheredgcAlgebraModel}),
 the systems of fixed points of the real ADE-actions on the 4-sphere, from Lemma \ref{GADEGWFixedInFourSphere},
 becomes the following:

\item   {\bf (i)}
  For the $(G_{\mathrm{ADE}})_{\Delta} \times (\mathbb{Z}_2)_{\mathrm{HW}}$-action the system of minimal dgc-algebras
  $$
    \mathrm{CE}\Big(
      \mathfrak{l}\big(
        ( S^4)^{(-)}
      \big)
    \Big)
    \maps
    \mathrm{Orb}_{ (G_{\mathrm{ADE}})_{\Delta} \times (\mathbb{Z}_2)_{\mathrm{HW}} }
    \xymatrix{\ar[r]&}
    \mathrm{dgcAlg}
  $$
  is, in the parts corresponding to those shown in Lemma \ref{GADEGWFixedInFourSphere}, given  by
  $$
    \hspace{-0.3cm}
    \raisebox{100pt}{
    \xymatrix@C=-16pt@R-4pt{
      &
      \frac{
        (G_{\mathrm{ADE}})_\Delta \times G_{\mathrm{HW}}
      }{
        \{e_{\mathrm{ADE}}\} \times \{e_{\mathrm{HW}}\}
      }
      \ar[rr]
      \ar[d]_{(g_{\mathrm{ADE}}, \sigma_{\mathrm{HW}} )}
      &&
      \frac{
        (G_{\mathrm{ADE}})_\Delta \times G_{\mathrm{HW}}
      }{
        \{e_{\mathrm{ADE}}\} \times G_{\mathrm{HW}}
      }
      \ar[d]^{g_{\mathrm{ADE}}}
      \\
      &
      \frac{
        (G_{\mathrm{ADE}})_{\Delta} \times G_{\mathrm{HW}}
      }{
        \{e_{\mathrm{ADE}}\} \times \{e_{\mathrm{HW}}\}
      }
      \ar[rr]
      \ar[d]_{(g_{\mathrm{ADE}}, e_{\mathrm{HW}} )}
      &&
      \frac{
        (G_{\mathrm{ADE}})_{\Delta} \times G_{\mathrm{HW}}
      }{
        \{e_{\mathrm{ADE}}\} \times G_{\mathrm{HW}}
      }
      \ar[d]^{g_{\mathrm{ADE}}}
      \\
      &
      \frac{
        (G_{\mathrm{ADE}})_{\Delta} \times G_{\mathrm{HW}}
      }{
        \{e_{\mathrm{ADE}}\} \times \{e_{\mathrm{HW}}\}
      }
      \ar[rr]
      \ar[ddl]
      \ar[dddd]|{ \phantom{ {A \atop A} \atop {A \atop A} } }
      &&
      \frac{
        (G_{\mathrm{ADE}})_{\mathrm{\Delta}} \times G_{\mathrm{HW}}
      }{
        \{e_{\mathrm{ADE}}\} \times G_{\mathrm{HW}}
      }
      \ar[ddl]
      \ar[dddd]
      \\
      \\
      \frac{
        (G_{\mathrm{ADE}})_{\Delta} \times G_{\mathrm{HW}}
      }{
        (\mathbb{Z}_{n+1})_\Delta \times \{e_{\mathrm{HW}}\}
      }
      \ar[rr]
      \ar[ddr]
      &&
      \frac{
        (G_{\mathrm{ADE}})_\Delta \times G_{\mathrm{HW}}
      }{
        (\mathbb{Z}_{n+1})_\Delta \times G_{\mathrm{HW}}
      }
      \ar[ddr]
      \\
      \\
      &
      \frac{
        (G_{\mathrm{ADE}})_\Delta \times G_{\mathrm{HW}}
      }{
        G_{\Delta} \times \{e_{\mathrm{HW}}\}
      }
      \ar[rr]
      &&
      \frac{
        (G_{\mathrm{ADE}})_{\Delta} \times G_{\mathrm{HW}}
      }{
        G_{\Delta} \times G_{\mathrm{HW}}
      }
    }
    }
    \xymatrix{\hspace{.2cm}\ar@{|->}[r]&\hspace{.1cm}}
    \hspace{-.0cm}
    \raisebox{100pt}{
    \xymatrix@C=-4pt@R22pt{
      &
      \mathrm{CE}(\mathfrak{l}(S^4))
      \ar[rr]^{
         \mbox{
        \tiny
        $
        \begin{aligned}
          \omega_4 & \mapsto 0
          \\
          \omega_7 & \mapsto 0
        \end{aligned}
        $}
      }
      \ar[d]_{
        \mbox{
        \tiny
        $
        \begin{aligned}
          \omega_4 & \mapsto - \omega_4
          \\
          \omega_7 & \mapsto + \omega_7
        \end{aligned}
        $}
      }
      &&
      \mathrm{CE}(\mathfrak{l}(S^3))
      \ar[d]^{
        h_3 \mapsto h_3
       }
      \\
      &
      \mathrm{CE}(\mathfrak{l}(S^4))
      \ar[rr]^{
         \mbox{
        \tiny
        $
        \begin{aligned}
          \omega_4 & \mapsto 0
          \\
          \omega_7 & \mapsto 0
        \end{aligned}
        $}
      }
      \ar[d]_{
        \mbox{
        \tiny
        $
        \begin{aligned}
          \omega_4 & \mapsto \omega_4
          \\
          \omega_7 & \mapsto \omega_7
        \end{aligned}
        $}
      }
      &&
      \mathrm{CE}(\mathfrak{l}(S^3))
      \ar[d]^{
        h_3 \mapsto h_3
       }
      \\
      &
      \mathrm{CE}(\mathfrak{l}(S^4))
      \ar[rr]|{
         \mbox{
        \tiny
        $
        \begin{aligned}
          \omega_4 & \mapsto 0
          \\
          \omega_7 & \mapsto 0
        \end{aligned}
        $}
      }
      \ar[ddl]|{
        \mbox{
        \tiny
        $
        \begin{aligned}
          \omega_4 & \mapsto 0
          \\
          \omega_7 & \mapsto 0
        \end{aligned}
        $}
      }
      \ar[dddd]|{ \phantom{ {A \atop A} \atop {A \atop A} } }
      &&
      \mathrm{CE}(\mathfrak{l}(S^3))
      \ar[ddl]|{ h_3 \mapsto 0 }
      \ar[dddd]
      \\
      \\
      \mathrm{CE}\left(\mathfrak{l}\left(S^2\right)\right)
      \ar[rr]|{
        \mbox{
        \tiny
        $
        \begin{aligned}
          \omega_2 & \mapsto 0
          \\
          \omega_3 & \mapsto 0
        \end{aligned}
        $
        }
      }
      \ar[ddr]|{
        \mbox{
        \tiny
        $
        \begin{aligned}
          \omega_2 & \mapsto 0
          \\
          \omega_3 & \mapsto 0
        \end{aligned}
        $
        }
      }
      &&
      \mathrm{CE}\left(\mathfrak{l}\left(S^1\right)\right)
      \ar[ddr]|{ h_1 \mapsto (0,0) }
      \\
      \\
      &
      \mathrm{CE}\left(\mathfrak{l}\left(S^1\right)\right)
      \ar[rr]|{ h_1 \mapsto (0,0) }
      &&
      \mathrm{CE}\left(\mathfrak{l}\left(S^0\right)\right)
    }
    }
  $$

\item  {\bf (ii)}
  For the $(G_{\mathrm{ADE}})_{L} \times (G_{\mathrm{ADE}})_{R}$-action, the system of minimal dgc-algebras
  $$
      \mathrm{CE}\Big(
      \mathfrak{l}\big(
        ( S^4)^{(-)}
      \big)
    \Big)
    \maps
    \mathrm{Orb}_{ (G_{\mathrm{ADE}})_{L} \times (G_{\mathrm{ADE}})_{R} }
    \xymatrix{\ar[r]&}
    \mathrm{dgcAlg}
  $$
  is, in the parts corresponding to those shown in Lemma \ref{GADEGWFixedInFourSphere}, given by
  $$
    \raisebox{45pt}{
    \xymatrix@C=-20pt@R-9pt{
      &
      \frac{
        (G_{\mathrm{ADE}})_{L} \times (G_{\mathrm{ADE}})_{R}
      }{
        \{e_{\mathrm{ADE}}\} \times \{e_{\mathrm{ADE}}\}
      }
      \ar[rr]
      \ar[d]_{ (g_{\mathrm{ADE}}, g'_{\mathrm{ADE}}) }
      &&
      \frac{
        (G_{\mathrm{ADE}})_{L} \times (G_{\mathrm{ADE}})_{R}
      }{
        \{e_{\mathrm{ADE}}\} \times (G_{\mathrm{ADE}})_{R}
      }
      \ar[d]_{ g_{\mathrm{ADE}} }
      \\
      &
      \frac{
        (G_{\mathrm{ADE}})_{L} \times (G_{\mathrm{ADE}})_{R}
      }{
        \{e_{\mathrm{ADE}}\} \times \{e_{\mathrm{ADE}}\}
      }
      \ar[dr]
      \ar[rr]
      \ar[dddd]
      &&
      \frac{
        (G_{\mathrm{ADE}})_{L} \times (G_{\mathrm{ADE}})_{R}
      }{
        \{e_{\mathrm{ADE}}\} \times (G_{\mathrm{ADE}})_{R}
      }
      \ar[dddd]
      \\
      &&
      \frac{
        (G_{\mathrm{ADE}})_{L} \times (G_{\mathrm{ADE}})_{R}
      }{
        (\mathbb{Z}_{n+1})_\Delta
      }
      \ar[dd]
      \\
      &&
      \\
      &&
      \frac{
        (G_{\mathrm{ADE}})_{L} \times (G_{\mathrm{ADE}})_{R}
      }{
        G_\Delta
      }
      \ar[dr]
      \\
      &
      \frac{
        (G_{\mathrm{ADE}})_{\Delta} \times G_{\mathrm{HW}}
      }{
        (G_{\mathrm{ADE}})_L \times \{ e_{\mathrm{ADE}} \}
      }
      \ar[rr]
      &&
      \frac{
        (G_{\mathrm{ADE}})_{L} \times (G_{\mathrm{ADE}})_{R}
      }{
        (G_{\mathrm{ADE}})_{L} \times (G_{\mathrm{ADE}})_{R}
      }
    }
    }
   \xymatrix{\hspace{.3cm}\ar@{|->}[r]&\hspace{0cm}}
    \raisebox{45pt}{
    \xymatrix@C=18pt@R19pt{
      &
      \mathrm{CE}(\mathfrak{l}(S^4))
      \ar[rr]^-{
        \mbox{
          \tiny
          $
          \begin{aligned}
            \omega_4 & \mapsto 0
            \\
            \omega_7 & \mapsto 0
          \end{aligned}
          $
        }
      }
      \ar[d]_-{
       \mbox{
          \tiny
          $
          \begin{aligned}
            \omega_4 & \mapsto + \omega_4
            \\
            \omega_7 & \mapsto + \omega_7
          \end{aligned}
         $
        }
      }
      &&
      \mathrm{CE}(\mathfrak{l}(S^0))
      \ar@{=}[d]
      \\
      &
      \mathrm{CE}(\mathfrak{l}(S^4))
      \ar[dr]|-{
        \mbox{
          \tiny
          $
          \begin{aligned}
            \omega_4 & \mapsto 0
            \\
            \vspace{-1mm}
            \omega_7 & \mapsto 0
          \end{aligned}
          $
      }
      }
      \ar[rr]|-{
        \mbox{
          \tiny
          $
          \begin{aligned}
            \omega_4 & \mapsto 0
            \\
            \omega_7 & \mapsto 0
          \end{aligned}
          $
        }
      }
      \ar[dddd]|-{
        \mbox{
          \tiny
          $
          \begin{aligned}
            \omega_4 & \mapsto 0
            \\
            \omega_7 & \mapsto 0
          \end{aligned}
          $
        }
      }
      &&
      \mathrm{CE}(\mathfrak{l}(S^0))
      \ar@{=}[dddd]
      \\
      &&
      \mathrm{CE}(\mathfrak{l}(S^2))
      \ar[dd]|-{
        \mbox{
          \tiny
          $
          \begin{aligned}
            \omega_2 & \mapsto 0
            \\
            \omega_3 & \mapsto 0
          \end{aligned}
          $
        }
      }
      \\
      &&
      \\
      &&
      \mathrm{CE}(\mathfrak{l}(S^1))
      \ar[dr]|{ h_1 \mapsto 0 }
      \\
      &
      \mathrm{CE}(\mathfrak{l}(S^0))
      \ar@{=}[rr]
      &&
      \mathrm{CE}(\mathfrak{l}(S^0))
    }
    }
  $$
\end{prop}
\begin{proof}
  In both cases the vertical maps at the top are obtained by using Lemma \ref{ADEActionOn4SphereRationallyTrivial}
  in Lemma \ref{GADEGWFixedInFourSphere}. That all the other maps appearing in Lemma \ref{GADEGWFixedInFourSphere}
  are represented by zero-maps on the minimal dgc-algebra models reflects the basic fact that
  every map from a sphere of lower dimension into one of higher dimensions is null homotopic.
\end{proof}

Prop. \ref{ADEWEquivariant4SpheredgAlgebraicCoefficients} serves to determine the explicit coefficients of
real ADE-equivariant rational cohomotopy in degree 4. We will make repeated use of this in the proofs in Sec.
\ref{ADEEquivariantMBraneSuperCucycles}.

\newpage

\section{Real ADE-equivariant brane cocycles}
\label{ADEEquivariantMBraneSuperCucycles}

In the previous subsection we discussed real ADE-actions both
on $\mathbb{R}^{10,1\vert \mathbf{32}}$ (Sec. \ref{ADESingularitiesInSuperSpacetime}) and on $S^4$ (Sec. \ref{ADEEquivariantRationalCohomotopy}).
This now allows us to discuss the possible equivariant enhancements (Example \ref{EquivariantEnhancementOfSuperCocycles}) of the M2/M5-brane cocycle (Prop. \ref{M2M5SuperCocycle})
that are compatible with these actions.

\begin{theorem}[Equivariant enhancements of the fundamental brane cocycles]
  \label{RealADEEquivariantEndhancementOfM2M5Cocycle}

With respect to the real ADE-actions \eqref{SpacetimeActions} on $D = 11$, $\mathcal{N} = 1$ super spacetime from Theorem \ref{SuperADESingularitiesIn11dSuperSpacetime}
and the real ADE-actions \eqref{ActionsOnThe4SphereInRPlusH} on the 4-sphere from Def.\ \ref{SuspendedHopfAction}, we have
non-trivial cocycles in real ADE-equivariant rational cohomotopy of superspaces as shown in
 \hyperlink{EquivariantEnhancementsList}{\bf Table 3},
providing equivariant enhancement (Example \ref{EquivariantEnhancementOfSuperCocycles})
of the fundamental M2/M5-cocycle (Prop. \ref{M2M5SuperCocycle}) as well as of the zero-cocycle on super spacetime.

\end{theorem}

\begin{proof}
We need to produce data as discussed in Example \ref{EquivariantEnhancementOfSuperCocycles}.
First of all this means to check that the fundamental brane cocycles at each stage (from Example \ref{FundamentalF1Cocycle} and Prop. \ref{M2M5SuperCocycle}) are plain equivariant in the first place, hence that, in the notation of Remark \ref{EquivarianceCommutingDiagram}, we have

\vspace{-.3cm}

$$
  \xymatrix{
    \mathbb{R}^{10,1\vert \mathbf{32}}
    \ar@(ul,ur)[]^G
    \ar[rr]^-{\mu_{{}_{M2/M5}}}
    &&
    S^4
    \ar@(ul,ur)[]^G
  },
  \phantom{AAAAA}
    \xymatrix{
      \mathbb{R}^{9,1\vert \mathbf{16}}
      \ar@(ul,ur)[]^{(\mathbb{Z}_{n+1})_{\Delta}}
      \ar[rr]^{\mu_{{}_{F1}}^{H/I}}
      &&
      S^3
      \ar@(ul,ur)[]^{(\mathbb{Z}_{n+1})_{\Delta}}
    }
   \!\!\!\!\! ,
  \phantom{AAAAA}
    \xymatrix{
      \mathbb{R}^{2,1\vert \mathbf{2}}
      \ar@(ul,ur)[]^{ G_{\mathrm{HW}} }
      \ar[rr]^{\mu_{{}_{F1}}^{D=3}}
      &&
      S^2
      \ar@(ul,ur)[]^{ G_{\mathrm{HW}} }
    }
   \!\!\! .
$$

\vspace{-.0cm}

We check this in Prop. \ref{PlainEquivarianceM2M5}, Prop. \ref{PlainEquivarianceF1}, and Prop. \ref{PlainEquivarianceF1OnM2} below.

Now first consider the case of simple singularities (Def.\ \ref{SimpleSingularities}),
where the system of fixed loci in superspacetime is constant away from the trivial subgroup (Example \ref{ConstantSystemsOfFixedSubspaces}).
By Prop. \ref{ADEWEquivariant4SpheredgAlgebraicCoefficients}, in these cases also the system of coefficients for real ADE-equivariant
rational cohomotopy is constant away from the trivial subgroup; and hence a cocycle is entirely determined by the single component
corresponding to unique morphism between the two extreme cases of $G$-orbits:

\vspace{-.1cm}

 \begin{equation}
   \label{EquivariantCocycleGlobalComponent}
  \raisebox{20pt}{
  \xymatrix@R=15pt@C=2em{
    G/\{e\}
    \ar[d]
    \ar@(ul,ur)[]^{G}
    &
    \mathbb{R}^{10,1\vert \mathbf{32}}
    \ar@(ul,ur)[]^{G}
    \ar[rrr]^-{ \mu_{\{e\}} }_<<<<<<{\ }="s"
    &&\;\;\;&
    S^4
    \ar@(ul,ur)[]^{G}
    \\
    G/G
    &
    \left( \mathbb{R}^{10,1\vert \mathbf{32}} \right)^G
    \ar@{^{(}->}[u]
    \ar[rrr]_{  \mu_{G}  }^>>>>>>{\ }="t"
    &&&
    \left( S^4 \right)^G.
    \ar[u]
    &
    \ar@{=>}^\eta "s"; "t"
  }}
\end{equation}
Prop. \ref{TrivializationsOfRestrictionsOfM2M5Cocycle} and Example
\ref{ComponentsForEquivariantEnhancementOfFundamentalBraneCocycles} below establish the
 possible data describing such diagrams. In particular, it shows that the choice of homotopy
 $\eta$ may depend in each case only on the bosonic fixed locus. This means that all cocycles
 on simple singularities are fixed by their component \eqref{EquivariantCocycleGlobalComponent}.

Analogously, for non-simple singularities (Def.\ \ref{SimpleSingularities}) a cocycle involves this kind of
data at each orbit type where the bosonic singular locus changes (see again Example \ref{ConstantSystemsOfFixedSubspaces}).
For intersections of two simple singularities (Def.\ \ref{IntersectionOfSimpleSingularities}) this means that
the cocycle is now determined by its components on four morphisms in the orbit category, as shown here:

\vspace{-3mm}
\hspace{-.8cm}
  \scalebox{.9}{
  $
  \xymatrix@C=-3.2pt@R=4pt{
    &&&&
    &&& &&&&&&&&&
    S^4
    \ar@(ul,ur)[]^{ G_1 \times G_2 }
    \\
    \\
    &&  &&
    &&&
    \mathbb{R}^{10,1\vert \mathbf{32}}
    \ar@(ul,ur)[]^{ G_1 \times G_2 }
    \ar[uurrrrrrrrr]^{ \color{blue} \mu_{ {}_{\{e_1\} \times \{e_2\}} }  }_<<<<<<<<<<<<{\ }="s1"_<<<<<<<<<<{\ }="s2"
    \\
    &
    \frac{G_1 \times G_2}{ \{e_1\} \times \{e_2\} }
    \ar[ddl]
    \ar[ddr]
    &&&
    &&&&&&&&& \left( S^4\right)^{G_1 \times \{e_2\}}
    \ar@{..>}[uuurrr]
    &&& &&&
    \left( S^4 \right)^{ \{e_1\} \times G_2 }
    \ar[uuulll]
    \\
    \\
    \frac{G_1 \times G_2}{ G_1 \times \{e_2\} }
    \ar[dr]
    &&
    \frac{G_1 \times G_2}{ \{e_1\} \times G_2 }
    \ar[dl]
    &&
    \left( \mathbb{R}^{10,1\vert \mathbf{32}} \right)^{G_1 \times \{e_2\}}
    \ar@{^{(}->}[uuurrr]
    \ar@{..>}[uurrrrrrrrr]|<<<<<<<<{ \color{blue} \mu_{{}_{ G_1 \times \{e_2\} }} }^>>>>>>>>>>>>>>>>>>>>>>>>>>>>>>>>>>>>>>{\ }="t2"_<<<<<<<<<<<<<<<<<<<{\ }="s4"
    &&& &&&
    \left( \mathbb{R}^{10,1\vert \mathbf{32}} \right)^{\{e_1\} \times G_2}
    \ar@{_{(}->}[uuulll]
    \ar[uurrrrrrrrr]|>>>>>>>>>>>>>{\color{blue} \mu_{{}_{ \{e_1\} \times  G_2}}  }^>>>>>>>>>>>>>>>>>>>>>>>>>>>>>>>>>>>>{\ }="t1"_<<<<<<<<<{\ }="s3"
    \\
    & \frac{G_1 \times G_2}{ G_1 \times G_2 } &&&
    &&&   &&&&&&&&& \left( S^4 \right)^{ G_1 \times G_2 }
    \ar[uuurrr]
    \ar@{..>}[uuulll]
    \\
    \\
    &&&&
    &&&
    \left( \mathbb{R}^{10,1\vert \mathbf{32}} \right)^{G_1 \times G_2}
    \ar@{^{(}->}[uuurrr]
    \ar@{_{(}->}[uuulll]
    \ar[uurrrrrrrrr]|{  }^>>>>>>>>>>>>>>>>>>>>>>>>>>>>>>>>>>>{\ }="t3"^<<<<<<<<<<<<<<{\ }="t4"
    \ar@{=>}|<<<<<<<<<<<{ \color{cyan} \phantom{{A \atop A}} \eta^{ \{e_1\} \times \{e_1\} }_{ \{e_1\} \times G_2 }   \phantom{{A \atop A}} } "s1"; "t1"
    \ar@{::>}|<<<<<{ \color{cyan} \phantom{{A \atop A}} \eta^{ \{e_1\} \times \{e_2\} }_{ G_1 \times \{e_2\} }  \phantom{{A \atop A}} } "s2"; "t2"
    \ar@{=>}|{ \color{cyan} \phantom{{A \atop A}} \eta^{ \{e_1\} \times G_2 }_{ G_1 \times G_2 } } "s3"; "t3"
    \ar@{::>}|<<<<<{ \color{cyan} \phantom{{A \atop A}} \eta^{ G_1 \times \{e_2\} }_{G_1 \times G_2} } "s4"; "t4"
  }
  $
  }
The data assigned to each square is determined by Example
\ref{ComponentsForEquivariantEnhancementOfFundamentalBraneCocycles},
which shows that the pasting composition of the
consecutive homotopies shown vanishes. This implies the 2-commutativity of the diagram.
\end{proof}

\newpage

\hypertarget{EquivariantEnhancementsList}{$\,$}

\vspace{-1.5cm}

\hspace{-1.3cm}
\scalebox{.8}{
\begin{tabularx}{1.315\textwidth}{|X|c|c|c|}
  \hline
  \begin{tabular}{c}
    \hspace{-5mm} \bf Black brane
    \\
   \hspace{-5mm} \bf species
  \end{tabular}
  &
  \begin{tabular}{c}
    \bf Type of
    \\
    \bf singularity
    \\
    \bf in $\mathbb{R}^{10,1}$
  \end{tabular}
  &
  \multicolumn{2}{|c|}{
      \bf Systems of {\color{blue} fundamental brane species} at singular loci
  }
  \\
  \hline
  \hline
  {\rm M2}
  &
  $(G_{\mathrm{ADE}})_{\Delta}$
  &
  \raisebox{20pt}{
  $
  \xymatrix@R=2pt{
    \mathrm{ST}
  &
    \mathbb{R}^{10,1\vert \mathbf{32}}
    \ar@(ul,ur)[]^{ (G_{\mathrm{ADE}})_{\Delta} }
    \ar[rrrr]^{ \color{blue} \mu_{{}_{M2/M5}}}_<<<<{\ }="s"
    &&&&
    S^4
    \ar@(ul,ur)[]^{ ( G_{\mathrm{ADE}} )_{\Delta} }
    \\
    \\
    \\
    \\
    \mathrm{M2}
    &
    \mathbb{R}^{2,1\vert 8 \cdot \mathbf{2}}
    \ar@{^{(}->}[uuuu]
    \ar[rrrr]_{0}^>>>>>>>{\ }="t"
    &&&&
    S^1
    \ar[uuuu]_0
    \ar@{=>}|{ \color{cyan} \phantom{{A \atop A} \atop {A \atop A}} \mathrm{svol}_{2+1}
     \phantom{{A \atop A} \atop {A \atop A}} } "s"; "t"
  }
  $
  }
  &
  \raisebox{20pt}{
  $
  \xymatrix@R=2pt{
    \mathrm{ST}
    &
    \mathbb{R}^{10,1\vert \mathbf{32}}
    \ar@(ul,ur)[]^{ (\mathbb{Z}_{n+1})_{\Delta} }
    \ar[rrrr]^{ \color{blue} \mu_{{}_{M2/M5}}  }_<<<<{\ }="s"
    &&&&
    S^4
    \ar@(ul,ur)[]^{ ( \mathbb{Z}_{n+1} )_{\Delta} }
    \\
    \\
    \\
    \\
    \mathrm{M2}
    &
    \mathbb{R}^{2,1\vert 8 \cdot \mathbf{2}}
    \ar@{^{(}->}[uuuu]
    \ar[rrrr]_{ \color{blue} \mu_{{}_{F1}}^{D=3}  }^>>>>>>>{\ }="t"
    &&&&
    S^2
    \ar[uuuu]_0
    \ar@{=>}|{ \color{cyan} \phantom{{A \atop A} \atop {A \atop A}} \mathrm{svol}_{2+1}
    \phantom{{A \atop A} \atop {A \atop A}} } "s"; "t"
  }
  $
  }
  \\
  \hline
  {\rm MO5}
  &
  $(\mathbb{Z}_2)_{\mathrm{W}}$
  &
  \raisebox{20pt}{
  $
  \xymatrix@R=2pt{
    \mathrm{ST}
    &
    \mathbb{R}^{10,1\vert \mathbf{32}}
    \ar@(ul,ur)[]^{ G_{\mathrm{W}} }
    \ar[rrrr]^{ \color{blue} \mu_{{}_{M2/M5}} }_<<<<{\ }="s"
    &&&&
    S^4
    \ar@(ul,ur)[]^{ G_{\mathrm{W}} }
    \\
    \\
    \\
    \\
    \mathrm{M5}
    &
    \mathbb{R}^{5,1\vert 2 \cdot \mathbf{8}}
    \ar@{^{(}->}[uuuu]
    \ar[rrrr]_-{ \color{blue}  0  }^>>>>>>>{\ }="t"
    &&&&
    S^1
    \ar[uuuu]_0
    \ar@{=>}|{ \color{cyan} \phantom{{A \atop A} \atop {A \atop A}} \mathrm{svol}_{5,1}
     \phantom{{A \atop A} \atop {A \atop A}} } "s"; "t"
  }
  $
  }
  &
  \\
  \hline
  $\mathrm{MO9}$
  &
  $(\mathbb{Z}_2)_{\mathrm{HW}}$
  &
  &
  \raisebox{20pt}{
  $
  \xymatrix@R=2pt{
    \mathrm{ST}
    &
    \mathbb{R}^{10,1\vert \mathbf{32}}
    \ar@(ul,ur)[]^{ G_{\mathrm{HW}} }
    \ar[rrrr]^{ \color{blue}  0 }_<<<<{\ }="s"
    &&&&
    S^4
    \ar@(ul,ur)[]^{ G_{\mathrm{HW}} }
    \\
    \\
    \\
    \\
    \mathrm{MO9}
    &
    \mathbb{R}^{9,1\vert \mathbf{16}}
    \ar@{^{(}->}[uuuu]
    \ar[rrrr]_{ \color{blue} \mu_{{}_{\mathrm{F1}}}^{H/I}  }^>>>>>>>{\ }="t"
    &&&&
    S^3
    \ar[uuuu]_0
    \ar@{=>}|{ \color{cyan} \phantom{A \atop A} 0  \phantom{A \atop A} } "s"; "t"
  }
  $
  }
  \\
  \hline
  $\mathrm{MO1}$
  &
  $(\mathbb{Z}_2)_{\mathrm{MW}}$
  &
  \raisebox{20pt}{
  $
  \xymatrix@R=2pt{
    \mathrm{ST}
    &
    \mathbb{R}^{10,1\vert \mathbf{32}}
    \ar@(ul,ur)[]^{ G_{\mathrm{MW}} }
    \ar[rrrr]^{ \color{blue} \mu_{{}_{M2/M5}} }_<<<<{\ }="s"
    &&&&
    S^4
    \ar@(ul,ur)[]^{ G_{\mathrm{MW}} }
    \\
    \\
    \\
    \\
    \mathrm{MW}
    &
    \mathbb{R}^{1,1\vert 16 \cdot \mathbf{1}}
    \ar@{^{(}->}[uuuu]
    \ar[rrrr]_-{ \color{blue}  0 }^>>>>>>>{\ }="t"
    &&&&
    S^1
    \ar[uuuu]_0
    \ar@{=>}|{
      \color{cyan}
      \phantom{{A \atop A} \atop {A \atop A}}
      0
      \phantom{{A \atop A} \atop {A \atop A}}
    } "s"; "t"
  }
  $
  }
  &
  \\
  \hline
  $\mathrm{MK6}$
  &
  $(G_{\mathrm{ADE}})_R$
  &&
  \raisebox{20pt}{
  $
  \xymatrix@R=2pt{
    \mathrm{ST}
    &
    \mathbb{R}^{10,1\vert \mathbf{32}}
    \ar@(ul,ur)[]^{ (G_{\mathrm{ADE}})_{R} }
    \ar[rrrr]^{ \color{blue} 0 }_<<<<{\ }="s"
    &&&&
    S^4
    \ar@(ul,ur)[]^{ (G_{\mathrm{ADE}})_R }
    \\
    \\
    \\
    \\
    \mathrm{MK6}
    &
    \mathbb{R}^{6,1\vert \mathbf{16} }
    \ar@{^{(}->}[uuuu]
    \ar[rrrr]_{ \color{blue} \mu_{{}_{-2}} }^>>>>>>>{\ }="t"
    &&&&
    S^0
    \ar[uuuu]_{ 0 }
    \ar@{=>}|{ \color{cyan} \phantom{A \atop A} 0 \phantom{ A \atop A} } "s"; "t"
  }
 $
 }
 \\
 \hline
  \begin{tabular}{ccc}
   \hspace{-4mm} $\mathrm{M2}$ & \hspace{-5mm}$\dashv$& \hspace{-5mm} $\mathrm{MO5}$
    \\
    & \hspace{-3mm}$\mathrm{M1}$
  \end{tabular}
  &
  $(\mathbb{Z}_{n+1})_\Delta  \hspace{-1mm}\times \hspace{-1mm}(\mathbb{Z}_2)_{\mathrm{W}}$
  &
  \multicolumn{2}{|c|}{
  \hspace{-40pt}
  \raisebox{57pt}{
  $
  \xymatrix@C=6pt@R=6pt{
    &&&&
    &&& &&&&&&&&&
    S^4
    \ar@(ul,ur)[]^{ (\mathbb{Z}_{n+1})_{\Delta} \times G_{\mathrm{W}} }
    \\
    \\
    &&  &&
    &&&
    \mathbb{R}^{10,1\vert \mathbf{32}}
    \ar@(ul,ur)[]^{ ( \mathbb{Z}_{n+1} )_{\Delta} \times G_{\mathrm{W}} }
    \ar[uurrrrrrrrr]^{ \color{blue} \mu_{{}_{M2/M5}}}_<<<<<<<<<<<<{\ }="s1"_<<<<<<<<<<{\ }="s2"
    \\
    &
    \mathrm{ST}
    &&&
    &&&&&&&&& S^2
    \ar@{..>}[uuurrr]|0
    &&& &&&
    S^1
    \ar[uuulll]|0
    \\
    \\
    \mathrm{M2} &&\mathrm{M5}&&
    \mathbb{R}^{2,1\vert 8 \cdot \mathbf{2}}
    \ar@{^{(}->}[uuurrr]
    \ar@{..>}[uurrrrrrrrr]|<<<<<<<<<<<<{ \color{blue} \mu_{{}_{F1}}^{D=3}}^>>>>>>>>>>>>>>>>>>>>>>>>>>>>>>>>>>>>>>{\ }="t2"_<<<<<<<<<<<<<<<<<<<{\ }="s4"
    &&& &&&
    \mathbb{R}^{5,1\vert 2 \cdot \mathbf{8}}
    \ar@{_{(}->}[uuulll]
    \ar[uurrrrrrrrr]|>>>>>>>>>>>>>{\color{blue} 0 }^>>>>>>>>>>>>>>>>>>>>>>>>>>>>>>>>>>>>{\ }="t1"_<<<<<<<<<{\ }="s3"
    \\
    & \mathrm{M1} &&&
    &&&   &&&&&&&&& S^0
    \ar[uuurrr]|0
    \ar@{..>}[uuulll]|0
    \\
    \\
    &&&&
    &&&
    {\mathbb{R}^{1,1\vert 8 \cdot \mathbf{1}}}
    \ar@{^{(}->}[uuurrr]
    \ar@{_{(}->}[uuulll]
    \ar[uurrrrrrrrr]|>>>>>>>>>>>>{ \color{blue} 0  }^>>>>>>>>>>>>>>>>>>>>>>>>>>>>>>>>>>>{\ }="t3"^<<<<<<<<<<<<<<{\ }="t4"
    \ar@{=>}|<<<<<<<<<<<{ \color{cyan} \phantom{{A \atop A}} \mathrm{svol}_{5 + 1} } "s1"; "t1"
    \ar@{::>}|<<<<<{ \color{cyan} \phantom{{A \atop A}} \mathrm{svol}_{2 + 1} } "s2"; "t2"
    \ar@{=>}|{ \color{cyan} \phantom{{A \atop A}} 0  \phantom{{A \atop A}} } "s3"; "t3"
    \ar@{::>}|<<<<<{ \color{cyan} \phantom{{A \atop A}} \mathrm{svol}_{1+1} } "s4"; "t4"
  }
  $
  }
  }
  \\
 \hline
  \begin{tabular}{ccc}
    \hspace{-4mm} $\mathrm{M2}$ & \hspace{-6mm}$\dashv$& \hspace{-7mm}$\mathrm{MO9}$
    \\
    & \hspace{-4mm}$\mathrm{NS1}_H$
  \end{tabular}
  &
  $(\mathbb{Z}_{n+1})_\Delta \hspace{-1mm}\times \hspace{-1mm} (\mathbb{Z}_2)_{\mathrm{HW}}$
  &
 \multicolumn{2}{|c|}{
  \hspace{-40pt}
  \raisebox{57pt}{
  $
  \xymatrix@C=6pt@R=6pt{
    &&&&
    &&& &&&&&&&&&
    S^4
    \ar@(ul,ur)[]^{ (\mathbb{Z}_{n+1})_{\Delta} \times G_{\mathrm{HW}} }
    \\
    \\
    &&  &&
    &&&
    \mathbb{R}^{10,1\vert \mathbf{32}}
    \ar@(ul,ur)[]^{ ( \mathbb{Z}_{n+1} )_{\Delta} \times G_{\mathrm{HW}} }
    \ar[uurrrrrrrrr]^{ \color{blue} 0  }_<<<<<<<<<<<<{\ }="s1"_<<<<<<<<<<{\ }="s2"
    \\
    &
    \mathrm{ST}
    &&&
    &&&&&&&&& S^2
    \ar@{..>}[uuurrr]|0
    &&& &&&
    S^3
    \ar[uuulll]|0
    \\
    \\
    \mathrm{M2} &&\mathrm{MO9}&&
    \mathbb{R}^{2,1\vert 8 \cdot \mathbf{2}}
    \ar@{^{(}->}[uuurrr]
    \ar@{..>}[uurrrrrrrrr]|<<<<<<<<{ \color{blue} \mu_{{}_{F1}}^{D=3} }^>>>>>>>>>>>>>>>>>>>>>>>>>>>>>>>>>>>>>>{\ }="t2"_<<<<<<<<<<<<<<<<<<<{\ }="s4"
    &&& &&&
    \mathbb{R}^{9,1 \vert \mathbf{16}}
    \ar@{_{(}->}[uuulll]
    \ar[uurrrrrrrrr]|>>>>>>>>>>>>>{\color{blue} \mu^{H/I}_{{}_{F1}}}^>>>>>>>>>>>>>>>>>>>>>>>>>>>>>>>>>>>>{\ }="t1"_<<<<<<<<<{\ }="s3"
    \\
    & \mathrm{NS1}_H &&&
    &&&   &&&&&&&&& S^1
    \ar[uuurrr]|0
    \ar@{..>}[uuulll]|0
    \\
    \\
    &&&&
    &&&
    {\mathbb{R}^{1,1\vert 8 \cdot \mathbf{1}}}
    \ar@{^{(}->}[uuurrr]
    \ar@{_{(}->}[uuulll]
    \ar[uurrrrrrrrr]|{0}^>>>>>>>>>>>>>>>>>>>>>>>>>>>>>>>>>>>{\ }="t3"^<<<<<<<<<<<<<<{\ }="t4"
    \ar@{=>}|<<<<<<<<<<<{ \color{cyan} \phantom{{A \atop A}} 0  \phantom{{A \atop A}} } "s1"; "t1"
    \ar@{::>}|<<<<<{ \color{cyan} \phantom{{A \atop A}} 0  \phantom{{A \atop A}} } "s2"; "t2"
    \ar@{=>}|{ \color{cyan} \phantom{{A \atop A}} \mathrm{svol}_{1+1} } "s3"; "t3"
    \ar@{::>}|<<<<<{ \color{cyan} \phantom{{A \atop A}} \mathrm{svol}_{1+1} } "s4"; "t4"
  }
  $
  }
  }
  \\
  \hline
  \end{tabularx}
}

\vspace{.1cm}
{\footnotesize
\noindent {\bf {Table 3}. Cocycles in equivariant rational cohomotopy of $D = 11$, $\mathcal{N} = 1$ super spacetime},
according to Theorem \ref{RealADEEquivariantEndhancementOfM2M5Cocycle}. The singularities (black branes) are
from Thm. \ref{SuperADESingularitiesIn11dSuperSpacetime} and Prop. \ref{NonSimpleRealSingularities}; the component
cocycles $\mu_{\cdots}$ (fundamental branes) are from the old brane scan (Prop. \ref{TheOldBraneScan});
the homotopies $\mathrm{svol}_{p+1}$ are their
Green--Schwarz action functionals (Prop. \ref{TrivializationsOfRestrictionsOfM2M5Cocycle}).
}

\newpage

\subsection{ Plain equivariance of the fundamental brane cocycles }

Part of the data of an equivariant enhancement of a supercocycle is the \emph{property} that the cocycle map be
plain equivariant, hence that it intertwines the group actions on both sides. Here we check that this is the case
for the cocycles appearing in Theorem \ref{RealADEEquivariantEndhancementOfM2M5Cocycle}.

\begin{lemma}[Real ADE-action on M-brane super-cocycles]
  \label{M2CocycleIsADEInvariant}
  With respect to the actions from Theorem \ref{SuperADESingularitiesIn11dSuperSpacetime},
  the M2-cocycle (Def.\ \ref{M2M5SuperCochains}) is \emph{invariant} under all $g \in G_{\mathrm{ADE}}$ and changes sign under the non-trivial element $ \sigma \in G_{\mathrm{W}}$,
  while the $M5$-cochain is invariant under both:
  \begin{equation}
    \label{TransformationOfM2M5Cochains}
  \mbox{
  \begin{tabular}{|c||c|c|}
    \hline
    & $g \in G_{\mathrm{ADE}}$ & $e \neq \sigma \in G_{\mathrm{W}}$
    \\
    \hline
    $\mu$ & $\rho(g)^\ast(\mu)$ & $\rho(\sigma)^\ast(\mu)$
    \\
    \hline
    \hline
    $\mu_{{}_{M2}}$ & $\phantom{-}\mu_{{}_{M2}}$ & $-\mu_{{}_{M2}}$
    \\
    \hline
    $\mu_{{}_{M5}}$ & $\phantom{-}\mu_{{}_{M5}}$ & $\phantom{-}\mu_{{}_{M5}}$
    \\
    \hline
  \end{tabular}
  }
  \end{equation}
\end{lemma}
\begin{proof}
  The following is the argument using the Dirac representation (Prop. \ref{RealSpinorsViaMajoranaConditionsOnDiracRepresentations}).
  For the reflection operation we have from \eqref{CMinusMajoranaRepresentationForPin} that
  $$
    \rho(\sigma)^\ast(\psi) = i \Gamma_{10} \psi\;,
    \phantom{AAAA}
    \rho(\sigma)^\ast(e^a) =
    \left\{
      \begin{array}{cc}
        e^a & \mbox{if $a \neq 10$},
        \\
        -e^a & \mbox{otherwise}.
      \end{array}
    \right.
  $$
  Extending this to an algebra homomorphism yields:
  \begin{equation}
    \label{HWTransformationOnCocycles}
    \begin{aligned}
      \rho(\sigma)^\ast(\mu_{{}_{M2}})
      & =
      \rho(\sigma)^\ast( \tfrac{i}{2} \overline{\psi} \Gamma_{a_1 a_2} \psi \wedge e^{a_1} \wedge e^{a_2} )
      \\
      & =
      \tfrac{i}{2}
      \underset{a_1, a_2 \neq 10}{\sum}
      (\overline{\Gamma_{10} \psi}) \Gamma_{a_1 a_2} (\Gamma_{10} \psi) \wedge e^{a_1} \wedge e^{a_2} + \cdots
      \\
      & =
      \tfrac{i}{2}
      \underset{a_1, a_2 \neq 10}{\sum}
      \psi^\dagger \underset{ = - \Gamma_{10}}{\underbrace{(\Gamma_{10})^\dagger}}
       \Gamma_0 \Gamma_{a_1 a_2} \Gamma_{10} \psi  \wedge e^{a_1} \wedge e^{a_2} + \cdots
      \\
      & =
      (-1)
      \tfrac{i}{2}
      \underset{a_1, a_2 \neq 10}{\sum}
      \psi^\dagger \Gamma_0 \Gamma_{a_1 a_2} \underset{ = 1 }{\underbrace{(-\Gamma_{10})\Gamma_{10}}}
      \psi  \wedge e^{a_1} \wedge e^{a_2} + \cdots
      \\
      & = - \mu_{{}_{M2}}
      \,.
    \end{aligned}
  \end{equation}
  Here the first three equalities just spell out the definition.
  In the second line we decomposed the sum over indices into summands that involve the 10th index and those that do not.
  We display only the first, as the argument for the second is the same except for \emph{two} extra signs that appear, and hence cancel.

  Under the brace in the third line we use (\ref{CliffordAdjointsInDiracRepresentation})
  and then in the next line we commute $(-\Gamma_{10})$ past $\Gamma_0 \Gamma_{a_1} \Gamma_{a_2}$, which picks up a total minus sign due to \eqref{CliffordAdjointsInDiracRepresentation}
  and since $a_1,a_2 \neq 10$ in this first summand, by construction.
  Finally we cancel the product of $-\Gamma_{10}$ with $\Gamma_{10}$, using again \eqref{CliffordAdjointsInDiracRepresentation}. In total this leaves an overall sign.

  Since $g \in G_{\mathrm{ADE}}$ acts by an even number of such spatial reflections, this also implies $\rho(g)^\ast(\mu_{{}_{M2}}) = \mu_{{}_{M2}}$.

  Finally, from this the statements for $\mu_{{}_{M5}}$ follows by the property of respecting the CE-differential, as in (\ref{omega7TransformationFromOmega4}).
\end{proof}

\begin{prop}[Real ADE-equivariance of the joint M2/M5-brane super-cocycle]
  \label{PlainEquivarianceM2M5}
  The combined M2/M5 cocycle $\mathbb{R}^{10,1\vert \mathbf{32}} \xrightarrow{\mu_{{}_{M2/M5}}} S^4$ (Prop. \ref{M2M5SuperCocycle}) is \emph{equivariant} with respect to the real ADE-actions \eqref{SpacetimeActions} on $\mathbb{R}^{10,1\vert \mathbf{32}}$ from Theorem \ref{SuperADESingularitiesIn11dSuperSpacetime} and
  the real ADE actions \eqref{ActionsOnThe4SphereInRPlusH} on $S^4$ from  Def.\ \ref{SuspendedHopfAction}. In the notation of Remark \ref{EquivarianceCommutingDiagram}
  this means that
  $$
    \xymatrix{
      \mathbb{R}^{10,1\vert \mathbf{32}}
      \ar@(ul,ur)[]^{ G_{\mathrm{ADE}} }
      \ar[rr]^{\mu_{{}_{M2/M5}}}
      &&
      S^4\;\;,
      \ar@(ul,ur)[]^{ G_{\mathrm{ADE}} }
    }
  \qquad \quad
    \xymatrix{
      \mathbb{R}^{10,1\vert \mathbf{32}}
     \ \ar@(ul,ur)[]^{ G_{\mathrm{HW}} }
      \ar[rr]^{\mu_{{}_{M2/M5}}}
      &&
      S^4\;\;,
      \ar@(ul,ur)[]^{ G_{\mathrm{HW}} }
    }
    \quad \;
\text{  and}
\; \quad
    \xymatrix{
      \mathbb{R}^{10,1\vert \mathbf{32}}
      \ar@(ul,ur)[]^{ G_{\mathrm{ADE},\mathrm{HW}} }
      \ar[rr]^{\mu_{{}_{M2/M5}}}
      &&
      \mathfrak{l}(S^4)\;.
      \ar@(ul,ur)[]^{ G_{\mathrm{ADE},\mathrm{HW}} }
    }
  $$
\end{prop}
\begin{proof}
  This follows by the combination of Lemma \ref{ADEActionOn4SphereRationallyTrivial} with Lemma \ref{M2CocycleIsADEInvariant}:
  Under the map (\ref{M2M5CocycleMap}),
  the table of transformation properties of generators of $S^4$ (\ref{omega7TransformationFromOmega4})
  turns into the table of transformation properties of the M-brane super-cochains (\ref{TransformationOfM2M5Cochains}).
\end{proof}

\begin{prop}[M2-cocycle vanishes on $\mathrm{MW}$, $\mathrm{M5}$ and $\mathrm{MO9}$]
  \label{M2CocycleVanishesOnMO9}
  The restriction of the M2-brane cocycle
  $
    \mu_{{}_{M2}} \in \mathrm{CE}\left( \mathbb{R}^{10,1\vert \mathbf{32}}  \right)
  $
  along the inclusion $\mathbb{R}^{p,1\vert \mathbf{N}} \hookrightarrow \mathbb{R}^{10,1\vert \mathbf{32}}$
  of the MW, M5, and MO9 fixed subspace from Prop. \ref{ClassificationOfZ2Actions} vanishes identically:
  $$
    \mu_{{}_{M2}}\vert_{{}_{\mathbb{R}^{p,1\vert \mathbf{N}}}}
    \;=\;
    0
    \,.
  $$
\end{prop}
\begin{proof}
  By Lemma \ref{M2CocycleIsADEInvariant} the M2-cocycle changes sign under the corresponding $p$-brane involutions.
  Hence it has to vanish on the fixed space of these involutions.

  For illustration, we now spell this out in more detail for the case of the MO9:
  without loss of generality, we may take the MO9 to be at $x_{10} = 0$ (by Lemma \ref{pBraneInvolutionConjugacy}).
  Consider then the corresponding decomposition of the M2-cocycle in the form into
  $$
    \mu_{{}_{M2}}
    \;=\;
    \Big(
       \;
      \underset{ =: \mu_{{}_{F1}}^H }{\underbrace{i
        \underset{a \leq 9}{\sum} \overline{\psi} \Gamma_{a\, 10} \psi \wedge e^{a}
      }}
 \;    \Big)
    \wedge e^{10}
    \;+\;
    \underset{
      =: \mu_{{}_{D2}}
    }{
    \underbrace{
      \tfrac{i}{2}
      \underset{a_1, a_2 \leq 9}{\sum} \overline{\psi} \Gamma_{a_1 a_2} \psi \wedge e^{a_1} e^{a_{2}}
    }}
    \,.
  $$

  \vspace{-1mm}
\noindent  The restriction of the first summand to $\mathbb{R}^{9,1\vert \mathbf{16}}$ vanishes simply because $e^{10}$ vanishes on
  the underlying $\mathbb{R}^{9,1}$.
  On the other hand, the restriction of the second summand  to $\mathbb{R}^{9,1\vert \mathbf{16}}$ vanishes
  because its fermionic factor vanishes after restriction to $\mathbf{16}$:
  By the proof of Lemma \ref{ExistenceOfpBraneInvolutions}, $\mathbf{16}$ is the $+1$ eigenspace of $\mathbf{\Gamma}_{10}$
  in $\mathbf{32}$. Hence with $\mathbf{\Gamma}_{10} \psi = \psi$ we find that
  $\mu_{{}_{M2}}|_{{}_{\mathbb{R}^{9,1\vert \mathbf{16}}}}$ is equal to minus itself,

  \vspace{-2mm}
  $$
    \begin{aligned}
      \mu_{{}_{M2}}|_{{}_{\mathbb{R}^{9,1\vert \mathbf{16}}}}
      & =
      \tfrac{i}{2}
      \underset{a_1, a_2 \leq 9}{\sum}
      \overline{ \underset{ = \overline{\psi} }{\underbrace{\mathbf{\Gamma}_{10} \psi}} }
      \Gamma_{a_1 } \Gamma_{a_2 }
      \psi \wedge e^{a_1} \wedge e^{a_{2}}
      \\
      & =
      \tfrac{i}{2}
      \underset{a_1, a_2 \leq 9}{\sum}
      \psi^\dagger \mathbf{\Gamma}_{10} \, \Gamma_0\Gamma_{a_1} \Gamma_{a_2} \psi \wedge e^{a_1} \wedge e^{a_{2}}
      \\
      & =
      -
      \tfrac{i}{2}
      \underset{a_1, a_2 \leq 9}{\sum}
      \psi^\dagger \mathbf{\Gamma}_0 \Gamma_{a_1} \Gamma_{a_2}
      \,
      \underset{= + \psi}{\underbrace{\mathbf{\Gamma}_{10} \psi}}
      \wedge e^{a_1} \wedge e^{a_{2}}
      \\
      & =
      - \mu_{{}_{M2}}|_{{}_{\mathbb{R}^{9,1\vert \mathbf{16}}}}\;.
    \end{aligned}
  $$
  Here we spelled out the Dirac conjugate \eqref{DiracConjugate} and then picked up a minus sign from commuting $\mathbf{\Gamma}_{10}$
  with $\Gamma_0 \Gamma_{a_1} \Gamma_{a_2}$.
\end{proof}

Similarly we have the following:

\begin{prop}[ADE-equivariance of the heterotic F1 super-cocycle]
  \label{PlainEquivarianceF1}
  The heterotic/type I string cocycle from Example \ref{FundamentalF1Cocycle} is equivariant with respect to the
  residual $(\mathbb{Z}_{n+1})_\Delta$-actions \eqref{SpacetimeActions} on $\mathbb{R}^{9,1\vert \mathbf{16}}$ from Theorem \ref{SuperADESingularitiesIn11dSuperSpacetime} and
  the action
  \eqref{ActionsOnThe4SphereInRPlusH} on $S^3 = \left( S^4\right)^{ G_{\mathrm{HW}} }$ from  Def.\ \ref{SuspendedHopfAction}. In the notation of Remark \ref{EquivarianceCommutingDiagram}
  this means that
  $$
    \xymatrix{
      \mathbb{R}^{9,1\vert \mathbf{16}}
      \ar@(ul,ur)[]^{(\mathbb{Z}_{n+1})_{\Delta}}
      \ar[rr]^{\mu_{{}_{F1}}^{H/I}}
      &&
      S^3
      \ar@(ul,ur)[]^{(\mathbb{Z}_{n+1})_{\Delta}}
    }\!\!\!\!\!\!\!.
  $$
\end{prop}

\begin{proof}
  The action on $S^3$ is through $\mathrm{SO}(3)$ (see \eqref{ImaginaryQuaternionsRotationAction}),
  hence is manifestly an orientation-preserving isometry.   Therefore the same kind of argument as in the proof of
  Lemma \ref{ADEActionOn4SphereRationallyTrivial} shows that the action on the minimal dgc-model for $S^3$ is trivial.
  By Example \ref{InvarianceAsEquivariance} this means that equivariance in this case means that the cocycle is in fact
  \emph{invariant} under the action on spacetime.
  Indeed, since the spacetime action $(\mathbb{Z}_{n+1})_{\Delta}$ is orientation-preserving, it factors through the canonical
  $\mathrm{Spin}(10,1)$-action, and under this the cocycle is invariant, by its very origin from the old brane scan, Prop. \ref{TheOldBraneScan}.
\end{proof}

\begin{prop}[Reality of the F1 super-cocycle on the M2-brane]
  \label{PlainEquivarianceF1OnM2}
  The fundamental string cocycle in $D = 3$ from Example \ref{FundamentalF1Cocycle} is equivariant with respect to the
  residual $G_{\mathrm{HW}}$-actions \eqref{SpacetimeActions} on $\mathbb{R}^{2,1\vert \mathbf{16}}$ from Theorem \ref{SuperADESingularitiesIn11dSuperSpacetime} and
  the action
  \eqref{ActionsOnThe4SphereInRPlusH} on $S^2 = \left( S^4\right)^{(\mathbb{Z}_{n+1})_{\Delta}}$ from
   Def.\ \ref{SuspendedHopfAction}. In the notation of Remark \ref{EquivarianceCommutingDiagram}
  this means that
  $$
    \xymatrix{
      \mathbb{R}^{2,1\vert \mathbf{2}}
      \ar@(ul,ur)[]^{G_{\mathrm{HW}}}
      \ar[rr]^{\mu_{{}_{F1}}^{D=3}}
      &&
      S^2
      \ar@(ul,ur)[]^{G_{\mathrm{HW}}}
    }
    \!\!\!.
  $$
\end{prop}

\begin{proof}

  The same kind of computation as in \eqref{HWTransformationOnCocycles} shows that $\mu_{{}_{F1}}^{D=3}$ is \emph{invariant}
  under $G_{\mathrm{HW}}$.
  Moreover, the same kind of argument as in Lemma \ref{ADEActionOn4SphereRationallyTrivial}
  leading to \eqref{TransformationOfGeneratorsFor4Sphere} shows that the rational generators of the 2-sphere from Lemma \ref{SpheredgcAlgebraModel}
  transform as
  \begin{equation}
    \label{TransformationOfGeneratorsFor2Sphere}
  \mbox{
  \begin{tabular}{|c||c|c|}
    \hline
    & $g \in G_{\mathrm{ADE}}$ & $e \neq \sigma \in G_{\mathrm{HW}}$
    \\
    \hline
    $\omega$ & $\rho(g)^\ast(\omega)$ & $\rho(\sigma)^\ast(\omega)$
    \\
    \hline
    \hline
    $\omega_2$ & $\phantom{-}\omega_2$ & $-\omega_2$
    \\
    \hline
    $\omega_3$ & $\phantom{-}\omega_3$ & $\phantom{-}\omega_3$
    \\
    \hline
  \end{tabular}
  }
  \end{equation}
  Hence, as opposed to the degree 2 generator $\omega_2$, the degree 3 generator $\omega_3$ is
   invariant, too, and the claim follows.
\end{proof}

\medskip

This concludes our discussion of the plain real ADE-equivariance of the fundamental brane cocycles, under the various action appearing in Theorem \ref{RealADEEquivariantEndhancementOfM2M5Cocycle}. This is the \emph{extra property} that an equivariant enhancement has to satisfy.
Now we turn to the \emph{extra structure} that equivariant enhancement involves (the ``extra degrees of freedom'').

\medskip

\subsection{ Components of equivariant enhancement: The GS-action functional}
\label{FundamentalBraneCocyclesOnSystemsOfIntersectingBlackBranes}

Here we work out the component data of equivariant enhancements of super-cocycles (Example \ref{EquivariantEnhancementOfSuperCocycles})
of relevance in Theorem \ref{RealADEEquivariantEndhancementOfM2M5Cocycle},
that is associated with morphisms in the orbit category (Def.\ \ref{OrbitCategory}) along which the bosonic dimension of the fixed locus
in super-spacetime decreases (see Example \ref{ConstantSystemsOfFixedSubspaces}). The main result of this section is
Prop. \ref{TrivializationsOfRestrictionsOfM2M5Cocycle} below, which says, in the terminology of Sec. \ref{Physics}, that this data is given by
\emph{fundamental brane cocycles trivializing over their own black brane
superspacetime locus via their Green--Schwarz action functional}.
Notice that the Green-Schwarz Lagrangian density of a brane,
when restricted this way to a solitonic configuration of that brane, is the local density of its
\emph{brane instanton contribution} \cite{BeckerBeckerStrominger95, HarveyMoore99}.
We conclude this section in Example \ref{ComponentsForEquivariantEnhancementOfFundamentalBraneCocycles}
by explaining how this is used in the proof
of Theorem \ref{RealADEEquivariantEndhancementOfM2M5Cocycle}.

\medskip

\begin{remark}[Left-invariant differential forms on the super Minkowski super Lie group]
    \label{SuperMinkowskiLeftInvariantDifferentialForms}
    If we regard a super Minkowski spacetime $\mathbb{R}^{p,1\vert \mathbf{N}}$
    (Def.\ \ref{MinkowskiSuper}) as a super Lie algebra, via Example
    \ref{SuperLieAlgebraAsSuperSpaces}, we can exponentiate this super Lie algebra to
    obtain a simply-connected super Lie group. It turns out that the underlying super
    manifold of this group is just $\R^{p,1|{\bf N}}$, so we abuse notation by
    denoting the algebra and group by the same symbol.
    This super manifold comes with canonical coordinate functions
    $$
      \{x^a\}_{a = 0}^p\;, \qquad \{ \theta^\alpha \}_{\alpha = 1}^N
    $$
    as well as their de Rham differentials
    $$
      \{d_{\mathrm{dR}} x^a\}_{a = 0}^p\;,  \qquad \{ d_{\mathrm{dR}}\theta^\alpha \}_{\alpha = 1}^N
      \,.
    $$
    But the non-triviality of the super Lie bracket implies that, unlike ordinary (bosonic) Minkowski spacetime,
    the differential forms $d_{\mathrm{dR}}x^a$ are \emph{not} left-invariant (nor right-invariant)
    under the left or right action of the super Lie group on itself. In physics jargon, we say they
    are \emph{not supersymmetric}. Instead, a basis of left-invariant (hence supersymmetric) differential
    1-forms on the super Minkowski super Lie group
    is given by
    \begin{equation}
      \label{LeftInvariant1FormsOnSuperMinkowski}
      \begin{aligned}
        e^a & := d_{\mathrm{dR}}x^a + \overline{\theta} \Gamma^a d_{\mathrm{dR}}\theta\;,
        \\
        \psi^\alpha & := d_{\mathrm{dR}}\theta^\alpha\;.
      \end{aligned}
    \end{equation}
    Of course the Chevalley-Eilenberg algebra of a super Lie group is isomorphic to the sub dgc-algebra of its de Rham algebra
    on the left-invariant differential forms, and under this correspondence \eqref{LeftInvariant1FormsOnSuperMinkowski} readily yields
    \eqref{CEAlgebraSuperMinkowski}.
\end{remark}

\begin{defn}[Supersymmetric volume form]
  \label{SupersymmetricVolumeForm}
  Let $\mathbb{R}^{p,1\vert \mathbf{N}}$ be a super Minkowski spacetime (Def.\ \ref{MinkowskiSuper}).
  Then the \emph{supersymmetric volume form} on $\mathbb{R}^{p,1\vert \mathbf{N}}$ is defined by
  $$
    \mathrm{svol}_{p+1} \;:=\; e^0 \wedge e^1 \wedge \cdots \wedge e^p
    \;\in\; \mathrm{CE}( \mathbb{R}^{p,1\vert \mathbf{N}})\;.
  $$
  \end{defn}
  Via remark \ref{SuperMinkowskiLeftInvariantDifferentialForms} this may equivalently be expressed in terms of
  plain differential forms on the super Minkowski supermanifold as
  \begin{equation}
    \label{SupersymmetricVolumeFormDecomposition}
    \mathrm{svol}_{p+1}
    \;=\;
    \underset{
      \mathrm{vol}_{p+1}
    }{
    \underbrace{
      d_{\mathrm{dR}}x^0 \wedge d_{\mathrm{dR}}x^1 \wedge \cdots \wedge d_{\mathrm{dR}}x^p
    }}
    +
    \underset{
      \Theta_{p+1}
    }{
    \underbrace{
      (\overline{\theta} \Gamma^a d_{\mathrm{dR}}\theta) \wedge ( \cdots )
    }
    }
    \,,
  \end{equation}

  \vspace{-3mm}
\noindent  where the first summand is the ordinary (bosonic) volume form, while the second summand is a fermionic correction that
  makes the sum be left-invariant (hence supersymmetric).

\begin{lemma}[Supersymmetric volume form on black $p$-brane trivializes fundamental $p$-brane cocycle]
  \label{VolumeFormOnBlackpBraneTrivializedFundamentalpBraneCocycle}
  Let
  $
    \mathbb{R}^{p,1\vert \mathbf{N}} \hookrightarrow \mathbb{R}^{10,1\vert \mathbf{32}}
  $
  be one of the fixed point BPS subspaces (Def.\ \ref{BPSSuperLieSubalgebra}) in the classification from Theorem \ref{SuperADESingularitiesIn11dSuperSpacetime}.
  Then $\pm 1$ times the supersymmetric form on $\mathbb{R}^{p,1\vert \mathbf{N}}$ (Def.\ \ref{SupersymmetricVolumeForm})
  provides a trivialization for the $p$-brane cocycle on $\mathbb{R}^{p,1\vert \mathbf{N}}$, in that
  $$
    d(\mathrm{svol}_{p+1})
    \;=\;
    \pm
    \tfrac{1}{p!} \big(\overline{\psi} \mathbf{\Gamma}_{a_1 \cdots a_p} \psi \big) \wedge e^{a_1} \wedge \cdots e^{a_p}
    \;\;\;
    \in
    \mathrm{CE}\big( \mathbb{R}^{p,1\vert \mathbf{N}}  \big)\;.
  $$
\end{lemma}
\begin{proof}
The fixed superspaces $\mathbb{R}^{p,1\vert \mathbf{N}}$ in the classification of Theorem
\ref{SuperADESingularitiesIn11dSuperSpacetime}
are, in particular, fixed by the corresponding $p$-brane involution (Def.\ \ref{pBraneInvolution}). Hence
the proof of Lemma \ref{ExistenceOfpBraneInvolutions} shows that the restriction of the complementary product of Clifford generators
to the fixed Spin sub-representation $\mathbf{N} \hookrightarrow \mathbf{32}$ is the identity, possibly up to a sign:
$
  \mathbf{\Gamma}_{p+1} \mathbf{\Gamma}_{p+2} \cdots \mathbf{\Gamma}_{10} \vert_{{}_{\mathbf{N}}}
  =
  \pm 1
$.
But also
$
  \mathbf{\Gamma}_0 \mathbf{\Gamma}_1 \cdots \mathbf{\Gamma}_{10} = \pm 1
$,
since $\mathbf{32}$ is chiral.
Together this implies that
$
  \mathbf{\Gamma}_0 \mathbf{\Gamma}_1 \cdots \mathbf{\Gamma}_p\vert_{{}_{\mathbf{N}}}
  =
  \pm 1
$
and hence that
\begin{equation}
  \label{EveryGammaProductOfTheOthersOnFixedPBrane}
  \begin{aligned}
    \mathbf{\Gamma}_{a_0}\vert_{{}_{\mathbf{N}}}
    & =
    \pm
    \tfrac{1}{p!}
    \epsilon_{ a_0 a_1 \cdots a_p  } \mathbf{\Gamma}^{a_1 \cdots a_p}
  \end{aligned}
\end{equation}
with the same sign $\pm$ for all $0 \leq a_0 \leq p$.
Using this we compute as follows:
$$
  \begin{aligned}
    d (\mathrm{svol}_{p+1})
    & =
    \tfrac{1}{(p+1)!} \epsilon_{a_0 \cdots a_p} d \left( e^{a_0} \wedge \cdots \wedge e^{a_p} \right)
    \\
    & \propto \;
    \tfrac{1}{p!} \epsilon_{a_0 a_1 \cdots a_p} \left( \overline{\psi} \mathbf{\Gamma}^{a_0} \psi \right) \wedge e^{a_1} \wedge \cdots e^{a_p}
    \\
    &
    \propto \;
    \pm \tfrac{1}{p!^2}\;
      \underset{
        = \delta_{a_1 \cdots a_p}^{b_1 \cdots b_p}
      }{
      \underbrace{
        \epsilon_{a_0 a_1 \cdots a_p}\epsilon^{a_0 b_1 \cdots b_p}
      }}
    \left( \overline{\psi} \mathbf{\Gamma}_{b_1 \cdots b_p} \psi \right) \wedge e^{a_1} \wedge \cdots e^{a_p}
    \\
    & \propto  \; \pm \tfrac{1}{p!} \left(\overline{\psi} \mathbf{\Gamma}_{a_1 \cdots a_p} \psi \right) \wedge e^{a_1} \wedge \cdots \wedge e^{a_p}\;.
  \end{aligned}
$$
Here the first step is just a combinatorial rewriting of the supersymmetric volume form, the second step
collects the results of applying the CE-differential \eqref{CEAlgebraSuperMinkowski} via its derivation property, the
third step uses \eqref{EveryGammaProductOfTheOthersOnFixedPBrane}, and in the fourth step we
used the combinatorics of the skew-symmetrized Kronecker symbol. The omitted proportionality
factors in each step are just integer powers of $i$
that are fixed by spinor conventions and using that the result must be real.
\end{proof}

\begin{prop}[Super volume form is Green--Schwarz Lagrangian on super embedding]
  \label{GreenSchwarzFromSuperVolumeForm}
  Let
  \begin{equation}
    \label{SuperembeddingCanonical}
    \mathbb{R}^{p,1\vert \mathbf{N}}
    \xymatrix{\ar@{^{(}->}[r]&}
     \mathbb{R}^{10,1\vert \mathbf{32}}
  \end{equation}
  be one of the fixed point BPS subspaces (Def.\ \ref{BPSSuperLieSubalgebra}) in the classification from Theorem \ref{SuperADESingularitiesIn11dSuperSpacetime}.
  Then its supersymmetric volume form (Def.\ \ref{SupersymmetricVolumeForm}) is the restriction of the Green--Schwarz-type Lagrangian
  for a super $p$-brane along the embedding \eqref{SuperembeddingCanonical}:
  \begin{equation}
    \label{SuperVolumeFormAsLagrangian}
    \mathrm{svol}_{p+1} \;=\; \mathbf{L}_{GS,p+1}\vert_{{}_{\mathbb{R}^{p,1\vert \mathbf{N}}}}
    \,.
  \end{equation}
\end{prop}
\begin{proof}
  For any embedding \eqref{SuperembeddingCanonical},
  the decomposition \eqref{SupersymmetricVolumeFormDecomposition} says that $\mathrm{svol}_{p+1}$ is
  the sum
  $$
    \mathrm{svol}_{p+1} \;=\; \mathbf{L}_{\mathrm{NG},p+1}\vert_{{}_{\mathbb{R}^{p,1\vert \mathbf{N}}}} + \Theta_{p+1}
  $$
  of the restriction of the Nambu-Goto Lagrangian $\mathbf{L}_{\mathrm{NG},p+1}$, which is the plain bosonic volume form
  $$
    \mathrm{vol}_{p+1} \;=\; \mathbf{L}_{\mathrm{NG}}\vert_{{}_{\mathbb{R}^{p,1\vert \mathbf{N}}}}
    \,,
  $$
  with a fermionic term $\Theta_{p+1}$.
  But if the embedding is, moreover, that of a BPS fixed locus from Theorem \ref{SuperADESingularitiesIn11dSuperSpacetime},
  then, since  $d(\mathrm{vol}_{p+1}) = 0$ (by Remark \ref{SuperMinkowskiLeftInvariantDifferentialForms}),
  Lemma \ref{VolumeFormOnBlackpBraneTrivializedFundamentalpBraneCocycle} says that $\Theta_{p+1}$ is a fermionic potential for the
  $p$-brane cocycle:
  $$
    d \Theta_{p+1} \;=\; \mu_p \;\propto\; \tfrac{1}{p!} \left(\overline{\psi} \Gamma_{a_1 \cdots a_p} \psi \right) \wedge e^{a_1} \wedge \cdots e^{a_p}
    \,,
  $$
  hence it is a \emph{WZW-term} for this cocycle. Now the Green--Schwarz action functional comes from a  very particular
  choice of such a potential, but \eqref{EveryGammaProductOfTheOthersOnFixedPBrane} implies that this is precisely the choice
  provided here by the supersymmetric volume form. We make this explicit for $p = 2$,
  and the other cases work analogously:
  $$
    \begin{aligned}
      \Theta_{2+1}
      & :=
      \mathrm{svol}_{2+1} - \mathrm{vol}_{2+1}
      \\
      & =
      \tfrac{1}{3!}
         \epsilon_{a_0 a_1 a_2}
         (d x^{a_0} +  \overline{\theta}\Gamma^{a_0}d \theta)
       \wedge
       (d x^{a_1} +  \overline{\theta}\Gamma^{a_1}d \theta)
       \wedge
       (d x^{a_2} +  \overline{\theta}\Gamma^{a_2}d \theta)
       \,-\,
       dx^0 \wedge dx^1 \wedge dx^2
      \\
      & =
      \tfrac{1}{3!}\;
      \underset{
        = \overline{\theta}\Gamma^{a_1 a_2} d \theta
      }{
      \underbrace{
        \epsilon_{a_0 a_1 a_2}
        (\overline{\theta}\Gamma^{a_0}d \theta)
      }
      }
      \big(
        3 d x^{a_1}
        \wedge
        ( d x^{a_2} +  \overline{\theta}\Gamma^{a_2}d \theta)
        +
        (\overline{\theta}\Gamma^{a_1}d \theta)
        \wedge
        (\overline{\theta}\Gamma^{a_2}d \theta)
      \big)
      \\
      &  =
      (\overline{\theta} \Gamma_{a_1 a_2} d \theta)
      \big(
        \tfrac{1}{2}
        d x^{a_1}
        \wedge
        ( d x^{a_2} +  \overline{\theta}\Gamma^{a_2}d \theta)
        +
        \tfrac{1}{6}
        (\overline{\theta}\Gamma^{a_1}d \theta)
        \wedge
        (\overline{\theta}\Gamma^{a_2}d \theta)
      \big)
      \,,
    \end{aligned}
  $$
  where under the brace we used \eqref{EveryGammaProductOfTheOthersOnFixedPBrane}. This is indeed the form of the Green--Schwarz-WZW type Lagrangian term in the
  sigma model for the super 2-brane \cite[(2.1)]{deWitHoppeNicolai88} \cite[(3)]{DasguptaNicolaiPlefka03}. The only difference is
  that here the bosonic indices range only over $a_i \in \{0, 1, 2\}$ instead of $a_i \in \{0,\cdots, 10\}$, and the fermionic coordinates
  range only over $\mathbf{16}$ inside $\mathbf{32}$. But this means exactly that we have the restriction of the full Green--Schwarz-WZW term along the embedding
  \eqref{SuperVolumeFormAsLagrangian}, as claimed.
\end{proof}

\begin{prop}[Fundamental brane cocycles on their black brane cobounded by  Green--Schwarz action]
\label{TrivializationsOfRestrictionsOfM2M5Cocycle}

\noindent {\bf (i)}
  Let
  $
    \mathbb{R}^{p,1\vert \mathbf{N}} \hookrightarrow \mathbb{R}^{10,1\vert \mathbf{32}}
  $
  be one of the fixed point BPS subspaces (Def.\ \ref{BPSSuperLieSubalgebra}) in the classification from Theorem \ref{SuperADESingularitiesIn11dSuperSpacetime}
  with $p \leq 5$, hence the $\mathrm{MO1}$ or $\mathrm{M2}$ or $\mathrm{MO5}$.
Then the restriction of the M2/M5-brane cocycle (Prop. \ref{M2M5SuperCocycle}) to this subspace is homotopic to the trivial cocycle,
and the homotopy is given by
the corresponding Green--Schwarz Lagrangian,
according to Prop. \ref{GreenSchwarzFromSuperVolumeForm}.
$$
  \raisebox{20pt}{
  \xymatrix@C=2em{
    \mathbb{R}^{10,1\vert \mathbf{32}}
    \ar[rrr]^-{\mu_{{}_{M2/M5}}}
    &&&
    S^4
    \\
    \mathbb{R}^{p,1 \vert 16}
    \ar@{^{(}->}[u]_>{\ }="s"
    \ar[rrr]^>{\ }="t"
    &&&
    \ast
    \ar@{^{(}->}[u]
    \ar@{=>}^{\;\; \pm \mathrm{svol}_{p+1}} "s"; "t"
  }}
$$
On the other hand, for the cases with $p \geq 6$ in the classification from Theorem \ref{SuperADESingularitiesIn11dSuperSpacetime},
no such trivializing homotopy exists.

\item {\bf (ii)}
Let
  $
    \mathbb{R}^{1,1\vert 8 \cdot \mathbf{1}} \hookrightarrow \mathbb{R}^{9,1,1\vert \mathbf{16}}
  $
the relative fixed locus inclusion corresponding to the item $\mathrm{NS1}_H$ in Prop. \ref{NonSimpleRealSingularities}.
Then the restriction of the fundamental heterotic string cocycle (Example \ref{FundamentalF1Cocycle}) along
that inclusion has a trivializing homotopy given by its Green--Schwarz action functional, according to Prop. \ref{SuperembeddingCanonical}:
$$
  \raisebox{20pt}{
  \xymatrix@C=2em{
    \mathbb{R}^{9,1\vert \mathbf{16}}
    \ar[rrr]^-{\mu^{\mathrm{het}}_{{}_{F1}}}
    &&&
    S^3
    \\
    \mathbb{R}^{1,1 \vert 8 \cdot \mathbf{1}}
    \ar@{^{(}->}[u]_>{\ }="s"
    \ar[rrr]^>{\ }="t"
    &&&
    \ast
    \ar@{^{(}->}[u]
    \ar@{=>}^{\;\; \pm \mathrm{svol}_{1+1}} "s"; "t"
  }}
$$
\end{prop}

\begin{proof}
  First observe the existence of trivializing coboundaries, as claimed:
  \begin{itemize}
  \vspace{-3mm}
    \item For $p = 1$ both the restrictions of $\mu_{{}_{M2}}$ and $\mu_{{}_{M5}}$ vanish identically, for degree reasons,
     while the restriction of $\mu_{{}_{F1}}$ is trivialized by the super volume form according to Lemma \ref{VolumeFormOnBlackpBraneTrivializedFundamentalpBraneCocycle}.
   \vspace{-3mm}
    \item For $p = 2$ the cochain $\mu_{{}_{M5}}$ still vanishes identically by degree reasons,
    while $\mu_{{}_{M2}}$ has a coboundary,
    given by the volume form, according to Lemma \ref{VolumeFormOnBlackpBraneTrivializedFundamentalpBraneCocycle}.
   \vspace{-3mm}
    \item For $p = 5$, the cochain $\mu_{{}_{M2}}$ vanishes by Prop. \ref{M2CocycleVanishesOnMO9},
      while now $\mu_{{}_{M5}}$ has a coboundary, again by Lemma \ref{VolumeFormOnBlackpBraneTrivializedFundamentalpBraneCocycle}.
  \end{itemize}
  \vspace{-2mm}
  By Example \ref{dgCoboiundariesHomotopies} these coboundaries yield the required homotopy.
This concludes the cases where the trivializing homotopy does exist.

\noindent We now turn to the cases where it does not exist:
  \begin{itemize}
   \vspace{-2mm}
    \item
     For restriction to $p = 6$ the M2-cocycle remains non-trivial,
     by the old brane scan (Prop. \ref{TheOldBraneScan}), which is sufficient for the M2/M5-cocycle to be non-trivial.
    \vspace{-3mm}
    \item
    For $p = 9$ the $\mu_{{}_{M2}}$-component of $\mu_{{}_{M2/M5}}$ vanishes, by Prop. \ref{M2CocycleVanishesOnMO9},
    and hence here the homotopy would exhibit a coboundary for the NS5-brane cocycle (this is studied from the 
    topological point of view in \cite{NS5}).  
    But, by the old brane scan (Prop. \ref{TheOldBraneScan}), this does not exist
    (meaning that the fundamental NS5-brane does exist in heterotic superspacetime).
  \end{itemize}

\vspace{-.5cm}
\end{proof}

We summarize in the following example how Prop. \ref{TrivializationsOfRestrictionsOfM2M5Cocycle} is used
iteratively to prove the existence of the equivariantly enhanced cocycles in Theorem \ref{RealADEEquivariantEndhancementOfM2M5Cocycle}:

\begin{example}[Components for equivariant enhancement of fundamental brane cocycles]
  \label{ComponentsForEquivariantEnhancementOfFundamentalBraneCocycles}
  Consider a diagram in rational super homotopy theory (Def.\ \ref{RationalSuperHomotopyTheory}) of the form
  \begin{equation}
  \label{eq-cptdiag}
    \xymatrix@R=.9em@C=3em{
      \mathbb{R}^{d, 1 \vert \mathbf{N}}
      \ar[rr]^{\mu}_<<<<<<<{\ }="s"
      &&
      S^{n}
      \\
      \\
      \mathbb{R}^{p, 1 \vert \mathbf{N}/\mathbf{k}}
      \ar@{^{(}->}[uu]^{\iota}
      \ar@{-->}[rr]_{\mu'}^>>>>>>{\ }="t"
      &&
      S^{k \lt n}
      \ar[uu]_0
      \ar@{==>}^{\eta} "s"; "t"
    }
  \end{equation}
  where $\mathbb{R}^{d,1\vert \mathbf{N}}$ is a super Minkowski spacetime (Def.\ \ref{MinkowskiSuper}), the
  left vertical inclusion is that of a fixed BPS subspace from Theorem \ref{SuperADESingularitiesIn11dSuperSpacetime}, Prop. \ref{NonSimpleRealSingularities},
  and the right vertical inclusion that of a fixed locus in a sphere, as appearing in Prop. \ref{ADEWEquivariant4SpheredgAlgebraicCoefficients}.

\begin{itemize}
\vspace{-2mm}
\item  We discuss the data involved in choosing the dashed maps, if the solid maps are given.
  The sphere coefficients may be represented via their minimal dgc-algebra models, according to Example \ref{SpheredgcAlgebraModel}.
  Hence if the sphere dimension is odd, then a morphism from a super Minkowski spacetime to it is equivalently a Spin-invariant cocycle
  in that degree, and classified by the old brane scan (Prop. \ref{TheOldBraneScan}).
  If the sphere coefficient is even, then such a morphism is equivalently one such cocycle in that degree,
  again classified by the old brane scan, together with a second element that trivialized the product of that cocycle with itself.
  In the cases considered below, that cocycle happens to vanish, so that the second element is a cocycle, and classified by the old brane scan.
  Hence in each case of interest, the choice for $\mu'$ is classified by an element
  \begin{equation}
    \label{Choice1}
    [\mu'] \;\in\; H^\bullet( \mathbb{R}^{p,1\vert \mathbf{N}/\mathbf{k}} )^{\mathrm{Spin}(p,1)}
  \end{equation}
  and these are controlled by the old brane scan.

\vspace{-2mm}
\item  We consider the cases where the right vertical map in \eqref{eq-cptdiag} is induced by the inclusion of a
  sphere of lower dimension into that of higher dimension. These inclusions are all null-homotopic
  and hence are represented by the zero-homomorphism on the representing dgc-algebras, as in Prop. \ref{ADEWEquivariant4SpheredgAlgebraicCoefficients}.
  This means that, independently of the choice of the cocycle $\mu'$,
  the homotopy $\eta$ is a null homotopy trivializing the restriction of
  $\mu$ along the fixed super subspace inclusion. By Example
  \ref{dgCoboiundariesHomotopies} these are given by coboundaries $\alpha$
  in the corresponding dg-algebra:
  $
    d \alpha = \iota^\ast(\mu)
    \in
    \mathrm{CE}\left( \mathbb{R}^{p,1\vert \mathbf{N}/\mathbf{k}} \right)
  $.
  By Prop. \ref{TrivializationsOfRestrictionsOfM2M5Cocycle}, in all cases considered here,
  one such choice is the super volume form $\mathrm{svol}_{p+1}$, hence the Green--Schwarz Lagrangian (Prop. \ref{GreenSchwarzFromSuperVolumeForm}).

\vspace{-2mm}
 \item  If $\alpha$ is any other choice, it follows that the difference $\mathrm{svol}_{p+1} - \alpha $ is a cocycle, and it is
  non-trivial as a cocycle precisely if the corresponding two homotopies are not themselves related by a higher homotopy.
  Hence, in the examples considered, the space of choices for $\eta$ is in each case the vector space
  $H^{p+1}\left(\mathbb{R}^{p,1\vert \mathbf{N}/\mathbf{k}}\right)$, again given by the old brane scan:
  \begin{equation}
    \label{Choice2}
    [\eta] = [\mathrm{svol}_{p+1}] + H^{p+1}\big(\mathbb{R}^{p,1\vert \mathbf{N}/\mathbf{k}}\big)
    \,.
  \end{equation}
  \end{itemize}
\end{example}

\medskip
This concludes the analysis of the available choices of components for
equivariant enhancement of fundamental brane cocycles,
and hence completes the proof of Theorem \ref{RealADEEquivariantEndhancementOfM2M5Cocycle}.


\appendix

\section{Mathematics background and conventions}

For ease of reference, here we collect some standard background material that we use in the main text:
on \emph{Spacetime and Spin} in Sec. \ref{SpacetimesAndSpin} and on \emph{Homotopy and Cohomology} in Sec. \ref{HomotopyTheory}.

\subsection{Spacetime and Spin}
  \label{SpacetimesAndSpin}

For reference and to fix some essential conventions, we briefly recall some details on real Spin representations in Lorentzian signature.

\begin{defn}[Spin geometry of Minkowski spacetime]
  \label{LorentzGroupsAndTheirSpinCovers}
For $p \in \mathbb{N}$ we write $\mathbb{R}^{p,1}$ for the corresponding Minkowski spacetime.
The underlying real vector space is $\mathbb{R}^{p+1}$.
\item {\bf (i)} With its canonical coordinate functions labeled as
$
  (x^0, x^1, \cdots, x^p)
$,
the inner product (Minkowski metric) is taken with the mostly plus
signature: for two vectors $u$ and $v$ in $\R^{p,1}$, we have
\begin{equation}
  \label{MinkowskiMetric}
  \eta(u,v) = -u^0 v^0 + u^1 v^1 + \cdots + u^p v^p
  \,.
\end{equation}

\item {\bf (ii)} We write
\begin{equation}
  \label{GeneralCliffordAlgebraSigns}
  \Cl(p,1)
  \;:=\;
  \R \langle \mathbf{\Gamma}_0, \mathbf{\Gamma}_1, \cdots, \mathbf{\Gamma}_p \rangle
  /
  \left(
    \mathbf{\Gamma}_a \cdot \mathbf{\Gamma}_b
    +
    \mathbf{\Gamma}_b \cdot \mathbf{\Gamma}_a
    =
    + 2 \eta_{a,b}
  \right)
\end{equation}
for the Clifford algebra of $\R^{p,1}$. This is the quotient of the free real
associative algebra on $p+1$ generators $\mathbf{\Gamma}_a$ by the Clifford relation,
which says that  their anticommutator is twice the corresponding entry in the Minkowski
metric. We write
\begin{equation}
  \label{SkewSymmetrizedCliffordAlgebra}
  \mathbf{\Gamma}_{a_1 \cdots a_q}
  \;:=\;
  \tfrac{1}{q!}
  \underset{\sigma \in \Sigma_q}{\sum} (-1)^{{\vert \sigma\vert}}
  \mathbf{\Gamma}_{a_{\sigma(1)}} \cdot \mathbf{\Gamma}_{a_{\sigma(2)}} \cdots \cdot \mathbf{\Gamma}_{a_{\sigma(q)}}
\end{equation}
for the skew-symmetrized products of the Clifford generators.

\item {\bf (iii)} Recall the various groups acting on Minkowski spacetime:
\begin{itemize}

\item the  \emph{Lorentz group} ${\rm O}(p,1) \hookrightarrow \mathrm{GL}(p+1)$ is the subgroup of linear transformations that preserve the Minkowski metric (\ref{MinkowskiMetric});

\item the \emph{orthochronous Lorentz group} ${\rm O}^+(p,1) \hookrightarrow {\rm O}(p,1)$
  is the subgroup of transformations that preserve time-orientation;

\item the \emph{Pin group} $\Pin(p,1) \to {\rm O}(p,1)$ is the double cover of the Lorentz
  group given by the standard Clifford algebraic construction: the group $\Pin(p,1)$
  is the subgroup of invertible elements of the Clifford algebra generated by unit
  vectors in $\R^{p,1}$:
  \[ \Pin(p,1) : = \langle u \in \R^{p,1} \, : \, \eta(u,u) = \pm 1 \rangle \subseteq
  \Cl(p,1) . \]
  A unit vector $u \in \Pin(p,1)$ maps to the reflection in ${\rm O}(p,1)$ through
  the hyperplane orthogonal to $u$, and this map on generators extends to a
  homomorphism $\Pin(p,1) \to {\rm O}(p,1)$. It is well known that this homomorphism
  is a double cover.

\vspace{-2mm}
\item the \emph{orthochronous Pin group} $\Pin^+(p,1) \to {\rm O}^+(p,1)$ is subgroup of
  $\Pin(p,1)$ that double covers the orthochronous Lorentz group ${\rm O}^+(p,1)$.

\item the \emph{Spin group} $\Spin(p,1) \to \SO^+(p,1)$ is the subgroup of $\Pin(p,1)$ that
  double covers the connected Lorentz group $\SO^+(p,1)$. In terms of the Clifford
  algebra, it is the subgroup of invertible elements generated by products of pairs
  of unit vectors with the same sign:
  \[
  \Spin(p,1) = \langle uv \in \Cl(p,1) : u,v \in \R^{p,1} \mbox{ and } \eta(u,u) = \eta(v,v) = \pm 1 \rangle .
   \]
  The Lie group $\Spin(p,1)$ is connected and simply-connected.

\item The  \emph{ Lie algebra of the Spin group}
is the Lie subalgebra of the Clifford algebra on commutators of vectors:
  \[
   {\rm Lie}(\Spin(p,1)) \simeq \{ [u,v] \in \Cl(p,1) : u,v \in \R^{p,1} \} .
   \]
  Note that this subspace is actually a Lie subalgebra with respect to the commutator in
  $\Cl(p,1)$. The double cover map $\Spin(p,1) \to \SO^+(p,1)$ induces an isomorphism of Lie algebras
  \[
  {\rm Lie}(\Spin(p,1) \simeq \so(p,1)
   \]
  with the Lorentz Lie algebra $\so(p,1)$ of the connected Lorentz group $\SO^+(p,1)$.
\end{itemize}
\end{defn}

Note that because the Spin group $\Spin(p,1)$ is connected and simply-connected, we
can describe a representation in two ways:

  \vspace{-2mm}
\begin{enumerate}[{\bf 1.}]
\item We can give the action of the generators $uv \in \Spin(p,1)$ on a vector space,
  where $u$ and $v$ are unit vectors of the same sign.

  \vspace{-2mm}
\item We can give the action of the Lie algebra $\so(p,1)$ on a vector space, and
  exponentiate to get the action of the Lie group $\Spin(p,1)$. In particular, it
  suffices to give the action of a basis of the Lie algebra $\{ [u,v] \in \Cl(p,1) :
  u,v \in \R^{p,1} \}$. A natural choice is the basis given by skew-symmetrized
  products of two gamma matrices, $\{ \Gamma_{ab} \}$.
\end{enumerate}

\begin{remark}[Technology for real Spin representations]
 \label{WaysOfGoingAboutRealSpinReps}

There are two alternative ways of constructing and handling the real Spin
representations that appear in the super-Minkowski spacetimes in  Def.\ \ref{MinkowskiSuper}:
\begin{itemize}

  \vspace{-2mm}
\item One may carve out real Spin representations from complex Dirac or Weyl
  representations by imposing a reality condition, called the \emph{Majorana condition}.
  This is the standard method used in the physics literature. A textbook reference
  for standard conventions is \cite[Sec. II.7]{CDF}, while a conceptual account is
  in \cite{FF}.  We recall this as Prop. \ref{RealSpinorsViaMajoranaConditionsOnDiracRepresentations} below;
  this serves for
  comparing the results in Sections \ref{ADESingularitiesInSuperSpacetime}, \ref{ADEEquivariantRationalCohomotopy},  and \ref{ADEEquivariantMBraneSuperCucycles} to the
  bulk of the string theory literature.

  \vspace{-2mm}
  \item Alternatively, one may use the real normed division algebras and matrices over
    them. The most famous example of this is identifying 4-dimensional spacetime,
    $\R^{3,1}$, with the $2 \times 2$ complex hermitian matrices, and generating the
    Weyl representations of $\Spin(3,1)$ on $\C^2$ from the action of these
    matrices. Yet this sort of construction continues to work for normed division
    algebras other than $\C$, and for spacetimes other than dimension 4. We recall
    this approach in Sec. \ref{SpinViaDivisionAlgebras} below; this serves to
    streamline the proofs of the theorems in Sections \ref{ADESingularitiesInSuperSpacetime}, \ref{ADEEquivariantRationalCohomotopy},  and \ref{ADEEquivariantMBraneSuperCucycles}.
 \end{itemize}
\end{remark}

\subsubsection*{Real Pin-representations via Majorana condition}
\label{MajoranaReps}

\begin{prop}[Real spinors via Majorana conditions on Dirac representations]
  \label{RealSpinorsViaMajoranaConditionsOnDiracRepresentations}
  Let
  $$
    p + 1 \in \{2\nu, 2\nu+1\}, \;\;\; \nu \in \mathbb{N}, \;\;\; 2\nu \geq 4
    \,.
  $$
  and let $N = 2^\nu$.
\item  {\bf (i) Dirac representations} (as in \cite[Sec. II.7.1]{CDF}): There exist complex matrices
$$
  \Gamma_a  \;\in\; {\rm End}_{\mathbb{C}}\left( \mathbb{C}^N  \right),
  a \in \{0,1, \cdots, p\}
$$
with the following properties:
\begin{equation}
  \label{CliffordAdjointsInDiracRepresentation}
\begin{aligned}
  & \phantom{AAAAA}  \Gamma_a \Gamma_b + \Gamma_b \Gamma_a = - 2 \eta_{a b}\;,
  \\
  &\left( \Gamma_0\right)^2 = +1\;,
  \phantom{AAAA}
  \left( \Gamma_a\right)^2 = -1
  \,,
\\
 &(\Gamma_0)^\dagger =  \Gamma_0\;,
  \phantom{AAAAA}
  (\Gamma_a)^\dagger = - \Gamma_a\;,
  \phantom{AA}
  \mbox{for $a,b \in \{1,\cdots, p\}$}.
\end{aligned}
\end{equation}
\item{\bf (ii) Charge conjugation matrices} (as in \cite[Sec. II.7.2]{CDF}):
Moreover, there exist \emph{charge conjugation matrices}
$$
  C_{(\pm)} \in \mathrm{End}_{\mathbb{C}}( \mathbb{C}^{N} )
$$
with real entries $(C_{(\pm)})^\ast = C_{(\pm)}$ and  related to the above $\Gamma$-matrices by
$
  C_{(\pm)} \Gamma = \pm \Gamma_a^t C_{(\pm)}
$
according to the following table:
\begin{equation}
\label{ChoicesOfChargeConjugation}
\mbox{
\begin{tabular}{|c|c|c|}
  \hline
  $p + 1$ & $C_{(+)}$ & $C_{(-)}$
  \\
  \hline
  \hline
  $3+1$ & $\ast$ & $\ast$
  \\
  \hline
  $4+1$ & $\ast$ &
  \\
  \hline
  $5+1$ & $\ast$ & $\ast$
  \\
  \hline
  $6+1$ &  & $\ast$
  \\
  \hline
  $7+1$ & $\ast$ & $\ast$
  \\
  \hline
  $8+1$ & $\ast$ &
  \\
  \hline
  $9+1$ & $\ast$ & $\ast$
  \\
  \hline
  $10+1$ &  & $\ast$
  \\
  \hline
\end{tabular}
}
\end{equation}
\item{\bf (iii) Majorana condition} (as in \cite[Sec. II.7.3]{CDF}):
Given a Dirac spinor $\psi \in \mathbb{C}^{N}$ we say that its \emph{Dirac conjugate} is
\begin{equation}
  \label{DiracConjugate}
  \overline{\psi}
  \;:=\;
  \psi^\dagger \Gamma_0
  \,.
\end{equation}
This $\psi$ is called a \emph{Majorana spinor} if its Dirac conjugate equals its \emph{Majorana conjugate},
which means
\begin{equation}
  \label{MajoranaCondition}
  \psi^t C = \psi^\dagger \Gamma_0
  \; \xymatrix{\ar@{<=>}[r]&} \;
  \mbox{$\psi$ is Majorana.}
\end{equation}
\item{\bf (iv) Majorana Spin representations} (see \cite{FF}): The subspace of Majorana spinors inside $\mathbb{C}^{N}$
$$
\mathbf{N} \subset \mathbb{C}^{N}.
$$
is preserved by multiplication by the $\Gamma_{a b}$. This set is a basis for
$\so(p,1)$ and this defines a real representation of $\Spin(p,1)$ on ${\bf N}$
with dimension $N = 2^\nu$.
The Dirac conjugation \eqref{DiracConjugate} induces on ${\bf N}$
the following quadratic and $\mathrm{Spin}(p,1)$-equivariant spinor-to-vector pairing
\begin{equation}
  \label{ViaDiracConjugateSpinorToVectorPairing}
  \xymatrix@R=-1pt{
    \mathbf{N}
    \ar[rr]^{ \overline{(-)}\Gamma(-) }
    &&
    \mathbb{R}^{p,1}
    \\
    \psi
    \ar@{|->}[rr]
    &&
    \left( \overline{\psi} \Gamma^a \psi \right)_{a = 0}^p
  }
\end{equation}

\item{\bf (v) Majorana Pin representations}:
For charge conjugation matrix $C_{(+)}$, the action of a single
$\Gamma_a$ preserves the Majorana condition. But for $C_{(-)}$ it does not.
For $C_{(-)}$ instead the product $i \Gamma_a$ preserves the Majorana condition. We will write
$$
  \mathbf{\Gamma}_a := i {\Gamma}_a
  \,.
$$
Instead of the relations \eqref{CliffordAdjointsInDiracRepresentation},
the relations satisfied by these boldface gamma matrices are the following:
\begin{equation}
  \label{CMinusMajoranaRepresentationForPin}
\begin{aligned}
  & \phantom{AAAA}
  \mathbf{\Gamma}_a \mathbf{\Gamma}_b + \mathbf{\Gamma}_b \mathbf{\Gamma}_a = + 2\eta_{a b}\;,
  \\
  &\left( \mathbf{\Gamma}_0\right)^2 = -1\;,
  \phantom{AAAA}\;
  \left( \mathbf{\Gamma}_a\right)^2 = +1\;,
 \\
  &(\mathbf{\Gamma}_0)^\dagger = - \mathbf{\Gamma}_0\;,
  \phantom{AAAA}
  (\mathbf{\Gamma}_a)^\dagger = + \mathbf{\Gamma}_a\;,
   \phantom{AA}
    \mbox{for $a \in \{1,\cdots, p\}$},
\end{aligned}
\end{equation}
Since now, for $C_{(-)}$, the subspace of Majorana spinors inside $\mathbb{C}^{N}$
is preserved by the action of each $\mathbf{\Gamma}_a$, equipped with this action it is a real
representation of the Pin group (Remark \ref{WaysOfGoingAboutRealSpinReps})
$
  \mathbf{N} \subset \mathbb{C}^{N}
$
of real dimension $N = 2^\nu$.
\end{prop}
\begin{example}[Real Spin representations]
  \label{RelevantExamplesOfRealSpinRepresentations}
  The following are the irreducible real representations (up to isomorphism)
  of $\mathrm{Spin}(p,1)$ (Def.\ \ref{LorentzGroupsAndTheirSpinCovers}) for values of $p$ of relevance in the main text,
  obtainable via Prop. \ref{RealSpinorsViaMajoranaConditionsOnDiracRepresentations}:
  $$
  \renewcommand{\arraystretch}{1.2}
  \begin{tabular}{|c||c|}
     \hline
     \begin{tabular}{c}
       Spacetime dimension
       \\
       $p + 1$
     \end{tabular}
     &
     \begin{tabular}{c}
       Supersymmetry
       \\
       $\mathbf{N}$
     \end{tabular}
     \\
     \hline
     \hline
     $10 + 1$ & $\mathbf{32}$
     \\
     \hline
     $9 +1$ & $\mathbf{16}$, $\overline{\mathbf{16}}$
     \\
     \hline
     $6 + 1$ & $\mathbf{16}$
     \\
     \hline
     $5 + 1$ & $\mathbf{8}$, $\overline{\mathbf{8}}$
     \\
     \hline
     $4 + 1$ & $\mathbf{8}$
     \\
     \hline
     $3+1$ & $\mathbf{4}$
     \\
     \hline
     $2+1$ & $\mathbf{2}$
     \\
     \hline
     $1 + 1$ & $\mathbf{1}$, $\overline{\mathbf{1}}$
     \\
     \hline
  \end{tabular}
  \renewcommand{\arraystretch}{1.2}
  $$
  We are particularly interested in the $\mathbf{32}$ of $\mathrm{Spin}(10,1)$.
  Notice that by (\ref{ChoicesOfChargeConjugation})
  the charge conjugation matrix in $D = 10+1 $ is $C_{(-)}$ and hence the gamma matrices representing the $\mathrm{Pin}(10,1)$-action on $\mathbf{32}$  are those from
  (\ref{CMinusMajoranaRepresentationForPin}).
\end{example}

\begin{remark}[Notation for Irrep decomposition -- Number of supersymmetries]
  \label{NumberOfSupersymmetries}
  Given irreducible real spinor representations $\mathbf{N}$ or $\overline{\mathbf{N}}$ as in Example \ref{RelevantExamplesOfRealSpinRepresentations},
  a general real spinor representation $\Delta$ is a direct sum of these. The multiplicities of the
  direct summands is traditionally denoted by $\mathcal{N}$ or $\mathcal{N}_{\pm} \in \mathbb{N}$:
  $$
    \Delta \;=\;  \mathcal{N} \cdot \mathbf{N}
     \phantom{AAA}
     \mbox{or}
     \phantom{AAA}
    \Delta \;=\; \mathcal{N}_+ \cdot \mathbf{N}
      \,\oplus\,
      \mathcal{N}_- \cdot \overline{\mathbf{N}}
    \,.
  $$
  Hence if the irreducible representations are understood, any other representation may be denoted simply by
  $$
    \mathcal{N}
    \phantom{AA}
    \mbox{or}
    \phantom{AA}
    \left( \mathcal{N}_+, \mathcal{N}_-\right)
    \,.
  $$
  When these real spinor representations serve as constituents of super Minkowski
  spacetimes (Def.\ \ref{MinkowskiSuper}) one calls the natural numbers $\mathcal{N}$ or
  $\mathcal{N}_{\pm}$ the \emph{number of supersymmetries}.
\end{remark}

\subsubsection*{Spinor representations  via normed division algebras}
\label{SpinViaDivisionAlgebras}

The observation that real $\mathrm{Spin}(p,1)$-representations for $p+1 \in \{3,4,5,6,7,10,11\}$
may be related to the real division algebras is due to \cite{KugoTownsend82}.
A comprehensive account is given in \cite{BaezHuerta10, BaezHuerta11}. Here we briefly recall the facts that we need.

\begin{defn}[Cayley--Dickson construction]
  \label{CayleyDicksonConstruction}
  Let $A$ be a real star-algebra (unitual, but not necessarily commutative nor associative), with star involution denoted by $\overline{(-)}$.
  Then its \emph{Cayley--Dickson double} $\mathrm{CD}(A)$ is the real star algebra obtained by adjoining a new generator $\ell$
  subject to the following relations:
  \begin{equation}
    \label{RelationsForCayleyDickson}
    \ell^2 = -1
    \,,
    \phantom{AA}
  \mbox{and}
    \phantom{AA}
    a (\ell b) = (a \overline{b}) \ell
    \,,
    \phantom{AA}
    (a \ell) b = \ell (a \overline{b})
    \,,
    \phantom{AA}
    (\ell a) (b \ell^{-1}) = \overline{a b}
  \end{equation}
  for all $a, b \in A$.
  This implies that the underlying real vector space is
  $$
    \mathrm{CD}(A) \simeq_{\mathbb{R}} A \oplus \ell A
    \,.
  $$
\end{defn}
\begin{example}[The four real normed division algebras]
  \label{TheFourRealNormedDivisionAlgebras}
  The first iterations of the Cayley--Dickson construction
  (Def.\ \ref{CayleyDicksonConstruction}) yield the real algebras of
   \vspace{-2mm}
    \item {\bf 1.} real numbers $\mathbb{R}$,
    \vspace{-2mm}
    \item {\bf 2.} complex numbers $\mathbb{C} \simeq \mathrm{CD}(\mathbb{R})$,
     \vspace{-2mm}
    \item {\bf 3.} quaternions $\mathbb{H} \simeq \mathrm{CD}(\mathbb{C})$,
     \vspace{-2mm}
    \item {\bf 4.} octonions $\mathbb{O} \simeq \mathrm{CD}(\mathbb{H})$.

\vskip .1em
\noindent These four algebras also happen to be precisely the finite-dimensional
`normed division algebras' over the real numbers. Recall that a normed division
algebra $\K$ is a real algebra, not necessarily associative, with unit 1 and equipped
with a norm $| \cdot |$ such that:
\[ |xy| = |x||y| \mbox{ for all } x, y \in \K . \]
We say that an algebra equipped with such a norm is \emph{normed}. Note that being normed immediately implies that $\K$ has no zero divisors, so $\K$ is indeed a division algebra.

Remarkably, there are only four normed division algebras: $\R$, $\C$, $\H$ and $\O$,
constructed above.  In the first step of this construction, going from $\mathbb{R}$
to $\mathbb{C}$, the adjoined generator $\ell$ is identified with the imaginary unit
$i \in \mathbb{C}$.  In the second step the adjoined generator is usually denoted
$j$, leading to the imaginary quaternions subject to the relations
  $$
  i j = k\,, \phantom{AA} j i = - k\;,
  $$
  and their cyclic permutations.
  When working with the octonions, we will exclusively use the Cayley--Dickson presentation,
  and hence in the main text $\ell$ always denotes a unit octonion orthogonal to $i$, $j$ and $k := i j$.
  Notice simple but important relations implied by (\ref{RelationsForCayleyDickson}), such as
  $
    \ell^{-1} = - \ell
   $,
  which lead to manipulations such as
  \begin{equation}
    \label{SomeRelationAmongImaginaryQuaternions}
    \overline{k} \ell = -k\ell = k \ell^{-1}
    \,.
  \end{equation}
\end{example}
\begin{prop}[Basic properties of the quaternions]
  \label{PropertiesQuaternions}
  We collect some well-known facts about quaternions (Example \ref{TheFourRealNormedDivisionAlgebras}):

   \vspace{-2mm}
    \item {\bf (i)}
      The quaternions $\mathbb{H}$ are isomorphic to $\mathbb{R}^4$ as a normed vector space:
      \begin{equation}
        \label{QuaternionsAsNormedVectorSpace}
        \mathbb{H} \simeq \mathbb{R}^4\;.
      \end{equation}

 \vspace{-3mm}
    \item {\bf (ii)} A quaternion $q \in \mathbb{H}$ of unit norm $\vert q\vert = 1$
      is also called a \emph{unit quaternion}, for short.  As a submanifold of
      $\mathbb{H}$, the space of unit quaternions is the 3-sphere
      $$
        S(\mathbb{H}) \simeq S^3
        \,.
      $$
      Quaternion multiplication turns $S(\H)$ into a Lie group. This group is isomorphic to $\mathrm{SU}(2)$:
      \begin{equation}
        \label{UnitQuaternionsAndSU2}
        S(\mathbb{H}) \simeq \mathrm{SU}(2)
        \,.
      \end{equation}

 \vspace{-3mm}
    \item {\bf (iii)}
      Thanks to quaternion multiplication, the group $\SU(2)$ acts on $\mathbb{H}$ in two ways (Def.\ \ref{GSpace}):
      \begin{equation}
        \label{SU2ActionsOnQuaternions}
        \xymatrix@R=-2pt{
          \SU(2) \times \mathbb{H} \ar[rr]^{\rho_L} && \mathbb{H}
          \\
          (q, v) \ar@{|->}[rr] && q v
        }
        \qquad
        \text{and}
        \qquad
        \xymatrix@R=-2pt{
           \SU(2) \times \mathbb{H}  \ar[rr]^{\rho_R} && \mathbb{H}\;.
          \\
          (q, v ) \ar@{|->}[rr] &&  v \overline{q}
        }
      \end{equation}
      These actions commute with each other because $\H$ is associative, and they
      preserve the norm because $\H$ is normed:
      $$
        \vert q v\vert = \underset{=1}{\underbrace{{\vert q\vert}}} \, \vert v\vert ={\vert v\vert }
        \,,
      $$
      with a similar calculation for the right action. Finally, in either case
      $\SU(2)$ acts on $\H$ by orientation-preserving transformations, because
      $\SU(2)$ is connected.  In summary, the two actions $\rho_{L,R}$ of
      $\mathrm{SU}(2)$ factor through the special orthogonal group in 4 dimensions:
      \begin{equation}
        \label{UnitQuaternionsActOrthogonallyOnQuaternions}
       \xymatrix{
        \rho_{L,R}
        \maps
        \mathrm{SU}(2)
        \ar[r] &
        \mathrm{SO}(4)
               }.
      \end{equation}

       \vspace{-3mm}
    \item {\bf (iv)} Because the actions $\rho_L$ and $\rho_R$ commute with each
      other, they define an action of $\SU(2) \times \SU(2)$ on $\H$. Restricting
      this to the diagonal $\SU(2)$ subgroup, we get an action of $\SU(2)$ on $\H$:
      \[ \begin{array}{rcl}
        \SU(2) \times \H & \longrightarrow & \H\;. \\
        (q,v) & \mapsto & qv\bar{q}
      \end{array}
      \]
      This action is trivial on the real quaternions, and preserves the 3-dimensional
      subspace of imaginary quaternions. In fact, $\H$ decomposes into the
      irreducible representations:
      \[
       \H \simeq \R \oplus \Im(\H) .
      \]
      The action of $\SU(2)$ on the summand $\Im(\H)$ preserves the norm, and this
      induces the famous homomorphism
      \begin{equation}
      \label{ImaginaryQuaternionsRotationAction}
      \SU(2) \to \SO(3) ,
    \end{equation}
     a double cover of $\SO(3)$.
\end{prop}


The finite subgroups of $\SU(2)$ are of particular interest in the main text:
\begin{remark}[The finite subgroups of $\mathrm{SU}(2)$ \cite{Klein1884}]
  \label{ADEGroups}
  The finite subgroups of $\mathrm{SU}(2)$
  are given, up to conjugacy, by the following classification (where $n \in \mathbb{N}$):
  \begin{center}
  \begin{tabular}{|c|c|c|}
    \hline
    \begin{tabular}{c}
      \bf Label
    \end{tabular}
    &
    \begin{tabular}{c}
      \bf Finite
      \\
      \bf subgroup
      \\
      \bf of $\mathrm{SU}(2)$
    \end{tabular}
    &
     \begin{tabular}{c}
       \bf Name of
      \\
      \bf group
    \end{tabular}
    \\
    \hline
    \hline
    $\mathbb{A}_{\mathrlap{n}}$ &
    $\phantom{2}\mathbb{Z}_{\mathrlap{n+1}}$ &   Cyclic
    \\
    \hline
    $\mathbb{D}_{\mathrlap{n+4}}$ &
     $2\mathrm{D}_{\mathrlap{n+2}}$ & Binary dihedral
    \\
    \hline
    $\mathbb{E}_{\mathrlap{6}}$ & $2\mathrm{T}$ & Binary tetrahedral
    \\
    \hline
    $\mathbb{E}_{\mathrlap{7}}$ & $2\mathrm{O}$ & Binary octahedral
    \\
    \hline
    $\mathbb{E}_{\mathrlap{8}}$ & $2\mathrm{I}$ & Binary icosahedral
    \\
    \hline
  \end{tabular}
  \end{center}
The full proof for the case of finite subgroups of $\mathrm{SL}(2,\mathbb{C})$ is given in \cite{MBD1916}, recalled in detail in \cite[Sec. 2]{Serrano14}.
Full proof for the case of $\mathrm{SO}(3)$ is also spelled out in \cite[Theorem 11]{Rees05}; from this the proof for the case of $\mathrm{SU}(2)$
is spelled out in \cite[Theorem 4]{Keenan03}.
\end{remark}

\begin{defn}[Hopf fibration]
  \label{HopfFibration}
  Let $\mathbb{K}$ be one of the four normed division algebras (Example \ref{TheFourRealNormedDivisionAlgebras}).
  Then the corresponding \emph{Hopf fibration} is the map between unit spheres given by
  \begin{equation}
    \label{HopfMap}
    \xymatrix@R=-2pt{
      S(\mathbb{K}^2)
      \ar[rrr]^-{H_{\mathbb{K}}}
      &&&
      S\left( \mathbb{R} \oplus \mathbb{K} \right)
      \\
      (x,y)
        \ar@{|->}[rrr]
        &&&
      \left( \vert y\vert^2 - \vert x\vert^2,   2 \, x  \overline{y} \right)\,,
    }
  \end{equation}
  where $S(V)$ denotes the unit sphere inside the normed vector space $V$. The image
  lies in $S(\R \oplus \K)$ because the normed division algebra is normed.
\end{defn}
Hence we have
\begin{center}
\begin{tabular}{|c|ccc|}
  \hline
  \begin{tabular}{c}
    \bf Normed
    \\
    \bf algebra
  \end{tabular}
  &
  \multicolumn{3}{c|}{
    \bf Hopf fibration

    }
  \\
  \hline
  \hline
  $\mathbb{R}$ & $S^1$ &  $\overset{\cdot 2}{\longrightarrow}$ &  $S^1$
  \\
  \hline
  $\mathbb{C}$ & $S^3$ & $\overset{}{\longrightarrow}$ &  $S^2$
  \\
  \hline
  $\mathbb{H}$ & $S^7$ &  $\overset{}{\longrightarrow}$ &  $S^4$
  \\
  \hline
  $\mathbb{O}$ & $S^{15}$ & $\overset{}{\longrightarrow}$ &  $S^8$
  \\
  \hline
\end{tabular}
\end{center}

The key statement for us is the following:
\begin{prop}[Real Spin representations via real normed division algebras
(see {\cite{BaezHuerta10, BaezHuerta11}})]
  \label{RealSpinRepresentationFromDivisionAlgebra}
  Let $\K \in \{\mathbb{R}, \mathbb{C}, \mathbb{H}, \mathbb{O}\}$ be one of the
  normed division algebras (Example \ref{TheFourRealNormedDivisionAlgebras}).
  Write $\h_2(\K)$ for the real vector space of $2 \times 2$ hermitian matrices with
  coefficients in $\K$:
  \[ \h_2(\K) := \left\{ \begin{pmatrix} t + x & \overline{y} \\ y & t -
      x \end{pmatrix} : t,x \in \R, \, y \in \K \right\} . \]
  Let $k$ denote the dimension of $\K$. Then:
  \begin{enumerate}

  \vspace{-2mm}
    \item There is an isomorphism of inner product spaces (``forming Pauli matrices over $\mathbb{K}$'')
    $$
      \left(
        \h_2(\K), -\mathrm{det}
      \right)
        \stackrel{\simeq}{\longrightarrow}
        (\mathbb{R}^{k+1,1}, \eta)
          $$
    identifying $\mathbb{R}^{k+1,1}$ equipped with its Minkowski inner product
    $$
      \eta(A,B) := -A^0 B^0 + A^1 B^1 + \cdots + A^{k+1} B^{k+1}, \mbox{ for } A, B \in \R^{k+1,1}
    $$
    with the space of $2 \times 2$ hermitian matrices equipped with the negative of
    the determinant operation.

\vspace{-2mm}
  \item Let ${\bf N}$ and $\overline{\bf N}$ both denote the vector space $\K^2$. Then ${\bf N} \oplus
    \overline{\bf N}$ is a module of the Clifford algebra $\Cl(k+1,1)$, with the action of a
    vector in $A \in \R^{k+1,1}$ given by
    \[ \Gamma(A) (\psi, \phi) = (\tilde{A}_L \phi, A_L \psi) \]
    for any element $(\psi, \phi) \in {\bf N} \oplus \overline{\bf N}$, where we are
    using the identification of vectors in $\R^{k+1,1}$ with $2 \times 2$ hermitian
    matrices. Here $\widetilde {(-)}$ is the operation $\widetilde{A} = A - \tr(A)
    1$, and $(-)_L$ denotes the linear map given by left multiplication by a matrix.

\vspace{-2mm}
  \item Realizing the Spin group $\Spin(k+1,1)$ inside the Clifford algebra
    $\Cl(k+1,1)$ by the standard construction, this induces irreducible
    representations $\rho$ and $\overline{\rho}$ of $\Spin(k+1,1)$ on ${\bf N}$ and
    $\overline{\bf N}$, respectively. Explicitly, recall that $\Spin(k+1,1)$ is the
    subgroup of the Clifford algebra generated by products of pairs of unit vectors
    of the same sign:
    \[ \Spin(k+1, 1) = \langle AB \in \Cl(k+1,1) \, : \, A, B \in \R^{k+1,1} \mbox{
    and } \eta(A,A) = \eta(B,B) = \pm 1 \rangle . \]
    Then restricting the Clifford action to these elements, a generator $AB$ of
    $\Spin(k+1,1)$ acts as
    \[ \rho(AB) = \tilde{A}_L B_L \mbox{ on } {\bf N} \]
    and as
    \[ \overline{\rho}(AB) = A_L \tilde{B}_L \mbox{ on } \overline{\bf N}, \]
    where again $\widetilde {(-)}$ is the operation $\widetilde{A} = A - \tr(A)
    1$, and $(-)_L$ denotes the linear map given by left multiplication by a matrix.




\vspace{-2mm}
      \item Moving up by one dimension, there is an isomorphism of inner product spaces
  $$
    \left\{
      \left(
        \begin{array}{cc}
          x^0 & \widetilde{A}
          \\
          A & - x^0
        \end{array}
      \right)
      \;:\;
      a \in \mathbb{R}
      \,,
      A \in \h_2( \mathbb{K} )
    \right\}
    \;\simeq\;
    \mathbb{R}^{k+2,1}
  $$
  between the subspace on the right of $4 \times 4$ matrices over $\K$, equipped with
  the inner product given by $-\mathrm{det}(A) + a^2$, and Minkowski spacetime
  $\R^{k+2,1}$.

\vspace{-2mm}
\item Let $\mathscr{N}$ denote the vector space $\K^4$. Then $\mathscr{N}$ is
  a module of the Clifford algebra $\Cl(k+2,1)$ with the action of a vector
  $\mathcal{A} \in \R^{k+2,1}$ given by:
  \[
  \Gamma(\mathcal{A}) \Psi = \mathcal{A}_L \Psi
  \]
  for any element $\Psi \in \mathscr{N}$. Here we are using the identification of
  vectors in $\R^{k+2,1}$ with a subspace of $4 \times 4$ matrices over $\K$, and
  $(-)_L$ denotes the linear operator given by left multiplication by a matrix.

\vspace{-2mm}
\item Realizing the Spin group $\Spin(k+2,1)$ inside the Clifford algebra
  $\Cl(k+2,1)$ by the standard construction, this induces an irreducible representation
  $\rho$ of $\Spin(k+2,1)$ on $\mathscr{N}$.
  Explicitly, recall that $\Spin(k+2,1)$ is the subgroup of the
  Clifford algebra generated by products of pairs of unit vectors of the same sign:
  \[ \Spin(k+2, 1) = \langle \A \B \in \Cl(k+2,1) \, : \, \A, \B \in \R^{k+2,1}
  \mbox{ and } \eta(\A,\A) = \eta(\B,\B) = \pm 1 \rangle . \]
  Then restricting the Clifford action to these elements, a generator $\A\B$ of
  $\Spin(k+2,1)$ acts as
  \[
   \rho(\A\B) = \A_L \B_L \mbox{ on } \mathscr{N}
   \]
  where again $(-)_L$ denotes the linear map given by left multiplication by a
  matrix.

\vspace{-2mm}
\item The representations ${\bf N}$, $\overline{\bf N}$ and $\mathscr{N}$
  constructed above are the irreducible real spinor representations in the following
  table (and as in Example \ref{RelevantExamplesOfRealSpinRepresentations}):
 \begin{center}
\begin{tabular}{|c|c|ccccc|c|}
  \hline
  \multirow{2}{*}{
  \begin{tabular}{c}
   \bf Dimension
   \\
   $D = p + 1$
  \end{tabular}
  }
  &
  \multirow{2}{*}{
  \begin{tabular}{c}
    \bf Real irreps of
    \\
    $\mathrm{Spin}(p,1)$
  \end{tabular}
  }
  &
  \multicolumn{5}{c|}{
    \begin{tabular}{c}
      \bf Clifford modules via
      \\
      \bf real normed division algebra
    \end{tabular}
  }
  \\
  &&
  $\{\widetilde{\sigma}\sigma\}$
  &&
  \scalebox{.6}{
  $
    \left\{
      \left(
        \begin{array}{cc}
          \!\!\!1\! & \!0\!\!\!
          \\
          \!\!\!0\! & \!-1\!\!\!
        \end{array}
      \right)\;, \;
      \left\{
      \left(
        \begin{array}{cc}
          \!\!\!0\! & \!\tilde \sigma\!\!\!
          \\
          \!\!\!\sigma\! & \!0\!\!\!
        \end{array}
      \right)
      \right\}
  \right\}$
  }
  &&
  $\{\sigma \widetilde \sigma\}$
  \\
  \hline
  \hline
  $10+1$ & $\phantom{{A \top A} \atop {A \atop A}}\mathbf{32}\phantom{{A \top A} \atop {A \atop A}}$ &  && $\mathbb{O}^4$ &&
  \\
  \hline
  $9+1$ & $\mathbf{16}$, $\overline{\mathbf{16}}$ & $\mathbb{O}^2$ && $\phantom{{A \top A} \atop {A \atop A}}$ && $\mathbb{O}^2$
  \\
  \hline
  $6+1$ & $\phantom{{A \top A} \atop {A \atop A}}\mathbf{16}\phantom{{A \top A} \atop {A \atop A}}$ &  && $\mathbb{H}^4$ &&
  \\
  \hline
  $5+1$ & $\mathbf{8}$, $\overline{\mathbf{8}}$ & $\mathbb{H}^2$ && $\phantom{{A \top A} \atop {A \atop A}}$  && $\mathbb{H}^2$
  \\
  \hline
  $4+1$ & $\phantom{{A \top A} \atop {A \atop A}}\mathbf{8}\phantom{{A \top A} \atop {A \atop A}}$ &  && $\mathbb{C}^4$ &&
  \\
  \hline
  $3+1$ & $\phantom{{A \top A} \atop {A \atop A}}\mathbf{4}\phantom{{A \top A} \atop {A \atop A}}$ & $\mathbb{C}^2$ &$\!\!\!\!\simeq\!\!\!\!\!\!\!\!\!\!$& $\mathbb{R}^4$ &$\!\!\!\!\!\!\!\!\!\!\simeq\!\!\!\!$& $\mathbb{C}^2$
  \\
  \hline
  $2+1$ & $\phantom{{A \top A} \atop {A \atop A}} \mathbf{2} \phantom{{A \top A} \atop {A \atop A}}$ & $\mathbb{R}^2$ && $\simeq$ && $\mathbb{R}^2$
  \\
  \hline
\end{tabular}
\end{center}
Here the symbol ``$\simeq$'' in the last two lines denotes isomorphism of \emph{real} representations.
  \end{enumerate}
\end{prop}

\begin{example}[The octonionic presentation of $\mathbf{32}$]
\label{OctonionPresentationOf32OfPin11}
We can identify the 32-dimensional vector space $\bf 32$ with the space $\O^4$:
$$
  \mathbf{32} \simeq \O^4.
$$
Under this identification, the Clifford algebra $\mathrm{Cl}(10,1)$ (see \eqref{GeneralCliffordAlgebraSigns})
acts on $\bf 32$ by left multiplication by the following $4 \times 4$ matrices with entries in the octonions, written as $2 \times 2$ matrices with $2 \times 2$ blocks:
\begin{equation}
  \label{11dGammaOct}
  \mathbf{\Gamma}^0 := \begin{pmatrix} 0 & -1 \\ 1 & 0 \end{pmatrix}, \quad
  \mathbf{\Gamma}^1 := \begin{pmatrix} 0 & \tau \\ \tau & 0 \end{pmatrix}, \quad
  \mathbf{\Gamma}^2 := \begin{pmatrix} 0 & \varepsilon \\ \varepsilon & 0 \end{pmatrix}, \quad
  \mathbf{\Gamma}^{i+2} := \begin{pmatrix} 0 & Je_i \\ Je_i & 0 \end{pmatrix}, \quad
  \mathbf{\Gamma}^{10} := \begin{pmatrix} 1 & 0 \\ 0 & -1 \end{pmatrix} .
\end{equation}
Here, besides the imaginary octonions $e_1, \ldots, e_7$, we have used the $2 \times 2$ real matrices:
\[
\tau := \begin{pmatrix} 1 & 0 \\ 0 & -1 \end{pmatrix}, \quad
\varepsilon := \begin{pmatrix} 0 & 1 \\ 1 & 0 \end{pmatrix}, \quad
J := \begin{pmatrix} 0 & -1 \\ 1 & 0 \end{pmatrix} .
\]
Defining the tensor product of matrices $A$ and $B$ to be the matrix
\[ A \otimes B =
\begin{pmatrix}
  a_{11} B & a_{12} B & \cdots & a_{1n} B \\
  a_{21} B & a_{22} B & \cdots & a_{2n} B \\
  \vdots & & \ddots & \vdots  \\
  a_{n1} B & a_{n2} B & \cdots & a_{nn} B \\
\end{pmatrix},
\]
we can rewrite the octonionic gamma matrices (\ref{11dGammaOct}) as follows:
\begin{equation}
  \label{11dGammaAsOctonionicAndTensor}
  \mathbf{\Gamma}^0 = J \otimes 1, \qquad
  \mathbf{\Gamma}^1 = \varepsilon \otimes \tau, \qquad
  \mathbf{\Gamma}^2 = \varepsilon \otimes \varepsilon, \qquad
  \mathbf{\Gamma}^{i+2} = \varepsilon \otimes Je_i, \qquad
  \mathbf{\Gamma}^{10} = \tau \otimes 1.
\end{equation}
\end{example}

\subsection{Homotopy and cohomology}
 \label{HomotopyTheory}

For reference, here we collect some basics of abstract homotopy theory
and of the associated generalized cohomology theories.

\subsubsection*{Homotopy theory}

We briefly recall some basics of homotopy theory, as well as some basic examples of relevance in the main text.
For a self-contained introductory account of abstract homotopy theory see \cite{Schreiber17b}. For minimal
background on language of categories required,  see \cite[around Remark 3.3]{Schreiber17a}, and for a
comprehensive reference see \cite{Borceux94}. For going deep and far into homotopy theory, see \cite{Lurie09}.
For exposition of the foundational role of homotopy theory see \cite{Shulman17}.

\begin{defn}[Category with weak equivalences (e.g. {\cite[Def.\ 2.1]{Schreiber17b}})]
  \label{CategoryWithWeakEquivalences}
  A \emph{category with weak equivalences} is a category $\mathcal{C}$
  equipped with a choice of sub-class $W \subset \mathrm{Mor}(\mathcal{C})$ of its
  morphisms, called the \emph{weak equivalences}, such that
  \begin{enumerate}

  \vspace{-3mm}
    \item $W$ contains all the identity morphisms;

  \vspace{-3mm}
    \item if $f,g \in W$ are composable with composite $g \circ f $, and if two elements in the set $\{f, g,  g \circ f \}$
      are weak equivalences, then also the third is.
  \end{enumerate}
  A category with weak equivalences may also be called a \emph{homotopy theory}.
\end{defn}
\begin{defn}[Homotopy categories (e.g. {\cite[Def.\ 2.30]{Schreiber17b}})]
  \label{CategoryHomotopy}
  Given a category with weak equivalences $(\mathcal{C},W)$ (Def.\ \ref{CategoryWithWeakEquivalences}), then its \emph{homotopy category} is the
  category $\mathrm(Ho)( \mathcal{C}[W^{-1}] )$ equipped with a functor
  \begin{equation}
    \label{LocalizationFunctor}
    \gamma \maps \mathcal{C} \longrightarrow \mathrm{Ho}\left(\mathcal{C}\left[W^{-1}\right]\right)
    \,,
  \end{equation}
  called the \emph{localization functor},
  such that
  \begin{enumerate}

\vspace{-3mm}
    \item $\gamma$ sends weak equivalences to actual isomorphisms;

 \vspace{-3mm}
    \item $\left(\mathrm{Ho}\left(\mathcal{C}\left[W^{-1}\right]\right)), \gamma\right)$ is the universal solution with this property, in that if
    $
      F \maps \mathcal{C} \rightarrow \mathcal{D}
    $
    is any other functor, to any other category, such that it sends the weak equivalences $W$ to actual isomorphisms,
    then $F$ actually factors through $\gamma$, up to natural isomorphism
    $$
      \xymatrix{
        \mathcal{C} \ar[d]_\gamma \ar[rr]^-F_<<<<<<<<<<{\ }="s" && \mathcal{D}
        \\
        \mathrm{Ho}\left( \mathcal{C}\left[W^{-1}\right] \right)
        \ar@{-->}[urr]^{\ }="t"
        \ar@{=>}_{\simeq} "s"; "t"
      }
    $$
    and this factorization is unique up to unique isomorphism.
  \end{enumerate}
\end{defn}

The following are basic examples of homotopy theories.

\begin{defn}[Compactly generated topological spaces]
  \label{TopologicalSpaces}
  By a \emph{topological space} we will always mean a \emph{compactly generated topological space} (e.g. {\cite[Def.\ 3.35]{Schreiber17b}}).
  We write ``$\mathrm{Spaces}$'' for the category whose objects are compactly generated topological spaces, and whose morphism are continuous functions.
  For $X,Y$ two such topological spaces, the space
  \begin{equation}
    \label{MappingSpace}
    \mathrm{Maps}(X,Y) \in \mathrm{Spaces}
  \end{equation}
  of continuous functions between them is itself
  naturally a compactly generated topological space (e.g. {\cite[Def.\ 3.39]{Schreiber17b}}) satisfying the universal properties
  of a mapping space (e.g. {\cite[Def.\ 3.41]{Schreiber17b}}).
\end{defn}
\begin{defn}[Classical homotopy theory (e.g. {\cite[Def.\ 3.11]{Schreiber17b}})]
  \label{ClassicalHomotopyCategories}
    A continuous function $f \maps X \to Y$ between topological spaces (Def.\ \ref{TopologicalSpaces}) is called
    a \emph{weak homotopy equivalence}  if it induce a bijection between connected components
  $$
    \pi_0( f ) \maps \pi_0\left(X_1\right)
   \xymatrix{\ar[r]^{\simeq}&} \pi_0\left(X_2\right)
  $$
  and, for every $n \in \mathbb{N}$, $n \geq 1$ and every base point $x \in X$,
   it induces an isomorphism between the $n$th homotopy groups
 $$
   \pi_n( f,x ) \maps \pi_n\left(X_1,x\right)
    \xymatrix{\ar[r]^{\simeq}&}
    \pi_n\left(X_2,f(x)\right)
    \,.
  $$
The resulting homotopy category (Def.\ \ref{CategoryHomotopy})
$$
  \mathrm{Ho}\left(
    \mathrm{Spaces}
  \right)
  \;:=\;
  \mathrm{Ho}\big(\mathrm{Spaces}\big[ \left\{\mbox{weak homotopy equivalences}\right\}^{-1} \big]\big)
$$
is also called the \emph{classical homotopy category}.
\end{defn}
\begin{example}[Based path spaces]
  \label{BasedPathSpaces}
  For $X$ any topological space (Def.\ \ref{TopologicalSpaces}), equipped with a base point $x \in X$, write
  $$
    P_x X \subset \mathrm{Maps}([0,1], X)
  $$
  for its \emph{based path space}, the subspace of the space of continuous functions
  $
    \gamma \maps [0,1] \longrightarrow X
  $
  from the interval to $X$ which take  $0 \in [0,1]$ to the base point $\gamma(0) = x$.
  There is then the endpoint evaluation map
  $$
    \xymatrix@R=-2pt{
      P_x X
      \ar[r]^{\mathrm{ev}_1}
      &
      X\;.
      \\
      \gamma \ar@{|->}[r] & \gamma(1)
    }
  $$
  Moreover, there is the unique map to the point
  $
    P_x X \longrightarrow \ast
  $.
  This is a weak homotopy equivalence (Def.\ \ref{ClassicalHomotopyCategories}).
  Observe that a continuous function $\widehat f$ into a based path space is equivalently a continuous function into $X$ equipped with a homotopy (Def.\ \ref{EquivariantHomotopy})
  to the function that is constant on the base point:
  $$
    \raisebox{20pt}{
    \xymatrix{
      &&
        P_x x
        \ar[d]^{\mathrm{ev_1}}
      \\
      Y
      \ar[urr]^{\widehat f}
      \ar[rr]_{f}
      &&
      X
    }
    }
    \phantom{AAAA}
    \xymatrix{
      \ar@{<->}[rrr]^{ (\widehat f(y))(t) = \eta(y,t) }
      &&&
    }
    \phantom{AAAA}
    \raisebox{20pt}{
    \xymatrix{
      &&
        \ast
        \ar[d]^{x}
      \\
      Y
      \ar@/^1pc/[urr]_>>>>>{\ }="s"
      \ar[rr]_{f}^{\ }="t"
      &&
      X
      \ar@{=>}^\eta "s"; "t"
    }
    }
  $$
\end{example}

It turns out that the classical homotopy category is an extremely rich structure. In order to get a handle
 on these categories, one may filter them in various ways such as
to study homotopy types in controlled approximations. A key instance of this is the rational approximation.
We recall this as Prop. \ref{SullivanEquivalence} below.

\begin{defn}[Rational homotopy theory (e.g. {\cite{Hess06}})]
  \label{RationalHomotopyTheory}

    \vspace{-3mm}
\item  {\bf (i)}
  A continuous function $f \maps X \to Y$ between topological spaces (Def.\ \ref{TopologicalSpaces}) is called a \emph{rational weak homotopy equivalences}
  if it induces a bijection between connected components
  $$
    \pi_0( f ) \maps \pi_0\left(X_1\right)
   \xymatrix{\ar[r]^{\simeq}&} \pi_0\left(X_2\right)
  $$
  and for every $n \in \mathbb{N}$, $n \geq 1$ and every base point $x \in X^H$ they induce an isomorphism between the rationalized $n$th homotopy groups
 $$
   \pi_n( f,x ) \otimes \mathbb{Q}
     \maps \pi_n\left(X_1,x\right)\otimes \mathbb{Q}
    \xymatrix{\ar[r]^{\simeq}&}
    \pi_n\left(X_2,f(x)\right) \otimes \mathbb{Q}
    \,.
  $$

  \vspace{-2mm}
\item {\bf (ii)} The resulting homotopy category (Def.\ \ref{CategoryHomotopy})
$$
  \mathrm{Ho}\left(
    \mathrm{Spaces}_{\mathbb{R}}
  \right)
  \;:=\;
  \mathrm{Ho}
  \left(
    \mathrm{Spaces}\left[ \left\{ \mbox{rational weak homotopy equivalences} \right\} \right]
  \right)
$$
is also called the \emph{rational homotopy category}.

\vspace{-1mm}
\item {\bf (iii)} We also consider the full subcategory
$$
  \mathrm{Ho}\left(
    \mathrm{Spaces}_{\mathbb{Q},\mathrm{nil}, \mathrm{fin}}
  \right)
  \xymatrix{\ar@{^{(}->}[r]&}
  \mathrm{Ho}\left(
    \mathrm{Spaces}_{\mathbb{R}}
  \right)
$$
on those spaces $X$ which are
\begin{itemize}
\vspace{-2mm}
\item of  \emph{finite rational type}
i.e. $H^1(X,\mathbb{Q})$ and $\pi_{k\geq 2}(X) \otimes \mathbb{Q}$
are finite-dimensional $\mathbb{Q}$-vector spaces for all $k \geq 2$;

\vspace{-3mm}
\item \emph{nilpotent}
   i.e. the fundamental group $\pi_1(X)$ is a nilpotent group and such that its action on the higher
   rational homotopy groups is nilpotent (i.e., making them nilpotent $\pi_1(X)$-modules).
\end{itemize}
\end{defn}
The key point about rational homotopy theory (Def.\ \ref{RationalHomotopyTheory}) is that it may be modeled by dg-algebraic means:
\begin{defn}[Rational DG-algebraic homotopy theory (e.g. {\cite{Hess06}})]
  \label{dgAlgebrasAnddgModules}
  \vspace{-2mm}
\item {\bf (i)}  We write
  $\mathrm{dgcAlg}$ for the category whose objects are differential graded-commutative $\mathbb{R}$-algebras
  and whose morphisms are dg-algebra homomorphisms.
  A morphism $\phi \maps A_1 \to A_2$ is called a \emph{quasi-isomorphism} if it induces isomorphisms on all
  cochain cohomology groups:
  $$
    H^n(\phi) \maps H^n(A_1) \overset{\simeq}{\longrightarrow} H^n( A_2 )
    \,.
  $$

  \vspace{-2mm}
\item {\bf (ii)}  We write the corresponding homotopy category (Def.\ \ref{CategoryHomotopy})  as
  $$
    \mathrm{Ho}\left( \mathrm{dgcAlg}^{\mathrm{op}} \right)
    \;:=\;
    \mathrm{Ho}\big( \mathrm{dgcAlg}^{\mathrm{op}}
    \big[ \left\{ \mbox{quasi-isomorphisms}\right\}^{-1} \big] \big)
    \,.
  $$
 \item {\bf (iii)}  We also consider the full subcategory
  $$
    \mathrm{Ho}\big(\mathrm{dgcAlg}^{\mathrm{op}}_{\mathrm{fin}, \mathrm{cn}} \big)
    \xymatrix{\ar@{^{(}->}[r] &}
    \mathrm{Ho}\big(\mathrm{dgcAlg}^{\mathrm{op}}\big)
  $$
  on those algebras $A$ which are
  \begin{itemize}

  \vspace{-2mm}
    \item \emph{of finite type}
       in that they are equivalent to a DGC-algebra that is degreewise finitely generated;

    \vspace{-3mm}
    \item
     \emph{connected}  in that the unit inclusion $\mathbb{Q} \to A$ induces an isomorphism $\mathbb{Q} \simeq H^0(A)$.
  \end{itemize}

  \vspace{-3mm}
  \item {\bf (iv)}
   Finally we write\footnote{ This just reflects the fact that a map from one disjoint union of connected spaces to another is simply a tuple of
   maps between connected spaces, one from each connected component of the domain to a connected component of the codomain. }
   \begin{equation}
     \label{GrothendieckConstructionOnConnecteddgcAlgebras}
     \underset{S \in \mathrm{Set}}{\int} \mathrm{Ho}\big(\mathrm{dgcAlg}^{\mathrm{op}}_{\mathrm{fin}, \mathrm{cn}} \big)^S
   \end{equation}
   for the category whose objects are pairs consisting of a set $S$ and an $S$-indexed tuple of objects of the homotopy category
   of connected finite-type dgc-algebras, and whose morphism are pairs consisting of a function between these sets and a
   tuple of homomorphisms between the corresponding dgc-algebras.
\end{defn}
The following is the classical statement of rational homotopy theory:
\begin{prop}[DG-model for rational homotopy theory ({\cite{Sullivan77, BousfieldGuggenheim76}}, see {\cite[Thm 2.1.10]{VBM18}})]
  \label{SullivanEquivalence}

  \vspace{-2mm}
\item {\bf (i)} There is an adjunction (\cite[Sec. 3]{Borceux94})
  $$
    \xymatrix{
      \mathrm{Ho}(\mathrm{Spaces})
      \ar@{->}@<+6pt>[rr]^{\mathcal{O}}
      \ar@{<-}@<-6pt>[rr]_{\mathcal{S}}^{\bot}
      &&
      \mathrm{Ho}\left( \mathrm{dgcAlg}^{\mathrm{op}}\right)
    }
  $$
  between the classical homotopy category of topological spaces
  (Def.\ \ref{ClassicalHomotopyCategories}) and the opposite of the
  homotopy category of DGC-algebras (Def.\ \ref{dgAlgebrasAnddgModules}), where
  $\mathcal{O}$ denotes the derived functor of forming the DGC-algebra of polynomial differential
  forms of a topological space.

  \vspace{-1mm}
  \item {\bf (ii)} This adjunction restricts to an equivalence of categories (\cite[Sec. 1]{Borceux94})
  \begin{equation}
    \label{EquivalenceSullivan}
    \xymatrix{
      \mathrm{Ho}(\mathrm{Spaces}_{\mathbb{Q}, \mathrm{cn},\mathrm{nil},\mathrm{fin}})
      \ar@{->}@<+6pt>[rr]^{\mathcal{O}}
      \ar@{<-}@<-6pt>[rr]_{\mathcal{S}}^{\simeq}
      &&
      \mathrm{Ho}\big(\mathrm{dgcAlg}^{\mathrm{op}}_{\mathrm{fin}, \mathrm{cn}}\big)
    }
  \end{equation}
  between the \emph{rational} homotopy category of connected nilpotent spaces of finite type (Def.\ \ref{RationalHomotopyTheory})
  and the homotopy category of connected DGC-algebras of finite type (Def.\ \ref{dgAlgebrasAnddgModules}).
  \item {\bf (iii)} Dropping the connectedness assumption on the left, this extends to an equivalence
  \begin{equation}
    \label{NonConnectedSullivanEquivalence}
    \xymatrix{
      \mathrm{Ho}(\mathrm{Spaces}_{\mathbb{Q}, \mathrm{nil},\mathrm{fin}})
      \ar@{->}@<+6pt>[rr]^{\mathcal{O}}
      \ar@{<-}@<-6pt>[rr]_{\mathcal{S}}^{\simeq}
      &&
      \underset{S \in \mathrm{Set}}{\int}
        \mathrm{Ho}\big(\mathrm{dgcAlg}^{\mathrm{op}}_{\mathrm{fin}, \mathrm{cn}}\big)^S
    }
  \end{equation}
  with the category \eqref{GrothendieckConstructionOnConnecteddgcAlgebras} on the right.
\end{prop}

\begin{example}[Minimal DGC-algebra model for the $n$-spheres]
  \label{SpheredgcAlgebraModel}
  Under the equivalence \ref{SullivanEquivalence} the minimal DGC-algebra models of  the $n$-spheres are, up to isomorphism
  as follows:

 \item {\bf (i)}
  The minimal dgc-algebra model for the 0-sphere consists of two copies of the plain algebra of real numbers:
  \begin{equation}
    \label{CEAlgebraFor0Sphere}
    \mathcal{O}(S^0)
    \;=\;
    \mathrm{CE}(\mathfrak{l}(S^0))
    :=
    \{\mathbb{R}, \mathbb{R}\}
    \,.
  \end{equation}

\item  {\bf (ii)}
  The minimal dgc-algebra model for the odd-dimensional spheres $S^{2n+1}$ are
  \begin{equation}
    \label{CEAlgebraFor3Sphere}
    \mathcal{O}(S^{2n+1})
    \;=\;
    \mathrm{CE}\left(\mathfrak{l}\left(S^{2n+1}\right)\right)
    :=
    \mathbb{R}[ h_{2n+1} ]/
    {\small \left(
        d h_{2n + 1}  = 0
    \right)}\;.
  \end{equation}

  \item {\bf (iii)}
  The minimal dgc-algebra model for the positive even-dimensional spheres $S^{2n+2}$ are
  \begin{equation}
    \label{CEAlgebraFor4Sphere}
    \mathcal{O}(S^{2n+2})
    \;=\;
    \mathrm{CE}\left(\mathfrak{l}\left(S^4\right)\right)
    :=
    \mathbb{R}[ \omega_{2n+2}, \omega_{4n+3} ]/
    {\small \left(
      \begin{aligned}
        d \omega_{2n+2}  & = 0
        \\[-2mm]
        d \omega_{4n+3} & = - \tfrac{1}{2} \omega_{2n+2} \wedge \omega_{2n+2}
      \end{aligned}
    \right)}\;.
  \end{equation}

\item  {\bf (iv)}
  Hence for $k \in \mathbb{N}$, there is a canonical map
  \begin{equation}
    \label{SphereRationalMap}
    \xymatrix@R=-2pt{
      S^{4k + 3}
      \ar[rr]
      &&
      S^{2 k + 2}
      \\
      0 && \omega_{2k+2} \ar@{|->}[ll]
      \\
      h_{4k+3} && \omega_{4k+3} \ar@{|->}[ll]
    }
  \end{equation}
  which represents a non-torsion homotopy class. For $k \in \{0,1,3\}$ this is the (rational image of)
  the complex, quaternionnic, or octonionic Hopf fibration (Def.\ \ref{HopfFibration}), respectively.
\end{example}

\begin{example}[DG-Cocycles as maps in rational homotopy theory]
  \label{dgCocyclesAsMaps}
  Let
  $$
    \mathrm{CE}( b^n \mathbb{R} )
    \;:=\;
    \mathbb{R}[ \underset{ \!\!\!\!\!\!\!\!\mathrm{deg} = n+ 1\!\!\!\!\!\!\!\! }{\underbrace{c}} ]
    \;\;
    \in
    \mathrm{dgcAlg}
  $$
  be the dgc-algebra (Def.\ \ref{dgAlgebrasAnddgModules}) whose underlying graded-commutative algebra
  is freely generated from a single generator in degree $n+1$, and whose differential vanishes.
  Under the Sullivan equivalence (Prop. \ref{SullivanEquivalence}) these are minimal models of the
  Eilenberg-MacLane spaces (Example \ref{ExamplesOfCohomologyTheories})
  $$
    B^{n+1} \mathbb{R} = K(\mathbb{R},n+1) \in \mathrm{Spaces}
  $$
  in that
  $
    \mathcal{O}( B^{n+1} \mathbb{R} ) \simeq \mathrm{CE}( b^n \mathbb{R} )
    $.
  Then for $A \in \mathrm{dgcAlg}$ any dgc-algebra, a dg-algebra homomorphism of the form
  $$
    \xymatrix{
      A \ar@{<-}[r]^-{ \;\;\mu^\ast } & \mathrm{CE}(b^n \mathbb{R})
    }
    ,
  $$
  which, under the Sullivan equivalence (Prop. \ref{SullivanEquivalence}), is a model for a map of spaces
  $$
    \mathcal{S}(A) \longrightarrow B^{n+1} \mathbb{R}
    \,,
  $$
  is equivalently an element $\mu \in A$ of degree $n+1$, which is closed $d \mu = 0 \; \in A$. Hence this is
  a \emph{cocycle} in the cochain cohomology of the cochain complex underlying $A$.

  Now under the Sullivan equivalence (Prop. \ref{SullivanEquivalence}), the dgc-algebra on the right is a model for the
  Eilenberg-MacLane space $K(\mathbb{R}, n+1)$
  $$
    \mathrm{CE}(b^n \mathbb{R}) \simeq \mathcal{O}( K(\mathbb{R}, n+1) )
  $$
  and hence the dg-cocycle $\mu$ is realized equivalently as map of spaces of the form
  \begin{equation}
    \label{MapFordgcocycle}
    \xymatrix{
      \mathcal{S}(A)
      \ar[rr]^-{ \mu := \mathcal{S}(\mu^\ast) }
      &&
      K(\mathbb{R}, n+1)
    }
    \,.
  \end{equation}
\end{example}

\begin{example}[DG-coboundaries as homotopies in rational homotopy theory]
  \label{dgCoboiundariesHomotopies}
  Let
  $$
    \xymatrix{
      A_1 \ar@{<-}[rr]^{ \mu_0^\ast,\, \mu_1^\ast } &&  \mathrm{CE}( b^n \mathbb{R} )
    }
  $$
  be two dg-algebra homomorphisms as in Example \ref{dgCocyclesAsMaps}, hence equivalently two
  dg-cocycles of degree $n+1$ in the given dgc-algebra $A$.
  Then a \emph{dg-homotopy} between these homomorphisms is a dg-algebra homomorphism of the form
  $$
    \xymatrix{
      A \otimes \Omega_{\mathrm{poly}}^\bullet([0,1])
      \ar@{<-}[rr]^-{ \eta^\ast }
      &&
      \mathrm{CE}( b^n \mathbb{R} )
    }
  $$
  to the tensor product algebra of $A$ with the
  de Rham algebra $\Omega_{\mathrm{poly}}^\bullet([0,1])$ of polynomial differential forms on the unit interval,
   such that its restriction to the endpoints of the interval reproduces the given homomorphisms, respectively.
  Explicitly, if we write $t \in \Omega^0_{\mathrm{poly}}([0,1])$ for the canonical coordinate function, this means
  equivalently that $\eta^\ast$  corresponds to an element
  $$
    \eta = \alpha + d t \wedge \beta
    \,,\;\;\;
    \in A \otimes \Omega^\bullet_{\mathrm{poly}}([0,1])
    \;\;\;\;
    \alpha, \beta \in A \otimes \mathbb{R}[t]
  $$
  of degree $n+1$, such that $d \eta = 0 \,\in A \otimes \Omega^\bullet_{\mathrm{poly}}([0,1])$, hence such that
  $$
    d (\alpha(t)) = 0 \;\; \in A
    \;,
    \phantom{AAA}
    d (\beta(t)) =  \frac{\partial}{\partial t} \alpha(t)\;,
  $$
  and satisfying
  $
    \alpha(0) = \mu_0
   $
   and $
    \alpha(1) = \mu_1
    $.
  For example, if $\omega \in A$ is a coboundary between the two cocycles, in the sense of the cochain cohomology of $A$
  \begin{equation}
    \label{dgCoboundary}
    d \omega = \mu_1 - \mu_0 \,\;\;\in A\;,
  \end{equation}
  then we get such an $\eta$ by setting
  $$
    \eta := (1-t) \mu_0 + t \mu_1 + dt \wedge \omega
    \,.
  $$
  Therefore, under the Sullivan equivalence (Prop. \ref{SullivanEquivalence})
  a coboundary \eqref{dgCoboundary} between dg-cocycles corresponds to a homotopy
  (Def.\ \ref{EquivariantHomotopy}) between the corresponding maps of spaces \eqref{MapFordgcocycle}:
  $$
    \xymatrix{
      \mathcal{S}(A)
      \ar@/^1.3pc/[rrr]^-{ \mu_0 }_-{\ }="s"
      \ar@/_1.3pc/[rrr]_-{ \mu_1 }^-{\ }="t"
      &&&
      K(\mathbb{R},n+1)\;.
      \ar@{=>}^{ \mathcal{S}(\eta^\ast) } "s"; "t"
    }
  $$
\end{example}

\subsubsection*{Cohomology}
\label{Coh}

Our main interest in homotopy theories (Def.\ \ref{CategoryWithWeakEquivalences}) here is that
each flavor of homotopy theory induces a corresponding \emph{generalized cohomology theory}
(Def.\ \ref{CohomoloyFromHomotopy} below). This includes \emph{Eilenberg-Steenrod-type
generalized cohomology theories} (Example \ref{ExamplesOfCohomologyTheories} below),
which are often just called ``generalized cohomology theories'', for short, but is in fact much more
general than that: all kinds of \emph{differential} and/or \emph{twisted} and/or \emph{non-abelian}
and/or \emph{equivariant} and/or \emph{orbifolded} and/or ... concepts of cohomology theories
arise via the simple Definition \ref{CohomoloyFromHomotopy} from a suitably chosen ambient
homotopy theory (see \cite{GrS} for recent developments).

\medskip
In the main text we are interested in this general concept of generalized cohomology in order to set up and study the
cohomology theory \emph{equivariant rational cohomotopy of superspaces} (Sec. \ref{ADEEquivariantRationalCohomotopy}).

\begin{defn}[Generalized cohomology theories from homotopy theory]
  \label{CohomoloyFromHomotopy}
Every homotopy theory induces a corresponding generalized cohomology theory: given a category with weak equivalences $(\mathcal{C}, W)$ (Def.\ \ref{CategoryWithWeakEquivalences})
and any object $A \in \mathcal{C}$ then

\begin{itemize}
\vspace{-3mm}
  \item a morphism $c \maps X  \to A$ in $\mathcal{C}$ is an \emph{$A$-valued cocycle on $X$};

  \vspace{-3mm}
  \item the equivalence relation on such morphisms induced by the localization functor (\ref{LocalizationFunctor}) is the \emph{coboundary relation};

  \vspace{-3mm}
  \item the image of $[c] := \gamma(c)$ in the morphisms of the homotopy category $\mathrm{Ho}\left( \mathcal{C}[W]^{-1} \right)$ (Def.\ \ref{CategoryHomotopy}) is the \emph{cohomology class} of the cocycle.
\end{itemize}
Hence the set of $A$-valued cohomology classes on $X$ is
$$
  H(X,A) := \mathrm{Hom}_{\mathrm{Ho}\left( \mathcal{C}[W]^{-1}\right)}\left( X,A\right)
  \,.
$$
\end{defn}

\begin{example}[Examples of generalized cohomology theories]
  \label{ExamplesOfCohomologyTheories}
  Examples of generalized cohomology theories arising from homotopy theories via Def.\ \ref{CohomoloyFromHomotopy} include the following:

  \begin{itemize}
  \vspace{-2mm}
    \item For $(\mathcal{C}, W)$ the category of spectra with \emph{stable weak homotopy equivalences}
    (see e.g. \cite[Def.\ I.4.1]{Schreiber17c}), the
    corresponding cohomology theories are equivalently the {\bf abelian generalized cohomology theories} in the sense of the Eilenberg-Steenrod axioms.
    This is the statement of the \emph{Brown representability theorem} (see e.g. \cite[Sec. 1]{Schreiber17d}).
    For instance
    \begin{itemize}
      \item if $A = \Sigma^n H \mathbb{Z}  \in \mathrm{Spectra}$ is an Eilenberg-MacLane spectrum (e.g. \cite[Def.\ II.6.3]{Schreiber17c}), then this is {\bf ordinary cohomology};
      \item
        if
        \begin{equation}
          \label{KTheorySpectrum}
          A := \mathrm{KU} := (\mathrm{KU}_k)_{k \in \mathbb{Z}} :=
          \left\{
            \begin{array}{ccc}
              B U \times \mathbb{Z} &\vert& k \; \mbox{even}
              \\
              U &\vert& k \; \mbox{odd}
            \end{array}
          \right.
        \end{equation}
        this is {\bf K-theory} (also called \emph{complex topological K-theory} for emphasis, to distinguish from a wealth of variants, such as \eqref{RealKTheory} below)
        which measure D-brane charge in type II string theory \cite{Witten98, FreedWitten99, MooreWitten00, EvslinSati06, Evslin06}.
    \end{itemize}

   \vspace{-3mm}
    \item For $(\mathcal{C},W)$ the category of spaces with $W$ the class of \emph{weak homotopy equivalences} (Def.\ \ref{ClassicalHomotopyCategories}),
    the corresponding cohomology theories are called {\bf non-abelian cohomology}. For instance
    \begin{itemize}
      \item if $A = B G \in \mathrm{Spaces}$ is the
    classifying space of a topological group $G$, then the corresponding cohomology theory is {\bf nonabelian $G$-cohomology} in degree 1,
    classifying $G$-principal bundles (in physics: $G$-instanton sectors);
          \item if $A = S^n \in \mathrm{Spaces}$ is an $n$-sphere, then the corresponding non-abelian cohomology theory is called {\bf cohomotopy} \cite{Spanier49}.
    \end{itemize}

  \vspace{-3mm}
    \item For $(\mathcal{C},W)$ the opposite category of dgc-algebras  with $W$ the class of
     quasi-isomorphisms (Def.\ \ref{dgAlgebrasAnddgModules}), we have that the corresponding cohomology theory is
     simply {\bf cochain cohomology} of the underlying cochain complexes (see Example \ref{dgCoboiundariesHomotopies} and Example \ref{dgCocyclesAsMaps})

    \vspace{-3mm}
    \item For $(\mathcal{C},W)$ the $G$-equivariant homotopy category (Def.\ \ref{HomotopyCategoriesOfGSpaces})
    or the category of $G$-fixed point systems (Def.\ \ref{HomotopyCategoryOfGFixedPointSystems}),
    whose homotopy categories are equivalent by Prop. \ref{ElmendorfTheorem},
    the corresponding cohomology theory is called {\bf Bredon equivariant cohomology},
    after \cite{Bredon67}. For instance:
    \begin{itemize}
      \item if $A \in \mathrm{Spaces}$ represents some cohomology theory, then
      that space equipped with a $\mathbb{Z}_2$-action (Example \ref{Z2ActionsAreInvolutions})
            \vspace{-3mm}
      \begin{equation}
        \label{RealCoefficient}
        \xymatrix{ A \ar@(ul,ur)[]^{ \mathbb{Z}_2 } } \in \mathbb{Z}_2 \mathrm{Spaces}
      \end{equation}
      represents a corresponding {\bf real cohomology theory} on \emph{real spaces} (Example \ref{RealSpace}).
      A prominent example is real K-theory \eqref{RealKTheory}.
    \end{itemize}

  \vspace{-3mm}
    \item For $(\mathcal{C},W)$ the $G$-equivariant \emph{stable} homotopy category (of spectra with $G$-actions),
     the complex K-theory spectrum \eqref{KTheorySpectrum} equipped with $\mathbb{Z}_2$-action
     \vspace{-5mm}
     \begin{equation}
       \label{RealKTheory}
       \mathrm{MR} := \xymatrix{ \mathrm{KU} \ar@(ul,ur)[]^{ \mathbb{Z}_2 } }
       ,
     \end{equation}
     where $e \neq \sigma \in \mathbb{Z}_2$ acts by complex
     conjugation, represents the real cohomology theory \eqref{RealCoefficient} called {\bf real K-theory} \cite{Atiyah66, HuKriz01}, which measures D-brane charge in
     type I string theory, hence in type II string theory in the presence of O-planes, hence on orientifolds \cite[Sec. 5.2]{Witten98}, \cite{Gukov99, Hori99, DFM09, DMR13}. See \cite{GrS2} for recent advances in differential KO-theory.
  \end{itemize}
\end{example}

\begin{remark}[Extra structure on cohomology]
  \label{ExtraStructureOnCohomology}
  For a given coefficient object $A \in \mathcal{C}$ in Def.\ \ref{CohomoloyFromHomotopy}, the induced generalized
  cohomology $H(-,A)$ a priori is only a set. This set inherits an extra algebraic structure to the extent that
  $A \in \mathrm{Ho}\left( \mathcal{C}[W]^{-1}  \right)$ is equipped with such extra structure. For instance,
   if $A$ carries the structure of an (abelian) group in the homotopy category, then
  $H(-,A)$ takes values in (abelian) cohomology \emph{groups}. This is often considered by default.
We highlight that this fails for key   examples of cohomology theories, such as notably for cohomotopy theory
  (Example \ref{ExamplesOfCohomologyTheories}),
  (except in those special degrees where the sphere coefficients happen to admit group structure).
  However, the minimum structure one will usually want to retain is that $A$ is equipped with a
   \emph{point}, namely with a morphism
  $$
  \xymatrix{
    \ast \ar[r]^{\mathrm{pt}_A} & A}
  $$
  in the homotopy category, from the terminal object $\ast$, making it a ``pointed object''.
  In this case  the cohomology sets $H(X,A)$ are also canonically pointed sets,
  namely by the unique cocycle $X \to \ast \xrightarrow{\mathrm{pt}_A} A$ that factors through $\mathrm{pt}_A$.
 This is then called the \emph{trivial cocycle}, while all other cocycles are \emph{non-trivial}.
\end{remark}

\medskip
\medskip

\noindent {\bf Acknowledgements.}
We thank
Vincent Braunack-Mayer,
David Corfield,
Mike Duff,
David Roberts,
and
Christian Saemann for useful comments.
We thank the anonymous referee for careful reading
and useful suggestions.

\medskip


\end{document}